\documentclass[journal,table]{IEEEtran}
\usepackage[latin1]{inputenc}
\usepackage[noadjust]{cite}
\usepackage{amsmath}
\usepackage{amsfonts}
\usepackage{amssymb}
\usepackage{mathrsfs}
\usepackage{graphicx}
\usepackage{float}
\usepackage{subcaption}
\usepackage{booktabs}
\usepackage[english]{babel}
\usepackage{amsthm}
\usepackage{comment}
\usepackage{mathtools}
\usepackage{bbm}
\usepackage{rotating}
\usepackage{todonotes}
\usepackage{framed}
\usepackage{tikz}
\usetikzlibrary{patterns}

\pgfdeclarepatternformonly{south west lines}{\pgfqpoint{-0pt}{-0pt}}{\pgfqpoint{5pt}{5pt}}{\pgfqpoint{5pt}{5pt}}{
        \pgfsetlinewidth{0.4pt}
        \pgfpathmoveto{\pgfqpoint{0pt}{0pt}}
        \pgfpathlineto{\pgfqpoint{5pt}{5pt}}
        \pgfpathmoveto{\pgfqpoint{4.8pt}{-.2pt}}
        \pgfpathlineto{\pgfqpoint{5.2pt}{.2pt}}
        \pgfpathmoveto{\pgfqpoint{-.2pt}{4.8pt}}
        \pgfpathlineto{\pgfqpoint{.2pt}{5.2pt}}
        \pgfusepath{stroke}}

\pgfdeclarepatternformonly{south east lines}{\pgfqpoint{-0pt}{-0pt}}{\pgfqpoint{3pt}{3pt}}{\pgfqpoint{3pt}{3pt}}{
    \pgfsetlinewidth{0.4pt}
    \pgfpathmoveto{\pgfqpoint{0pt}{3pt}}
    \pgfpathlineto{\pgfqpoint{3pt}{0pt}}
    \pgfpathmoveto{\pgfqpoint{.2pt}{-.2pt}}
    \pgfpathlineto{\pgfqpoint{-.2pt}{.2pt}}
    \pgfpathmoveto{\pgfqpoint{3.2pt}{2.8pt}}
    \pgfpathlineto{\pgfqpoint{2.8pt}{3.2pt}}
    \pgfusepath{stroke}}

\allowdisplaybreaks

\newtheorem{theorem}{Theorem}
\newtheorem{corollary}[theorem]{Corollary}
\newtheorem{lemma}[theorem]{Lemma}
\newtheorem{axiom}{Axiom}

\newtheorem{definition}{Definition}

\newtheoremstyle{example}
  {3pt plus 1pt minus 2pt} 
  {3pt plus 1pt minus 2pt} 
  {} 
  {} 
  {\bfseries} 
  {.} 
  {.3em} 
  {\thmname{#1}\thmnumber{ #2}\thmnote{ (#3)}} 
\theoremstyle{example}
\newtheorem{example}{Example}

\newcommand{\w}{{\boldsymbol w}}

\newcommand{\y}{{\bf y}}

\newcommand{\btheta}{\boldsymbol\theta}
\newcommand{\boldeta}{\boldsymbol\eta}

\newcommand{\bSigma}{\boldsymbol\Sigma}
\newcommand{\bmu}{\boldsymbol\mu}

\newcommand{\mA}{\mathcal{A}}
\newcommand{\mB}{\mathcal{B}}
\newcommand{\mC}{\mathcal{C}}
\newcommand{\mF}{\mathcal{F}}
\newcommand{\warp}{\chi}

\DeclareMathOperator*{\argmin}{arg\,min}

\DeclarePairedDelimiterX{\infdivx}[2]{(}{)}{%
  #1\;\delimsize\|\;#2%
}

\title{Fusion of Probability Density Functions
\vspace{1mm}}
\author{
\IEEEauthorblockN{G\"unther Koliander\IEEEauthorrefmark{1}, Yousef El-Laham\IEEEauthorrefmark{2}, Petar M. Djuri\'c\IEEEauthorrefmark{2}, and Franz Hlawatsch\IEEEauthorrefmark{3}}\\[3mm]
    \IEEEauthorblockA{\normalsize\IEEEauthorrefmark{1}Acoustics Research Institute, Austrian Academy of Sciences
    (gkoliander@kfs.oeaw.ac.at)}\\[1mm]
    \IEEEauthorblockA{\IEEEauthorrefmark{2}Department of Electrical and Computer Engineering, Stony Brook University
    (\{petar.djuric, yousef.ellaham\}@stonybrook.edu})\\[1mm]
    \IEEEauthorblockA{\IEEEauthorrefmark{3}Institute of Telecommunications, TU Wien
    (franz.hlawatsch@tuwien.ac.at)}
    \thanks{This work was supported in part by the Vienna Science and Technology Fund (WWTF) under grant  MA16-053, by the Austrian Science Fund (FWF) under grants P32055-N31  and Y1199, and
    by the Growing Convergence Research Program of the National Science Foundation under grant
    2021002.  
    }
    }
\date{\today}

\begin{document}
\maketitle

\begin{abstract}
Fusing probabilistic information is a fundamental task in signal and data processing with relevance to many fields of technology and science. In this work, we investigate the fusion of multiple probability density functions (pdfs) of a continuous random variable or vector. Although the case of continuous random variables and the problem of pdf fusion frequently arise in multisensor signal processing, statistical inference, and machine learning, a 
universally accepted method for pdf fusion does not exist. The diversity of approaches, perspectives, and solutions related to pdf fusion motivates a unified presentation of the theory and methodology of the field. We discuss three different approaches to fusing pdfs. In the axiomatic approach, the fusion rule is defined indirectly by a set of properties (axioms). In the optimization approach, it is the result of minimizing an objective function that involves an information-theoretic divergence or a distance measure. In the supra-Bayesian approach, the fusion center interprets the pdfs to be fused as random observations. Our work is partly a survey, reviewing in a structured and coherent fashion many of the concepts and methods that have been developed in the literature. In addition, we present new results  
for each of the three approaches. Our original contributions include new fusion rules, axioms, and axiomatic and optimization-based characterizations; a new formulation of supra-Bayesian fusion in terms of finite-dimensional parametrizations; and a study of supra-Bayesian fusion of posterior pdfs for linear Gaussian models.

\end{abstract}


\renewcommand{\textfraction }{0.01}

\section{Introduction}
\label{sec: introduction}


The fusion of multiple probabilistic descriptions of a random quantity is a fundamental task with 
applications in many fields 
including multisensor signal processing 
\cite{
bandyopadhyay2018distributed, 
Clark10, 
hu2011diffusion, 
maddern2016real, 
da2021recent, 
fantacci2018robust, 
meyer2015distributed, 
Uney13}, 
machine learning \cite{lakshminarayanan2017simple,lee2009neural,lu2020ensemble,thorgeirsson2021probabilistic},
robotics \cite{maddern2016real}, 
smart environments \cite{alam2017data}, 
medicine \cite{kats2019soft},
transportation \cite{kolosz2013modelling},
precision agriculture \cite{chlingaryan2018machine}, 
pharmacology \cite{li2013combined}, 
weather forecasting \cite{murphy1984probability,marty2015combining},  
economics \cite{mitchell2005evaluating,moral2015model},
and financial engineering \cite{barak2017fusion}. 
While this task has been studied for several decades, an in-depth treatment with a focus on 
continuous random variables and, accordingly, on the fusion of probability density functions (pdfs) appears to be lacking.
The present paper attempts to fill this gap. 
Our focus on continuous random variables is motivated by the fact that continuously distributed quantities are the primary object of interest in many applications.

The fusion of pdfs can be considered in different contexts, and several different techniques for this task have been proposed in the literature.
Our treatment is partly a survey of existing concepts and techniques, with an emphasis on a structured and coherent presentation. 
In addition, we present numerous original contributions related to axiomatic, optimization-based, and Bayesian approaches to pdf fusion.


\subsection{Motivation}
\label{sec:motivation}

The field of pdf fusion is multifaceted and somewhat fuzzy: there are many possible approaches to the problem of finding a pdf fusion rule, and there is no universally accepted measure of performance \cite{genest1986combining,hall2007combining}. 
An appropriate fusion rule and performance measure 
depend on the scenario and application. This situation can be 
aggravated by the fact that different fusion rules
can lead to very different results.

Although in specific applications certain pdf fusion rules have been established and found to be useful,  
the rationales of these rules and their possible alternatives
are not always obvious. Thus, it is both theoretically interesting and practically relevant to study the problem of pdf fusion and the existing viewpoints and solutions in a general way that abstracts from specific applications, and to put these viewpoints and solutions into a higher-level perspective. 
Our hope is that this 
analysis will 
support an informed choice of a pdf fusion rule for specific scenarios and applications.
Accordingly, rather than considering
a single framework or method for pdf fusion, this paper 
reviews the 
different approaches 
that have been developed over several decades in different disciplines and by different communities. In addition, these approaches are categorized into three fundamental approaches to principled pdf fusion, which we
term the axiomatic, optimization, and supra-Bayesian approaches.

Fusing pdfs is a special variant of
the general task of ``data fusion'' or ``information fusion,'' and one may ask why it can be advantageous to perform data/information fusion at the level of pdfs. 
Possible answers 
include the following \cite{mitchell2005evaluating,hall2007combining,gneiting2008probabilistic}:

\begin{itemize}

\vspace{1.5mm}

\item
A pdf constitutes a {\em complete} probabilistic description of a continuous random variable or random vector. In addition to its mean or its mode (which can be used as point estimates 
of the random variable or vector), this description includes further important information such as effective support, multimodality, tail decay, and a detailed characterization of the ``dispersion'' around the mean. 
Moreover, it enables
the calculation of quantitative measures of the accuracy of point estimates.

\vspace{1.5mm}

\item
A pdf provides a standardized and ``genesis-agnostic'' representation of the state of information of an agent or sensor, i.e., it abstracts from the intricacies of the 
processing employed by the agent or sensor to obtain it from the raw data. This ``no questions asked'' characteristic
enables or facilitates an information fusion even between heterogeneous agents, which employ different sensing modalities
and/or different types of data preprocessing. Furthermore, the lack of a transparent relation to the raw data is a desirable feature in 
privacy-sensitive applications.

\vspace{1.5mm}

\item
Because a pdf provides a standardized, genesis-agnostic representation, pdf fusion is 
well suited to 
a decentralized (peer-to-peer) network topology. In decentralized, possibly ad-hoc networks, a distributed in-network type of processing is used where each agent communicates with a limited set of neighboring agents and, typically, little or no information about the characteristics of far-away agents
is available locally. The pdf format here facilitates the dissemination of information through the network.

\vspace{1.5mm}

\item
Computationally efficient pdf fusion algorithms based on parameteric pdf representations are available.
For example, the fusion of Gaussian pdfs reduces to fusing the corresponding means and variances or covariance matrices. 
More generally, there are efficient algorithms for fusing Gaussian mixture pdfs. In distributed implementations, parametric pdf representations enable pdf fusion with low or moderate communication cost. Thus, pdf fusion is attractive because detailed probabilistic information can be fused with moderate complexity in terms of computation and communication.

\end{itemize}

\subsection{Probabilistic Opinion Pooling}
\label{sec:POP}

Consider $K$ ``agents,'' ``experts,'' or ``models,'' each providing an ``opinion'' about  an unknown random object that
may be a scalar or  vector. In the probabilistic setting studied in this work, the opinions provided by the agents
are not point estimates of the 
random object but 
probability distributions. 
More specifically, we focus on the case of a continuous random variable or vector $\btheta$, where the opinion of agent $k$ is
expressed by a pdf
$q_k(\btheta)$. 

The problem studied in this paper is to combine, or fuse, the 
pdfs of the $K$ agents, $q_k(\btheta)$ for $k = 1,2,\ldots,K$, into an aggregate pdf $q(\btheta)$.
This problem is traditionally referred to as \emph{probabilistic opinion pooling},
although that term is also used for the fusion of discrete (categorical) distributions. 
We assume that the combination of the agent pdfs $q_k(\btheta)$ 
is done by a central agent or unit, termed a ``fusion center,'' which has access to all the agent pdfs.
The function employed by the fusion center to map the $q_k(\btheta)$ into the aggregate (fused) pdf $q(\btheta)$ is termed a \emph{fusion rule} or a \emph{pooling function}. 
Many different pooling functions have been proposed in the literature, based on various models and considerations.
Important examples include the linear pooling function (a weighted arithmetic mean, also known as arithmetic mean density) 
\cite{stone1961opinion, Bailey12} and the log-linear pooling function (a weighted geometric mean, 
also referred to as Chernoff fusion or geometric mean density) \cite{genest1984characterization, hurley2002information, Lehrer19, Bailey12, Gunay16}.
For Gaussian pdfs, the covariance intersection technique \cite{hurley2002information, julier1997non} is an instance of a log-linear pooling function.
These and several other pooling functions will be discussed in later sections.

An alternative to the centralized setting for probabilistic opinion pooling described above would be a decentralized network of agents without a dedicated fusion center \cite{Chang10, Lehrer19, Uney13,urteaga2016sequential}. 
Here, the agents communicate their pdfs only locally, i.e., to neighboring agents, and each agent can be considered to act as a local fusion center. In this ``in-network'' or ``network-centric'' type of probabilistic opinion pooling, the agents use a distributed communication-and-fusion protocol, such as flooding, consensus, gossip, or diffusion, to disseminate their local pdfs through the network and emulate a given overall pooling function. 
This relies on a suitable pdf representation such as a Gaussian, Gaussian mixture, or particle representation.
The fusion methods we discuss in this work are also relevant to decentralized probabilistic opinion pooling. 
We note, however, that there are numerous methods for in-network signal and information processing in which the local processing results that are being combined are not pdfs. 
For example, some methods combine local likelihood functions \cite{bandyopadhyay2018distributed,hlinka2014consensus,savic2014belief} or messages within a message passing algorithm such as belief propagation \cite{wymeersch2009cooperative,meyer2015distributed}, or certain iterated quantities within a networkwide adaptation-diffusion procedure \cite{lopes2008diffusion}, to name a few.


\subsection{Relevance and Applications}
\label{sec:relevance}
Probabilistic opinion pooling is a fundamental and elementary functionality with widespread applications. 
Historically, the first motivation was to combine expert opinions into an aggregate opinion \cite{stone1961opinion}. 
Nowadays it is more likely that the different probability distributions
do not represent the opinions of multiple experts but originate from the use of multiple sensors, models, or data sets.
In particular, probabilistic opinion pooling is often formulated 
in a Bayesian setting as the fusion of local \emph{posterior} pdfs that are produced by multiple agents using
local implementations of Bayesian inference \cite{Punska99}. 
The ideal aggregate pdf here is 
the global posterior pdf, which takes into account all the data available to the agents.
However, the calculation of the global posterior pdf generally requires additional knowledge besides the local posterior pdfs, such as the local likelihood functions, 
the prior pdfs used by the agents, and possible statistical dependencies between
the agents. 
By contrast, probabilistic opinion pooling requires only the local posterior pdfs. 
In many settings, it is easily and widely applicable because it does not make any assumptions about the local inference methods,
the types of the sensors, or the nature of the local data, which can all be different at different agents.

From the viewpoint of the processed data, there is a wide range of scenarios for
probabilistic opinion pooling. Two extreme cases are particularly important: all the agents process different data, or they 
process exactly the same data. Furthermore, the processing may be carried out with completely unrelated models but with the same objective (e.g., predicting future observations or classifying observations).


Current applications of
probabilistic opinion pooling include, but are not limited to, the following selection:


\begin{itemize}

\vspace{1.5mm}

\item 
In multisensor signal processing applications of probabilistic opinion pooling, multiple sensors derive local pdfs 
based on local observations and either submit these pdfs (or finite-dimensional representations thereof)
to a fusion center or fuse them in a distributed, peer-to-peer manner
\cite{Bailey12, 
bandyopadhyay2018distributed, 
battistelli2014kullback, Chang10, 
deng2012sequential, 
Gunay16, hu2011diffusion, julier1997non, 
Lehrer19, 
Punska99, tang2018information}.
In particular, probabilistic opinion pooling 
plays an important role in multisensor target tracking 
\cite{
Clark10, 
da2021recent, 
fantacci2018robust,
hu2011diffusion, 
Lehrer19, li-battistelli2021distributed, 
li2020arithmetic, 
li-hlawatsch2021distributed, 
Uney13, 
yi2020distributed}.
For tracking an unknown number of targets, probabilistic opinion pooling has recently also been applied to the ``multiobject'' pdfs or
to the probability hypothesis densities (i.e., the ``densities'' of the first moment measures)
of finite point processes, also known as random finite sets 
\cite{Clark10, 
da2021recent, 
fantacci2018robust, gao2020fusion, 
gostar2020cooperative,
li-battistelli2021distributed, li-battistelli2019computationally, 
li2020arithmetic, 
li-hlawatsch2021distributed, mahler2000optimal,
Uney13, 
yi2020distributed}.
Although in this work we do not consider finite point processes, much of our discussion is also relevant in that domain. 
The application of probabilistic opinion pooling to multisensor target tracking will be discussed in more detail in Section \ref{sec: examples_TT}.

\vspace{1.5mm}

\item 
In probabilistic machine learning, 
several scenarios suggest
the combination of probability distributions.
For example,
the concept of ensemble learning \cite{lakshminarayanan2017simple}
is based on applying 
multiple learning algorithms whose outputs
are combined to obtain an aggregate
result that is more accurate than that of any of the individual learning algorithms in the ensemble.
Furthermore,
in federated learning \cite{li2020federated,yurochkin2019bayesian,thorgeirsson2021probabilistic}, multiple  edge devices  learn statistical models individually from their  local data sets without explicitly exchanging these data sets, and
a fusion center aggregates the learned models without having access to the original data. 
This is attractive for privacy-sensitive applications, since no private data have to be shared.
More details on probabilistic machine learning are provided in Section \ref{sec: examples_GenMod}.

\vspace{1.5mm}

\item 
The main goal in the combination of forecasts \cite{clemen1989combining,armstrong2001combining}
is the estimation of a parameter by combining several different models.
To this end, certain methods perform a fusion
of pdfs and usually refer to it as ``combining density forecasts'' \cite{wallis2005combining,mitchell2005evaluating,winkler2019probability}.
This application will be addressed in more detail in Section \ref{sec: examples_Forec}. 

\vspace{1.5mm}

\item 
In Bayesian model averaging, several different models are used to derive different posterior pdfs based on the same data \cite{hoeting1999bayesian,fragoso2018bayesian}.
An aggregate
pdf is derived
as a weighted average of the individual pdfs, where the weights are given by the posterior probabilities of the models.
Bayesian model averaging has been widely used in phylogenetics \cite{posada2008jmodeltest,darriba2012jmodeltest}, economics \cite{moral2015model,steel2020model},
ecology \cite{dormann2018model},
and many other fields \cite{turkheimer2003undecidability,montgomery2010bayesian}.

\vspace{1.5mm}

\item 
Traditional implementations of Monte Carlo-based inference schemes do not easily scale to large data sets (``big data'').
A common expedient
then is to partition the data set into subsets and obtain a partial
posterior pdf approximation for
each subset. 
The partial
approximations are subsequently fused into an approximation of the overall
posterior pdf, which, thereby, takes into account the full data set \cite{neiswanger2014asymptotically,wang2013parallel,bardenet2017markov}. 
More details on this application are given in Section~\ref{sec: examples_GenMod}.
Another approach \cite{rabin2011wasserstein, srivastava2015wasp} directly fuses  sample representations of distributions by interpreting these samples as a weighted sum of Dirac measures. 


\end{itemize}

To focus the scope of the present work, we assume for the most part that the fusion center does not have any additional data
about the random vector $\btheta$ beyond the pdfs provided by the agents. (Here, an
exception is given by the supra-Bayesian setting studied in 
Sections~\ref{sec: fusion_of_distributions} and \ref{sec:lingaumeam}, where we assume that the fusion center knows a statistical model related to $\btheta$.)
In particular, the fusion center cannot access any training data that were 
used by the agents, e.g., to derive a global posterior pdf, and it does not have any validation data that it could use to validate the agents' pdfs.
Thus, although the fundamental problems are similar, we will not consider several ensemble learning methods such as stacking 
\cite{wolpert1992stacked,breiman1996stacked} or many other machine learning settings 
related to probabilistic fusion \cite{shazeer2017outrageously,hoang2019collective,liu2020gaussian}.
Furthermore, given our focus on pdfs rather than discrete probability distributions, we will not touch upon methods  tailored to the combination of classifiers, 
another large and 
growing field \cite{wozniak2014survey}.
Finally, we are interested in obtaining a pdf and not merely a point estimate of $\btheta$.
This is motivated by the fact that the pdf of $\btheta$ contains all the probabilistic information about $\btheta$ and can thus be used to obtain point estimates or other types of statistics.
Hence, certain works on multimodel inference \cite{burnham2002practical} and
the combination of forecasts \cite{clemen1989combining,armstrong2001combining,wallis2011combining} share some ideas with the present work but ultimately have a different focus.

\subsection{Approaches to Probabilistic Opinion Pooling}
\label{sec:approaches}

Although the probabilistic opinion pooling problem may appear simple and elementary, 
no single pooling function is universally accepted or uniformly best.
Generally speaking, we would like the pooling function to involve the agent pdfs $q_k(\btheta)$ in a way that follows some rationale.
This rationale and the resulting choice of a pooling function may depend on the overall problem setting, application-specific aspects, side constraints, 
additional information available to the fusion center, and other considerations. 
The probabilistic opinion pooling problem has been studied for many decades, 
and substantial research efforts have been dedicated to the definition or derivation of pooling functions. 
One of the earliest works is \cite{stone1961opinion}, where the linear pooling function
was introduced. 
Several survey articles on probabilistic opinion pooling with detailed literature reviews
have been published \cite{genest1986combining,cooke1991experts,clemen1999combining,dietrichprobabilistic}, however often with a focus on discrete random variables. 

In this work, we consider three principled approaches to defining a pooling function for pdfs.
In what we call the \emph{axiomatic approach}, the pooling function is defined indirectly by a set of properties (axioms) that it is required to satisfy. 
For example, it may be reasonable to require that the aggregate
pdf $q(\btheta)$ does not depend on the indexing order of the agent pdfs $q_k(\btheta)$, or that for equal $q_k(\btheta)$---i.e., unanimity among all the agents---the aggregate 
pdf $q(\btheta)$ conforms to that unanimous opinion.
Most of the early literature in the field was dedicated to the axiomatic
approach \cite{genest1986combining,dietrichprobabilistic}.
An axiomatic approach is also adopted
in the literature based on \textit{imprecise probabilities} \cite{stewart2018probabilistic}.
There, the idea is to define pooling operators that map from the agents' probability mass functions (pmfs) to a set of pmfs rather than a single pmf. 
To the best of our knowledge, the concept of imprecise probabilities has so far been considered only for discrete probability spaces \cite{stewart2018probabilistic,stewart2018learning}. 

In the \emph{optimization approach}, the pooling function is the result of an optimization, i.e., the minimization or maximization of an objective function. 
Usually, the idea is that the aggregate pdf $q(\btheta)$ should be as close as possible to all the agent pdfs $q_k(\btheta)$ simultaneously. 
This can be formulated as a minimization involving an information-theoretic divergence \cite{abbas2009kullback,garg2004generalized,da2019kullback} or a distance 
measure \cite{garg2004generalized,agueh2011barycenters}.
The resulting optimum $q(\btheta)$ can typically be interpreted as an ``average'' of the $q_k(\btheta)$.

Finally, the \emph{supra-Bayesian approach} considers the fusion center as a Bayesian observer that interprets the agent pdfs $q_k(\btheta)$
as random observations. 
This Bayesian observer builds on additional information about the dependence of these pdfs on
$\btheta$ 
(represented by the
conditional probability distribution $p(q_1, \dots, q_K  \,\vert\, \btheta)$ of the random functions
$q_1, \dots, q_K$ given $\btheta$) to calculate a posterior pdf, which  then constitutes the  fusion result \cite{winkler1968consensus,morris1977combining}.
Most of the early literature \cite{winkler1981combining,lindley1983reconciliation,clemen1985limits} describes $p(q_1, \dots, q_K  \,\vert\, \btheta)$ implicitly
by assuming that the joint distribution of the errors $\bmu_k-\btheta$ 
(where $\bmu_k$ is the expectation of $\btheta$ induced by the pdf
$q_k$) is multivariate Gaussian.
This 
reduces the fusion problem to the calculation of the posterior pdf for a simple Bayesian linear Gaussian model where the $\bmu_k$ 
are treated as observations at the fusion center and the covariance structure is known.
The practically most 
important scenario in the supra-Bayesian approach is where each agent has access to certain random observations that are statistically dependent on the random vector $\btheta$,
and both the agents and the fusion center have knowledge of a prior distribution of $\btheta$ and of the local likelihood functions of 
the agents. 
The agent pdf $q_k(\btheta)$ is here
given by the agent's local posterior pdf.
The fusion center is also aware of any statistical dependencies between the observations of different agents, which are described
by a global likelihood function.

\subsection{Contributions and Paper Organization}
\label{sec:contributions}

The diversity of approaches, perspectives, and solutions related to probabilistic opinion pooling motivates a survey that presents the theory and methodology
of the field
in a coherent manner. 
The present paper attempts to answer this call.
In addition, it provides a number of original contributions and results, including the following:

\begin{itemize}

\vspace{1.5mm}

\item 
A rigorous and coherent
treatment of probabilistic opinion pooling for a \emph{continuous} random vector $\btheta$ and, accordingly,
for the fusion of pdfs. In particular, for the first time, the axiomatic approach is rigorously and thoroughly discussed for 
pdfs (Section~\ref{sec: pop_axiomatic_appr}).
So far, the focus in the literature has mostly been on discrete probability distributions, and it has been claimed that analogous results hold for pdfs. 
Although this is indeed often the case, the non-atomic structure of the pdf setting sometimes allows for stronger or different results.

\vspace{1.5mm}

\item 
The definition of a new pooling function, referred to as ``generalized multiplicative pooling function'' (Section~\ref{sec:genmult}).

\vspace{1.5mm}

\item 
Two new axioms for pooling functions, referred to as ``factorization preservation'' and ``generalized Bayesianity''
(Axioms~\ref{ax:FP} and \ref{ax:genbayes} in Section~\ref{subsec: axioms_opinion_pooling}).

\vspace{1.5mm}

\item 
Several new theorems presenting axiomatic characterizations of pooling functions for pdfs and related results
(Theorems~\ref{th:linearpooling}, \ref{th:genlinpool}, and \ref{th:mult_pooling}--\ref{th:implstruct} in Section~\ref{subsec:rel_axioms} and Appendices~\ref{app:proof_linearpooling}--\ref{app:thimplstruct}). 
These theorems are partly adaptations of existing results formulated
for discrete probability distributions 
and partly entirely new results. 

\vspace{1.5mm}

\item 
Proofs of the following results: 
the pooling function minimizing the weighted sum of $\alpha$-divergences is given by the weighted H\"older mean;
the pooling function minimizing the weighted sum of Pearson $\chi^2$-divergences is given by the weighted harmonic mean;
the pooling function minimizing the weighted sum of $L_2$ distances is given by the weighted arithmetic mean
(Theorems~\ref{th:optimal_hellinger} and \ref{th:optimal_l2}
in Sections~\ref{sec: pop_optim_appr_hellinger} and \ref{subsec: frechet_means}
and Appendices~\ref{appendix: proof_weighted_hellinger_simplex} and \ref{appendix: proof_weighted_l2_simplex}).
Furthermore, we derive the solution to the problem of minimizing a general class of weighted symmetric distance functions (Theorem~\ref{th: generalized_pooling_frechet} in Section~\ref{subsec: frechet_means} and Appendix~\ref{appendix: proof_weighted_general_frechet}).

\vspace{1.5mm}

\item 
A new framework of supra-Bayesian fusion of posterior pdfs in terms of finite-dimensional 
``local statistics'' (Sections~\ref{sec:abstract-bayesian} through \ref{sec:bayagdepmeam}). 
This includes an explicit pooling function for the case of agents collecting conditionally independent observations
(Theorem~\ref{th:bayesindep} in Section~\ref{sec:exbayes}),
a formal definition of and result for finite-dimensional supra-Bayesian fusion (Definition~\ref{def:supra} and Theorem~\ref{th:supraBayes} in Section~\ref{sec:abstract-bayesian}),
and 
a general procedure for establishing a fusion rule for the case of agents collecting conditionally dependent observations (Section~\ref{sec:bayagdepmeam}).
 
\vspace{1.5mm}

\item 
A detailed study of supra-Bayesian fusion of posterior pdfs for linear Gaussian  models (Section~\ref{sec:lingaumeam}), including the derivation of explicit pooling functions 
and fusion rules (Sections~\ref{sec:gaussscalar} and \ref{sec:gaussvector}, Appendices~\ref{appendix:scalarbayesfusionall} and \ref{appendix:vectorbayesfusion}).

\vspace{1.5mm}

\end{itemize}

The paper's structure is as follows. 
In Section \ref{sec: examples}, we illustrate the applicability and relevance of probabilistic opinion pooling by discussing three specific example applications.
In Section~\ref{sec: probabilistic_opinion_pooling}, we formulate the probabilistic opinion pooling problem for pdfs and
present a collection of specific pooling functions.
Section~\ref{sec: pop_axiomatic_appr} discusses the axiomatic approach to opinion pooling and provides
several new characterization theorems. 
In Section~\ref{sec: pop_optim_appr}, we consider the optimization approach to opinion pooling.
We describe various optimization criteria and show that they partly lead to the same pooling functions as the axiomatic approach
and partly to different pooling functions such as the family of H\"older means. 
The 
fusion of Gaussian distributions using the pooling functions from Sections~\ref{sec: probabilistic_opinion_pooling} and \ref{sec: pop_optim_appr}
is considered in Section~\ref{sec: opinion_pooling_with_gaussians}. 
Section~\ref{sec: choosing_the_weights} addresses
the choice of the weights involved in the two most prominent and popular pooling functions, namely, the linear and log-linear pooling functions,
as well as the choice of the parameter involved in the H\"older pooling function. 
In Section~\ref{sec: fusion_of_distributions}, we present
a new view of the supra-Bayesian pooling approach using finite-dimensional parametrizations.
The results of Section~\ref{sec: fusion_of_distributions} are specialized to linear Gaussian  models in Section~\ref{sec:lingaumeam}. 
The model of  Section~\ref{sec:lingaumeam} includes as a special case the supra-Bayesian setting presented in \cite{winkler1981combining,lindley1983reconciliation,clemen1985limits}.
We broaden this setting significantly and 
present detailed fusion rules.
In Section~\ref{sec:outlook}, we provide suggestions for future research, and in Section~\ref{sec: conclusion},
a summary
of our main insights and results.  
Detailed proofs of our main results are provided in several appendices.

\subsection{Notation}
\label{sec:notation}

We will use the following basic notation. 
Vectors are denoted by boldface lower-case letters (e.g., $\mathbf{t}$ and $\btheta$),
matrices
by boldface upper-case letters (e.g., $\mathbf{H}$ and $\boldsymbol{\Sigma}$),
and sets and events by calligraphic letters (e.g., $\mathcal{A}$).
The transpose is
written as $(\cdot)^\intercal\!$.
We write $\mathbf{I}_{d}$ for the identity matrix of dimension $d$,
$\mathbf{0}_{d_1\times d_2}$ for the $d_1\times d_2$ zero matrix,
$\mathbf{1}_{d}$ for the all-one vector of dimension $d$,
and $\otimes$ for the Kronecker product.
The symbol $\mathcal{P}$ denotes the set of all pdfs, 
and $\mathcal{S}_K$ denotes
the probability simplex on $[0,1]^K$, i.e., the set of all  $(w_1, \ldots, w_K) \in [0,1]^K$ with $\sum_{k=1}^K w_k=1$.
For a set or event $\mathcal{A}$, 
we denote the complement as $\mathcal{A}^c$, 
the indicator function as $\mathbbm{1}_{\mathcal{A}}$, and the Lebesgue measure as $\lvert \mathcal{A}\rvert$.
Further notation is listed in Table~\ref{tab:TableOfNotation}.

\begin{table}[t]%
\centering 
\begin{tabular}{p{1.9cm} c p{5.25cm}}
\toprule\\[-3.5mm]
\multicolumn{3}{c}{\textbf{\emph{Probabilistic opinion pooling}}} \\[.3mm]
\toprule\\[-2.8mm]
$q_k(\btheta)$ & --- & pdf of agent $k$
\\[.5mm]
$q(\btheta)$ & --- & aggregate (fused) pdf
\\[.5mm]
$Q_k({\cal A})$ & --- & probability of event ${\cal A}$ 
according to 
$q_k(\btheta)$ \\[.5mm]
$Q({\cal A})$ & --- & probability of event ${\cal A}$ according to 
$q(\btheta)$\\[.5mm]
$\bmu_{q_k}$ & --- & mean associated with 
$q_k(\btheta)$ \\[.5mm]
$\bmu_{q}$ & --- & mean associated with 
$q(\btheta)$ \\[.5mm]
$\bSigma_{q_k}$ & --- & covariance matrix associated with 
$q_k(\btheta)$ \\[.5mm]
$\bSigma_{q}$ & --- & covariance matrix associated with 
$q(\btheta)$ \\[1.5mm]
\bottomrule\\[-2.5mm]
\multicolumn{3}{c}{\textbf{\emph{Supra-Bayesian framework}}} \\[.3mm]
\toprule\\[-2.8mm]
$\y_k$ & --- & local observation vector of agent $k$
\\[.5mm]
$\y$ & --- & global observation vector (stacking all $\y_k$)\\[.5mm]
${\bf t}_k$ & --- & local statistic of agent $k$
\\[.5mm]
${\bf t}$ & --- & stacked vector of all local statistics ${\bf t}_k$ \\[.5mm]
$p(\btheta)$ & --- & prior pdf \\[.5mm]
$\ell_k(\btheta)$ & --- & local observation likelihood function of\\[-.2mm]
&& agent $k$
\\[.5mm]
$\ell(\btheta)$ & --- &  global observation likelihood function \\[.5mm]
$\lambda_k(\btheta)$ & --- & local ${\bf t}_k$-likelihood function of agent $k$
\\[.5mm]
$\lambda(\btheta)$ & --- &  global ${\bf t}$-likelihood function \\[.5mm]
$\pi_k(\btheta)$ & --- & local posterior pdf of agent $k$\\[1.5mm]
\bottomrule\\[-2.5mm]
\multicolumn{3}{c}{\textbf{\emph{General notation}}} \\[.3mm]
\toprule\\[-2.8mm]
$g[\hspace{.2mm}\cdot\hspace{.2mm}]$ & --- & pooling function 
\\[.5mm]
$g[q_1, \ldots, q_K](\btheta)$ & --- & fused pdf resulting from application of pooling\\[-.2mm]
&& function $g$ to pdfs $q_1(\btheta), \ldots, q_K(\btheta)$ 
\\[.5mm]
$\mathbb{E}_\psi[\cdot]$ & --- & expectation operator with respect to pdf $\psi(\btheta)$ \\[.5mm]
$\mathbb{E}[\cdot]$ & --- & expectation operator with respect to the joint\\[-.2mm]
&& pdf of all involved random variables \\[.5mm]
$\mathcal{N}(\btheta; \bmu, \bSigma)$ & --- & pdf of a Gaussian random vector $\btheta$ with mean\\[-.2mm]
&& $\bmu$ and covariance matrix $\bSigma$ \\[.6mm]
\bottomrule
\end{tabular}
\caption{Notation}
\label{tab:TableOfNotation}
\end{table}

\section{Illustrative Applications}
\label{sec: examples}

To illustrate the broad applicability of probabilistic opinion pooling or, more concretely, 
of the fusion of pdfs, we consider three illustrative applications in more detail.

\subsection{Target Tracking}
\label{sec: examples_TT}

Target tracking aims to 
estimate the time-varying state (e.g., position and velocity) 
of a ``target'' from a sequence of observations \cite{anderson1979optimal, ristic2003beyond}.
Applications 
include 
aeronautical and maritime situational awareness, surveillance, autonomous driving, biomedical analytics, 
remote sensing, and robotics.
The performance of target tracking can be enhanced by using
multiple sensors.
This can be done in an
optimal manner 
if the multisensor 
observation model
is completely known, including possible statistical dependencies between the observations. However, 
in many cases, 
a simplified
approach to multisensor target tracking 
based on probabilistic opinion pooling is adopted. 
Each sensor node operates a Bayesian
filter that, at each time step, calculates a local posterior pdf of the current 
state based solely on the observation of that sensor.
Fig.~\ref{fig:targettracking} illustrates the 
local posterior pdfs of two sensor nodes at two different time steps.
The local posterior pdfs of the various sensor nodes are then fused using,
typically, log-linear pooling 
or its second-order version known as covariance intersection
\cite{julier1997non, chong2001convex, hurley2002information, Chang10, hu2011diffusion, Bailey12, deng2012sequential, battistelli2014kullback, Gunay16, Lehrer19} 
(see Sections \ref{sec:loglinear} and \ref{subsec: fusion_gaussian_log_linear_pooling}).
This 
approach 
is practically convenient because
(i) the multisensor fusion is decoupled from the filtering, and 
(ii) it works for any choice of Bayesian filter methods used at the sensor nodes
and for any sensing modalities, even when they are different at different sensor nodes. 
These  
characteristics 
make the probabilistic opinion pooling approach well suited to heterogeneous and/or decentralized sensor networks.

\begin{figure}[t]
\centering
\includegraphics[trim={0cm 6cm 7cm 0cm}, clip, width=\linewidth]{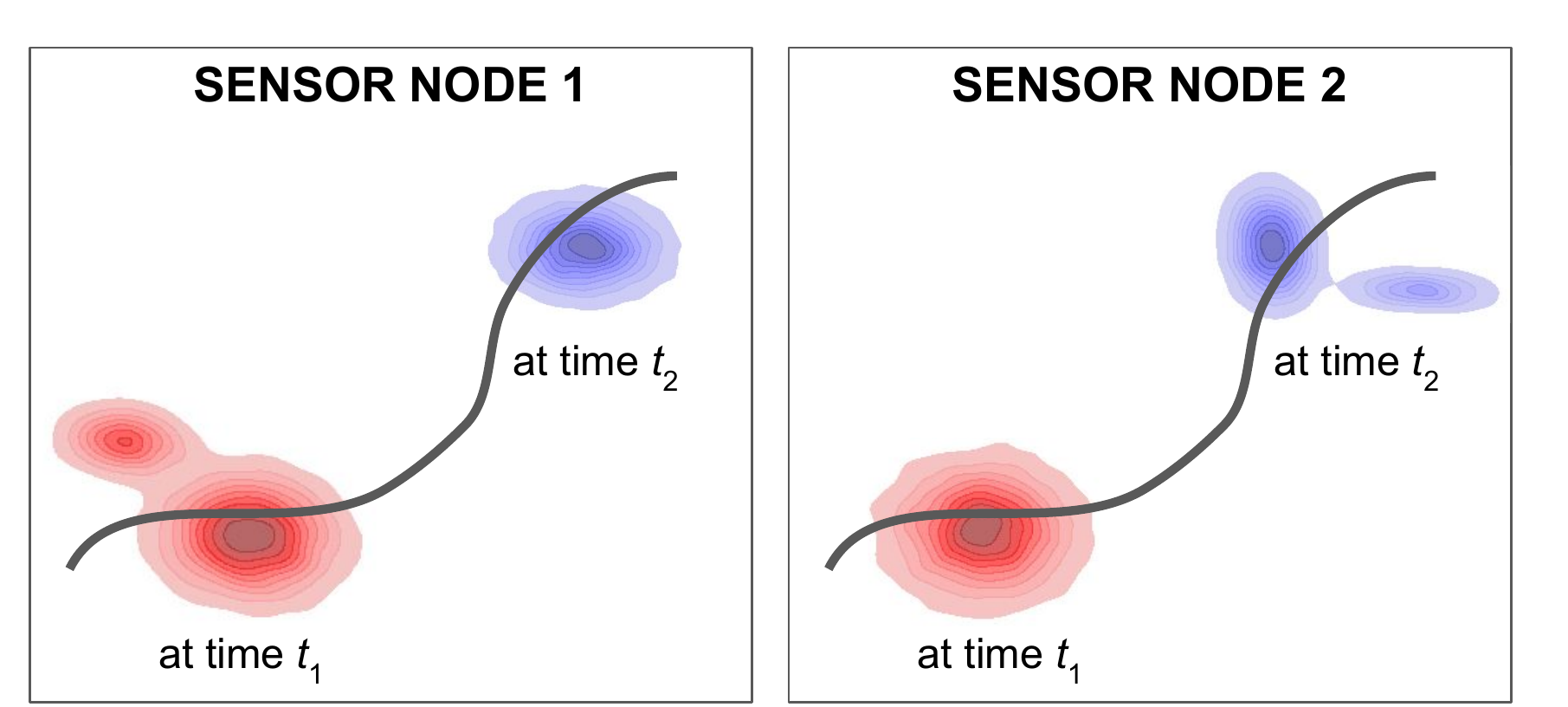}
\caption{Schematic illustration of the state trajectory 
of a target and the local posterior pdfs of two sensor nodes at two different time steps.}
\label{fig:targettracking}
\vspace{-2mm}
\end{figure}

A nontrivial extension of target tracking 
is \emph{multitarget} tracking, which involves an unknown time-varying number of targets
and a more complicated observation model
\cite{Bar-Shalom95,blackman1999design,Mahler07,Bar-Shalom11,Challa2011fundamentals,mahler2014advances,Koch14,meyer2018message}.
More specifically, targets can appear and disappear randomly, and there are missed detections (i.e., some sensors do not produce observations for some of the targets), clutter or false-alarm observations (which are not related to any target), and an observation-origin uncertainty (i.e., the sensor nodes do not know whether a given observation originated from a target, and from which target, or is clutter).
Probabilistic opinion pooling can be used both for ``vector-based'' multitarget tracking methods, 
which describe the joint state of the targets by a random vector, and for ``set-based'' methods, 
which describe it by a random finite set or equivalently a finite point process
\cite{Daley03, Mahler07, mahler2014advances}. 
In the vector-based case, 
the target states are fused individually using, typically, log-linear pooling or covariance intersection.
This presupposes an association of the target states across the 
sensors \cite{kaplan2008assignment,maresca2014maritime}.


In set-based methods, on the other hand, probabilistic opinion pooling is applied either to the 
posterior \emph{multiobject pdfs} or to the posterior \emph{probability hypothesis densities} (PHDs) of the 
sensor nodes,
which provide two alternative \emph{joint} descriptions of all the target states \cite{Mahler07, mahler2014advances}.
Here, both log-linear pooling---also termed
geometric average fusion, exponential mixture density, generalized covariance intersection, or Kullback-Leibler averaging
\cite{mahler2000optimal, Clark10, uney2010monte, Uney13, mahler2013toward, da2021recent, da2019kullback, li-fan2019second, li-battistelli2021distributed}---and
linear pooling (see Section \ref{subsec: fusion_rules_axioms_lop})---also termed arithmetic average fusion and minimum information loss fusion 
\cite{yu2016distributed, Li-T_19, li-fan2019second, gostar2020cooperative, li2020arithmetic, yi2020distributed, gao2020multiobject, da2021recent, li-hlawatsch2021distributed, li-battistelli2021distributed, da2019kullback}---have been used.
Log-linear pooling is more sensitive to 
missed detections 
whereas
linear pooling is more sensitive to 
clutter.
Regarding this sensitivity
tradeoff, 
we note
that pooling functions that are intermediate between the linear and log-linear ones are provided by the family of H\"older pooling functions to be presented in Section \ref{sec:holder}.

Finally, both log-linear and linear pooling have recently been generalized to 
multitarget tracking methods based on \emph{labeled} random finite sets, which track the identities 
of the targets in addition to their states 
\cite{fantacci2015consensus, fantacci2018robust, li-yi2018robust, li-battistelli2019computationally, gao2020fusion,kropfreiter2020probabilistic}.
Some of these methods require a label association step that is similar in spirit to the target
association step required by vector-based methods
\cite{li-yi2018robust, li-battistelli2019computationally, kropfreiter2020probabilistic, gao2020fusion}.


\subsection{Probabilistic Machine Learning}
\label{sec: examples_GenMod}

Probabilistic machine learning  \cite{ghahramani2015probabilistic,gal2016dropout} 
has recently seen applications in many different areas including quantum molecular dynamics \cite{krems2019bayesian}, 
disease detection \cite{leibig2017leveraging},
medical diagnosis \cite{begoli2019need},
scene understanding \cite{kendall2015bayesian},
and geotechnical engineering \cite{ching2019constructing}.
In machine learning, \textit{uncertainty quantification} for predictive models is required
for problems that involve
risk assessment.
Unfortunately, classical machine learning models do not account for parameter uncertainty, which makes them more susceptible to failure when dealing with unseen and/or unrelated data \cite{ovadia2019can}. This is a prominent issue for deep learning models \cite{abdar2021review}.
One way to account for predictive uncertainty in machine learning is to adopt a Bayesian framework: using training data, a prior pdf over the model parameters is updated to obtain a posterior pdf.
This posterior pdf is then used to calculate a predictive pdf
for unobserved data (test data).
This pdf is often represented in parametric form---e.g., a Gaussian pdf is parameterized by its mean and covariance matrix---or by a set of samples. Examples of Bayesian machine learning models include Bayesian linear regression, Bayesian neural networks \cite{wang2016towards, wilson2020bayesian}, Gaussian processes \cite{rasmussen2006gaussian}, and deep Gaussian processes \cite{salimbeni2017doubly}.

\begin{figure}[t]
\centering
\includegraphics[trim={0cm 5.25cm 13.5cm 0cm}, clip, width=\linewidth]{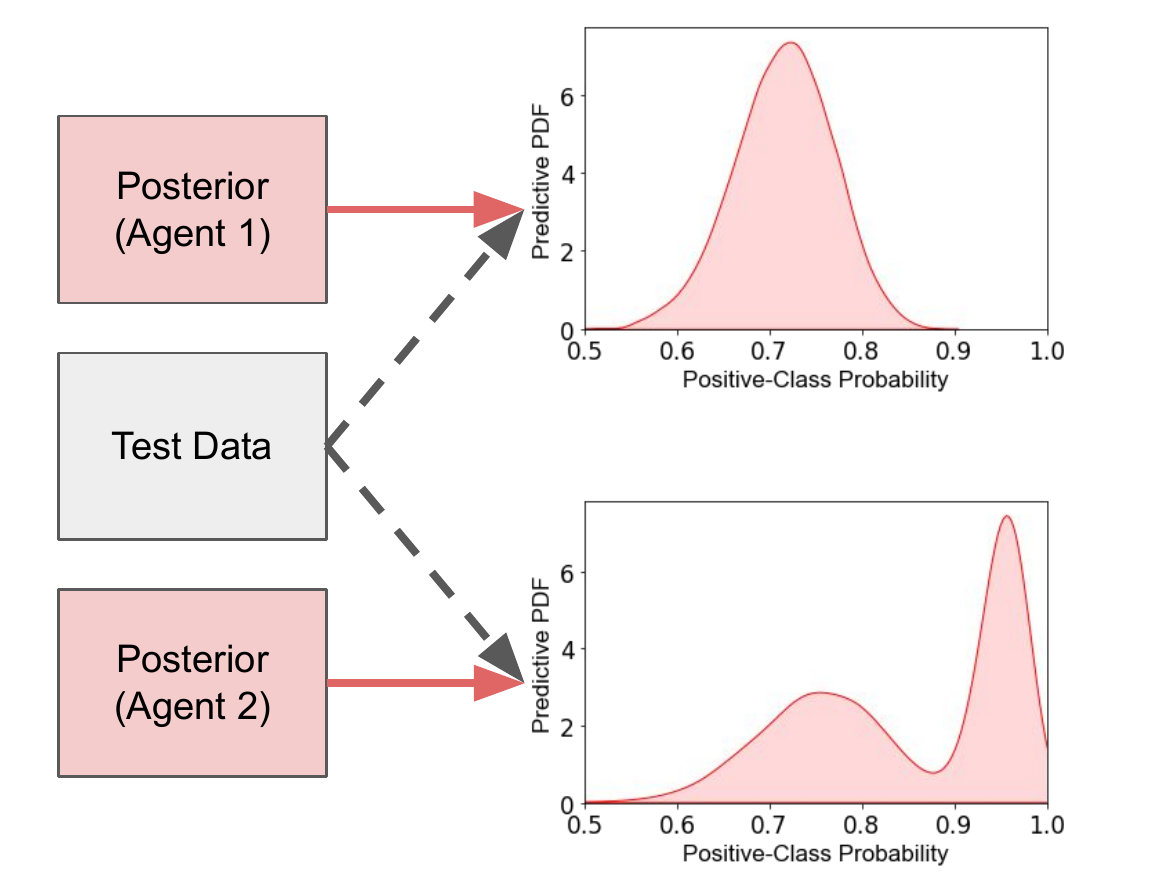}
\caption{Bayesian machine learning in the context of binary classification with two agents. Each agent obtains a posterior pdf from training data and uses it to derive its predictive pdf of the probability that test data belong to the  positive class.
These predictive pdfs are subsequently combined to obtain an aggregate predictive pdf.}
\label{fig:bml}
\vspace{-2mm}
\end{figure}



In certain scenarios of probabilistic machine learning, probabilistic opinion pooling can be used
to resolve practical challenges. 
For example, the choice of a model (or an architecture, or a set of parameters) is frequently not obvious, and thus there is a model uncertainty that has to be taken into account to ensure robustness and generalization.
A class of methods dealing with this issue is known as \textit{ensemble learning}. 
The learning is carried out by a collection of algorithms based on different models,
and the final result of classification, regression, or clustering is obtained by combining the individual results \cite{hastie2009friedman,murphy2012machine,sagi2018ensemble,opitz1999popular,polikar2006ensemble,rothe2016comparison}. 
The combination of the results of individual probabilistic learning algorithms can be implemented via probabilistic opinion pooling, i.e., by fusing the predictive pdfs produced by the individual algorithms. 
An example in the context of binary classification is shown in Fig.~\ref{fig:bml}.
Probabilistic opinion pooling in ensemble learning has been successfully applied, e.g., in the context of deep ensembles \cite{lakshminarayanan2017simple}, neural network ensembles \cite{lee2009neural}, and ensemble Gaussian processes \cite{lu2020ensemble}.  
Note that in ensemble learning, unlike in multisensor signal processing and, in particular, target tracking as discussed in the previous subsection, all the algorithms may operate on the same set of data.
%


Another practical challenge in machine learning is posed by privacy-sensitive scenarios.
Here, local (private) data observed at individual nodes may not be disseminated across the nodes or to a fusion center, and thus can be used only to train local models at the respective nodes. This framework, 
often referred to as \textit{federated learning},
requires the combination of local models at a fusion center \cite{li2020federated,yurochkin2019bayesian,thorgeirsson2021probabilistic,savazzi2020federated}.
Although in many instances of federated learning, 
updates are also communicated from the fusion center to the nodes, several works consider problem settings along the lines of probabilistic opinion pooling.
For example, agnostic federated learning \cite{mohri2019agnostic} combines sample representations of probability distributions trained on private data into an aggregate distribution.


Finally, the application of machine learning methods to ``big data'' scenarios calls for divide-and-conquer strategies that partition the data to much smaller sets, perform learning on each set, and combine the respective predictive or posterior distributions 
\cite{neiswanger2014asymptotically,wang2013parallel,bardenet2017markov,vehtari2020expectation}. 
Here, a focus has so far been on Markov chain Monte Carlo (MCMC) samplers for Bayesian inference \cite{neiswanger2014asymptotically,wang2013parallel,bardenet2017markov}.
For example, in \cite{neiswanger2014asymptotically}, the idea is to 
generate a ``subposterior'' for each small dataset and combine the subposteriors using the multiplicative pooling function (see Section~\ref{sec:multpool}).
Each subposterior is initially represented by a set of samples produced by an MCMC sampler but is then
converted into a continuous pdf given by a kernel density estimate. The different pdfs
are finally fused to form an approximation to the overall posterior pdf. 
This approach can be motivated by the fact, to be shown in Section~\ref{sec:exbayes}, that under a suitable conditional independence assumption a  multiplicative pooling function 
operating on the subposteriors
gives the overall posterior pdf.
The use of probabilistic machine learning has so far been restricted
by the fact that many popular methods of machine learning do not provide probabilistic results. 
However, we expect that the outcomes of recent and ongoing 
research will 
remove this limitation and thereby increase the successful application of probabilistic opinion pooling in this field.

\subsection{Forecasting}
\label{sec: examples_Forec}

The goal of forecasting is to predict future values of some variable of interest based on present and past observed data \cite{brockwell2016introduction}. An issue that may limit the performance of forecasting is a lack of confidence in the underlying model. This issue can be addressed by the \emph{combination of forecasts}, which fuses the forecasting results obtained with several different models \cite{clemen1989combining,armstrong2001combining,winkler2019probability}. While classical work has considered
point forecasts, 
\textit{probabilistic forecasting} uses a description of the variable of interest in terms of probability distributions. 
Here, for a long time, the focus was on discrete probability distributions
\cite{murphy1984probability}, and accordingly continuous random variables were approximated by discrete random variables through quantization. For example, in meteorology, the amount of precipitation was binned into a finite number of
categories \cite{bermowitz1979automated}. 
\begin{figure}[t]
\centering
\includegraphics[trim={0cm 0cm 0cm 0cm}, clip, width=\linewidth]{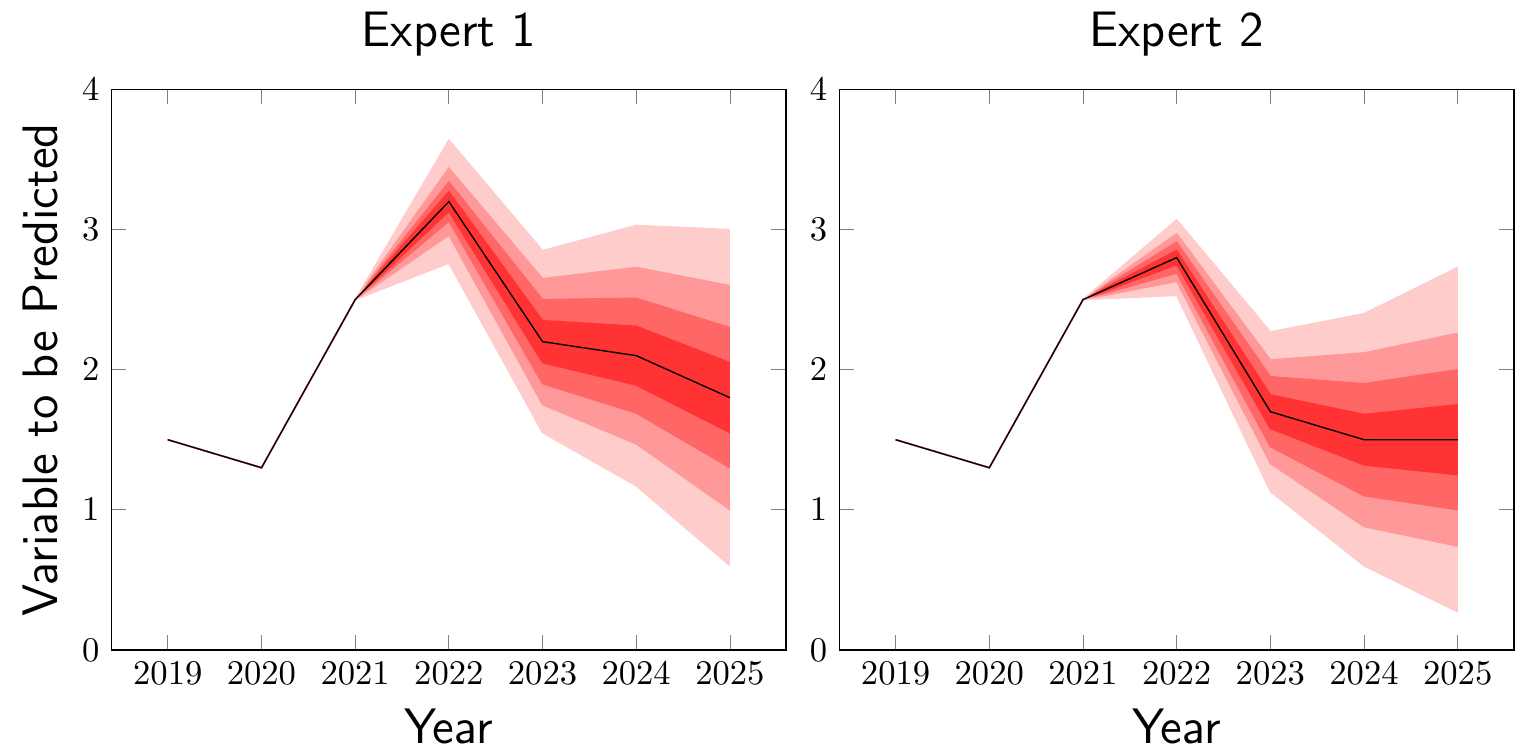}
\caption{Density forecasts of a variable (e.g., inflation) beyond 2021 made by two experts, visualized as fan charts.
The values of past years are already observed and thus fixed while predictions farther into the future become increasingly uncertain.}
\label{fig:forecasting}
\vspace{-2mm}
\end{figure}

By contrast, the idea of \textit{density forecasting} is to predict   continuous random variables directly in terms of their pdfs \cite{tay2000density}.
This is visualized by
Fig.~\ref{fig:forecasting}, which shows a fan chart representation of density forecasts made by two experts.
Density forecasting
was suggested already more than 50 years ago \cite{freiberger1965formulation,epstein1969stochastic}. However, the combination of density forecasts \cite{gneiting2013combining}---which is a special setting of the fusion of pdfs---was considered only much later.
Suggestions to combine density forecasts started with \cite{wallis2005combining,mitchell2005evaluating}, which discussed the optimization of the weights in the
linear pooling function based on training data.
At about the same time, the use of Bayesian model averaging \cite{hoeting1999bayesian} in forecasting was proposed  \cite{raftery2005using}, again resulting in a linear pooling function.
Also subsequent work focused on linear pooling \cite{hall2007combining}.
Nonlinear pooling functions were mostly obtained by a preprocessing of the individual pdfs (e.g., in the spread-adjusted linear pool \cite{glahn2009mos}) or by a postprocessing of the aggregate pdf (e.g., in the  Beta-transformed linear pool \cite{gneiting2013combining}).
Recently, the combination of density forecasts has also been studied in a nonparametric Bayesian setting based on the Beta-transformed linear pool \cite{bassetti2018bayesian}.

While the combination of density forecasts has the same goal 
as pdf fusion---namely, to fuse pdfs from different sources---there are two distinctive features. First, realizations of the random variable to be predicted
are observed on a regular basis, which enables an evaluation of density forecasts and their combinations based on new data. A significant part of the literature focuses on this aspect. Although beyond the scope of our work, such an evaluation can obviously be performed also within the general setting of pdf fusion if the required data are available. Second, forecasts usually concern one-dimensional random variables. This implies that the combination of forecasts can be formulated in terms of the one-dimensional cumulative distribution function (cdf), and more specific properties such as calibration \cite{gneiting2013combining} can be studied.
Also the combination of forecasts---in particular, the choice of weights---is often based on new data and the evaluation of the fused one-dimensional cdf \cite{clements2011combining}.

Probabilistic forecasting has been used in the broad domains of meteorology  \cite{murphy1984probability} and economics \cite{zarnowitz1969new,mitchell2005evaluating,tay2000density,steel2020model} and, more specifically and more recently,
in many disciplines including wind forecasting \cite{tastu2013probabilistic,baran2016mixture}, 
electric load forecasting \cite{ho2016probabilistic}, electricity price forecasting \cite{nowotarski2018recent}, and solar power forecasting \cite{doubleday2020probabilistic}.
The combination of density forecasts has, e.g., been considered in \cite{mitchell2005evaluating,baran2016mixture,doubleday2020probabilistic}, and we conjecture that successful deployments of this  variant of pdf fusion will emerge
in many further applications of probabilistic forecasting.


\section{Probabilistic Opinion Pooling}
\label{sec: probabilistic_opinion_pooling}

\subsection{Basic Framework} 

In probabilistic opinion pooling, we are interested in fusing  
the pdfs of $K$ agents or ``experts" into a single pdf. 
Let $\btheta\in\Theta\subseteq\mathbb{R}^{d_\theta}$ be a continuous random variable or vector defined on some probability space.%
\footnote{
Our results extend to arbitrary probability measures that are absolutely continuous with respect to a $\sigma$-finite non-atomic measure. However, to keep the presentation more easily accessible, we present all results in the familiar setting of pdfs on $\mathbb{R}^{d_\theta}$.}
Furthermore, let the pdf $q_{k}(\btheta)\in\mathcal{P}$ represent the {\it opinion} of the $k$th agent. 
The sequence of all opinions $(q_1, q_2, \ldots, q_K)\in\mathcal{P}^K$ is called the \textit{opinion profile}. 
We consider \textit{events} to be (measurable) subsets of $\Theta$. 
The probability of an event ${\cal A} \subseteq \Theta$ according to the opinion of the $k$th agent is given by
\begin{equation*}
    Q_{k}({\cal A}) = \int_{{\cal A}} q_{k}(\btheta) \mathrm{d}\btheta.
\end{equation*}
Given an opinion profile $(q_1, q_2, \ldots, q_K)$, a \textit{pooling function} $g\colon \mathcal{P}^K \rightarrow \mathcal{P}$ is used to fuse the agents' pdfs $q_k(\btheta)$ into a single pdf
\begin{equation*}
    q(\btheta)=g[q_1,\ldots,q_K](\btheta).
\end{equation*}
The probability of an event ${\cal A} \subseteq \Theta$ according to the fused pdf $q(\btheta)$ is then given by
\begin{equation*}
    Q({\cal A}) = \int_{\cal A} q(\btheta) \mathrm{d}\btheta.
\end{equation*}
The fused pdf $q(\btheta)$ summarizes the opinions of the $K$ agents and will be referred to as
the \textit{aggregate pdf}. The fusion of the agent opinions via the pooling function is done (at least virtually) at a \textit{fusion center}.

\subsection{Pooling Functions}
\label{subsec: fusion_rules_axioms}
Over the years, 
many different 
pooling functions $g$ have been proposed.  
We summarize some of them
in the following. 
These pooling functions will be reconsidered in later sections.

\subsubsection{Linear Pooling}
\label{subsec: fusion_rules_axioms_lop}
The most popular pooling function is the \textit{linear pooling function}, which was introduced in \cite{stone1961opinion}. 
Linear pooling aggregates the agent opinions through a weighted arithmetic average, i.e.,
\begin{equation}
    \label{eq: linear_pooling}
    g[q_1,\ldots,q_K](\btheta) = \sum_{k=1}^K w_{k} q_{k}(\btheta),
\end{equation}
where $(w_1,\ldots,w_K)\in \mathcal{S}_K$.

One can establish a connection between linear opinion pooling and model averaging \cite{hoeting1999bayesian}. 
Let us consider the joint distribution $q(\btheta, M)$ of the unknown random vector
$\btheta$ and a discrete ``model" random variable $M\in\{M_1,\ldots, M_K\}$. 
Furthermore, let $q(\btheta|M_k)$ denote the pdf of  $\btheta$ conditioned on model $M_k$ and  $P(M_k)$ denote
the probability of $M_k$. 
Then the marginal pdf of $\btheta$ is given by 
\begin{equation}
    \label{eq: model_averaging}
    q(\btheta) = \sum_{k=1}^K P(M_k) q(\btheta|M_k).
\end{equation}
This is equivalent to the linear pooling operation \eqref{eq: linear_pooling}, wherein the agent pdf $q_k(\btheta)$
is interpreted as the pdf of  $\btheta$ under model $M_k$,
the weight $w_k$  equals the probability of $M_k$, and the aggregate pdf $q(\btheta)$ is the marginal pdf of $\btheta$.

\subsubsection{Generalized Linear Pooling}
The \textit{generalized linear pooling function} defined in \cite{genest1984pooling}  
includes an arbitrary  pdf $q_0$ in the weighted arithmetic average \eqref{eq: linear_pooling}, i.e., 
\begin{equation}
    \label{eq:gen_linear_pooling}
    g[q_1,\ldots,q_K](\btheta) = \sum_{k=0}^K w_{k} q_{k}(\btheta),
\end{equation}
where $(w_0,\ldots,w_K) \in \mathcal{S}_{K+1}$.
We note that in the general, measure-theoretic formulation of generalized linear opinion pooling in \cite{genest1984pooling}, some weights $w_i$ are allowed to be negative. 
However, in the setting of fusing pdfs, this would result in a fusion rule $g$ that does not give a valid (nonnegative) pdf for all possible opinion profiles $(q_1,\ldots,q_K)$. 
Thus, we restrict to nonnegative weights. 
One possible interpretation of the pdf $q_0$ is as the opinion of the fusion center. Alternatively,  $q_0$ can be interpreted as a regularization.

\subsubsection{Log-linear Pooling}
\label{sec:loglinear}
Another popular pooling function is the \textit{log-linear pooling function} \cite{genest1984characterization}. 
This function aggregates the agent opinions using a weighted geometric average, i.e.,
\begin{equation}
    \label{eq: log_linear_pooling}
    g[q_1,\ldots,q_K](\btheta) = c \prod_{k=1}^K \left(q_k(\btheta)\right)^{w_k},
\end{equation}
where $c$ is a normalization factor given by 
\begin{equation}
    \label{eq:log_linear_pooling_normaliz}
    c = \frac{1}{\int_{\Theta} \prod_{k=1}^K \left(q_k(\btheta)\right)^{w_k} \mathrm{d}\btheta},
\end{equation}
and $(w_1,\ldots,w_K)\in \mathcal{S}_K$. 
To avoid the possibility of the integral in \eqref{eq:log_linear_pooling_normaliz} being zero and, thus, $c$ being undefined, this pooling function is usually only defined for pdfs that are positive on the domain $\Theta$.
We will refer to opinion profiles $(q_1, \dots, q_K)$ that satisfy
\begin{equation}
\label{eq:nonnegopprof}
    q_k(\btheta) > 0 \quad \text{ for all } \btheta \in \Theta\,
\end{equation}
as \textit{positive opinion profiles}.

The pooling function is called ``log-linear" because it is a linear function of the agent pdfs in the log-domain, i.e., the logarithm of the right-hand side of \eqref{eq: log_linear_pooling} is
\begin{align*}
    \log\left(c \prod_{k=1}^K \left(q_k(\btheta)\right)^{w_k}\right) = \log(c)+\sum_{k=1}^K w_k \log(q_k(\btheta)),
\end{align*}
which is a weighted arithmetic average (up to the additive constant $\log (c)$). 
We will therefore refer to the powers  $w_1,\ldots,w_K$ as ``weights."

\subsubsection{Generalized Log-linear Pooling}
Similar to the generalized linear pooling function, 
a generalization of the log-linear pooling function can be obtained by including an arbitrary function $\xi_0$ as an additional factor.
However, in contrast to the generalized linear pooling function, $\xi_0$ is not necessarily a pdf.
More specifically, the \textit{generalized log-linear pooling function} \cite{genest1986characterization} is defined as
\begin{equation}
    \label{eq:gen_log_linear_pooling}
    g[q_1,\ldots,q_K](\btheta) = c\, \xi_0(\btheta) \prod_{k=1}^K \left(q_k(\btheta)\right)^{w_k},
\end{equation}
where 
\begin{equation*}
    c = \frac{1}{\int_{\Theta} \xi_0(\btheta) \prod_{k=1}^K \left(q_k(\btheta)\right)^{w_k} \mathrm{d}\btheta},
\end{equation*}
$\xi_0$ is a bounded, positive function,
and $(w_1,\ldots,w_K)\in \mathcal{S}_K$.
Here, we again restrict to positive opinion profiles. 
The function $\xi_0$ can be used, e.g., to include the opinion of the fusion center or to regularize the fused density.

\subsubsection{H\"{o}lder Pooling}
\label{sec:holder}
 The following pooling function was apparently first suggested in \cite{cooke1991experts} as a generalization of the linear and log-linear pooling functions:
\begin{equation}
    \label{eq:holderpooling}
    g[q_1,\ldots,q_K](\btheta) = c\Bigg(\sum_{k=1}^K w_k(q_k(\btheta))^\alpha\Bigg)^{1/\alpha},
\end{equation}
where 
\begin{equation*}
    c= \frac{1}{\int_{\Theta}\big(\sum_{k=1}^K w_k(q_k(\btheta))^\alpha\big)^{1/\alpha}\mathrm{d}\btheta}
\end{equation*}
and 
$\alpha\in \mathbb{R}\setminus \{0\}$.
While for $\alpha \geq 1$ it can be shown that $c$ is defined for arbitrary opinion profiles, in the other cases we have to restrict to opinion profiles such that $c$ is defined.
Because the pooling function
in \eqref{eq:holderpooling} is the weighted H\"{o}lder mean (also called the generalized average) \cite{bullen2013handbook} of the agent pdfs  $q_k(\btheta)$,
we will refer to \eqref{eq:holderpooling} as the  \emph{H\"{o}lder pooling function}. 
The linear and log-linear pooling functions are special cases of the H\"older pooling function for $\alpha=1$ and $\alpha \to 0$, respectively.

\subsubsection{Inverse-linear Pooling}
The \emph{inverse-linear pooling function} (weighted harmonic average) is defined as
\begin{equation}
    \label{eq:inversepooling}
    g[q_1,\ldots,q_K](\btheta) = c\, \Bigg(\sum_{k=1}^K  \frac{w_k}{q_k(\btheta)}\Bigg)^{-1},
\end{equation}
where 
\begin{equation*}
    c= \frac{1}{\int_{\Theta}\big(\sum_{k=1}^K \frac{w_k}{q_k(\btheta)}\big)^{-1}\mathrm{d}\btheta} .
\end{equation*}
This is the special case of the H\"{o}lder pooling function for $\alpha=-1$.

\subsubsection{Multiplicative Pooling}
\label{sec:multpool}
The \textit{multiplicative pooling function}, proposed in \cite{dietrichprobabilistic} for pmfs, is defined as
\begin{equation}
    \label{eq: multiplicative_pooling}
    g[q_1,\ldots,q_K](\btheta) = c \left(q_0(\btheta)\right)^{1-K} \prod_{k=1}^K q_k(\btheta),
\end{equation}
where 
\begin{equation*}
    c = \frac{1}{\int_{\Theta} \left(q_0(\btheta)\right)^{1-K} \prod_{k=1}^K  q_k(\btheta) \mathrm{d}\btheta},
\end{equation*}
and $q_0$ is a   positive pdf
called the \emph{calibrating pdf}.
Here, we restrict to  positive opinion profiles and further assume that $q_k(\btheta)/q_0(\btheta)$ is bounded for all $k = 1, \dots, K$.
These assumptions  guarantee that the normalization constant $c$ is well-defined and nonzero.
In Section~\ref{sec:exbayes}, we will show that within the supra-Bayesian framework, the multiplicative pooling function is the correct fusion rule for combining posterior pdfs
in the case of conditionally independent observations. 
In that case, the calibrating pdf $q_0$ is the prior pdf used by the agents to form their posterior pdfs.

\subsubsection{Generalized Multiplicative Pooling}
\label{sec:genmult}
We propose another pooling function that is a generalization of both the generalized log-linear pooling function and the multiplicative pooling function.
In addition to a calibrating pdf $q_0$, we also allow for arbitrary weights in the  generalized log-linear pooling function \eqref{eq:gen_log_linear_pooling}.
More specifically, we define the \textit{generalized multiplicative pooling function} as
\begin{equation}
    \label{eq:gen_mult_pooling}
    g[q_1,\ldots,q_K](\btheta) = c \left(q_0(\btheta)\right)^{1-\sum_{k=1}^K w_k} \prod_{k=1}^K \left(q_k(\btheta)\right)^{w_k},
\end{equation}
where
\begin{equation*}
    c = \frac{1}{\int_{\Theta} \left(q_0(\btheta)\right)^{1-\sum_{k=1}^K w_k} \prod_{k=1}^K \left(q_k(\btheta)\right)^{w_k} \mathrm{d}\btheta},
\end{equation*}
$q_0$ is a positive calibrating pdf,
and the weights $w_1, \ldots, w_K\in \mathbb{R}$ are arbitrary real numbers. 
We again restrict to positive opinion profiles and assume that $(q_k(\btheta)/q_0(\btheta))^{w_k}$ is bounded for all $k = 1, \dots, K$.
In Section~\ref{sec:gaussscalar}, we will show that within the supra-Bayesian framework with a linear Gaussian  model, the generalized multiplicative pooling function is the correct fusion rule for combining posterior pdfs.

\subsubsection{Dictatorship Pooling}
The {\it dictatorship pooling function} maps the opinion profile to a single agent opinion, i.e., 
\begin{equation}
    \label{eq: dictatorship}
    g[q_1, \ldots, q_K](\btheta) = q_k(\btheta),
\end{equation}
for some fixed $k\in\{1,\ldots,K\}$. 
Although this function is a valid pooling function, one would not normally expect it to be a good choice.
\subsubsection{Dogmatic Pooling}
The
\textit{dogmatic pooling function}  enforces a fixed  pdf $q_0$  independently of the opinion profile, i.e., 
\begin{equation}
 \label{eq: dogmatic_pooling_function}
     g[q_1, \ldots, q_K](\btheta) = q_0(\btheta)\,.
\end{equation}
Again, this pooling function will not be suitable in most applications.

\section{The Axiomatic Approach} 
\label{sec: pop_axiomatic_appr}
Fundamentally, we would like the pooling function $g[q_1,\ldots,q_K]$ to depend directly on all the agent pdfs
$q_k$ in a way that follows some rationale. 
One principled approach to probabilistic opinion pooling is the axiomatic approach, which seeks to determine all 
pooling functions that satisfy a set of desirable properties (axioms). 
In this section, we first formulate some
axioms and then rigorously analyze the relationships between these axioms and the pooling functions presented in Section~\ref{subsec: fusion_rules_axioms}.

\subsection{Axioms}
\label{subsec: axioms_opinion_pooling}
To begin, one basic restriction we may impose 
on the pooling function is that it be
a symmetric function, i.e., a function whose arguments can be interchanged without altering the output of the function. This means that there is no ``natural order" of the agents, and all agents are treated equally. This is formally stated in the following axiom:
\begin{axiom}
\label{ax:S}
{\bf (Symmetry)}
For all permutations  $\beta\colon \mathcal{K}\rightarrow \mathcal{K}$ of the set $\mathcal{K}=\{1,\ldots,K\}$ and  all opinion profiles $(q_1, \ldots, q_K)$, the pooling function $g$ satisfies
\begin{equation*}
    g[q_1,\ldots, q_K](\btheta)=g[q_{\beta(1)},\ldots, q_{\beta(K)}](\btheta).
\end{equation*}
\end{axiom}
A symmetric pooling function seems to be desirable and natural since it treats the pdfs of the agents equally at the fusion center. 
However, if certain agents are known a priori to be more ``reliable" or ``informative" than other agents, then it may be reasonable to 
emphasize them in the pooling function. For example, in the linear or log-linear pooling function, we may assign larger weights $w_k$.
If this is done in a fixed manner, the pooling function is no longer symmetric. 
On the other hand, if the weights are chosen adaptively such that each weight
is an explicit function of the opinion profile and this adaptation rule involves each agent in the same way, then all agents are treated equally and
the resulting pooling function is still symmetric.
This will be further discussed in Section~\ref{sec: choosing_the_weights}.

Another basic property for a pooling function is the preservation of agreement among agents. For instance, if each of the agents believes that a certain event ${\cal A}\subset\Theta$ 
is a null event, i.e., the probability of 
${\cal A}$ is 0 according to all the agents, then ${\cal A}$ should also be a null event according to the aggregate pdf. 
This property is called the \textit{zero preservation property} (ZPP) \cite{mcconway1981marginalization}:
\begin{axiom}
\label{ax:ZPP}
{\bf (Zero Preservation)} 
For any event ${\cal A}\subset\Theta$, 
if $Q_k({\cal A})=0$ 
for all $k$, then $Q({\cal A})=0$.
\end{axiom}

The next property, termed \textit{unanimity preservation} \cite{dietrichprobabilistic}, 
 asserts that if the opinions of the agents are identical, then the aggregate pdf should conform to that unanimous opinion.
\begin{axiom}
\label{ax:UP}
{\bf (Unanimity Preservation)}
If for all events ${\cal A}\subseteq\Theta$, 
the probabilities $Q_k({\cal A})=p_{\cal A}$ 
coincide for all $k$, then $Q({\cal A})=p_{\cal A}$. 
Equivalently, if\,%
\footnote{We consider two pdfs to be equal if they are equal almost everywhere with respect to the Lebesgue measure.}
$q_k(\btheta)=q_0(\btheta)$ for all $k$ and some pdf $q_0(\btheta)$, then $q(\btheta)=q_0(\btheta)$.
\end{axiom}

Another property that may be desirable in a pooling function is the \textit{strong setwise function property} (SSFP) \cite{mcconway1981marginalization}. 
The SSFP states that the probability of an event
${\cal A}\subseteq\Theta$ 
according to the aggregate pdf $q(\btheta)$ can be expressed as a function of the probabilities of that event according to each agent, i.e., $Q_1({\cal A}),\ldots,Q_K({\cal A})$. 
\begin{axiom}
\label{ax:SSFP}
{\bf (Strong Setwise Function Property)}
There exists a function $h\colon  [0,1]^K\rightarrow [0,1]$ such that for  all opinion profiles $(q_1, \ldots, q_K)$ and for all events ${\cal A}\subseteq\Theta$,
\begin{equation}
    \label{eq: global_function_h}
    Q({\cal A})=h(Q_1({\cal A}),\ldots, Q_K({\cal A})).
\end{equation}
\end{axiom}
We note that this axiom is in general not equivalent to the property that there exists a function $\tilde{h}\colon  [0,\infty)^K\rightarrow [0,\infty)$ such that for  
all opinion profiles $(q_1, \ldots, q_K)$ and each point $\btheta \in \Theta$
\begin{equation}
    \label{eq: global_function_htilde}
    q(\btheta) = \tilde{h}(q_1(\btheta), \dots, q_K(\btheta)).
\end{equation}
In particular, for the case that $\Theta$ has finite Lebesgue measure $\lvert \Theta \rvert$, the dogmatic pooling function $q(\btheta)=1/\lvert \Theta \rvert$ for $\btheta \in \Theta$ trivially satisfies \eqref{eq: global_function_htilde}  but not \eqref{eq: global_function_h} (as a simple consequence of Theorem~\ref{th:linearpooling} below).

A more relaxed criterion than the SSFP is the \textit{weak setwise function property}  (WSFP) \cite{mcconway1981marginalization}. 
The WSFP states
that the probability of an event according to the aggregate pdf is a function of the probabilities of that event according to each agent \textit{and} the event itself.
\begin{axiom}
\label{ax:WSFP}
{\bf (Weak Setwise Function Property)}
For all events ${\cal A}\subseteq\Theta$, 
there exists a generally ${\cal A}$-dependent 
function $h_{{\cal A}}\colon [0,1]^K\rightarrow [0,1]$ such that for  all opinion profiles $(q_1, \ldots, q_K)$
\begin{equation}
\label{eq:wsft}
Q({\cal A})=h_{{\cal A}}(Q_1({\cal A}),\ldots, Q_K({\cal A})).
\end{equation}
\end{axiom}
The WSFP is also equivalent to the so-called  \textit{marginalization property}, which states that marginalization and fusion are commutative operations. 
Formulating the marginalization property requires a  measure-theoretic language that is beyond the scope of this paper. 
We thus omit a discussion of the marginalization property and refer the interested reader to \cite{mcconway1981marginalization} and \cite{genest1984pooling}. 

Another relaxation of the SSFP is the  \textit{likelihood principle} \cite{genest1984characterization}. 
Here, the  value of the aggregate pdf $q(\btheta)$ at some point $\btheta$ may only depend on the  values of all $q_k(\btheta)$ at the same  $\btheta$ up to a normalization constant that can depend on the opinion profile. 
\begin{axiom}
\label{ax:LP}
{\bf (Likelihood Principle)}
There exists a function $h\colon  [0,\infty)^K\rightarrow [0,\infty)$ such that for  all opinion profiles $(q_1, \ldots, q_K)$ and each point $\btheta \in \Theta$
\begin{equation*}
    q(\btheta) = \frac{h(q_1(\btheta), \dots, q_K(\btheta))}{\int_{\Theta}h(q_1(\btheta'), \dots, q_K(\btheta') ) \mathrm{d}\btheta'}.
\end{equation*}
\end{axiom}
The name ``likelihood principle" is motivated by viewing the pdfs as normalized likelihood functions: 
in this viewpoint, the idea is that the fused likelihood at $\btheta$ should only depend on the local likelihoods  at $\btheta$ up to normalization \cite{genest1984characterization}.
Note that \eqref{eq: global_function_htilde} is a significantly stronger assumption because the function $\tilde{h}$ in  \eqref{eq: global_function_htilde} has to normalize to one.

We can also formulate a weak version of the likelihood principle, where the function $h$ may depend on $\btheta$ \cite{genest1984characterization}.
\begin{axiom}
\label{ax:WLP}
{\bf (Weak Likelihood Principle)}
For all $\btheta\in \Theta$, there exists a generally $\btheta$-dependent function $h_{\btheta}\colon  [0,\infty)^K\rightarrow [0,\infty)$ such that for  all opinion profiles $(q_1, \ldots, q_K)$
\begin{equation*}
    q(\btheta) = \frac{h_{\btheta}(q_1(\btheta), \dots, q_K(\btheta))}{\int_{\Theta}h_{\btheta}(q_1(\btheta'), \dots, q_K(\btheta') ) \mathrm{d}\btheta'}.
\end{equation*}
\end{axiom}

Another important axiom is \textit{independence preservation}%
\footnote{Independence preservation should not be confused with the WSFP, which is sometimes referred to as the independence or eventwise independence property (e.g., \cite{dietrichprobabilistic}).}
\cite{laddaga1977lehrer}.
This axiom asserts that if all the agents agree that two events ${\cal A}, {\cal B}\subseteq\Theta$ 
are independent, then these events should be independent also according to the aggregate pdf. 
\begin{axiom}
\label{ax:IP}
{\bf (Independence Preservation)} 
For any events ${\cal A}, {\cal B}\subseteq \Theta$, 
if
\begin{equation*}
Q_{k}({\cal A} \cap {\cal B})=Q_{k}({\cal A})Q_{k}({\cal B})
\end{equation*}
for all $k\in\{1,\ldots,K\}$, then $Q({\cal A} \cap {\cal B})=Q({\cal A})Q({\cal B})$. 
\end{axiom}

A relaxation of independence preservation which, to the best of our knowledge, has not been considered before is to assume the preservation of a given factorization structure. 
\begin{axiom}
\label{ax:FP}
{\bf (Factorization Preservation)} 
For any functions $f_1 \colon \Theta \to \mathbb{R}^{d_1}$ and $f_2\colon \Theta \to \mathbb{R}^{d_2}$, if there exist functions $q_{k,1}$ and $q_{k,2}$ such that
\begin{equation*}
   q_{k}(\btheta)=q_{k,1}(f_1(\btheta)) q_{k,2}(f_2(\btheta))
\end{equation*}
for all $k\in\{1,\ldots,K\}$, then there exist functions $q_{\text{\em a},1}$ and $q_{\text{\em a},2}$ such that 
\begin{equation*}
   q(\btheta)=q_{\text{\em a},1}(f_1(\btheta)) q_{\text{\em a},2}(f_2(\btheta))\,.
\end{equation*}
\end{axiom}
This axiom expresses,
in particular, preservation of the independence of components of $\btheta$.
Assume that $\btheta = (\btheta_1, \btheta_2)$ and all agent pdfs factor according to 
$q_{k}(\btheta)=q_{k,1}(\btheta_1) q_{k,2}(\btheta_2)$.
We can choose $f_1(\btheta)= \btheta_1$ and $f_2(\btheta)= \btheta_2$, and  factorization preservation then implies that also the aggregate pdf  preserves the independence of $\btheta_1$ and $\btheta_2$, i.e., $q(\btheta)=q_{\text{a},1}(\btheta_1) q_{\text{a},2}(\btheta_2)$. 

The final axioms we consider are motivated by Bayesian updating of probabilities. 
More specifically, we interpret
each agent pdf $q_k(\btheta)$ as the agent's  belief about an unknown quantity $\btheta$ after observing some data.
When observing new (additional) data,  $q_k(\btheta)$ is updated by multiplying it by a  likelihood function $\ell\colon \Theta\rightarrow[0,\infty)$, which relates the  agent's new data to $\btheta$. 
The updated belief of the $k$th agent, $q_k^{(\ell)}(\btheta)$, is thus given as
\begin{equation}
    \label{eq: discrete_bayesian_update}
    q_k^{(\ell)}(\btheta)=\frac{\ell(\btheta)q_k(\btheta)}{\int_\Theta \ell(\btheta')q_k(\btheta')\mathrm{d}\btheta'}.
\end{equation}
To avoid degenerate cases, one usually assumes in the following axioms that all pdfs are positive on the domain $\Theta$. 
Thus, we restrict  the statements of the axioms to positive opinion profiles.
The first axiom related to the Bayesian framework is known as \textit{external Bayesianity} \cite{madansky1964externally,genest1984characterization}. 
\begin{axiom}
\label{ax:EB}
{\bf (External Bayesianity)} For all  functions $\ell\colon\Theta\rightarrow [0,\infty)$ and all positive opinion profiles $(q_1, \dots, q_K)$ satisfying 
$0 < \int_\Theta \ell(\btheta)q_k(\btheta)\mathrm{d}\btheta<\infty$ for all $k \in \{1, \dots, K\}$, we have 
\begin{equation*}
   q^{(\ell)}(\btheta) = g[q_1^{(\ell)},\ldots,q_K^{(\ell)}](\btheta),
\end{equation*}
where $q_k^{(\ell)}$ is defined in \eqref{eq: discrete_bayesian_update} and
\begin{equation}
\label{eq:fused_bayesian_update}
    q^{(\ell)}(\btheta) = \frac{\ell(\btheta)q(\btheta)}{\int_{\Theta}\ell(\btheta')q(\btheta')\mathrm{d}\btheta'},
\end{equation}
with $q(\btheta) = g[q_1,\ldots,q_K](\btheta)$.
\end{axiom}
This axiom is motivated by the following Bayesian scenario: 
Assume that $q_1, \dots, q_K$ are prior pdfs of $K$ agents. 
Some data are observed, and the resulting likelihood function $\ell$ is provided to all agents.
Then, a pooling function $g$ satisfying external Bayesianity 
gives the same fusion result if it first aggregates the priors $q_k$ into a fused prior $q$ and then $q$ is updated 
according to \eqref{eq:fused_bayesian_update}, or if it aggregates the posterior pdfs $q_k^{(\ell)}$ resulting from all agents updating their priors 
according to \eqref{eq: discrete_bayesian_update}.
Thus, external Bayesianity states that pdf updating and fusion are commutative operations.
Such a property is desirable in applications where the agents share identical data (i.e., a global likelihood function) but have distinct prior distributions \cite{rufo2012log}.  

A second axiom related to the Bayesian framework is known as \textit{individualized Bayesianity} \cite{dietrichprobabilistic}. 
This axiom is motivated by the idea of combining posterior probabilities, where each agent's posterior probability is based on private data (i.e., a local likelihood function) in contrast to all agents sharing identical data. 
\begin{axiom}\label{ax:indbayes}
\label{ax:IB}
{\bf (Individualized Bayesianity)} 
For all $k \in \{1, \dots, K\}$, all bounded, positive%
\footnote{The assumption of boundedness and positivity
is needed
to obtain the 
characterization theorems involving individualized Bayesianity
in Section~\ref{subsec:rel_axioms}.}
functions $\ell\colon\Theta\rightarrow [0,\infty)$, and all positive opinion profiles $(q_1, \dots, q_K)$,
we have
\begin{equation}
\label{eq:indivbayes}
    q^{(\ell)}(\btheta) =g[q_1,\ldots,q_{k-1},q_{k}^{(\ell)},q_{k+1},\ldots,q_K](\btheta),
\end{equation}
 where 
$q_{k}^{(\ell)}$ and $q^{(\ell)}$ are defined by \eqref{eq: discrete_bayesian_update} and \eqref{eq:fused_bayesian_update}, respectively.
\end{axiom}
This axiom is motivated by a scenario that is partly different from the scenario motivating external Bayesianity.
We again assume that $q_1, \dots, q_K$ are prior pdfs of the
agents. For some arbitrary but fixed $k$, the $k$th agent observes (private) data in terms of a likelihood function $\ell$.
Then, a pooling function $g$ satisfying individualized Bayesianity gives the same fusion result if it first
aggregates the priors $q_k$ into a fused prior $q$ and then $q$ is updated
according to \eqref{eq:fused_bayesian_update}, or if it aggregates the priors of all but the $k$th agent and the posterior pdf $q_k^{(\ell)}$
resulting from the $k$th agent updating its prior
according to \eqref{eq: discrete_bayesian_update}.
Thus, individualized Bayesianity states that pdf updating
at a \textit{single} agent and fusion are commutative operations.

Finally, we state a novel axiom that generalizes individualized Bayesianity. 
We thus call it \textit{generalized Bayesianity}.
\begin{axiom}\label{ax:genbayes}
{\bf (Generalized Bayesianity)} 
\label{ax:GB}
For all bounded, positive functions $\ell_k\colon \Theta\rightarrow [0,\infty)$, $k\in \{1, \dots, K\}$,
there exists a fused likelihood function $h[\ell_1, \dots, \ell_K]$
such that for all positive opinion profiles $(q_1, \dots, q_K)$, 
we have
\begin{equation}
\label{eq:gnbayes}
    q^{(h[\ell_1, \dots, \ell_K])}(\btheta) = g[q_1^{(\ell_1)},\ldots,q_K^{(\ell_K)}](\btheta),
\end{equation}
where 
$q_{k}^{(\ell_k)}$ and $q^{(h[\ell_1, \dots, \ell_K])}$ are defined by \eqref{eq: discrete_bayesian_update} and \eqref{eq:fused_bayesian_update}, respectively.
\end{axiom}
This axiom states that fusing $q_1^{(\ell_1)},\ldots,q_K^{(\ell_K)}$, i.e., the result of updating $q_1, \dots, q_K$, is equivalent to updating $q$, i.e., 
the result of fusing  $q_1, \dots, q_K$, by a ``fused likelihood function" $h[\ell_1, \dots, \ell_K]$.
Note that the fused likelihood function is not allowed to depend on the opinion profile $(q_1, \dots, q_K)$.

The axioms related to the Bayesian framework presented above are not directly related to the supra-Bayesian approach presented in Sections~\ref{sec: fusion_of_distributions} and~\ref{sec:lingaumeam} below. 
More specifically, in the supra-Bayesian framework, we have explicit likelihood functions and thus the pooling function does not necessarily satisfy properties that relate to arbitrary likelihood functions as in the axioms above. 



\subsection{Relations between Axioms and Pooling Functions}
\label{subsec:rel_axioms}

Having presented various pooling functions in Section~\ref{subsec: fusion_rules_axioms} and various axioms in Section~\ref{subsec: axioms_opinion_pooling}, we next
 analyze which pooling functions satisfy which axioms and, conversely, which axioms imply which pooling functions.
 Our results are summarized in Table~\ref{tab:axioms}.
\begin{table*}[tbh]
    \centering
    \rowcolors{4}{gray!10}{}
    \begin{tabular}{lcccccccccccc}
        \toprule
         {} & \multicolumn{11}{c}{Axiom} \\
         \cmidrule(lr){2-13}
         Pooling Function  & 1 & 2 & 3 & 4 & 5 & 6 & 7 & 8 & 9 & 10 & 11 & 12\\
         \midrule
         Linear  
         & $*$        & \checkmark & \checkmark & \checkmark & \checkmark & \checkmark & \checkmark & 
         {}         & {}         & {}         & {}         & {}         \\
         Generalized Linear  
         & $*$        & {}         & {}         & {}         & \checkmark & {}         & \checkmark & 
         {}         & {}         & {}         & {}         & {}         \\
         Log-linear  
         & $*$        & n.a.       & \checkmark & {}         & {}         & \checkmark & \checkmark & 
         {}         & \checkmark & \checkmark & {}         & \checkmark \\
         Generalized Log-linear  
         & $*$        &  n.a.      & {}         & {}         & {}         & {}         & \checkmark & 
         {}         & \checkmark & \checkmark & {}         & \checkmark \\
         H\"{o}lder
         & $*$        &  n.a.      & \checkmark & {}         & {}         & \checkmark & \checkmark & 
         {}         & {}         & {}         & {}         & {}         \\
         Inverse-linear  
         & $*$        & n.a.       & \checkmark & {}         & {}         & \checkmark & \checkmark & 
         {}         & {}         & {}         & {}         & {}         \\
         Multiplicative  
         & \checkmark &  n.a.      & {}         & {}         & {}         & {}         & \checkmark & 
         {}         & \checkmark & {}         & \checkmark & \checkmark \\
         Generalized Multiplicative  
         & $*$        &  n.a.      & {}         & {}         & {}         & {}         & \checkmark & 
         {}         & \checkmark & {}         & {}         & \checkmark \\
         Dictatorship  
         & {}         & \checkmark & \checkmark & \checkmark & \checkmark & \checkmark & \checkmark & 
         \checkmark & \checkmark & \checkmark & {}         & \checkmark \\
         Dogmatic  
         & \checkmark & {}         & {}         & {}         & \checkmark & {}         & \checkmark & 
         {}         & {}         & {}         & {}         & \checkmark \\
         \bottomrule
    \end{tabular}
    \caption{Axioms satisfied by the pooling functions presented in Section~\ref{subsec: fusion_rules_axioms}. ($*$: satisfied if and only if all weights are equal.)
    }
    \label{tab:axioms}
\end{table*}
In what follows, we will abbreviate the various axioms as A1, A2, etc.

\begin{theorem}
\label{th:linearpooling}
    The linear pooling function in \eqref{eq: linear_pooling} satisfies the 
    ZPP (A\ref{ax:ZPP}), unanimity preservation (A\ref{ax:UP}), the SSFP (A\ref{ax:SSFP}), the WSFP (A\ref{ax:WSFP}), the likelihood principle (A\ref{ax:LP}), and the weak likelihood principle (A\ref{ax:WLP}).
    In addition, it satisfies the symmetry axiom (A\ref{ax:S}) if and only if all weights are equal, i.e., $w_1=w_2=\cdots=w_K = 1/K$. 
    Furthermore, for a pooling function $g$ the following statements are equivalent:
    \begin{enumerate}
    \renewcommand{\theenumi}{(\roman{enumi})}
    \renewcommand{\labelenumi}{(\roman{enumi})}
        \item \label{en:glinpool} $g$ is a linear pooling function;
        \item \label{en:gssfp} $g$ satisfies the SSFP (A\ref{ax:SSFP});
        \item \label{en:gwsfpzpp} $g$ satisfies the WSFP (A\ref{ax:WSFP}) and the ZPP (A\ref{ax:ZPP});
        \item \label{en:gwsfpup} $g$ satisfies the WSFP (A\ref{ax:WSFP}) and unanimity preservation (A\ref{ax:UP}).
    \end{enumerate}
\end{theorem}
The equivalence of \ref{en:glinpool}, \ref{en:gssfp}, and \ref{en:gwsfpzpp}  was first proven in \cite{mcconway1981marginalization} for pmfs and in \cite{genest1984pooling} for 
arbitrary probability
measures.
However, to the best of our knowledge, a proof for pdfs has not been provided so far.%
\footnote{
Note that the proof for arbitrary probability measures in \cite{genest1984pooling} does not imply the result for pdfs.
Indeed,
in our pdf framework, 
only probability measures that are absolutely continuous with respect to a fixed reference measure 
(usually the Lebesgue measure) are considered.
This implicates the following difference 
from the framework of \cite{genest1984pooling}: 
whereas we only assume that an axiom holds for all pdfs, \cite{genest1984pooling} assumes that it also holds for other probability measures such as, e.g., a Dirac measure.
Therefore,
if \cite{genest1984pooling} states that, e.g., the assumption \ref{en:gssfp} implies \ref{en:glinpool}, 
then this refers to a stronger version of \ref{en:gssfp}.}
In \cite{dietrichprobabilistic}, the equivalence of \ref{en:gwsfpup} and \ref{en:gwsfpzpp} was presented for pmfs. 
In Appendix~\ref{app:proof_linearpooling}, we give a proof of Theorem~\ref{th:linearpooling} for pdfs.

\begin{theorem}
    \label{th:genlinpool}
    The generalized linear pooling function in \eqref{eq:gen_linear_pooling} satisfies the 
    WSFP (A\ref{ax:WSFP}) and the weak likelihood principle (A\ref{ax:WLP}).
    Conversely, any pooling function that satisfies the WSFP (A\ref{ax:WSFP}) is a generalized linear pooling function.
    In addition, the generalized linear pooling function satisfies the symmetry axiom (A\ref{ax:S}) if and only if all weights except $w_0$ are equal, i.e., $w_1=w_2=\cdots=w_K$. 
\end{theorem}
The measure-theoretic equivalence of generalized linear pooling functions with possibly negative weights and pooling functions satisfying the  WSFP (A\ref{ax:WSFP}) was proven in \cite{genest1984pooling}.
However, in the case of the fusion of pdfs considered here, the generalized linear pooling functions cannot have negative weights.
We thus present a proof with the necessary adaptations in  Appendix~\ref{app:proof_genlinpool}.

We next turn to pooling functions that include multiplication of pdfs or of powers of pdfs. 
In this context, we restrict to positive opinion profiles, i.e., we assume that \eqref{eq:nonnegopprof} is satisfied.
Note that in this setting the ZPP (A\ref{ax:ZPP}) is not applicable
since $Q_k({\cal A})=0$ is not possible except for sets ${\cal A}$ of Lebesgue measure zero;
therefore, we will disregard
the ZPP in the following considerations.
\begin{theorem}
\label{th:loglin}
    The log-linear pooling function in \eqref{eq: log_linear_pooling} satisfies 
    unanimity preservation (A\ref{ax:UP}), the likelihood principle (A\ref{ax:LP}), the weak likelihood principle (A\ref{ax:WLP}), factorization preservation (A\ref{ax:FP}), external Bayesianity (A\ref{ax:EB}), and generalized Bayesianity (A\ref{ax:GB}).
    In addition, it satisfies the symmetry axiom (A\ref{ax:S}) if and only if all weights are equal, i.e., $w_1=w_2=\cdots=w_K = 1/K$. 
    Furthermore, for a pooling function $g$ the following statements are equivalent:
    \begin{enumerate}
    \renewcommand{\theenumi}{(\roman{enumi})}
    \renewcommand{\labelenumi}{(\roman{enumi})}
        \item \label{en:gloglinpool} $g$ is a log-linear pooling function;
        \item \label{en:glpeb} $g$ satisfies the likelihood principle (A\ref{ax:LP}) and external Bayesianity (A\ref{ax:EB});
        \item \label{en:gwlpupeb} $g$ satisfies unanimity preservation (A\ref{ax:UP}), the weak likelihood principle (A\ref{ax:WLP}), and external Bayesianity (A\ref{ax:EB}).
    \end{enumerate}
\end{theorem}
The equivalence of \ref{en:gloglinpool} and \ref{en:glpeb} was proven in \cite{genest1984characterization} and the equivalence of  \ref{en:gloglinpool} and \ref{en:gwlpupeb} in \cite{genest1986characterization}.
The remaining claimed axioms follow straightforwardly from the definition of the log-linear pooling function in \eqref{eq: log_linear_pooling}.

\begin{theorem}
    \label{th:genloglinpool}
    The generalized log-linear pooling function in \eqref{eq:gen_log_linear_pooling} satisfies 
    the weak likelihood principle (A\ref{ax:WLP}), factorization preservation (A\ref{ax:FP}), external Bayesianity (A\ref{ax:EB}), and generalized Bayesianity (A\ref{ax:GB}).
    In addition, it satisfies the symmetry axiom (A\ref{ax:S}) if and only if all weights except $w_0$ are equal, i.e., $w_1=w_2=\cdots=w_K$.
    Furthermore, for a pooling function $g$ the following statements are equivalent:
    \begin{enumerate}
    \renewcommand{\theenumi}{(\roman{enumi})}
    \renewcommand{\labelenumi}{(\roman{enumi})}
        \item \label{en:ggenloglinpool} $g$ is a generalized log-linear pooling function;
        \item \label{en:gwlpeb} $g$ satisfies the weak likelihood principle (A\ref{ax:WLP}) and external Bayesianity (A\ref{ax:EB}).
    \end{enumerate}
\end{theorem}
This characterization theorem was proven in \cite{genest1986characterization}.
We note that  the assumption of fusing pdfs (rather than general measures) is essential.
In particular, for pmfs axioms A\ref{ax:WLP} and A\ref{ax:EB} would imply only a ``modified'' generalized log-linear pooling function that may contain negative weights \cite{genest1986characterization}. 
In \cite{genest1986characterization}, one can also find a characterization of all pooling functions that satisfy external Bayesianity (A\ref{ax:EB}).
However, these pooling functions do not have a simple structure.

\begin{theorem}
    The H\"older pooling function in \eqref{eq:holderpooling} satisfies 
    unanimity preservation (A\ref{ax:UP}), the likelihood principle (A\ref{ax:LP}), and the weak likelihood principle (A\ref{ax:WLP}).
    In addition, it satisfies the symmetry axiom (A\ref{ax:S}) if and only if all weights are equal, i.e., $w_1=w_2=\cdots=w_K$.
\end{theorem}
The proof of this theorem is straightforward and thus omitted. 
Because the inverse-linear pooling function \eqref{eq:inversepooling} is a special case of the H\"older pooling function, it follows that it also satisfies A\ref{ax:UP}, A\ref{ax:LP}, and A\ref{ax:WLP}.


\begin{theorem}
\label{th:mult_pooling}
    The multiplicative pooling function in \eqref{eq: multiplicative_pooling} satisfies 
    the  symmetry axiom (A\ref{ax:S}), 
    the weak likelihood principle (A\ref{ax:WLP}), factorization preservation (A\ref{ax:FP}),  individualized Bayesianity (A\ref{ax:IB}), and generalized Bayesianity (A\ref{ax:GB}).
    Furthermore, for a pooling function $g$ the following statements are equivalent:
    \begin{enumerate}
    \renewcommand{\theenumi}{(\roman{enumi})}
    \renewcommand{\labelenumi}{(\roman{enumi})}
        \item \label{en:gmultpool} $g$ is a multiplicative pooling function with calibrating pdf $q_0$;
        \item \label{en:glib} $g$ satisfies individualized Bayesianity (A\ref{ax:IB}) and there exists a pdf $q_0(\btheta)$ such that $g[q_0, \dots, q_0](\btheta) = q_0(\btheta)$.
    \end{enumerate}
\end{theorem}
The claimed axioms follow straightforwardly from the definition of the pooling function.
A result similar to the equivalence of \ref{en:gmultpool} and \ref{en:glib} was proven for pmfs in \cite{dietrichprobabilistic}. 
We provide a proof for pdfs in Appendix~\ref{app:proof_mult_pooling}.

\begin{theorem}
    The generalized multiplicative pooling function in \eqref{eq:gen_mult_pooling} satisfies 
    the weak likelihood principle (A\ref{ax:WLP}), factorization preservation (A\ref{ax:FP}), and generalized Bayesianity (A\ref{ax:GB}).
    In addition, it satisfies the symmetry axiom (A\ref{ax:S}) if and only if all weights are equal, i.e., $w_1=w_2=\cdots=w_K$.
\end{theorem}
Again, the 
 claimed axioms follow straightforwardly from the definition of the pooling function.

\begin{theorem}
\label{th:dict}
    The dictatorship pooling function in \eqref{eq: dictatorship} satisfies the ZPP (A\ref{ax:ZPP}), unanimity preservation (A\ref{ax:UP}), the SSFP  (A\ref{ax:SSFP}), the  WSFP  (A\ref{ax:WSFP}), the likelihood principle (A\ref{ax:LP}), the weak likelihood principle (A\ref{ax:WLP}), independence preservation (A\ref{ax:IP}), factorization preservation (A\ref{ax:FP}), external Bayesianity (A\ref{ax:EB}), and generalized Bayesianity (A\ref{ax:GB}).
    Furthermore, for a pooling function $g$ the following statements are equivalent:
    \begin{enumerate}
    \renewcommand{\theenumi}{(\roman{enumi})}
    \renewcommand{\labelenumi}{(\roman{enumi})}
        \item \label{en:gdictpool} $g$ is a dictatorship pooling function;
        \item \label{en:gssfpip} $g$ satisfies the SSFP  (A\ref{ax:SSFP}) and independence preservation (A\ref{ax:IP});
        \item \label{en:gwsfpip} $g$ satisfies the WSFP  (A\ref{ax:WSFP}) and independence preservation (A\ref{ax:IP});
        \item \label{en:gssfpeb} $g$ satisfies the SSFP  (A\ref{ax:SSFP}) and external Bayesianity (A\ref{ax:EB});
        \item \label{en:gwsfpeb} $g$ satisfies the WSFP  (A\ref{ax:WSFP}) and external Bayesianity (A\ref{ax:EB});
        \item \label{en:gssfpgb} $g$ satisfies the SSFP  (A\ref{ax:SSFP}) and generalized Bayesianity (A\ref{ax:GB}).
    \end{enumerate}
\end{theorem}
Our statements regarding  the satisfied axioms follow easily from the definition of the dictatorship pooling function. 
The equivalence of \ref{en:gdictpool} and \ref{en:gssfpip} was proven in \cite[Theorem~3.1]{genest1984pooling}. 
In Appendix~\ref{app:proof_dict}, we strengthen this result and show that the WSFP---instead of the (stronger) SSFP---in combination with independence preservation suffices to axiomatically define the dictatorship pooling function, i.e., that \ref{en:gwsfpip} implies \ref{en:gssfpip}.
The equivalence of \ref{en:gdictpool} and \ref{en:gssfpeb} was proven in \cite{genest1984conflict}.
In fact, \cite{genest1984conflict} even states the equivalence of \ref{en:gdictpool} and \ref{en:gwsfpeb} by proving that the version of external Bayesianity considered in \cite{genest1984conflict} implies the ZPP.
However, our formulation of external Bayesianity assumes positive opinion profiles and thus the ZPP cannot be proven.
To close this gap, we further show in Appendix~\ref{app:proof_dict} that \ref{en:gwsfpeb} implies \ref{en:gssfpeb}.
Finally, we also show in Appendix~\ref{app:proof_dict} that \ref{en:gssfpgb} implies \ref{en:gdictpool}. 

We  note that the dictatorship pooling function is a special case of both the linear and log-linear pooling functions, when one of the weights is 1 and all the others are 0.
The fact that the dictatorship pooling function satisfies ten axioms  shows that 
a pooling function that satisfies many axioms
is not necessarily a useful pooling function.

Turning to the dogmatic pooling function, we first present a preliminary result that is proven in Appendix~\ref{app:prooflemgbwsfpdog}.
\begin{lemma}
\label{lem:gbwsfpdog}
    Assume that a pooling function $g$ satisfies the WSFP  (A\ref{ax:WSFP}) and  generalized Bayesianity (A\ref{ax:GB}). 
    Then $g$ is either a dogmatic pooling function or a dictatorship pooling function. 
\end{lemma}
The following characterization of the dogmatic pooling function now follows easily.
\begin{theorem}
\label{th:dogm}
    The dogmatic pooling function in \eqref{eq: dogmatic_pooling_function} satisfies the symmetry axiom (A\ref{ax:S}), the  WSFP  (A\ref{ax:WSFP}), the weak likelihood principle (A\ref{ax:WLP}), and generalized Bayesianity (A\ref{ax:GB}).
    Conversely, any pooling function that satisfies the symmetry axiom (A\ref{ax:S}), the  WSFP  (A\ref{ax:WSFP}), and generalized Bayesianity (A\ref{ax:GB}) is a dogmatic pooling function.
\end{theorem}
It is obvious that the dogmatic pooling function satisfies the stated axioms. 
The converse follows because by Lemma~\ref{lem:gbwsfpdog}  the pooling function must be either a dogmatic pooling function or a dictatorship pooling function, but of these only the dogmatic pooling function is symmetric.

Based on the theorems above, we can 
establish an implication structure for the different axioms from Section~\ref{subsec: axioms_opinion_pooling}, 
which indicates which axioms imply which other axioms.
 To formalize this structure, we will designate the set of all pooling functions that satisfy Axiom~$i$ as $\mathcal{F}_i$. 
The next theorem states the currently known implications.
Venn diagrams representing the implication structure are presented in Fig.~\ref{fig:three graphs}.

\begin{figure*}[tbh]
\centering
    \begin{subfigure}[b]{0.45\textwidth}
         \centering
        \begin{tikzpicture}
        \draw[rounded corners] (0, 0) rectangle (6, 4) {};
        \node at (5.7, 0.2)   (WLP) {$\mF_{\ref{ax:WLP}}$};
        \draw[rounded corners] (0.5, 0.5) rectangle (3, 3.5) {};
        \node at (2.7, 0.7)   (LP) {$\mF_{\ref{ax:LP}}$};
        \draw[rounded corners] (2, 1) rectangle (5, 3) {};
        \node at (4.7, 1.2)   (WSFP) {$\mF_{\ref{ax:WSFP}}$};
        \draw[rounded corners] (1, 2.5) rectangle (4, 5) {};
        \node at (3.65, 4.8)   (GB) {$\mF_{\ref{ax:GB}}$};
        \draw[rounded corners,pattern=south east lines, pattern color=gray] (2, 2.5) rectangle (3, 6) {};
        \node at (2.65, 5.8)   (EB) {$\mF_{\ref{ax:EB}}$};
        \draw[rounded corners,pattern=south west lines, pattern color=gray] (2, 1.5) rectangle (3, 3) {};
        \node at (2.7, 1.7)   (SSFP) {$\mF_{\ref{ax:SSFP}}$};
        \draw[rounded corners] (3.15, 3.25) rectangle (3.75, 3.75) {};
        \node at (3.45, 3.5)   (IB) {$\mF_{\ref{ax:IB}}$};
        \node at (2.5, 2.75)   (Dict) {Dict};
        \node at (3.5, 2.75)   (Dogm) {Dogm};
        \end{tikzpicture}
        \caption{}
        \label{fig:venn1}
     \end{subfigure}
     \begin{subfigure}[b]{0.45\textwidth}
        \centering
        \begin{tikzpicture}
        \draw[rounded corners] (0, 0) rectangle (5, 4) {};
        \node at (4.7, 0.2)   (WLP) {$\mF_{\ref{ax:WLP}}$};
        \draw[rounded corners] (0.5, 0.5) rectangle (3, 3.5) {};
        \node at (2.7, 0.7)   (LP) {$\mF_{\ref{ax:LP}}$};
        \draw[rounded corners] (2, 1) rectangle (4, 3) {};
        \node at (3.7, 1.2)   (WSFP) {$\mF_{\ref{ax:WSFP}}$};
        \draw[rounded corners,pattern=south east lines, pattern color=gray] 
        (-0.5, 2.5) -- (-0.5, 2.9) -- (1.5, 2.9) -- (1.5, 5) -- (3, 5) -- (3, 2.5) -- cycle {};
        \node at (2.65, 4.8)   (IP) {$\mF_{\ref{ax:IP}}$};
        \draw[rounded corners] (-1, 1.5) rectangle (3, 5) {};
        \node at (-0.7, 4.8)   (ZPP) {$\mF_{\ref{ax:ZPP}}$};
        \draw[rounded corners,pattern=south west lines, pattern color=gray] (2, 1.5) rectangle (3, 3) {};
        \node at (2.7, 1.7)   (SSFP) {$\mF_{\ref{ax:SSFP}}$};
        \draw[rounded corners] (-2, 1.5) rectangle (3, 3) {};
        \node at (-1.7, 1.7)   (UP) {$\mF_{\ref{ax:UP}}$};
        \node at (2.5, 2.75)   (Dict) {Dict};
        \end{tikzpicture}
        \caption{}
        \label{fig:venn2}
     \end{subfigure}
\caption{Venn diagrams representing
the implication structure for the axioms from Section~\ref{subsec: axioms_opinion_pooling}:
(a) A\ref{ax:SSFP}--A\ref{ax:WLP} and A\ref{ax:EB}--A\ref{ax:GB} as well as intersections resulting in dictatorship (Dict) or dogmatic (Dogm) pooling functions; 
(b) A\ref{ax:ZPP}--A\ref{ax:IP} as well as intersections resulting in dictatorship (Dict) pooling functions.
Note that the diagrams illustrate the currently known implications, and some regions that appear non-empty in the diagrams may actually be empty sets. 
For better visibility, the sets $\mF_{\ref{ax:SSFP}}$, $\mF_{\ref{ax:IP}}$, and $\mF_{\ref{ax:EB}}$ are highlighted by different line-patterns.}
\label{fig:three graphs}
\end{figure*}
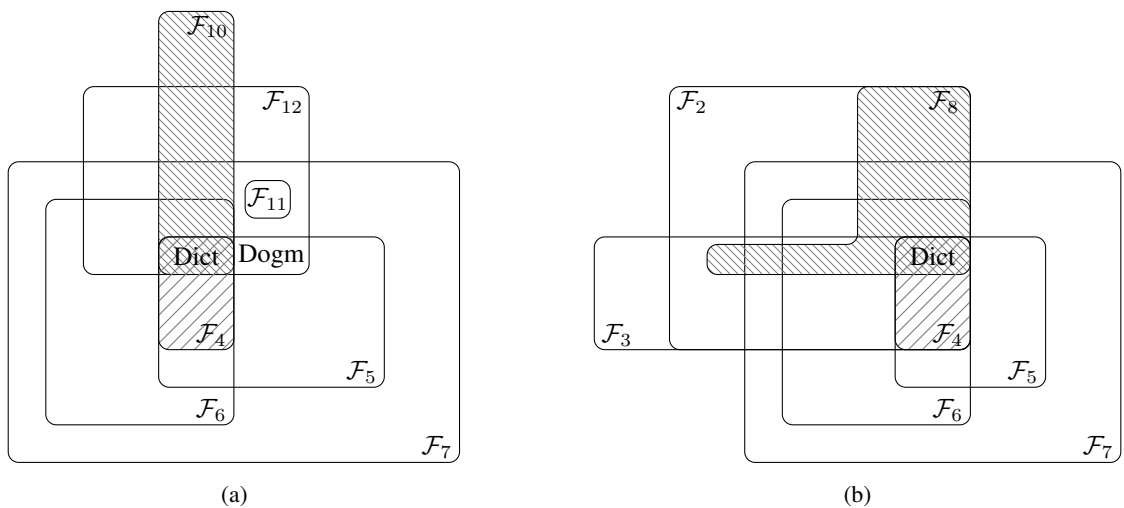

\begin{theorem}
\label{th:implstruct}
    For the axioms introduced in Section~\ref{subsec: axioms_opinion_pooling}, the following implications hold:
    \begin{enumerate}
    \renewcommand{\theenumi}{(\roman{enumi})}
    \renewcommand{\labelenumi}{(\roman{enumi})}
        \item \label{en:implSSFP}The SSFP (A\ref{ax:SSFP}) implies the 
    ZPP (A\ref{ax:ZPP}), unanimity preservation (A\ref{ax:UP}), the WSFP (A\ref{ax:WSFP}), the likelihood principle (A\ref{ax:LP}), and the weak likelihood principle (A\ref{ax:WLP}), i.e., 
        $\mF_{\ref{ax:SSFP}} \subseteq \mF_{\ref{ax:ZPP}} \cap \mF_{\ref{ax:UP}} \cap \mF_{\ref{ax:WSFP}} \cap \mF_{\ref{ax:LP}} \cap \mF_{\ref{ax:WLP}}$.
        Furthermore, $\mF_{\ref{ax:SSFP}} = \mF_{\ref{ax:ZPP}} \cap \mF_{\ref{ax:WSFP}} = \mF_{\ref{ax:UP}} \cap \mF_{\ref{ax:WSFP}}$.
        \item \label{en:implWSFP}The WSFP (A\ref{ax:WSFP}) implies the weak likelihood principle (A\ref{ax:WLP}), i.e., 
        $\mF_{\ref{ax:WSFP}} \subseteq \mF_{\ref{ax:WLP}}$.
        \item \label{en:implLP}The likelihood principle (A\ref{ax:LP}) implies the weak likelihood principle (A\ref{ax:WLP}), i.e., 
        $\mF_{\ref{ax:LP}} \subseteq \mF_{\ref{ax:WLP}}$.
        \item \label{en:implIP}Independence preservation (A\ref{ax:IP}) implies the  ZPP (A\ref{ax:ZPP}), i.e., $\mF_{\ref{ax:IP}} \subseteq \mF_{\ref{ax:ZPP}}$.
        \item \label{en:implIB}Individualized Bayesianity (A\ref{ax:IB}) implies generalized Bayesianity (A\ref{ax:GB}), i.e., 
        $\mF_{\ref{ax:IB}} \subseteq  \mF_{\ref{ax:GB}}$.
    \end{enumerate}
\end{theorem}
Most of these implications follow from our earlier theorems.
For completeness, we provide a proof of, or references for,
all implications in Appendix~\ref{app:thimplstruct}.

\section{The Optimization Approach} 
\label{sec: pop_optim_appr}
In the previous section, we identified pooling functions that satisfy certain axioms.
An alternative approach to establishing pooling functions for probabilistic opinion pooling is the optimization approach. 
Here, a pooling function is obtained by minimizing the weighted average of some discrepancy measure between the pdfs of the $K$ agents, $q_1(\btheta),\ldots, q_K(\btheta)$, and the aggregate pdf $q(\btheta)$. The underlying idea is to make the aggregate pdf as similar as possible to all the agent pdfs simultaneously. As we will see, the obtained $q(\btheta)$ turns out to be some sort of average of the 
agent pdfs $q_1(\btheta),\ldots, q_K(\btheta)$.

One class of discrepancy measures that can be considered are $f$-divergences. 
For a convex function $f\colon \mathbb{R}^+\rightarrow \mathbb{R}$ with $f(1)=0$, the $f$-divergence between two pdfs $q_k(\btheta)$ and $\varphi(\btheta)$ with common domain $\Theta$ is defined as \cite{csiszar1964informationstheoretische,csiszar1967information,vajda1972f,Sason_2018}
\begin{equation}
  \label{eq: f_divergence_pi_q}
    \mathcal{D}_f(q_k\|\varphi)=\int_{\Theta} \varphi(\btheta)f\bigg(\frac{q_k(\btheta)}{\varphi(\btheta)}\bigg)  \mathrm{d}\btheta.
\end{equation}
The fusion of the agent
pdfs $q_1(\btheta),\ldots,q_K(\btheta)$ can then be based on defining the aggregate pdf $q(\btheta)$ as the pdf 
that minimizes a weighted average of 
$f$-divergences:\footnote{
This 
minimization establishes a
conceptual link to a central problem in the field of robust 
hypothesis testing, namely, the identification of
a vector of ``least favorable'' pdfs within a given set of hypothesized pdfs. 
For two pdfs, this problem can be shown to be equivalent to the joint minimization of all $f$-divergences \eqref{eq: f_divergence_pi_q} 
for all twice differentiable convex functions $f$
\cite{veeravalli1994minimax, fauss2021minimax}.
The solution to this minimization can be interpreted as the pdfs that are maximally
similar within the set of hypothesized pdfs, which means that a statistical test between the respective pdfs is ``as hard as possible.'' 
It is interesting that an interpretation as a maximally similar pdf holds for both the optimization approach to pdf fusion and robust hypothesis testing.
}
\begin{equation}
    \label{eq: f_divergence_fusion_optimization}
    q =\argmin_{\varphi\in\mathcal{P}} \sum_{k=1}^K w_k  \mathcal{D}_f(q_k\|\varphi),
\end{equation}
where 
the weights satisfy $(w_1,\ldots,w_K)\in\mathcal{S}_K$.
In what follows, we consider some specific $f$-divergences and derive the associated pooling functions defined by \eqref{eq: f_divergence_fusion_optimization}.
These results are summarized in Table~\ref{tab: frechet_means}.

\begin{table*}[t]
\centering
\begin{tabular}{@{}ccccc@{}}
\toprule
Pooling Function &
$f(x)$ &
$\mathcal{D}_f(q_k\|\varphi)$ &
$\warp(x)$ &
$\|\warp(q_k)-\warp(\varphi)\|_2^2$ 
\\
\midrule
Linear: \ $q(\btheta) = \sum_{k=1}^K w_kq_k(\btheta)$ & 
$x \log x$ &
$\int_{\Theta} q_k(\btheta)\log\left(\frac{q_k(\btheta)}{\varphi(\btheta)}\right) \mathrm{d}\btheta$ & 
$x$ & 
$\int_{\Theta} \left(q_k(\btheta)-\varphi(\btheta)\right)^2 \mathrm{d}\btheta$ 
\\
Log-linear: \ $q(\btheta) \propto\prod_{k=1}^K \left(q_k(\btheta)\right)^{w_k}$ & 
$-\log x$ &
$\int_{\Theta} \varphi(\btheta)\log\left(\frac{\varphi(\btheta)}{q_k(\btheta)}\right) \mathrm{d}\btheta$ & 
$\log x$  & 
$\int_{\Theta} \left(\log q_k(\btheta)-\log \varphi(\btheta)\right)^2 \mathrm{d}\btheta$             
\\
Inverse-linear: \ $q(\btheta)\propto\left(\sum_{k=1}^K  \frac{w_k}{q_k(\btheta)}\right)^{-1}$ &
$\frac{1-x}{2 x}$ &
$\int_{\Theta}\frac{(q_k(\btheta)-\varphi(\btheta))^2}{q_k(\btheta)}\mathrm{d}\btheta$ & 
$\frac{1}{x}$   & 
$\int_{\Theta} \left(\frac{1}{q_k(\btheta)}-\frac{1}{\varphi(\btheta)}\right)^2 \mathrm{d}\btheta$   
\\
H\"{o}lder: \ $q(\btheta)\propto\left(\sum_{k=1}^K w_k(q_k(\btheta))^\alpha\right)^{1/\alpha}$ & 
$\frac{x^{\alpha}-1}{\alpha(\alpha-1)}$ &
$\frac{1}{\alpha (\alpha-1)}\int_{\Theta} \varphi(\btheta) \frac{(q_k(\btheta))^\alpha-(\varphi(\btheta))^\alpha}{(\varphi(\btheta))^\alpha}  \mathrm{d}\btheta$ & $x^\alpha$ &
$\int_{\Theta} \left((q_k(\btheta))^{\alpha}-(\varphi(\btheta))^{\alpha}\right)^2 \mathrm{d}\btheta$ 
\\ 
\bottomrule
\end{tabular}
\caption{Optimization-based definition of pooling functions: some pooling functions along with the underlying $f$-divergence $\mathcal{D}_{f}(q_k\|\varphi)$ and squared distance function $d^2(q_k, \varphi)=\|\warp(q_k)-\warp(\varphi)\|_2^2$ used in the optimization problems in \eqref{eq: f_divergence_fusion_optimization} and \eqref{eq: relaxed_Frechet_means}, respectively.}
\label{tab: frechet_means}
\end{table*}

\subsection{Kullback-Leibler Divergence}
\label{subsec: forward_KL_minimization_pop}
For $f(x)=x\log x$, the $f$-divergence is the Kullback-Leibler divergence (KLD) \cite{liese2006divergences}
\begin{equation}
    \label{eq: forward_KLD}
    \mathcal{D}_{\textrm{KL}}(q_k\|\varphi)=\int_{\Theta} q_k(\btheta)\log\left(\frac{q_k(\btheta)}{\varphi(\btheta)}\right) \mathrm{d}\btheta.
\end{equation}
Under this choice of divergence, the pooling function that solves the optimization problem in \eqref{eq: f_divergence_fusion_optimization} is the linear pooling function in \eqref{eq: linear_pooling}: 
\begin{theorem}
Let $f(x)=x\log x$ (i.e., $\mathcal{D}_f(q_k\|\varphi)=\mathcal{D}_{\rm KL}(q_k\|\varphi)$) and $(w_1,\ldots,w_K)\in\mathcal{S}_K$. 
Then, the solution to the optimization problem in \eqref{eq: f_divergence_fusion_optimization} is 
\begin{equation*}
    q(\btheta)=\sum_{k=1}^K w_k q_k(\btheta).
\end{equation*}
\end{theorem}
A proof of this theorem can be found in \cite{abbas2009kullback}. The proof is based on the fact that minimizing the weighted average of KLDs is equivalent to minimizing the cross-entropy
\begin{equation*}
H(q_{\rm mix}, \varphi)=  -\int_{\Theta} q_{\rm mix}(\btheta) \log\left(\varphi(\btheta)\right)\mathrm{d}\btheta 
\end{equation*}
between the mixture pdf $q_{\rm mix}(\btheta)=\sum_{k=1}^K w_kq_k(\btheta)$ and the pdf $\varphi\in\mathcal{P}${.} That is,
\begin{equation*}
    \argmin_{\varphi\in\mathcal{P}} \sum_{k=1}^K w_k  \mathcal{D}_{\rm KL}(q_k\|\varphi) = \argmin_{\varphi\in\mathcal{P}} H(q_{\rm mix},\varphi){.}
\end{equation*}
The cross-entropy $H(q_{\rm mix},\varphi)$ is minimized if and only if
$q_{\rm mix}(\btheta)$ and $\varphi(\btheta)$ are equal. This follows from the fact that
$H(q_{\rm mix},\varphi)$
is equal to 
the sum of the KLD between $q_{\rm mix}(\btheta)$ and $\varphi(\btheta)$ and the differential entropy of $q_{\rm mix}(\btheta)$ \cite[Chapter~2]{murphy2012machine}, i.e.,
\begin{equation*}
    H(q_{\rm mix}, \varphi) = {\cal D}_{\rm KL}(q_{\rm mix}\|\varphi) - \int_{\Theta} q_{\rm mix}(\btheta)\log(q_{\rm mix}(\btheta)){\rm d}\btheta.
\end{equation*}
Hence, $H(q_{\rm mix},\varphi)$ is minimized if and only if ${\cal D}_{\rm KL}(q_{\rm mix}\|\varphi)$ is minimized, which implies that
$\varphi(\btheta)
=q_{\rm mix}(\btheta)$.

\subsection{Reverse Kullback-Leibler Divergence}
Next, consider $f(x)=-\log x$. In this case, the $f$-divergence corresponds to the KLD whose arguments are reversed with respect to \eqref{eq: forward_KLD} \cite{van2014renyi}, i.e.,
\begin{equation*}
    \mathcal{D}_{\textrm{KL}}(\varphi\|q_k)=\int_{\Theta} \varphi(\btheta)\log\left(\frac{\varphi(\btheta)}{q_k(\btheta)}\right) \mathrm{d}\btheta.
\end{equation*}
We refer to $\mathcal{D}_{\textrm{KL}}(\varphi\|q_k)$ as the \textit{reverse} KLD. 
For the reverse KLD, the solution to the optimization problem in \eqref{eq: f_divergence_fusion_optimization} is the log-linear pooling function in \eqref{eq: log_linear_pooling}:
\begin{theorem}
Let $f(x)=-\log x$ (i.e., $\mathcal{D}_f(q_k\|\varphi)=\mathcal{D}_{\rm KL}(\varphi\| q_k)$)  and $(w_1,\ldots,w_K)\in\mathcal{S}_K$. 
Then, the solution to the optimization problem in \eqref{eq: f_divergence_fusion_optimization} is
\begin{equation*}
    \label{eq: optimal_KLD_reverse}
    q(\btheta) =  c\prod_{k=1}^K \left(q_k(\btheta)\right)^{w_k},
\end{equation*}
where $c=1 \big/ \!\int_{\Theta}\prod_{k=1}^K\left(q_k(\btheta)\right)^{w_k} \mathrm{d}\btheta$.
\end{theorem}
A
proof of this theorem can be found in \cite{abbas2009kullback} and \cite{dedecius2016sequential}. The idea behind the proof is to derive a lower bound on the weighted average of reverse KLDs using Jensen's inequality and then to show that the lower bound is achieved if and only if \eqref{eq: log_linear_pooling} is satisfied.

\subsection{$\alpha$-Divergences}
\label{sec: pop_optim_appr_hellinger}
We have shown that both the linear and log-linear pooling functions
can be derived using the optimization approach involving the KLD or reverse KLD, respectively. These two results can be extended to an entire family of divergences and a corresponding family of pooling
functions that are both
parameterized by a 
real parameter $\alpha$. Indeed, let us consider the $f$-divergence $\mathcal{D}_f(q_k\|\varphi)$ induced by
\begin{align*}
    f(x)&=f_\alpha(x)
    \triangleq \frac{x^\alpha-1}{\alpha(\alpha-1)},
\end{align*}
where $x>0$ and $\alpha\in\mathbb{R}\setminus{\{0, 1\}}$. This yields the family of 
$\alpha$-divergences defined as \cite{zhu1995information, minka2005divergence, cichocki2011generalized}
\begin{align}
    \mathcal{D}_\alpha (q_k\|\varphi)
    & \triangleq \mathcal{D}_{f_\alpha} (q_k\|\varphi)
    \label{eq: df_alpha_def} \\
    & =\frac{1}{\alpha(\alpha-1)}\int_{\Theta} \varphi(\btheta)
    \frac{(q_k(\btheta))^\alpha-(\varphi(\btheta))^\alpha}{(\varphi(\btheta))^\alpha}  \mathrm{d}\btheta \,.
    \label{eq: hellinger_divergence}
\end{align}
We remark that the $\alpha$-divergence 
equals the so-called Hellinger divergence up to a scaling factor and is also a one-to-one transformation of the R\'{e}nyi divergence 
\cite{Sason_2018}. 
Using the optimization approach for the  $\alpha$-divergences, we obtain the $\alpha$-parameterized family of H\"older pooling functions in \eqref{eq:holderpooling}. 
As noted earlier, this family comprises the linear, log-linear, and inverse-linear pooling functions as special cases.
\begin{theorem} \label{th:optimal_hellinger}
Let $f(x)=f_\alpha(x)=\frac{x^\alpha-1}{\alpha(\alpha-1)}$ (i.e., $\mathcal{D}_f(q_k\|\varphi)=\mathcal{D}_\alpha(q_k\|\varphi)$) with  $\alpha\in\mathbb{R}\setminus{\{0, 1\}}$ and $(w_1,\ldots,w_K)\in\mathcal{S}_K$. 
Then, the solution to the optimization problem in \eqref{eq: f_divergence_fusion_optimization} is 
\begin{equation}
    \label{eq: optimal_hellinger}
    q(\btheta)= c\left(\sum_{k=1}^K w_k(q_k(\btheta))^\alpha\right)^{1/\alpha},
\end{equation}
where $c= 1 \big/ \!\int_{\Theta}\big(\sum_{k=1}^K w_k(q_k(\btheta))^\alpha\big)^{1/\alpha}\mathrm{d}\btheta$.
\end{theorem}
Although this result was mentioned in \cite[Fig. 1]{garg2004generalized}, to the best of our knowledge, a proof does not exist in the literature. 
We provide a proof in Appendix~\ref{appendix: proof_weighted_hellinger_simplex}.

In the limiting case $\alpha\rightarrow 0$, the H\"{o}lder pooling function  \eqref{eq: optimal_hellinger} becomes the log-linear pooling function (weighted geometric average) 
in \eqref{eq: log_linear_pooling}, while for $\alpha = 1$ it equals the linear pooling function (weighted arithmetic average)
in \eqref{eq: linear_pooling}. 
These results are consistent with the fact that $\lim_{\alpha\rightarrow 0}  \mathcal{D}_\alpha(q_k\|\varphi)=\mathcal{D}_{\rm KL}(\varphi\|q_k)$ and $\lim_{\alpha\rightarrow 1}\mathcal{D}_\alpha(q_k\|\varphi)=\mathcal{D}_{\rm KL}(q_k\|\varphi)$ \cite{minka2005divergence}. 
For $\alpha=-1$,   the H\"{o}lder pooling function  \eqref{eq: optimal_hellinger} becomes the inverse-linear pooling function \eqref{eq:inversepooling}.
Furthermore,
the $\alpha$-divergence in the case $\alpha=2$ is (up to a scaling factor 2) equal to the Pearson $\chi^2$-divergence \cite{pearson1900x, Sason_2018}
\begin{align*}
    \chi^2(q_k, \varphi) 
    & \triangleq \int_{\Theta}\frac{(q_k(\btheta)-\varphi(\btheta))^2}{\varphi(\btheta)}\mathrm{d}\btheta
    \\
    & = \int_{\Theta}\varphi(\btheta)\frac{(q_k(\btheta))^2-(\varphi(\btheta))^2}{(\varphi(\btheta))^2}\mathrm{d}\btheta\,.
\end{align*}
The corresponding H\"{o}lder pooling function  \eqref{eq: optimal_hellinger} is thus
\begin{equation*}
    q(\btheta)= c\left(\sum_{k=1}^K w_k(q_k(\btheta))^2\right)^{1/2},
\end{equation*}
where $c= 1 \big/ \!\int_{\Theta}\big(\sum_{k=1}^K w_k(q_k(\btheta))^2\big)^{1/2}\mathrm{d}\btheta$.

\subsection{Reverse $\alpha$-Divergences}
\label{sec:optim_appr_revalpha}

As for the KLD, one can exchange the order of $q_k$ and $\varphi$ in the $\alpha$-divergence in \eqref{eq: hellinger_divergence}.
Again, this is equivalent to changing to a different $f$-divergence.
More precisely, it is stated in \cite[eq.~(1.13)]{csiszar1967information} (see also \cite[Prop.~2]{Sason_2018})  that 
\begin{equation}
\label{eq:invfdivergence}
    \mathcal{D}_f(\varphi \| q_k) = \mathcal{D}_{f^*}(q_k\|\varphi),
\end{equation}
where $f^*(x)= x f(1/x)$. 
Based on this result, we show in Appendix~\ref{app:revalphadiv} that 
\begin{equation*}
    \mathcal{D}_{\alpha}(\varphi \| q_k)
= \mathcal{D}_{\alpha^*}(q_k\|\varphi), 
\end{equation*}
where $\alpha^*=1-\alpha$.
%
%
Thus, Theorem~\ref{th:optimal_hellinger} implies the following result.
\begin{corollary}
The solution to the optimization problem 
\begin{equation}
    q =\argmin_{\varphi\in\mathcal{P}} \sum_{k=1}^K w_k  \mathcal{D}_{\alpha}(\varphi \| q_k)
    \label{eq:dalphainvedirection}
\end{equation}
is 
\begin{equation}
    q(\btheta)= c\left(\sum_{k=1}^K w_k(q_k(\btheta))^{\alpha^*}\right)^{1/\alpha^*},
    \label{eq:dalphainvedirectionsol}
\end{equation}
where $c= 1 \big/ \!\int_{\Theta}\big(\sum_{k=1}^K w_k(q_k(\btheta))^{\alpha^*}\big)^{1/\alpha^*}\mathrm{d}\btheta$ and  $\alpha^*=1-\alpha$.
\end{corollary}

In particular, the reverse $\alpha$-divergence for
$\alpha=2$ corresponds to the Pearson $\chi^2$-divergence in the reverse direction, i.e., $\chi^2(\varphi, q_k) = \int_{\Theta}\frac{(q_k(\btheta)-\varphi(\btheta))^2}{q_k(\btheta)}\mathrm{d}\btheta$.
In this case, $\alpha^*=1-2=-1$ and the corresponding H\"older pooling function  \eqref{eq:dalphainvedirectionsol}
is the  inverse-linear pooling function \eqref{eq:inversepooling}.

\subsection{Symmetric Discrepancy Measures}
\label{subsec: frechet_means}
As previously mentioned, the optimization approach defines pooling functions by minimizing a weighted average of discrepancy measures between the agent pdfs and the aggregate pdf. So far, our focus has been on minimizing a weighted average of $f$-divergences, where our
choices of $f$ 
yielded \emph{asymmetric} discrepancy measures. 
Through this approach, we 
derived pooling functions that are the weighted arithmetic, geometric, harmonic, and H\"{o}lder averages
of the agent pdfs. Interestingly, these fusion rules can also be derived using an alternative formulation, where the goal is to minimize a weighted average of \emph{symmetric} discrepancy measures (distance functions). Let $d(q_k, \varphi)$ be a symmetric function expressing
a distance between the $k$th agent pdf $q_k(\btheta)$ and the pdf $\varphi(\btheta)$, where symmetric
means that $d(q_k, \varphi)=d(\varphi, q_k)$. Then, we can define the aggregate pdf to be the solution to the following optimization problem:
\begin{equation}
\label{eq: frechet_mean_definition}
q(\btheta) = \argmin_{\varphi\in\mathcal{P}} \sum_{k=1}^K w_k d^2(q_k, \varphi),
\end{equation}
where $(w_1,\ldots,w_K)\in\mathcal{S}_K$. 
The resulting $q(\btheta)$ has been referred to as  \emph{Fr\'{e}chet mean} \cite{li2020arithmetic}. 

An important distance function is the $L_2$ distance function
defined as
\begin{equation}
    \label{eq: l2_distance_pdfs}
    \|q_k-\varphi\|_2 = \sqrt{\int_{\Theta} (q_k(\btheta)-\varphi(\btheta))^2 \,\mathrm{d}\btheta}.
\end{equation}
The linear pooling function can be obtained alternatively by minimizing a weighted average of squared $L_2$ distances:
\begin{theorem} 
\label{th:optimal_l2}
Let $d(q_k, \varphi)=\|q_k-\varphi\|_2$ and $(w_1,\ldots,w_K)\in\mathcal{S}_K$. Then, the solution to the optimization problem in \eqref{eq: frechet_mean_definition} is
\begin{equation*}
    q(\btheta)=\sum_{k=1}^K w_k q_k(\btheta).
\end{equation*}
\end{theorem}
This result was mentioned without proof in \cite[Fig. 1]{garg2004generalized}. We provide a proof in Appendix~\ref{appendix: proof_weighted_l2_simplex}. 

Unfortunately, for arbitrary distance functions $d(q_k, \varphi)$, an analytical solution to the optimization problem in \eqref{eq: frechet_mean_definition} does not exist. This is due to the difficulty in satisfying the constraint $\varphi\in{\cal P}$, which ensures
that the obtained aggregate pdf $q(\btheta)$ is a valid pdf.  To overcome this difficulty, following \cite{li2020arithmetic}, we can instead solve 
the unconstrained version of the optimization problem in \eqref{eq: frechet_mean_definition}, i.e.,
\begin{equation}
    \label{eq: relaxed_Frechet_means}
    \tilde q(\btheta) = \argmin_{\varphi} \sum_{k=1}^K w_k d^2(q_k, \varphi), 
\end{equation}
and then normalize the result, i.e.,
\begin{equation*}
    q(\btheta) = \frac{\tilde q(\btheta)}{\int_{\Theta} \tilde q(\btheta') d\btheta'}.
\end{equation*}
However, we emphasize that the obtained aggregate pdf $q(\btheta)$ is generally different from the solution of the constrained optimization problem in \eqref{eq: frechet_mean_definition}.

Using this unconstrained approach,
the minimization of the $L_2$ distance function \eqref{eq: l2_distance_pdfs}
results again in
the linear pooling function \cite{li2020arithmetic}.
Here, the solution  satisfies the constraint $q\in \mathcal{P}$ without explicitly enforcing it.
Furthermore, 
the log-linear \cite{li2020arithmetic}, inverse-linear, and H\"{o}lder pooling functions can be derived in an analogous manner using suitable distance functions. 
We can arrive at all of these results and many more  in a unified manner by considering the general class of distance functions $d(q_k, \varphi)$ defined as
\begin{equation}
    \label{eq: general_distance_function}
    \|\warp(q_k)-\warp(\varphi)\|_2 = \sqrt{\int_{\Theta} (\warp(q_k(\btheta))-\warp(\varphi(\btheta)))^2 \,\mathrm{d}\btheta},
\end{equation}
where $\warp\colon (0, \infty) \to (a,b)$ with $a\in \mathbb{R}\cup \{-\infty\}$ and $b\in \mathbb{R}\cup \{\infty\}$ is an invertible function. 
Solving the optimization problem  \eqref{eq: relaxed_Frechet_means} for the distance functions \eqref{eq: general_distance_function} leads to the rich class of pooling functions stated by the following result.
\begin{theorem}
\label{th: generalized_pooling_frechet}
Let $d(q_k, \varphi)=\|\warp(q_k)-\warp(\varphi)\|_2$ and $(w_1,\ldots,w_K)\in\mathcal{S}_K$. Then, the solution to the optimization problem in \eqref{eq: relaxed_Frechet_means} is
\begin{equation}
    \label{eq: general_unnormalized_pooling_function}
    \tilde q(\btheta)=\warp^{-1}\Bigg(\sum_{k=1}^K w_k \warp(q_k(\btheta))\Bigg).
\end{equation}
\end{theorem}
A proof is provided in Appendix~\ref{appendix: proof_weighted_general_frechet}, and  the functions $\warp$ leading to the linear, log-linear, inverse-linear, and H\"older pooling functions are listed in Table~\ref{tab: frechet_means}.
Note that the solution $\tilde q(\btheta)$ in  \eqref{eq: general_unnormalized_pooling_function} is always nonnegative because the domain of $\warp$ is $(0,\infty)$.

\section{Gaussian Densities}
\label{sec: opinion_pooling_with_gaussians}
In Sections~\ref{sec: probabilistic_opinion_pooling} and \ref{sec: pop_optim_appr}, we discussed a variety of pooling functions that can be used to fuse the pdfs of several agents into a single aggregate pdf. We now consider the practically important special case where the opinions of the agents are represented by Gaussian pdfs. {That is}, we assume that 
\begin{equation}
    \label{eq: agents_are_gaussian}
    q_k(\btheta)=\mathcal{N}(\btheta; \bmu_{q_k},\bSigma_{q_k}), \quad k=1,\ldots,K,
\end{equation}
where $\mathcal{N}(\btheta;\bmu_{q_k},\bSigma_{q_k})$ denotes a multivariate Gaussian pdf with mean ${\bmu_{q_k}=\mathbb{E}_{q_k}[\btheta]}$ and covariance matrix ${\bSigma_{q_k}=\mathbb{E}_{q_k}[(\btheta-\bmu_{q_k})(\btheta-\bmu_{q_k})^\intercal]}$. {An important aspect of the Gaussian case is the fact that each agent pdf $q_k(\btheta)$ is completely characterized by its first- and second-order moments $\bmu_{q_k}$ and $\bSigma_{q_k}$.}
\subsection{Linear Pooling}
\label{subsec: fusion_gaussian_linear_pooling}
The fusion of Gaussian pdfs using the linear pooling function in \eqref{eq: linear_pooling} results in an aggregate pdf that is a mixture of Gaussians, i.e., 
\begin{equation}
    \label{eq: aggregate_mixture_gauss}
    q(\btheta) = \sum_{k=1}^K w_k \mathcal{N}(\btheta;\bmu_{q_k},\bSigma_{q_k}). 
\end{equation}
A convenient property in this context is that the expected value of a function $h(\btheta)$ with respect to {the pdf $q(\btheta)$ in} \eqref{eq: linear_pooling} is the weighted average of the expected {values} of $h(\btheta)$ with respect to {the} agent pdfs $q_1(\btheta),\ldots,q_K(\btheta)$, i.e., $\mathbb{E}_q[h(\btheta)]=\sum_{k=1}^K w_k\mathbb{E}_{q_k}[h(\btheta)]$. This implies that the mean of {the} aggregate pdf in \eqref{eq: aggregate_mixture_gauss}, $\bmu_{q}=\mathbb{E}_q[\btheta]$, {is simply the weighted average of the agent means, i.e.,}
\begin{equation}
    \label{eq: mean_linear_pooling}
    \bmu_{q}=\sum_{k=1}^K w_k\bmu_{q_k}.
\end{equation}
Similarly, the covariance matrix of the aggregate pdf in \eqref{eq: aggregate_mixture_gauss}, $\bSigma_{q}= \mathbb{E}_q[(\btheta-\bmu_{q})(\btheta-\bmu_{q})^\intercal]$, is obtained as \cite{malladi1997new}
\begin{equation}
    \label{eq: cov_matrix_aggregate}
    \bSigma_{q}=\sum_{k=1}^Kw_k\left(\bSigma_{q_k}+(\bmu_{q_k}-\bmu_{q})(\bmu_{q_k}-\bmu_{q})^\intercal\right).
\end{equation}

Thus, the mean and covariance matrix of the aggregate pdf {$q(\btheta)$} can be calculated easily from the agent means and covariance matrices. 
 This is useful from a practical perspective because it provides a way for obtaining an estimate of the parameters (e.g., mean)  as well as a measure of  uncertainty for that estimate (e.g., covariance matrix). It is important to note, however, that since $q(\btheta)$ is a mixture of Gaussians and, therefore, is non-Gaussian, it is not fully characterized by its mean and covariance matrix. Indeed, a mixture of Gaussians can have properties that a Gaussian cannot have, including heavy tails, multiple modes, and nonzero skewness \cite{wang2015multivariate}.

In the case that the agent pdfs are Gaussian, the connection of linear opinion pooling to model averaging established in Section~\ref{subsec: fusion_rules_axioms_lop} extends to an estimation technique in the Kalman filtering literature called \emph{multiple model adaptive estimation} (MMAE) \cite{hanlon2000multiple}. MMAE {uses} a bank of Kalman filters to estimate an {unknown state (time-varying parameter)}, where each Kalman filter assumes a distinct model {describing the state's time evolution and its relation} to the observed data. In this context, $\bmu_{q_k}$ is the local state estimate provided by the $k$th Kalman filter {at a given time}, while $\bSigma_{q_k}$ is the covariance {of that estimate}. The local state estimates are then combined according to \eqref{eq: mean_linear_pooling} to obtain a final state estimate $\bmu_{q}$, whose covariance $\bSigma_{q}$ is determined by \eqref{eq: cov_matrix_aggregate}. Here, the weight $w_k$ {equals} the posterior probability of the model {assumed by} the $k$th Kalman filter.

\subsection{Log-linear Pooling}
\label{subsec: fusion_gaussian_log_linear_pooling}
The fusion of the 
Gaussian pdfs $q_k(\btheta)$ {in \eqref{eq: agents_are_gaussian}} by the log-linear pooling function in \eqref{eq: log_linear_pooling} results in an aggregate pdf that is also Gaussian, i.e.,
\begin{equation*}
    q(\btheta)=\mathcal{N}(\btheta; \bmu_{q}, \bSigma_{q}),
\end{equation*}
with mean vector
\begin{equation}
    \label{eq: mean_fusion_log_linear_gaussians}
    \bmu_{q} = \left(\sum_{k=1}^K w_k\bSigma_{q_k}^{-1}\right)^{-1}\sum_{j=1}^K w_j\bSigma_{q_j}^{-1}\bmu_{q_j}
\end{equation}
and covariance matrix
\begin{equation}
    \label{eq: covariance_fusion_log_linear_gaussians}
    \bSigma_{q} = \left(\sum_{k=1}^K w_k\bSigma_{q_k}^{-1}\right)^{-1}.
\end{equation}
Unlike the case of linear pooling, since the aggregate pdf $q(\btheta)$ is Gaussian, it is unimodal and symmetric {about the mean $\bmu_{q}$}, and it is moreover fully characterized by the mean $\bmu_{q}$ and covariance $\bSigma_{q}$. 

There is a strong link between log-linear pooling of Gaussian pdfs and a second-order fusion method called \emph{covariance intersection} \cite{julier1997non, hurley2002information}, which is often employed {in distributed (decentralized) Kalman filter implementations} \cite{chong2001convex, hu2011diffusion, deng2012sequential}. In the covariance intersection context, there are $K$ agents, each {of which uses its own local observations to form a local estimate of an unknown quantity $\btheta$}. The goal of covariance intersection is to fuse the {local} estimates in a way that does not underestimate the overall  covariance of the fused estimate. Let $\bmu_{q_k}$ be the local estimate of the $k$th agent, whose covariance is denoted by $\bSigma_{q_k}$. The fused state estimate $\bmu_{q}$ is determined according to \eqref{eq: mean_fusion_log_linear_gaussians}, while the corresponding  covariance matrix $\bSigma_{q}$ is {given by} \eqref{eq: covariance_fusion_log_linear_gaussians}. The weights $w_1,\ldots,w_K$ {used in \eqref{eq: mean_fusion_log_linear_gaussians} and \eqref{eq: covariance_fusion_log_linear_gaussians}} are typically chosen to minimize the determinant or the trace of $\bSigma_{q}$ \cite{hurley2002information}.

\subsection{Other Pooling Functions}
\label{subsec: fusion_gaussian_alpha_pooling}
Finally, we consider the H\"{o}lder pooling functions.  
The normalization factor $c$ in  the H\"{o}lder pooling function in \eqref{eq:holderpooling}  for general $\alpha\in\mathbb{R}\setminus{\{0, 1\}}$,
involves an intractable integral and cannot be evaluated, even if the agent pdfs are Gaussian. 
Therefore, typically, the aggregate pdf $q(\btheta)$ resulting from the 
H\"{o}lder pooling function is only known up to a normalization factor. 
Computing expected values with respect to $q(\btheta)$ 
would require the use of numerical integration techniques such as the trapezoidal quadrature rule or Monte Carlo methods \cite{robert2013monte}. Because numerical integration techniques are plagued by the \emph{curse of dimensionality} \cite{hinrichs2019curse}, computing expectations with respect to $q(\btheta)$ under the
H\"{o}lder pooling function becomes challenging when the dimension of $\btheta$ is large.

To illustrate the behavior of the linear and log-linear
pooling functions, and to demonstrate the effect of different choices of $\alpha$ on  
the H\"{o}lder pooling function, we present in Fig.~\ref{fig: alpha_fusion} simulation results for two different sets of $K=2$ Gaussian agent pdfs $q_k(\theta)$
with $\theta\in\mathbb{R}$. We used the trapezoidal quadrature rule to compute the normalization factor {of the aggregate pdf}. Fig.~\ref{fig: alpha_fusion_visual_diff_mean} shows the fusion of two Gaussian pdfs with different means but the same variance. In this case, the value of $\alpha$ in the H\"{o}lder pooling function controls the multimodality of the aggregate pdf, in the sense that smaller (larger) values of $\alpha$ attenuate (enhance) the modes of the agent pdfs in the aggregate pdf. Fig.~\ref{fig: alpha_fusion_visual_diff_var} shows the fusion of two Gaussian pdfs with the same mean but different variances. In this case, the value of $\alpha$ controls the shape
of the tails of the aggregate pdf, in the sense that smaller (larger) values of $\alpha$ lead to less heavy (heavier) tails. 

\begin{figure*}
     \centering
     \begin{subfigure}[b]{\textwidth}
         \centering
         \includegraphics[width=\textwidth]{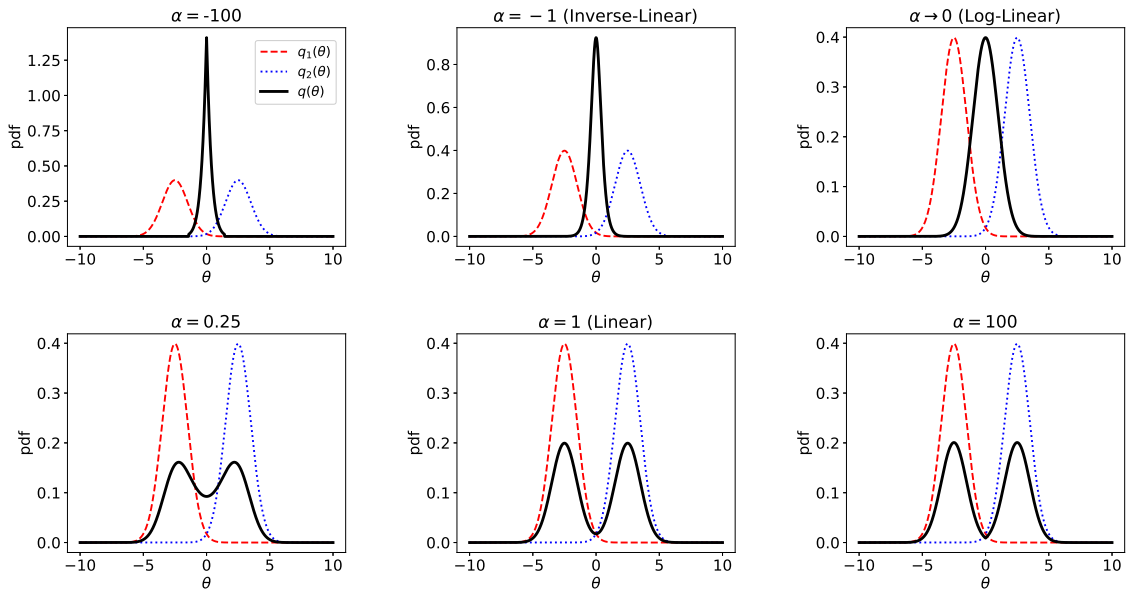}
         \vspace{-9mm}
         \caption{Different means, same variance.}
         \label{fig: alpha_fusion_visual_diff_mean}
     \end{subfigure}
     \hfill
     \begin{subfigure}[b]{\textwidth}
         \centering
          \includegraphics[width=\textwidth]{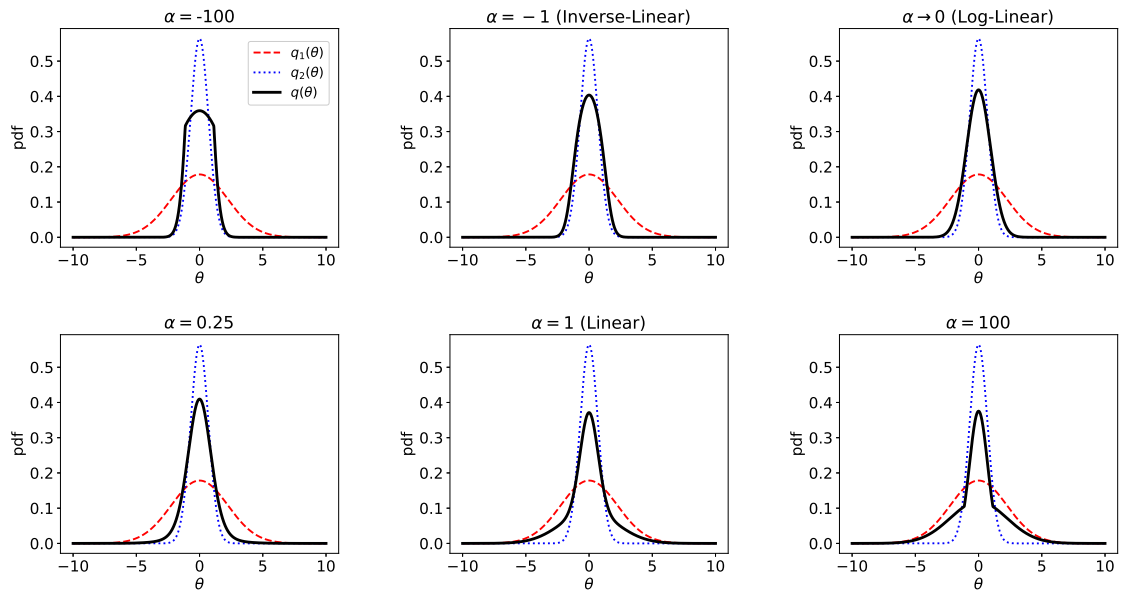}
          \vspace{-9mm}
         \caption{Same mean, different variances.}
        \label{fig: alpha_fusion_visual_diff_var}
     \end{subfigure}
     \vspace{-1mm}
     \caption{Results of H\"{o}lder pooling of two pdfs $q_1(\theta)$ and $q_2(\theta)$ using weights $w_1=w_2=0.5$ and different values of $\alpha$. The pdfs are defined as follows: (a) $q_1(\theta)=\mathcal{N}(\theta; -2.5, 1)$ and $q_2(\theta)=\mathcal{N}(\theta; 2.5, 1)$ (different means, same variance), and (b) $q_1(\btheta)=\mathcal{N}(\theta; 0, 5)$ and $q_2(\theta)=\mathcal{N}(\theta; 0, 0.5)$ (same mean, different variances). 
     Note that $\alpha=-1$, $\alpha\rightarrow 0$, and $\alpha= 1$ correspond to the inverse-linear, log-linear, and linear pooling functions, respectively.}
     \label{fig: alpha_fusion}
\end{figure*}

\section{Choosing the Pooling Parameters} 
\label{sec: choosing_the_weights}
An important consideration in opinion pooling is the choice of the parameters involved in the various pooling functions. While most of our
discussion will be in regard to the weights $w_1,\ldots, w_K$,
we also provide some insight on 
the choice of the parameter $\alpha$ in the H\"{o}lder pooling function. 

The problem of choosing the weights in probabilistic opinion pooling is well researched. The simplest approach is to assign equal weights to all agents, i.e., $w_k={1}/{K}$ for all $k$ \cite{wallis2005combining}. However, alternative strategies for assigning weights have been proposed for linear \cite{genest1990allocating,degroot1991optimal,clemen2008comment} and log-linear \cite{heskes1998selecting, rufo2012log, de2015choosing} pooling. These strategies are usually based on solving some optimization problem, where the definition of the objective function depends on how the weights are interpreted by the fusion center. In some instances, the optimization of the weights solely depends on the agent pdfs. In other scenarios, weight assignment takes into consideration data that are observed at the fusion center, and is based on a Bayesian interpretation involving likelihood functions or posterior distributions. These data-dependent methods have also been extended to the sequential case, where observed data are streamed and the weights are updated when new data become available \cite{genest1990allocating}. 

In the following, we describe several options for choosing the weights in linear and log-linear pooling. We focus on methods that do not assume that the fusion center has observed any data. 
At this point, it 
is important to emphasize that in both the axiomatic and optimization approaches to probabilistic opinion pooling, the weights $w_k$ were assumed fixed, i.e., not dependent on the agent pdfs $q_k(\btheta)$. If, on the other hand, the weights are chosen adaptively according to an additional optimization procedure that involves the agent pdfs $q_k(\btheta)$, then this implies a deviation from the strict mathematical framework established by both the axiomatic and optimization approaches. For example, the linear pooling function with adaptively chosen weights is no longer linear in the agent pdfs $q_k(\btheta)$.

\subsection{Linear Pooling}
\label{subsec: choice_of_weights_lop}
The problem of assigning the weights in the linear pooling function has been considered in many works; see \cite{genest1990allocating} for a review. One approach is based on interpreting the weight $w_k$ as a \emph{veridical probability}, i.e, as the probability that the true pdf of $\btheta$ is $q_k(\btheta)$ \cite{bunn1981two}. Accordingly, $w_k$ is chosen to equal a prior or posterior estimate of that probability. This approach is connected to the model-averaging view of linear opinion pooling mentioned in Section~\ref{subsec: fusion_rules_axioms_lop}, since in \eqref{eq: model_averaging}, $P(M_k)$ equals the probability that the model of the $k$th agent, $M_k$, is the correct one. When data are considered, the weights $w_k$ equal the posterior probabilities of the models $M_k$, and this is exactly how they are assigned in the MMAE algorithm mentioned in Section~\ref{subsec: fusion_gaussian_linear_pooling}  \cite{malladi1997new,hanlon2000multiple}.

Alternatively, the weights can be assigned according to the predictive performance of each agent by viewing the weights as \emph{outranking probabilities} \cite{bunn1982synthesis}. In this view, $w_k$ is the probability that predictions made based on $q_k(\btheta)$ will outperform the predictions based on the pdfs of the other agents. This rationale for choosing the weights requires consideration of data and a mechanism for assessing the predictive performance of the agents. 

Another idea is to interpret the weights as a measure of distance \cite{barlow1986combination}. Based on this interpretation, agents that have ``middle of the road" opinions are assigned higher weights, while those that have more extreme (controversial) opinions are assigned lower weights.
The opposite strategy would in principle also be possible, namely, giving more weight to controversial opinions.
Such  weight assignments can be achieved by assigning a nonnegative score $\gamma_k$ to each agent pdf $q_k(\btheta)$. 
For example, one can choose the score $\gamma_k$ to be inversely related to the \emph{maximum discrepancy} between agent $k$ and the other agents, i.e.,
\begin{equation}
    \label{eq: agents_max_discrepancy}
    \gamma_k = \frac{1}{\underset{j\in\{1,\ldots,K\}}{\max}\mathcal{D}_{\rm KL}(q_k\|q_{j})} \geq 0,\quad k=1,\ldots,K.
\end{equation}
Here, the KLD is used to measure the discrepancy between agents, although other divergences can be used instead. The weight of each agent is then obtained as a normalized version of 
$\gamma_k$,
i.e., 
\begin{equation*}
    w_k = \frac{\gamma_k}{\sum_{j=1}^K \gamma_j}, \quad k=1,\ldots,K. 
\end{equation*}

Finally, there are also iterative schemes for weight assignment, where each agent considers itself to be a fusion center and assigns weights to all the other agents. The weights are iteratively updated until a consensus is reached. In \cite{degroot1974reaching}, the weight vector of each agent is updated by multiplying it by a transition matrix, and under some conditions 
a consensus is reached asymptotically. The work \cite{carvalho2013consensual} builds on this idea, but updates the weights according to how closely the agent pdfs agree, using a scoring function similar to \eqref{eq: agents_max_discrepancy}. 

\subsection{Log-linear Pooling}
\label{subsec: choice_of_weights_llop}
The choice of the weights in the log-linear pooling function has been considered less intensely in the literature. Some of the aforementioned methods for linear opinion pooling can also be applied to log-linear opinion pooling; for example, the scoring rule in \eqref{eq: agents_max_discrepancy} is still reasonable. Moreover, as mentioned in Section~\ref{subsec: fusion_gaussian_log_linear_pooling}, for Gaussian agent pdfs, log-linear pooling corresponds to the covariance intersection fusion method. Here, the weights can be chosen using schemes proposed in the covariance intersection literature, such as minimizing the trace or determinant of the covariance matrix in \eqref{eq: covariance_fusion_log_linear_gaussians} \cite{hurley2002information}. 

One criterion proposed in the literature that does not require the consideration of data is the \emph{minimum KLD} criterion \cite{de2015choosing}. If there is no basis for determining the reliability of each agent, one can choose the weights such that the aggregate pdf is maximally close to all the agent pdfs simultaneously. This is the criterion that was used in Section~\ref{sec: pop_optim_appr} to find an optimal pooling function for given weights $w_k$.  Similarly to Section~\ref{subsec: forward_KL_minimization_pop}, the criterion can be formulated as a minimization of the average of the KLDs between the agent pdfs $q_k(\btheta)$ and the aggregate pdf $q(\btheta)$. Introducing the weight vector $\w\triangleq(w_1,\ldots,w_K)$, the optimal weights are defined as
\begin{equation*}
    \w^\star = \argmin_{\w \in \mathcal{S}_K} L(\w),
\end{equation*}
with
\begin{align*}
    L(\w) 
    & \triangleq \frac{1}{K}\sum_{k=1}^K \mathcal{D}_{\rm KL}(q_k \| q)
    \\
    & = \frac{1}{K}\sum_{k=1}^K \mathcal{D}_{\rm KL}\bigg( q_k \, \bigg\| \,c \prod_{\ell=1}^K \left(q_\ell(\btheta)\right)^{w_\ell}\bigg),
\end{align*}
where expression \eqref{eq: log_linear_pooling} was inserted  for $q(\btheta)$.
Using the KLD definition \eqref{eq: forward_KLD}, one can obtain \cite{de2015choosing} 
\begin{equation}
    \label{eq: sum_of_KLDs_worked_out}
    L(\w) = -\log c(\w)+ \frac{1}{K}\sum_{k=1}^K \sum_{j\neq k} 
    w_j\mathcal{D}_{\rm KL}(q_k \| q_j).
\end{equation}
Here, $c(\w)$ is the normalization factor in \eqref{eq:log_linear_pooling_normaliz}, which depends on $\w$. The objective function $L(\w)$ is convex, since the first term $-\log c(\w)$ is convex \cite{rufo2012log} and the second term is a linear function of $\w$. Therefore, tools from convex optimization can be used to compute the optimal weight vector $\w^\star$. We note that the minimum KLD criterion would also be a reasonable criterion for use with other pooling functions; however, the expression for $L(\w)$ in \eqref{eq: sum_of_KLDs_worked_out} applies specifically to the log-linear pooling function. 

Furthermore, we remark that if the average of the reverse KLDs, i.e.,
\begin{align}
    \tilde{L}(\w) 
    & \triangleq\frac{1}{K}\sum_{k=1}^K \mathcal{D}_{\rm KL}(q\|q_k)
    \notag
    \\
    & = \frac{1}{K}\sum_{k=1}^K \mathcal{D}_{\rm KL}\bigg( c \prod_{\ell=1}^K \left(q_\ell(\btheta)\right)^{w_\ell}\, \bigg\|\, q_k \bigg),
    \label{eq: sum_of_revKLDS}
\end{align}
was chosen as the objective function to be minimized, the optimal weights 
would be given by 
\begin{equation}
    \label{eq:wstarrevKL}
    \argmin_{\w \in \mathcal{S}_K} \tilde{L}(\w) = \bigg(\frac{1}{K},\ldots,\frac{1}{K}\bigg)\,.
\end{equation}
Indeed, let $q^\star(\btheta)$ be defined by \eqref{eq: log_linear_pooling} with weights $\w=(\frac{1}{K},\ldots,\frac{1}{K})$, i.e.,
\begin{equation}
\label{eq:qstarrevKL1}
    q^\star(\btheta) \triangleq c \prod_{k=1}^K \left(q_k(\btheta)\right)^{1/K}.
\end{equation}
By Theorem~\ref{eq: optimal_KLD_reverse},  $q^\star(\btheta)$
minimizes the objective function in \eqref{eq: sum_of_revKLDS} over all pdfs $\varphi$, i.e., 
\begin{equation}
\label{eq:qstarrevKL}
    q^\star = \argmin_{\varphi \in \mathcal{P}} \frac{1}{K}\sum_{k=1}^K \mathcal{D}_{\rm KL}(\varphi \| q_k)\,.
\end{equation}
Thus, we have 
\begin{align*}
    \tilde{L}(\w) 
    & \stackrel{\hidewidth \eqref{eq: sum_of_revKLDS}\hidewidth}= \frac{1}{K}\sum_{k=1}^K \mathcal{D}_{\rm KL}\bigg( c \prod_{\ell=1}^K \left(q_\ell(\btheta)\right)^{w_\ell} \bigg\| \,q_k\bigg) 
    \\
    & \stackrel{\hidewidth \eqref{eq:qstarrevKL} \hidewidth}\geq \frac{1}{K}\sum_{k=1}^K \mathcal{D}_{\rm KL}(q^\star \| q_k)
    \\
    & \stackrel{\hidewidth \eqref{eq:qstarrevKL1} \hidewidth}= \frac{1}{K}\sum_{k=1}^K \mathcal{D}_{\rm KL}\bigg( c \prod_{\ell=1}^K \left(q_\ell(\btheta)\right)^{1/K} \bigg\| \,q_k\bigg) 
    \\
    & \stackrel{\hidewidth \eqref{eq: sum_of_revKLDS}\hidewidth}= \tilde{L}\big(\big(\tfrac{1}{K},\ldots,\tfrac{1}{K}\big)\big)\,.
\end{align*}
Thus, for any $\w$, $ \tilde{L}(\w)$ is lower bounded by $\tilde{L}\big(\big(\tfrac{1}{K},\ldots,\tfrac{1}{K}\big)\big)$. 
This proves \eqref{eq:wstarrevKL}.

Other approaches minimize an alternative KLD criterion \cite{heskes1998selecting, rufo2012bayesian} or take a Bayesian approach by specifying a prior distribution over the weights \cite{de2015choosing}. However, these approaches require data to be available, and usually lead to closed form solutions only if the prior pdfs take the form of conjugate priors 
for the considered likelihood functions. 


\subsection{H\"{o}lder Pooling}
\label{subsec: choice_of_parameters_holder}

In addition to the 
weights, the 
parameter $\alpha$ involved in the H\"{o}lder pooling function in \eqref{eq:holderpooling} strongly
impacts the resulting aggregate pdf, as was demonstrated in Fig.~\ref{fig: alpha_fusion}.
An appropriate choice of $\alpha$ depends on the application at hand.
For example, in risk assessment, the choice of $\alpha$ is relevant to a quantification of uncertainty.
In a 
risk-averse scenario, one may opt to choose a larger value of $\alpha$, or at least a positive $\alpha$. Indeed, for any $\alpha > 0$, the supports of the agent pdfs are preserved by the fusion in the 
sense that the support of the aggregate pdf equals the union of the supports of all the agent 
pdfs. Furthermore, a larger $\alpha$ tends to yield a larger uncertainty in the aggregate pdf. This latter characteristic is related to the fact, shown in Fig.~\ref{fig: alpha_fusion},  that a larger $\alpha$ tends to promote multimodality and/or heavy-tailed properties in the aggregate pdf.

If one instead chooses a small value of $\alpha$,
then components of different agent pdfs that occur at different $\btheta$ locations will have substantially less influence on the aggregate pdf. This means, in particular, that an ``outlier behavior'' of one agent will tend to be attenuated in the fusion process. 
Furthermore, for $\alpha = 0$, if the pdf of any agent $k$ is zero for some $\btheta_0$, i.e., $q_k(\btheta_0)=0$, this implies that the aggregate pdf is also zero at $\btheta_0$ irrespectively of the values of the other agent pdfs.
This ``veto property'' can be problematic in certain situations.
Finally, for $\alpha < 0$, H\"{o}lder pooling is restricted to 
positive opinion profiles, which implies that all agents have to agree on the support $\Theta$ of $\btheta$.

H\"{o}lder pooling appears to be practically relevant mostly for values of $\alpha$ in $[0,1]$.
Here, we recall that $\alpha = 0$ and $\alpha = 1$ correspond to the log-linear pooling function and the linear pooling function, respectively; furthermore,  values of $\alpha$ between $0$ and $1$ correspond to pooling functions whose characteristics---e.g., with regard to multimodality and tail decay---are intermediate between those of the linear and log-linear pooling functions, as demonstrated by Fig.~\ref{fig: alpha_fusion}.
An application where this observation is potentially
relevant was considered in Section \ref{sec: examples_TT}.

\section{The Supra-Bayesian Framework}
\label{sec: fusion_of_distributions}

\newcommand{\meamsc}{y}
\newcommand{\meam}{\mathbf{\meamsc}}
\newcommand{\meamr}{\meam}
\newcommand{\sufstsc}{t}
\newcommand{\sufst}{\mathbf{\sufstsc}}
\newcommand{\sufstr}{\sufst}
\newcommand{\lik}{p}
\newcommand{\pri}{p}
\newcommand{\poste}{p}
\newcommand{\myproj}{\mathbf{S}}
\newcommand{\E}{\mathbb{E}}
\newcommand{\tlik}{\lambda}
\newcommand{\parfam}{\psi}
\newcommand{\Ourinv}{\mathbf{V}}
\newcommand{\ourinv}{\mathbf{v}^{\intercal}}
\newcommand{\ourinvt}{\mathbf{v}}

The supra-Bayesian framework is fundamentally different from the approaches discussed so far.
In this section, we consider $\btheta$ to be a random variable with prior pdf $\pri(\btheta)$ and 
assume that the fusion center  follows a Bayesian update rule to derive a posterior pdf.
Our focus will  be on scenarios where
observations (data) that depend on $\btheta$ are obtained by the agents but are not known to the fusion center.
We will start this section with
a formulation using conditionally independent observations,
and extend from there to the general supra-Bayesian framework.

\subsection{Agents Collecting Conditionally Independent Observations} \label{sec:exbayes}
Let us consider a scenario with $K$ agents where 
each agent $k\in \{1, \dots, K\}$ obtains observations $\meam_k\in \mathbb{R}^{d_{\meamsc_k}}$.
These observations are statistically related to the random vector
$\btheta \in \mathbb{R}^{d_{\theta}}$ according to the ``local'' likelihood functions $\lik(\meam_k \,\vert\,  \btheta)$.
We consider the observations fixed (i.e., already observed) and emphasize the dependence of $\lik(\meam_k \,\vert\,  \btheta)$ on $\btheta$ by writing the local likelihood functions as $\ell_k(\btheta) \triangleq \lik(\meam_k \,\vert\,  \btheta)$.
Furthermore, each agent has access to the prior pdf $\pri(\btheta)$ and is thus able to calculate its local posterior 
$\pi_k(\btheta) \triangleq \poste(\btheta \,\vert\, \meam_k)$ according to Bayes' rule:
\begin{equation} \label{eq:locpost}
    \pi_k(\btheta)
    = \poste(\btheta \,\vert\, \meam_k)
    = \frac{\ell_k(\btheta) \pri(\btheta) }%
    {\int_{\Theta} \ell_k(\btheta') \pri(\btheta')\,\mathrm{d}\btheta' }\,.
\end{equation}
We further assume that the local observations $\meam_k$  are conditionally independent given $\btheta$ for all $k\in \{1, \dots, K\}$.
This implies that the ``global'' likelihood function $\ell(\btheta) \triangleq \lik(\meam \,\vert\,  \btheta)$ for $\meam \triangleq [\meam_1^\intercal, \dots, \meam_K^\intercal]^{\intercal}$
factors into the local likelihood functions $\ell_k(\btheta) = \lik(\meam_k \,\vert\,  \btheta)$, i.e., 
\begin{equation} \label{eq:indepmeams}
    \ell(\btheta)
    = \prod_{k=1}^K \ell_k(\btheta) \,.
\end{equation}

The task of the fusion center is to fuse the local posteriors $\pi_k(\btheta)$ provided by the agents into an aggregate (fused) pdf 
 $g[\pi_1,\ldots,\pi_K](\btheta)$.
We assume that the fusion center is aware of  the statistical properties of all the observations (i.e., the conditional pdfs $\lik(\meam_k \,\vert\,  \btheta)$) and of the prior $\pri(\btheta)$ but does not have access to the observations $\meam_k$ directly.
From a Bayesian viewpoint, the best possible fusion result is the posterior pdf of $\btheta$ using the observations from all the agents
as represented by the total observation vector $\meam$, i.e., $\poste(\btheta \,\vert\, \meam)$. We will refer to $\poste(\btheta \,\vert\, \meam)$ as
\textit{oracle posterior}
because the fusion center does not know the observations $\meam$ explicitly.
Nevertheless, the following result shows that the fusion center is still able to fuse the $\pi_k(\btheta)$ into the oracle posterior $\poste(\btheta \,\vert\, \meam)$.
\begin{theorem}
\label{th:bayesindep}
Let $\btheta$ be a random vector
with prior $\pri(\btheta)$. 
Furthermore, let the local observations $\meam_1$, $\dots$, $\meam_K$ given $\btheta$ be mutually independent and distributed according to $\lik(\meam_k \,\vert\,  \btheta)$.
Then the global posterior $\poste(\btheta \,\vert\, \meam)$ with $\meam = [\meam_1^\intercal, \dots, \meam_K^\intercal]^{\intercal}$ is given by
\begin{equation}\label{eq:fusionindepbayes}
    \poste(\btheta \,\vert\, \meam)
    = g[\pi_1,\ldots,\pi_K](\btheta)
    =
    c \,(\pri(\btheta))^{1-K}  \prod_{k=1}^K \pi_k(\btheta)\,,
\end{equation}
where $c = 1/ \int_{\Theta} (\pri(\btheta))^{1-K}  \big(\prod_{k=1}^K \pi_k(\btheta)\big) \mathrm{d}\btheta$ is a normalization factor and
the local posteriors $\pi_k(\btheta)$ are given by \eqref{eq:locpost}.
\end{theorem}
\begin{proof}
    We recall that $\ell(\btheta) = \lik(\meam \,\vert\,  \btheta)$,
    $\ell_k(\btheta) = \lik(\meam_k \,\vert\,  \btheta)$, 
    and $\pi_k(\btheta) = \poste(\btheta \,\vert\, \meam_k)$.
    We have by Bayes' rule that 
    \begin{align}
        \poste(\btheta \,\vert\, \meam)
        & \propto  \pri(\btheta) \ell(\btheta) 
        \notag \\
        & \stackrel{\hidewidth \eqref{eq:indepmeams} \hidewidth }=   \pri(\btheta)\prod_{k=1}^K \ell_k(\btheta)
        \notag \\
        & \propto \pri(\btheta)  \prod_{k=1}^K \frac{\poste(\btheta \,\vert\, \meam_k)}{\pri(\btheta)}
        \notag \\
        & = (\pri(\btheta))^{1-K}  \prod_{k=1}^K \pi_k(\btheta)\,. 
        \label{eq:nonnormindepBayes}
    \end{align}
    Since $\poste(\btheta \,\vert\, \meam)$ is a conditional pdf, normalizing the function in \eqref{eq:nonnormindepBayes} gives \eqref{eq:fusionindepbayes}.
\end{proof}
The fusion rule in \eqref{eq:fusionindepbayes}
is recognized to be an instance of the  multiplicative pooling function in~\eqref{eq: multiplicative_pooling}, where the calibrating pdf $q_0(\btheta)$ is given by the prior $\pri(\btheta)$.
Thus, Theorem~\ref{th:bayesindep} states that the multiplicative pooling function applied to the local posteriors $\pi_k(\btheta)$
provides the oracle posterior $\poste(\btheta \,\vert\, \meam)$ in the case of conditionally independent local observations $\meam_k$.

We note that the fusion center could calculate $\poste(\btheta \,\vert\, \meam)$ equally well from the local  likelihood functions $\ell_k(\btheta) = \lik(\meam_k \,\vert\,  \btheta)$, rather than from
the local posteriors $\pi_k(\btheta)$.
Indeed, the fusion rule \eqref{eq:fusionindepbayes} can be interpreted as first dividing each local posterior $\pi_k(\btheta)$ by the prior $\pri(\btheta)$ to obtain the local  likelihood function $\lik(\meam_k \,\vert\,  \btheta)$, 
then fusing (multiplying) the local likelihood functions into the global likelihood function $\lik(\meam \,\vert\,  \btheta)$,  
and finally multiplying  by the prior to obtain the oracle posterior
$\poste(\btheta \,\vert\, \meam)$.
(This corresponds to reading the proof of Theorem~\ref{th:bayesindep} bottom up.)
Thus, in the present scenario of conditionally independent observations $\meam_k$, 
the agents may also communicate 
their local likelihood functions $\ell_k(\btheta)$ to
the fusion center, rather than their posteriors $\pi_k(\btheta)$.

\subsection{The Supra-Bayesian Framework and Local Statistics}
\label{sec:abstract-bayesian}

To generalize the scenario considered in Section~\ref{sec:exbayes}, we take the perspective of the fusion center.
In our Bayesian setting, the fusion center aims to calculate the posterior distribution of 
$\btheta$, given all the information it has access to.
However, in more general settings than the case of conditionally independent observations discussed in Section~\ref{sec:exbayes}, we cannot expect that the fusion center is able to calculate the oracle posterior $\poste(\btheta \,\vert\, \meam)$. This is because the fusion center does not have direct access to the observations $\meam_k$; rather,
it observes the effect of the $\meam_k$ only indirectly through the local posteriors 
$\pi_k(\btheta) = \poste(\btheta \,\vert\, \meam_k)$.
In addition to knowing the local posteriors 
$\pi_k(\btheta)$, the fusion center is aware of the prior $\pri(\btheta)$ and the  conditional distribution $\lik(\meam \,\vert\, \btheta)$ (as a function of $\meam$ and $\btheta$, not for the fixed, observed $\meam$). 
Finally, the fusion center knows how the agents derive their local posteriors $\pi_k(\btheta) = \poste(\btheta \,\vert\, \meam_k)$ given their local observations $\meam_k$,
i.e., it is aware that each $\pi_k$ depends on 
$\btheta$ in a well-defined probabilistic way, namely, by the two-step process of first generating a random 
$\meam_k$ given $\btheta$ according to the conditional pdf $\lik(\meam_k \,\vert\,  \btheta)$
and then deriving $\pi_k$ from $\meam_k$ using \eqref{eq:locpost}.

This setup can be formulated generically via an abstract ``observation model'' $\lik(\pi_1, \dots, \pi_K\,\vert\,  \btheta)$ in which the local posteriors $\pi_K$ are considered as ``observations.''
This approach is known in the literature as the supra-Bayesian model \cite{winkler1968consensus,morris1977combining}.
In this abstract setting, we no longer have to consider the intermediate step of generating the observations $\meam_k$ given $\btheta$, and we no longer have to assume that the local pdfs  $\pi_k$ are generated as posteriors. 
Instead, we directly define an observation model by specifying a probability distribution over the local pdfs $\pi_k$ given 
$\btheta$.
Thus, at the fusion center,  
the local pdfs of all agents are considered as observations, i.e.,
as random objects whose statistical relation to $\btheta$ is described by the ``likelihood function'' $\lik(\pi_1, \dots, \pi_K\,\vert\,  \btheta)$.
As always in Bayesian settings, we need in addition some prior $\pri(\btheta)$.
By Bayes' theorem, we can then express the posterior distribution of 
$\btheta$ given the local pdfs $\pi_k$ as
\begin{equation}
    \label{eq:bayesabstract}
    \poste(\btheta \,\vert\, \pi_1, \dots, \pi_K) = 
    \frac{\lik(\pi_1, \dots, \pi_K \,\vert\,  \btheta)\pri(\btheta) }%
    {\int_{{\Theta}}  \lik(\pi_1, \dots, \pi_K \,\vert\,  \btheta') \pri(\btheta')\,\mathrm{d}\btheta'},
\end{equation}
which is 
considered to be the supra-Bayesian fusion result, also to be referred to as ``supra-Bayesian posterior.''

For any given $\btheta$, $\lik(\pi_1, \dots, \pi_K \,\vert\,  \btheta)$ is a probability distribution over the infinite-dimensional space of functions that is given by the
$K$-fold Cartesian product of the space of all  pdfs $\mathcal{P}$.
It is both mathematically and practically convenient to restrict to a finite-dimensional subset of this space. 
Indeed,  a finite-dimensional parameterization is very often used in practical applications. In particular, if $\pi_k$ depends deterministically on some finite-dimensional observation $\meam_k$, then $\pi_k$ is obviously restricted to a finite-dimensional subset.
Thus, we will hereafter assume that each $\pi_k$ depends deterministically and in a one-to-one manner on a finite-dimensional random vector $\sufst_k \in \mathbb{R}^{d_{\sufstsc_k}}$.
Then the probability distribution  $\lik(\pi_1, \dots, \pi_K \,\vert\,  \btheta)$ simplifies to a conventional conditional pdf $\lik(\sufst_1, \dots, \sufst_K \,\vert\,  \btheta)$. This finite-dimensional setting is formalized by the following definition.

\begin{definition}
\label{def:supra}
    A \emph{finite-dimensional supra-Bayesian model} for a parameter $\btheta \in  \Theta \subseteq \mathbb{R}^{d_{\theta}}$ consists of:    
   \begin{itemize}
        \item a prior pdf $p(\btheta)$;
   \vspace{1mm}
   \item a conditional pdf $\lik(\sufstr \,\vert\,  \btheta)$, where $\sufstr = { [\sufstr_1^\intercal, \dots, \sufstr_K^\intercal] }^\intercal$ with $\sufst_k \in \mathbb{R}^{d_{\sufstsc_k}}$ for $k = 1,\ldots,K$; 
    \vspace{1mm}
       \item for each $k\in \{1, \dots, K\}$, a one-to-one mapping $\parfam_k\colon \mathbb{R}^{d_{\sufstsc_k}} \to \mathcal{P}$.
   \end{itemize}  
    The vectors $\sufst_k$ are referred to as \emph{local statistics} and the functions $\pi_k(\btheta)= \parfam_k[\sufst_k](\btheta)$ as \emph{local pdfs}.
\end{definition}
In a finite-dimensional supra-Bayesian model,
each local pdf $\pi_k$ is uniquely defined by a corresponding local statistic $\sufst_k$.
As a consequence,  the  conditional  distribution $\lik(\pi_1, \dots, \pi_K \,\vert\,  \btheta)$ 
is implicitly given by 
the conditional pdf
$\lik(\sufstr \,\vert\,  \btheta)$ with $\sufstr = {[\sufstr_1^\intercal, \dots,\sufstr_K^\intercal]}^\intercal$,
and we will refer to $\tlik(\btheta) \triangleq \lik(\sufstr \,\vert\,  \btheta)$ as global likelihood function.
The function $\parfam_k$ specifies which family of distributions $\pi_k$ belongs to.
For example, if we want to model the fact that $\pi_k(\btheta)$  belongs to the family of Gaussian distributions with fixed and known covariance matrix $\bSigma$, 
then we define $\parfam_k[\bmu_k](\btheta) = \mathcal{N}(\btheta; \bmu_{k}, \bSigma)$.
In this example, then,
$\sufstr_k = \bmu_k$.

In Definition~\ref{def:supra},
we further assumed that there is a one-to-one relation between the local  pdf $\pi_k$ and  $\sufstr_k$, i.e., two different vectors $\sufst_k$ and $\widetilde{\sufst}_k$ correspond to different pdfs $\pi_k$ and $\widetilde{\pi}_k$.
In addition to the fact that the pdf  $\pi_k$ is uniquely specified by the vector $\sufst_k$, this assumption also implies that we can uniquely determine  $\sufst_k$ from   $\pi_k$, i.e.,  $\sufst_k$ is a function of $\pi_k$ and we can thus interpret it as a statistic of $\pi_k$.
This justifies the designation of
the vectors $\sufst_k$ as local statistics.
In summary, the local statistic $\sufstr_k$ represents the information provided by the pdf $\pi_k$ of agent $k$ in 
a more accessible, finite-dimensional way.

The following result is an immediate consequence of our definition of a finite-dimensional supra-Bayesian model (Definition \ref{def:supra}) and Bayes' theorem:
the one-to-one relationship between $\pi_k$ and $\sufst_k$ for each $k\in \{1, \dots, K\}$ implies that $\poste(\btheta \,\vert\, \pi_1, \dots, \pi_K) = \poste(\btheta \,\vert\, \sufst)$, and Bayes' theorem implies that $\poste(\btheta \,\vert\, \sufst) \propto \lik(\sufst \,\vert\,  \btheta)\pri(\btheta)$.
\begin{theorem}
\label{th:supraBayes}
    In a  finite-dimensional supra-Bayesian model, the supra-Bayesian fusion result (or supra-Bayesian posterior) is given by
\vspace{-1mm} 
\begin{equation}\label{eq:globpost}
        \poste(\btheta \,\vert\, \pi_1, \dots, \pi_K)
        =
        \poste(\btheta \,\vert\, \sufst) 
        =
        \frac{\tlik(\btheta)\pri(\btheta)}{\int_{\Theta}     \tlik(\btheta')\pri(\btheta') \mathrm{d}\btheta'}\,,
 \vspace{-1mm} 
   \end{equation}
    where $\tlik(\btheta) = \lik(\sufst \,\vert\,  \btheta)$.
\end{theorem}

Since the fusion center  knows $\sufst$,  $\lik(\sufst \,\vert\,  \btheta)$, and $\pri(\btheta)$, it is able to calculate \eqref{eq:globpost}.
However, in general,  \eqref{eq:globpost} does not provide an explicit rule for fusing the pdfs $\pi_k(\btheta)$ into the supra-Bayesian posterior $\poste(\btheta \,\vert\, \sufst)$, 
i.e., it does not specify a pooling function $g$ such that $\poste(\btheta \,\vert\, \sufst) = g[\pi_1, \dots, \pi_K](\btheta)$.
Nevertheless, we can already deduce an interesting fact from the structure of \eqref{eq:globpost}:
The supra-Bayesian posterior is proportional to the product of the prior $\pri(\btheta)$ and the global likelihood function $\tlik(\btheta)$, and thus depends on the pdfs $\pi_k(\btheta)$ only indirectly via the global likelihood function $\tlik(\btheta)$.
Hence, the actual task in supra-Bayesian fusion is to establish a rule for obtaining the global likelihood function $\tlik(\btheta) = \lik(\sufst \,\vert\,  \btheta)$ from the local posteriors $\pi_k(\btheta)$ or, equivalently, from the vector of local statistics $\sufst = {[\sufstr_1^\intercal, \dots, \sufstr_K^\intercal]}^\intercal$.
In what follows, we will see that this approach can result in interesting fusion rules for specific scenarios. 
In particular, we will consider conditionally independent agents in Section~\ref{sec:bayindepag} and 
dependent agents in Section~\ref{sec:bayagdepmeam}. Furthermore, the special case given by the linear Gaussian  model will be studied in Section~\ref{sec:lingaumeam}.

\subsection{Supra-Bayesian Fusion for Conditionally Independent Agents} \label{sec:bayindepag}
Generalizing the scenario in Section~\ref{sec:exbayes}, 
we  assume
that, given $\btheta$, the information provided by each agent to the fusion center is conditionally independent of the information provided by 
the other agents.
In our finite-dimensional supra-Bayesian model
this means that the $\sufstr_k$ are conditionally independent given $\btheta$, i.e., the global likelihood function $\tlik(\btheta)$ factors according to
\begin{equation}
\label{eq:likindeppik}
    \tlik(\btheta) = \lik(\sufst \,\vert\,  \btheta) 
    = \prod_{k=1}^K \lik(\sufst_k \,\vert\,  \btheta)
    = \prod_{k=1}^K \tlik_k(\btheta),
\end{equation}
where we introduced the local likelihood functions $\tlik_k(\btheta) \triangleq \lik(\sufst_k \,\vert\,  \btheta)$.
Because conditional independence of the  $\sufstr_k$ is equivalent to conditional independence of the random local pdfs $\pi_k$,  we immediately obtain the following corollary by inserting \eqref{eq:likindeppik} into  \eqref{eq:globpost}.
\begin{corollary}
    In a finite-dimensional supra-Bayesian model where the local pdfs $\pi_k$ are conditionally independent given $\btheta$, the supra-Bayesian fusion result (or  supra-Bayesian posterior) is given 
    \vspace{-2mm}
    by
    \begin{equation}\label{eq:postindep}
        \poste(\btheta \,\vert\, \pi_1, \dots, \pi_K)
        =
        \poste(\btheta \,\vert\, \sufst) 
        =
        \frac{\big(\prod_{k=1}^K \tlik_k(\btheta)\big) \pri(\btheta)}
        {\int_{\Theta} \big(\prod_{k=1}^K \tlik_k(\btheta') \big) \pri(\btheta')\,\mathrm{d}\btheta'},
    \vspace{-2mm}
\end{equation}
where $\tlik_k(\btheta) = \lik(\sufst_k \,\vert\,  \btheta)$.
\end{corollary}

To establish a link to the scenario of Section~\ref{sec:exbayes}, let us consider the local statistics $\sufstr_k$ and the global likelihood function $\tlik(\btheta) = \lik(\sufst \,\vert\, \btheta)$ in that scenario.
Recall that in Section~\ref{sec:exbayes},
we assumed that each agent has observations $\meam_k\in \mathbb{R}^{d_{\meamsc_k}}$ related to 
$\btheta$ according to the local observation likelihood function $\ell_k(\btheta) = \lik(\meam_k \,\vert\,  \btheta)$, and these observations are 
conditionally independent given $\btheta$.
The local pdfs $\pi_k(\btheta)$---which, in this scenario, are the local posteriors $\poste(\btheta \,\vert\, \meam_k)$---are given by~\eqref{eq:locpost}, and they are thus
parametrized by the local observations $\meam_k$. 
However, in common observation models, the observations $\meam_k$ cannot be uniquely reconstructed  from the posterior pdf $\pi_k(\btheta)$. 
Indeed,  local statistics $\sufstr_k$ that parametrize the local posteriors $\pi_k(\btheta)$ in a one-to-one manner are usually obtained as some function
$T_k( \meamr_k)$ of the observations, where $T_k \colon \mathbb{R}^{d_{\meamsc_k}} \to \mathbb{R}^{d_{\sufstsc_k}}$ with $d_{\sufstsc_k} \leq d_{\meamsc_k}$ is in general not invertible.
The random variable  $\sufstr_k = T_k(\meamr_k)$ is then a sufficient statistic \cite[Sec.~6.2]{casella2002statistical} of $ \meamr_k$ for 
$\btheta$, i.e.,  
\begin{equation}
\label{eq:postemkeqpostetk}
    \poste(\btheta \,\vert\, \meam_k) = \poste(\btheta \,\vert\, \sufst_k)\,.
\end{equation}

Thus, our local statistic $\sufstr_k$ uniquely parametrizing the local posterior $\pi_k(\btheta)$ is given by
$\sufstr_k = T_k( \meamr_k)$, with a noninvertible, possibly dimension-reducing function $T_k$.
The local statistics $\sufstr_k$ given $\btheta$ are  conditionally independent for $k=1, \dots, K$ because they are deterministic functions of the conditionally independent observations $\meamr_k$.
Hence, the factorization \eqref{eq:likindeppik} holds,
and indeed we have a finite-dimensional supra-Bayesian model with a prior $\pri(\btheta)$, a likelihood function $\lik(\sufst \,\vert\,  \btheta)$, and local pdfs $\pi_k$ that are given by $\pi_k(\btheta) = \poste(\btheta \,\vert\, \sufst_k)$, i.e., $\parfam_k[\sufst_k](\btheta) = \poste(\btheta \,\vert\, \sufst_k)$.
Thus,
the supra-Bayesian fusion result $\poste(\btheta \,\vert\, \sufst)$ is given by the expression in \eqref{eq:postindep}.
We will now demonstrate that $\poste(\btheta \,\vert\, \sufst)$ coincides with the fusion result given in 
\eqref{eq:fusionindepbayes}.
Recalling that $\tlik_k(\btheta)= \lik(\sufst_k \,\vert\,  \btheta)$, the supra-Bayesian fusion result \eqref{eq:postindep} becomes
\begin{align}
    \poste(\btheta \,\vert\, \sufst)
    & \propto \bigg(\prod_{k=1}^K \lik(\sufst_k \,\vert\,  \btheta)\bigg) \pri(\btheta) 
    \notag \\
    & \propto \bigg(\prod_{k=1}^K \frac{\poste(\btheta \,\vert\, \sufst_k )}{\pri(\btheta)}\bigg) \pri(\btheta) 
    \notag \\
    & = (\pri(\btheta))^{1-K}  \prod_{k=1}^K \poste(\btheta \,\vert\, \sufst_k ), 
    \notag
\end{align}
where  we used   Bayes' theorem.
By~\eqref{eq:postemkeqpostetk}, we further have
\begin{align}
    \poste(\btheta \,\vert\, \sufst)
     & \propto  (\pri(\btheta))^{1-K} \prod_{k=1}^K  \poste(\btheta \,\vert\, \meam_k)
     \notag \\
     & =  (\pri(\btheta))^{1-K} \prod_{k=1}^K  \pi_k(\btheta)\,,
     \label{eq:globalpostaslocalpostbayes}
\end{align}
which indeed equals the fusion rule
\eqref{eq:fusionindepbayes}.
In particular, a comparison with \eqref{eq:fusionindepbayes} shows
that for conditionally independent $\meam_k$, the supra-Bayesian posterior $\poste(\btheta \,\vert\, \sufst)$ coincides with the oracle posterior  $\poste(\btheta \,\vert\, \meam)$.
Thus, in this case, $\sufst$ is a sufficient statistic of $\meam$ for 
$\btheta$.

\vspace{.5mm}

\begin{example}[Exponential Families]
\label{ex:expfam}
A convenient and versatile class of likelihood functions is given by exponential families \cite{brown1986fundamentals}.
We thus specialize
the results discussed above to these models.
    A local observation likelihood function of the exponential family type 
can be written as    
\begin{equation}
    \label{eq:expofamlik}
         \lik(\meam_k \,\vert\, \btheta) = h_k(\meam_k) \exp\big(\boldeta(\btheta)^{\intercal} T_k(\meam_k) -A_k(\btheta)\big),
    \end{equation}
    with some functions $h_k(\meam_k) \ge 0$, $\boldeta(\btheta) \in \mathbb{R}^{d_\theta}$, and $T_k(\meam_k) \in \mathbb{R}^{d_\theta}$. 
    The function $A_k(\btheta)$ is determined by 
    the other functions
    via the fact that $\lik(\meam_k \,\vert\, \btheta)$ is normalized.
    We assume that the
    observations $\meam_k$ are conditionally independent given $\btheta$.
    Furthermore, the fusion center is supposed to know the conditional pdfs 
    $\lik(\meam_k \,\vert\, \btheta)$ in terms of the functions $\boldeta$, $h_k$, $T_k$, and $A_k$ for all $k$ 
    (but, as always, it does not know the $\meamr_k$),
    and to be also aware of the prior $\pri(\btheta)$.

    It is known that the local statistic $\sufstr_k = T_k(\meamr_k)$ is a sufficient statistic of $\meamr_k$ for 
    $\btheta$ \cite[Prop.~1.5]{brown1986fundamentals}.
    To verify that there is a one-to-one relation between the local posterior $\pi_k$ and $\sufstr_k$, we have to show that $\sufstr_k$ can be recovered from $\pi_k$.
We have    
\begin{equation}
        \pi_k(\btheta) \propto \pri(\btheta) \lik(\meam_k \,\vert\, \btheta) 
        \propto \pri(\btheta) \exp(\boldeta(\btheta)^{\intercal} \sufstr_k -A_k(\btheta))\,.
        \label{eq:postexpofam}
    \end{equation}
    Then 
    \begin{equation}
    \label{eq:expfamtosolve}
        \log \bigg(\frac{\pi_k(\btheta)}{\pri(\btheta)} \exp(A_k(\btheta))\bigg)
        = \boldeta(\btheta)^{\intercal} \sufstr_k + C \,,
    \end{equation}
    where $C$
    is a 
    constant that does not depend on $\btheta$.
    To be able to solve \eqref{eq:expfamtosolve} for $\sufstr_k$ and $C$, we make the technical assumption that there exist
     $d_\theta+1$ different $\btheta_j$ such that the matrix
    \begin{equation}
    \label{eq:defmatB}
    \mathbf{B} \triangleq
        \begin{pmatrix}
             \boldeta(\btheta_1)^{\intercal} & 1 \\
             \vdots & \vdots \\
             \boldeta(\btheta_{d_\theta+1})^{\intercal} & 1 
        \end{pmatrix}
        \in \mathbb{R}^{(d_\theta+1) \times (d_\theta+1)}
    \end{equation}
    is nonsingular.  
    Then, evaluating \eqref{eq:expfamtosolve} at $\btheta_1, \ldots, \btheta_{d_\theta+1}$ gives a system of $d_\theta+1$ equations that can be written as 
    \begin{equation*}
        \mathbf{B} \binom{\sufst_k}{C} = 
        \begin{pmatrix}
             \log \bigg(\frac{\pi_k(\btheta_1)}{\pri(\btheta_1)} \exp(A_k(\btheta_1))\bigg) \\
             \vdots \\
             \log \bigg(\frac{\pi_k(\btheta_{d_\theta+1})}{\pri(\btheta_{d_\theta+1})} \exp(A_k(\btheta_{d_\theta+1}))\bigg)
        \end{pmatrix}.
    \end{equation*}
    Because $\mathbf{B}$ is nonsingular, this equation can be solved for $\sufst_k$ and $C$. Thus, we are able to recover $\sufst_k$ from $\pi_k$.
    We conclude that our exponential family model is a finite-dimensional supra-Bayesian model.
    
    Using \eqref{eq:postexpofam} in \eqref{eq:globalpostaslocalpostbayes}, the supra-Bayesian fusion result
    is obtained as
    \begin{align}
        \poste(\btheta \,\vert\, \sufst)
        & \propto 
        \pri(\btheta)^{1-K} \prod_{k=1}^K  \pri(\btheta) \exp\big(\boldeta(\btheta)^{\intercal} \sufst_k -A_k(\btheta)\big)
        \notag \\
        & = \pri(\btheta) \exp\big(\boldeta(\btheta)^{\intercal} \bar{\sufst} - \bar{A}(\btheta) \big),
        \label{eq:posteexpfam}
  \end{align}
with
    \begin{equation*}
        \bar{\sufst} = \sum_{k=1}^K \sufst_k, 
        \quad 
        \bar{A}(\btheta) = \sum_{k=1}^K A_k(\btheta)\,.
    \end{equation*}
    We see that, for conditionally independent observations $\meamr_k$, $\poste(\btheta \,\vert\, \sufst)$
    depends on the observations $\meamr_k$ only via the local statistics $\sufstr_k = T_k(\meamr_k)$, and furthermore,
    supra-Bayesian fusion 
    essentially amounts to the summation of the local statistics $\sufst_k$ and of the normalization functions $A_k(\btheta)$.
 
 This simple summation rule is augmented when the prior $\pri(\btheta)$ is chosen as 
    \begin{equation}
    \label{eq:priexpfam}
        \pri(\btheta) \propto \exp(\boldeta(\btheta)^{\intercal} \sufst_0 -A_0(\btheta)) ,
    \end{equation}
    for some vector $\sufst_0$ and function $A_0(\btheta)$. 
    Inserting \eqref{eq:priexpfam} into \eqref{eq:posteexpfam}, we obtain
    \begin{equation}
        \poste(\btheta \,\vert\, \sufst)
        \propto \exp\big(\boldeta(\btheta)^{\intercal} \sufst_{\text{post}} -A_{\text{post}}(\btheta)\big),
        \label{eq:fusedexpfam}
    \end{equation}
    with
    \begin{equation}
        \sufst_{\text{post}} = \bar{\sufst} + \sufst_0 = \sum_{k=0}^K \sufst_k 
        \label{eq:fusedexpfam_sum_t} 
    \end{equation}
and
    \begin{equation}
                A_{\text{post}}(\btheta) = \bar{A}(\btheta) + A_0(\btheta) 
        = \sum_{k=0}^K A_k(\btheta)\,.
    \label{eq:fusedexpfam_sum_A}
    \end{equation}
In particular, when all $A_k(\btheta)$ for $k=1, \dots, K$ are equal to the same $A(\btheta)$ and $A_0(\btheta)=a_0 A(\btheta)$, 
then the prior becomes the conjugate prior \cite[Def.~4.18]{brown1986fundamentals}
\begin{equation*}
        \pri(\btheta) \propto \exp(\boldeta(\btheta)^{\intercal} \sufst_0 -a_0 A(\btheta)) ,
\end{equation*}
with the two hyperparameters $\sufst_0$ and $a_0>0$.
Here, the supra-Bayesian fusion result simplifies to
\begin{equation*}
        \poste(\btheta \,\vert\, \sufst)
        \propto \exp\big(\boldeta(\btheta)^{\intercal} \sufst_{\text{post}} -(K+a_0)A(\btheta)\big).
\end{equation*}
We see that $\poste(\btheta \,\vert\, \sufst)$ has the same form as the prior $\pri(\btheta)$, while the hyperparameters $\sufst_0$ and $a_0$ are replaced by $\sufst_{\text{post}} = \sufst_0 + \bar{\sufst}$ and $a_0 + K$, respectively.
\end{example}

An important special case of the exponential family setting is given by linear Gaussian observations.
This case will be considered in Section~\ref{sec:lingaumeam}, both for conditionally dependent and independent observations (see in particular Example~\ref{ex:indepgauss} in Section~\ref{subsec:locstat}).

\subsection{Supra-Bayesian Fusion for 
Agents Collecting Dependent Observations} \label{sec:bayagdepmeam}
Similar to the setting of independent agents studied above, we consider $K$ agents that  obtain observations $\meamr_k \in \mathbb{R}^{d_{\meamsc_k}}$ distributed according to the local observation likelihood functions $\lik(\meam_k \,\vert\,  \btheta)$, with $k\in \{1, \dots, K\}$.
Again, each agent has access also to the prior pdf $\pri(\btheta)$, and the local  posterior pdfs $\pi_k(\btheta)$ are still
    given by \eqref{eq:locpost}.
However, in contrast to the previous subsection,
we do not assume that the observations are conditionally independent.
We assume that the fusion center is aware of the conditional pdf%
\footnote{Note that  the conditional pdfs $\lik(\meam_k \,\vert\,  \btheta)$ are  marginals of the conditional pdf $\lik(\meam \,\vert\,  \btheta)$.}
 $\lik(\meam \,\vert\,  \btheta)$ of all observations  $\meamr = {[\meamr_1^{\intercal}, \dots, \meamr_K^{\intercal}]}^{\intercal}$ given $\btheta$,
the prior pdf $\pri(\btheta)$, and
the local posterior pdfs $\pi_k(\btheta)= \poste(\btheta \,\vert\, \meam_k)$.
We emphasize that although the fusion center has access to $\lik(\meam \,\vert\,  \btheta)$ as a function of $\meam$ and $\btheta$, it does not know the global observation $\meam$ and thus cannot use $\lik(\meam \,\vert\,  \btheta)$ as a global likelihood function.

To establish a supra-Bayesian fusion scheme 
for this scenario, we again
consider a finite-dimensional supra-Bayesian model,
i.e., 
for each agent $k$ there exists a local statistic $\sufstr_k$ such that $\pi_k(\btheta) = \parfam_k[\sufst_k](\btheta)$, and there is a one-to-one relation between $\sufstr_k$ and the local posterior $\pi_k$.
Because $\pi_k(\btheta) = \poste(\btheta \,\vert\, \meam_k)$, the local pdf $\pi_k$ is also uniquely determined by $\meam_k$, and thus the one-to-one relation between $\pi_k$ and $\sufstr_k$ implies that there exists a function
$T_k \colon \mathbb{R}^{d_{\meamsc_k}} \to \mathbb{R}^{d_{\sufstsc_k}}$ such that
 $\sufstr_k =T_k( \meam_k)$.
As before,  the function $T_k$ is not one-to-one in general, i.e., it is not possible to recover  $\meamr_k$ from $\sufstr_k$.
However, $\sufstr_k$ is again a sufficient statistic of $\meamr_k$ for 
$\btheta$, i.e., $\poste(\btheta \,\vert\, \meam_k) = \poste(\btheta \,\vert\, \sufst_k)$.

Because the local observations $\meamr_k$ are subvectors of the global observation $\meamr = {[\meamr_1^{\intercal}, \dots, \meamr_K^{\intercal}]}^{\intercal}\in \mathbb{R}^{d_{\meamsc}} $, we can   introduce $T\colon \mathbb{R}^{d_{\meamsc}} \to \mathbb{R}^{\sum_{k=1}^K d_{\sufstsc_k}}$ as 
\begin{equation*}
    T(\meam) = {[T_1(  \meam_1 )^\intercal, \dots, T_K(  \meam_K )^\intercal]}^{\intercal}\,,
\end{equation*}
and thus we have  
\begin{equation*}
    \sufstr = {[\sufstr_1^\intercal, \dots, \sufstr_K^\intercal]}^\intercal = T(\meamr)\,.
\end{equation*}
The random vector $\sufstr$  summarizes all the information that the agents communicate
to the fusion center, and it is thus known to the fusion center (whereas $\meam$ is not).
Note that although each  $\sufstr_k$ is a sufficient statistic of $\meamr_k$ for 
$\btheta$, the global statistic $\sufstr$ is, in general, not a sufficient statistic of $\meamr$.
This is due to the fact that $\sufstr$ generally  does not capture all the dependencies between the individual  $\meamr_k$.

Because $\sufstr =  T(\meamr)$, we can use the general change-of-variables formula \cite[Sec.~3.4.3]{evangar92measure} to calculate the  conditional pdf $\lik(\sufst \,\vert\, \btheta)$ from the conditional pdf $\lik(\meam \,\vert\, \btheta)$, provided the function $T$ is differentiable.
Since $\sufstr$ summarizes the information communicated
by the agents to the fusion center, $\tlik(\btheta) = \lik(\sufst \,\vert\, \btheta)$
is the global likelihood function that the fusion center has to use in the calculation of the supra-Bayesian posterior $\poste (\btheta \,\vert\, \sufst)$ according to \eqref{eq:globpost}.
Therefore, to obtain the supra-Bayesian fusion rule $g[\pi_1, \dots, \pi_K]$, based on \eqref{eq:globpost}, we have to perform the following three steps:
\begin{enumerate}
    \item Identify the local statistics $\sufstr_k$ that uniquely represent the local posterior pdfs $\pi_k$ within the given statistical model;
    \item apply the general change-of-variables formula to transform the (known) conditional pdf $\lik(\meam \,\vert\, \btheta)$ into the global likelihood function $\tlik(\btheta) = \lik(\sufst \,\vert\, \btheta)$;
    \item calculate the supra-Bayesian posterior $\poste (\btheta \,\vert\, \sufst)$ according to \eqref{eq:globpost}. 
\end{enumerate}

While this three-step process can in principle be performed in any setting satisfying our assumptions,
an explicit characterization of the resulting supra-Bayesian  fusion rule (pooling function) $g[\pi_1, \dots, \pi_K]$   can only be derived for special cases.
The important case of a linear Gaussian  model will be explored in the following.

\section{Supra-Bayesian Fusion for the Linear Gaussian Model}
\label{sec:lingaumeam}

We consider supra-Bayesian pdf fusion for the linear observation model 
\begin{equation}
\meamr =  \mathbf{H}\btheta + \mathbf{n} ,
\label{eq:lingaussmeasmodel}
\end{equation}  
where $\mathbf{H}\in \mathbb{R}^{d_{\meamsc}\times d_{\theta}}$
is a known observation matrix
and $\mathbf{n}\in \mathbb{R}^{d_{\meamsc}}$ is additive zero-mean Gaussian noise with 
a known covariance matrix $\mathbf{\Sigma}$, i.e., $p(\mathbf{n}) = \mathcal{N}(\mathbf{n}; \mathbf{0}, \mathbf{\Sigma})$.
Thus,   $\meamr$  given $\btheta$ is Gaussian distributed with mean $ \mathbf{H}\btheta$ and covariance matrix $\mathbf{\Sigma}$, i.e., 
\begin{equation}
    \lik(\meam \,\vert\, \btheta) = \mathcal{N} (\meam; \mathbf{H}\btheta, \mathbf{\Sigma})\,.
    \label{eq:normaldistygiventh}
\end{equation}
The local observation at agent $k$ is given as $\meamr_k =  \mathbf{H}_k\btheta + \mathbf{n}_k \in \mathbb{R}^{d_{\meamsc_k}}$, where 
\begin{equation}
\label{eq:Hblockstruct}
    \mathbf{H} = 
    \begin{pmatrix}
        \mathbf{H}_1\\
        \mathbf{H}_2\\
         \vdots \\
        \mathbf{H}_K
    \end{pmatrix},
\end{equation} 
with $\mathbf{H}_k\in \mathbb{R}^{d_{\meamsc_k}\times d_{\theta}}$, and $\mathbf{n} = {[\mathbf{n}_1^{\intercal}, \dots, \mathbf{n}_K^{\intercal}]}^\intercal$.
Thus, each local observation $\meamr_k$ given $\btheta$ is again Gaussian with mean $\mathbf{H}_k\btheta$   and  covariance matrix  $\mathbf{\Sigma}_{k k} \in \mathbb{R}^{d_{\meamsc_k}\times d_{\meamsc_k}}$.
We note that the overall covariance matrix $\mathbf{\Sigma}$ is block-structured according to
\begin{equation}
\label{eq:covariancebayesgauss}
    \mathbf{\Sigma}= 
    \begin{pmatrix}
        \mathbf{\Sigma}_{1 1} & \cdots & \mathbf{\Sigma}_{1 K}
         \\ 
         \vdots & \ddots & \vdots 
         \\
        \mathbf{\Sigma}_{K 1} & \cdots & \mathbf{\Sigma}_{K K}
    \end{pmatrix} ,
\end{equation}
where the off-diagonal cross-covariance matrices $\mathbf{\Sigma}_{k k'}$ for $k\neq k'$ describe the conditional dependency between the observations of different agents.
The case of conditionally independent observations $\meam_k$ is obtained for $\mathbf{\Sigma}_{k k'}=\mathbf{0}$ for all $k\neq k'$.
For simplicity, we further assume that for all $k= 1, \dots, K$, $d_{\meamsc_k} \geq d_{\theta}$,  $\mathbf{H}_k$ has full rank, and  $\mathbf{\Sigma}_{k k}$ is positive definite.
The local observation likelihood functions are here given by
\begin{align}
\ell_k(\btheta)
    & = \lik(\meam_k \,\vert\, \btheta) 
    \notag \\
    & = \mathcal{N} (\meam_k; \mathbf{H}_k\btheta, \mathbf{\Sigma}_{k k})
    \notag \\
    & \propto \exp\bigg({-}\frac{(\meam_k-\mathbf{H}_k \btheta)^\intercal\mathbf{\Sigma}_{k k}^{-1}(\meam_k-\mathbf{H}_k\btheta)}{2}\bigg)\,.
    \label{eq:loclicgaussnons1} 
\end{align}

\subsection{Local Statistics}
\label{subsec:locstat}
We can rewrite \eqref{eq:loclicgaussnons1} as    
\begin{align}
\ell_k(\btheta)
    &  \propto \exp\bigg({-}\frac{(\btheta- \Ourinv_k\meam_k)^\intercal 
      \mathbf{H}_k^\intercal
    \mathbf{\Sigma}_{k k}^{-1}
    \mathbf{H}_k 
    (\btheta-\Ourinv_k\meam_k)}{2}\bigg)
    \notag \\
    & = \exp\bigg({-}\frac{(\btheta- \sufstr_k)^\intercal 
      \mathbf{H}_k^\intercal
    \mathbf{\Sigma}_{k k}^{-1}
    \mathbf{H}_k 
    (\btheta-\sufstr_k)}{2}\bigg),
    \label{eq:loclicgaussnons2}
\end{align}
where 
\begin{equation}\label{eq:defhkplus}
    \Ourinv_k
    = 
    (  \mathbf{H}_k^\intercal
    \mathbf{\Sigma}_{k k}^{-1}
    \mathbf{H}_k )^{-1}  \mathbf{H}_k^\intercal \mathbf{\Sigma}_{k k}^{-1}
\end{equation} 
and 
\begin{equation}
\label{eq:sufstrgauss}
    \sufstr_k = \Ourinv_k\meamr_k 
    = (  \mathbf{H}_k^\intercal
    \mathbf{\Sigma}_{k k}^{-1}
    \mathbf{H}_k )^{-1}  \mathbf{H}_k^\intercal \mathbf{\Sigma}_{k k}^{-1}\meamr_k.
\end{equation}
The proportionality in \eqref{eq:loclicgaussnons2} is as a function of $\btheta$, i.e., the proportionality constant will depend on $\meam_k$.

We claim that $\sufstr_k$ in \eqref{eq:sufstrgauss} qualifies as a local statistic 
in a finite-dimensional supra-Bayesian model. For a proof, we note that the local posteriors are again given as 
$\pi_k(\btheta)
    = \poste(\btheta \,\vert\, \meam_k)
    \propto \ell_k(\btheta) \pri(\btheta)$.
To see that there is a one-to-one relation between the local posterior $\pi_k$ and the finite-dimensional parameter $\sufstr_k \in \mathbb{R}^{  d_{\theta}}$, recall that the fusion center is aware of the prior $\pri(\btheta)$ and the  matrices $\mathbf{H}$ and  $\mathbf{\Sigma}$.
In particular, the fusion center is aware of
$\mathbf{H}_k$ and $\mathbf{\Sigma}_{k k}$, and thus it is able to recover from $\sufstr_k$  the local observation  likelihood function $\ell_k(\btheta)$ in  \eqref{eq:loclicgaussnons2}
and, in turn, the local posterior  $\pi_k(\btheta) \propto \ell_k(\btheta) \pri(\btheta)$.
Conversely, the fusion center is able to obtain  $\sufstr_k$  from the local posterior $\pi_k(\btheta)$ by first dividing by the prior $\pri(\btheta)$ and normalizing as a function of $\btheta$ (to obtain a function proportional to $\ell_k(\btheta)$), and finally calculating the mean of the resulting pdf in $\btheta$ (which is $\sufstr_k$ according to \eqref{eq:loclicgaussnons2}).
Thus, $\sufstr_k$ is related to $\pi_k$ in a one-to-one manner, and hence it is a local statistic.

\begin{example}[Conditionally Independent Agents]
\label{ex:indepgauss}
    In the case of conditionally independent agents, i.e., the observations $\meamr_k$ are conditionally independent given $\btheta$, we can easily calculate the  supra-Bayesian posterior.
    Indeed, the structure of the local likelihood function in \eqref{eq:loclicgaussnons2} shows  that we are in the exponential family setting of Example~\ref{ex:expfam}.
    More specifically, we can rewrite  \eqref{eq:loclicgaussnons2} as
    \begin{align}
        \ell_k(\btheta)
        &  \propto  
        \exp\bigg(
        \btheta^\intercal 
          \tilde{\sufstr}_k
        -\frac{\btheta^\intercal 
          \mathbf{H}_k^\intercal
        \mathbf{\Sigma}_{k k}^{-1}
        \mathbf{H}_k 
        \btheta}{2}\bigg),
        \label{eq:ellkgaussindep}
    \end{align}
    where
    \begin{equation*}
        \tilde{\sufstr}_k = \mathbf{H}_k^\intercal
            \mathbf{\Sigma}_{k k}^{-1}
            \mathbf{H}_k 
            \sufstr_k 
            = \mathbf{H}_k^\intercal \mathbf{\Sigma}_{k k}^{-1}\meamr_k
    \end{equation*}
    is a bijective trans\-for\-ma\-tion of $\sufstr_k$ and thus also a valid choice for a local statistic.
    Considering a Gaussian prior $\pri(\btheta)$ with mean $\bmu_0$ and covariance matrix $\mathbf{\Sigma}_{0}$, we can rewrite $\pri(\btheta)$
    as
    \begin{equation}
    \label{eq:priorgaussindep}
        \pri(\btheta) \propto  
            \exp\bigg(
            \btheta^\intercal 
              \tilde{\sufstr}_0
            -\frac{\btheta^\intercal 
            \mathbf{\Sigma}_{0}^{-1}
            \btheta}{2}\bigg),
    \end{equation}
    where $\tilde{\sufstr}_0 = \mathbf{\Sigma}_{0}^{-1} \bmu_0$.
    Comparing \eqref{eq:ellkgaussindep} with \eqref{eq:expofamlik} and 
    \eqref{eq:priorgaussindep} with \eqref{eq:priexpfam},
    we see that $\ell_k(\btheta) = \lik(\meam_k \,\vert\, \btheta)$  belongs to the exponential family \eqref{eq:expofamlik} with $\sufstr_k$   formally replaced by $\tilde{\sufstr}_k$ and 
    $A_k(\btheta) = \frac{\btheta^\intercal 
          \mathbf{H}_k^\intercal
        \mathbf{\Sigma}_{k k}^{-1}
        \mathbf{H}_k 
        \btheta}{2}$.
        Furthermore, $\pri(\btheta)$ conforms to \eqref{eq:priexpfam} with  $A_0(\btheta)  = \frac{\btheta^\intercal 
            \mathbf{\Sigma}_{0}^{-1}
            \btheta}{2}$.
    With our assumption of conditionally independent agents, we can use the result \eqref{eq:fusedexpfam}--\eqref{eq:fusedexpfam_sum_A} and obtain for the 
     supra-Bayesian fusion result
    \begin{align}
        & \poste(\btheta \,\vert\, \sufst)
        \notag \\*
            & \propto 
            \exp\bigg(\btheta^{\intercal} \bigg(\sum_{k=0}^K \tilde{\sufst}_k\bigg) - 
            \frac{
            \btheta^\intercal \big(
            \mathbf{\Sigma}_{0}^{-1} 
            + \sum_{k=1}^K   
            \mathbf{H}_k^\intercal
            \mathbf{\Sigma}_{k k}^{-1}
            \mathbf{H}_k 
            \big)
            \btheta
            }{2} \bigg).
            \label{eq:optbaygaussindep}
    \end{align}
    This is again a Gaussian pdf, with mean 
     \begin{equation*}
        \bmu_1
         = \mathbf{\Sigma}_1 \sum_{k=0}^K \tilde{\sufst}_k
         = \mathbf{\Sigma}_1
            \bigg(\mathbf{\Sigma}_{0}^{-1} \bmu_0
            + \sum_{k=1}^K \mathbf{H}_k^\intercal \mathbf{\Sigma}_{k k}^{-1}\meamr_k\bigg)
    \end{equation*}
    and covariance matrix 
    \begin{equation*}
        \mathbf{\Sigma}_1 = \bigg(
            \mathbf{\Sigma}_{0}^{-1} 
            + \sum_{k=1}^K
            \mathbf{H}_k^\intercal
            \mathbf{\Sigma}_{k k}^{-1}
            \mathbf{H}_k 
            \bigg)^{-1}.
    \end{equation*}
    It is straightforward to verify that \eqref{eq:optbaygaussindep} is equal to the  oracle  posterior $\poste(\btheta \,\vert\, \meamr)$.
    Thus, we see once again (cf.~Section~\ref{sec:bayindepag}) that although the  supra-Bayesian fusion result depends on the observations $\meamr_k$ only via the local statistics $\tilde{\sufstr}_k$, it still equals
    the oracle posterior $\poste(\btheta \,\vert\, \meamr)$, as if the fusion center had access to all observations $\meamr_k$ directly.
    As we will see below, this crucially depends on our
    assumption of conditionally independent agents and is no longer true
    if we assume conditional dependencies between the observations.
\end{example}

\subsection{Global Likelihood Function}
In the previous subsection, for the general linear Gaussian model with conditionally dependent $\meamr_k$,
we identified   local statistics  $\sufstr_k = T_k(\meamr_k) = \Ourinv_k\meamr_k $ that are related in a one-to-one manner to the local posteriors $\pi_k$. 
The next step according to our three-step program from Section~\ref{sec:bayagdepmeam} is to calculate the global likelihood function 
$\tlik(\btheta) = \lik(\sufst \,\vert\,  \btheta)$ by 
transforming the conditional pdf  $p(\meamr \, \vert \, \btheta)$ into the conditional pdf $p(\sufstr \, \vert \, \btheta)$.
According to \eqref{eq:normaldistygiventh},
the conditional pdf of   $\meamr$ given $\btheta$ is%
\footnote{
This conditional pdf only exists if  the covariance matrix $\mathbf{\Sigma}$ is positive definite.
However, the derivations that follow do not require the existence of a pdf and are also valid if $\mathbf{\Sigma}$ is positive semidefinite.
}
\begin{equation} 
 \label{eq:condpdfgauss}
    \lik (\meam \, \vert \, \btheta) 
    \propto \exp\bigg({-}\frac{(\meam- \mathbf{H} \btheta)^\intercal \mathbf{\Sigma}^{-1}(\meam- \mathbf{H} \btheta)}{2}\bigg).
\end{equation}
We further have that
\begin{equation}
\label{eq:sufstrvecgauss}
    \sufstr = {[\sufstr_1^\intercal, \dots, \sufstr_K^\intercal]}^\intercal = \Ourinv \meamr\,,
\end{equation}
where $\Ourinv = \operatorname{diag}(\Ourinv_1, \dots, \Ourinv_K)$ denotes the block-diagonal matrix with block entries $\Ourinv_k$ on the diagonal.
Thus, $\sufstr$ is a linear function of $\meamr$ and
hence $\sufstr$ given $\btheta$ is Gaussian and has mean $\Ourinv\mathbf{H}\btheta$ and covariance matrix 
\begin{equation}\label{eq:sigmatilde}
    \widetilde{\mathbf{\Sigma}} 
     = \Ourinv\mathbf{\Sigma}{\Ourinv}^{\intercal} .
\end{equation}
We assume that $\widetilde{\mathbf{\Sigma}}$ is nonsingular.
The mean can be simplified to
\begin{equation*}
    \Ourinv\mathbf{H} \btheta= 
    \begin{pmatrix}
        \Ourinv_1\mathbf{H}_1 \\
         \Ourinv_2\mathbf{H}_2 \\
         \vdots \\
         \Ourinv_K\mathbf{H}_K 
    \end{pmatrix} \btheta
    =
    \begin{pmatrix}
        \mathbf{I}_{d_{\theta}}\\
        \mathbf{I}_{d_{\theta}}\\
         \vdots \\
        \mathbf{I}_{d_{\theta}}
    \end{pmatrix}\btheta
    =
    \begin{pmatrix}
        \btheta\\
        \btheta\\
         \vdots \\
        \btheta
    \end{pmatrix}
    = 
    \mathbf{1}_{K} \otimes \btheta\,,
\end{equation*}
where we used \eqref{eq:Hblockstruct} and the fact that, by \eqref{eq:defhkplus}, 
\begin{equation}
\label{eq:ourinvleftinv}
    \Ourinv_k\mathbf{H}_k = \mathbf{I}_{d_{\theta}}\,.
\end{equation}
The global likelihood function $\tlik(\btheta)$ is thus obtained
as 
\begin{align}
    \tlik(\btheta) 
    &  = \lik(\sufst \, \vert \, \btheta)
    \notag \\
    & = \mathcal{N} (\sufst; \mathbf{1}_{K} \otimes \btheta, \widetilde{\mathbf{\Sigma}})
    \notag \\
    & \propto \exp\bigg({-}\frac{(\sufstr-\mathbf{1}_{K} \otimes \btheta)^\intercal \widetilde{\mathbf{\Sigma}}^{-1}(\sufstr-\mathbf{1}_{K} \otimes \btheta)}{2}\bigg)\,.
    \label{eq:likgausssufstatinep}
\end{align}
To summarize, for the linear Gaussian model, local statistics $\sufstr_k$ characterizing the local posteriors $\pi_k$ are given by \eqref{eq:sufstrgauss}, and the corresponding global likelihood function $\tlik(\btheta)  = \lik(\sufst \, \vert \, \btheta)$ is given by \eqref{eq:likgausssufstatinep}.

\subsection{Supra-Bayesian Fusion Rule for a Scalar $\theta$} 
\label{sec:gaussscalar}

After identifying local statistics $\sufstr_k$ and calculating the global likelihood function $\tlik(\btheta)  = \lik(\sufst \, \vert \, \btheta)$, the final step in the derivation
of the supra-Bayesian fusion rule is to calculate the supra-Bayesian posterior $\poste(\btheta \,\vert\, \sufst)$ according to \eqref{eq:globpost}.
We first develop the  supra-Bayesian fusion rule for the case that $d_{\theta}=1$, i.e., for a scalar random variable $\theta \in \mathbb{R}$.
Here, the observation matrix $\mathbf{H}$ reduces to a vector $\mathbf{h}\in \mathbb{R}^{d_{\meamsc}}$ and the observation model \eqref{eq:lingaussmeasmodel} is given by
\begin{equation*}
    \meamr =  \mathbf{h}\,\theta + \mathbf{n}\,.
\end{equation*}  
Similarly, the local observation at agent $k$ is given as $\meamr_k =  \mathbf{h}_k\theta + \mathbf{n}_k$ with $\mathbf{h}_k\in \mathbb{R}^{d_{\meamsc_k}}$,
and the local statistic at agent $k$ follows from \eqref{eq:sufstrgauss} as 
\begin{equation}
\label{eq:sufstsc_k}
    \sufstsc_k = \ourinv_k\meamr_k\in  \mathbb{R}\,,
\end{equation}
where $\Ourinv_k$  reduces to the (row) vector 
\begin{equation}
\label{eq:ourinv_k}
    \ourinv_k 
    =  \frac{1}{ \mathbf{h}_k^\intercal
    \mathbf{\Sigma}_{k k}^{-1}
    \mathbf{h}_k}
    \mathbf{h}_k^\intercal \mathbf{\Sigma}_{k k}^{-1}\,.
\end{equation}
Note that $\Ourinv = \operatorname{diag}(\ourinv_1, \dots, \ourinv_K)$ is still a matrix.
In this case, we can give the following explicit fusion rule,
which is derived in Appendix~\ref{appendix:scalarbayesfusion}.

\begin{theorem}
\label{th:scalarbayesfusion}
    For $d_{\theta}=1$, let $\ell_k(\theta) = \lik(\meam_k \,\vert\, \theta)$ denote the local observation likelihood functions given by \eqref{eq:loclicgaussnons1} 
    for $k=1, \dots, K$ 
    and let $\tlik(\theta) = \lik(\sufst \, \vert \, \theta)$ be the global likelihood function given by \eqref{eq:likgausssufstatinep}.
        Then 
        \begin{equation}
            \tlik(\theta)
            \propto
            \prod_{k=1}^K  (\ell_k(\theta))^{w_k},
            \label{eq:fusionliksgaussd1}
        \end{equation}
        where
        \begin{equation}\label{eq:powersbayes}
            w_k = \frac{\mathbf{1}_{K}^\intercal \widetilde{\mathbf{\Sigma}}^{-1}\mathbf{e}_k} {\mathbf{h}_k^\intercal
            \mathbf{\Sigma}_{k k}^{-1}
            \mathbf{h}_k}\,,
        \end{equation}
        with $\widetilde{\mathbf{\Sigma}} 
         = \Ourinv\mathbf{\Sigma}{\Ourinv}^{\intercal}$ and
        $\mathbf{e}_k$ denoting the $k$th unit vector in $\mathbb{R}^{K}$.
        Furthermore, for a given prior $\pri(\theta)$
        and local posteriors $\pi_k(\theta)= \poste(\theta \,\vert\, \meam_k) \propto \pri(\theta) \ell_k(\theta)$, the  supra-Bayesian fusion result $g[\pi_1, \dots, \pi_K](\theta) = \poste(\theta \,\vert\, \sufst) \propto \pri(\theta) \tlik(\theta)$ is given by
    \begin{align}
        g[\pi_1, \dots, \pi_K](\theta)
        & \propto 
        (\pri(\theta))^{1-\sum_{k=1}^K w_k}  \prod_{k=1}^K   (\pi_k( \theta))^{w_k}\,.
        \label{eq:fusedepbayesd1}
    \end{align}
\end{theorem}


We emphasize that in this theorem we do not assume that the observations $\meamr_k$ are conditionally independent given $\theta$.
Furthermore, it should be noted that the weights $w_k$ in \eqref{eq:powersbayes} do not generally sum to one, and they may be negative. 
Thus, the fusion rule \eqref{eq:fusedepbayesd1} is an instance of the generalized multiplicative pooling function in~\eqref{eq:gen_mult_pooling}.

Finally, if the prior $\pri(\theta)$ is Gaussian, we can show that the supra-Bayesian fusion result
$\poste(\theta \,\vert\, \sufst)$ is again Gaussian and reduce the fusion rule \eqref{eq:fusedepbayesd1} to a second-order rule involving only the mean and variance:
\begin{corollary}
\label{cor:fusegausspriord1}
Under the assumptions of Theorem~\ref{th:scalarbayesfusion},
let the prior $\pri(\theta)$ be Gaussian with mean $\mu_0$ and variance $\sigma_0^2$,
i.e., $\pri(\theta) = \mathcal{N}(\theta; \mu_0, \sigma_0^2)$.
Then the supra-Bayesian fusion result
$\poste(\theta \,\vert\, \sufst)$ is again Gaussian, i.e., $\poste(\theta \,\vert\, \sufst) =  \mathcal{N}(\theta; \mu_1, \sigma_1^2)$, with mean
\begin{equation}
    \mu_1 = \frac{\widehat{\sigma}^{2} \sigma_0^{2}}{ \widehat{\sigma}^{2} + \sigma_0^{2}} \mathbf{1}_{K}^\intercal \widetilde{\mathbf{\Sigma}}^{-1}\sufstr
    + 
    \frac{\widehat{\sigma}^{2}}{ \widehat{\sigma}^{2} + \sigma_0^{2}} \mu_0 
    \label{eq:mu1sc}
\end{equation}
and variance
\begin{equation*}
    \sigma_1^2 = \frac{\widehat{\sigma}^{2} \sigma_0^{2}}{ \widehat{\sigma}^{2} + \sigma_0^{2}},
\end{equation*}
where
\begin{equation}
\label{eq:sighatsc}
        \widehat{\sigma}^{2} = \frac{1}{\mathbf{1}_{K}^\intercal \widetilde{\mathbf{\Sigma}}^{-1} \mathbf{1}_{K}}
    \end{equation}
and $\sufstr= {[\sufstsc_1^\intercal, \dots, \sufstsc_K^\intercal]}^\intercal$ is given by \eqref{eq:sufstsc_k} and \eqref{eq:ourinv_k}.
\end{corollary}

As  mentioned before, the supra-Bayesian fusion result  $\poste(\theta \,\vert\, \sufst)$
is in general different from the oracle posterior $\poste(\theta \, \vert \, \meam)$.
 Indeed, the oracle posterior is proportional to the product of the prior $\pri(\theta)$ and the global observation likelihood function $\lik (\meam \, \vert \, \theta)$ in \eqref{eq:condpdfgauss}.
It can then easily be seen that the oracle posterior $\poste(\theta \, \vert \, \meam)$ is also Gaussian but with mean
\begin{equation}
    \mu_2 = \frac{\widehat{\sigma}_2^{2} \sigma_0^{2}}{ \widehat{\sigma}^{2} + \sigma_0^{2}} \mathbf{h}^\intercal {\mathbf{\Sigma}}^{-1}\meam
    + 
    \frac{\widehat{\sigma}_2^{2}}{ \widehat{\sigma}_2^{2} + \sigma_0^{2}} \mu_0 
    \label{eq:mu2sc}
\end{equation}
and variance
\begin{equation*}
    \sigma_2^2 = \frac{\widehat{\sigma}_2^{2} \sigma_0^{2}}{ \widehat{\sigma}_2^{2} + \sigma_0^{2}},
\end{equation*}
where
\begin{equation}
\label{eq:sig2hatsc}
    \widehat{\sigma}_2^{2} = \frac{1}{\mathbf{h}^\intercal {\mathbf{\Sigma}}^{-1} \mathbf{h}}\,.
\end{equation}
To better understand the difference, we note that in \eqref{eq:mu1sc} 
\begin{align*}
    \mathbf{1}_{K}^\intercal \widetilde{\mathbf{\Sigma}}^{-1}\sufstr
    & = 
    \mathbf{h}^\intercal \Ourinv^\intercal(\Ourinv\mathbf{\Sigma}{\Ourinv}^{\intercal})^{-1} \Ourinv \meam
\end{align*}
and in \eqref{eq:sighatsc}
\begin{align*}
    \mathbf{1}_{K}^\intercal \widetilde{\mathbf{\Sigma}}^{-1} \mathbf{1}_{K}
    & = 
    \mathbf{h}^\intercal \Ourinv^\intercal(\Ourinv\mathbf{\Sigma}{\Ourinv}^{\intercal})^{-1} \Ourinv \mathbf{h} \,,
\end{align*}
where we used \eqref{eq:sufstrvecgauss}--\eqref{eq:ourinvleftinv}.
Comparing with $\mathbf{h}^\intercal {\mathbf{\Sigma}}^{-1}\meam$ and $\mathbf{h}^\intercal {\mathbf{\Sigma}}^{-1} \mathbf{h}$ arising in \eqref{eq:mu2sc} and \eqref{eq:sig2hatsc}, respectively, 
we conclude that  the difference between the oracle posterior and the supra-Bayesian posterior is that  the  matrix ${\mathbf{\Sigma}}^{-1}$ is replaced by $\Ourinv^\intercal(\Ourinv\mathbf{\Sigma}{\Ourinv}^{\intercal})^{-1} \Ourinv$.

A simplified version of Theorem~\ref{th:scalarbayesfusion} has been shown in \cite{clemen1985limits} and is the setting of the early supra-Bayesian approaches. 
More specifically, it is assumed
in  \cite{clemen1985limits} 
that a fusion center obtains from $K$ agents estimates $\mu_k$ of a scalar random variable
$\theta$.
These estimates can be interpreted as our local statistics $\sufstsc_k$.
Furthermore, the fusion center has a Gaussian prior for $\theta$ and knows that the vector of the estimation errors of all agents, $\mathbf{u}= {[u_1, \dots, u_K]}^{\intercal}$ with $u_k = \mu_k - \theta$, also follows a Gaussian distribution with zero mean and some covariance matrix $\widetilde{\mathbf{\Sigma}}$ (in general,
the errors may be correlated).
Equivalently, conditionally on $\theta$, the estimates $\bmu = {[\mu_1, \dots, \mu_K]}^{\intercal}$ follow a 
Gaussian distribution with mean $\mathbf{1}_{K} \theta$ and the same covariance matrix $\widetilde{\mathbf{\Sigma}}$.
Thus, the setting in  \cite{clemen1985limits} directly assumes the conditional distribution of $\sufst$ given $\theta$ without starting from any detailed observation model.

To get a better intuition about the role of the weights $w_k$ and the meaning of negative weights in the setting of Theorem~\ref{th:scalarbayesfusion}, we will consider a specific example. 
\begin{example}[Private and Shared Observations]
\label{ex:specialweightsd1}
We assume that agent $k$ has $r_k$ private observations, i.e., observations that no other agent observes, and $r_0$ shared observations, i.e., observations that all agents know jointly.
The resulting total number of observations is thus $d_{\meamsc}=\sum_{k=1}^K (r_0 + r_k)$.
However, there are only $r_0 + \sum_{k=1}^K  r_k$ different observations.
We assume that these different observations given $\theta$ are independent and have variance one and mean $\theta$.
To embed this scenario into our linear model, we choose $\mathbf{h}_k = \mathbf{1}_{r_0+r_k}$
and the submatrices of the covariance matrix $\mathbf{\Sigma}$ in \eqref{eq:covariancebayesgauss} as
\begin{equation}
\mathbf{\Sigma}_{k k'} = 
    \begin{pmatrix}
        \mathbf{I}_{r_0} & \mathbf{0}_{r_0\times r_{k'}} \\
        \mathbf{0}_{r_k\times r_0} & \mathbf{0}_{r_k\times r_{k'}}
    \end{pmatrix}
    \in \mathbb{R}^{(r_0+r_k) \times (r_0+r_{k'})}
    \label{eq:covexample}
\end{equation}
for $k\neq k'$ and 
\begin{equation*}
    \mathbf{\Sigma}_{k k} = \mathbf{I}_{r_0+r_k}\,.
\end{equation*}
Thus, we have that 
\begin{equation*}
    \meam_k = \mathbf{1}_{r_0+r_k} \theta + \mathbf{n}_k\,,
\end{equation*}
where $\mathbf{n}_k$ is a vector of independent and identically distributed standard Gaussian random variables,
i.e., 
$\pri(\mathbf{n}_k) = \mathcal{N}(\mathbf{n}_k; \mathbf{0}_{(r_0+r_k)\times 1}, \mathbf{I}_{r_0+r_k})$.
The covariance structure \eqref{eq:covexample} between the  $\mathbf{n}_k$, for $k\in \{1, \dots, K\}$, implies that for $i\in \{1, \dots, r_0\}$ the $i$th entry of $\mathbf{n}_k$ 
and the $i$th entry of $\mathbf{n}_{k'}$ with $k'\neq k$ coincide with probability one:
\begin{equation*}
    \E[(n_{k,i} - n_{k',i})^2]
    = \underbrace{\E[n_{k,i}^2]}_{=1} + \underbrace{\E[n_{k',i}^2]}_{=1} - 2 \underbrace{\E[n_{k,i}n_{k',i}]}_{=1} = 0\,.
\end{equation*}
Thus, the first $r_0$ observations are the same for all agents.

With these choices and assuming that $r_k>0$ and $r_0>0$, a tedious but straightforward calculation (for details see Appendix~\ref{app:weightsd1}) shows that the weights $w_k$ in \eqref{eq:powersbayes} simplify to 
\begin{align}
    \label{eq: weight_wk_ex_1}
    w_k 
    & = 1 - \frac{K-1}{r_k} \bigg(\sum_{k'=0}^K \frac{1}{r_{k'}}\bigg)^{-1}.
\end{align}
In particular, we see that all weights are upper-bounded by $1$ and are emphasized according to their amount of independent information as given by $r_k$. 
More surprising is the possibility of negative weights for agents with few private observations (e.g., the setting $K=3$, $r_1=1$, and $r_0=r_2=r_3=4$ gives $w_1=-1/7$).
An explanation for this result is that  negatively weighting agents with few private observations can counteract the multiple-counting of the shared observations that are part of all agents' posteriors. 
More generally, it follows from \eqref{eq: weight_wk_ex_1} that $w_k\geq 0$ if and only if
\begin{align*}
    r_k \geq (K-1)\Bigg(\sum_{k'=0}^K \frac{1}{r_{k'}}\Bigg)^{-1}
\end{align*}
or, equivalently,
\begin{align*}
    \sum_{k'=0}^K \frac{r_k}{r_{k'}} \geq K  -  1 .
\end{align*}

The sum of all weights is given by
\begin{equation}
    \label{eq: sum_of_weights_expression_ex_1}
    \sum_{k=1}^K w_k = K - (K-1)\bigg(\sum_{k'=0}^K \frac{1}{r_{k'}}\bigg)^{-1}\sum_{k=1}^K\frac{1}{r_k} .
\end{equation}
From this expression, we readily conclude that
\begin{equation}
    \label{eq: sum_of_weights_bounds}
    1 \leq \sum_{k=1}^K w_k \leq K. 
\end{equation}
Indeed, this follows from the fact that the second term on the right-hand side of \eqref{eq: sum_of_weights_expression_ex_1}, 
$(K-1)\big(\sum_{k'=0}^K \frac{1}{r_{k'}}\big)^{-1}\sum_{k=1}^K\frac{1}{r_k}$,
is nonnegative 
and upper-bounded by $K-1$ since
$
    \big(\sum_{k'=0}^K \frac{1}{r_{k'}}\big)^{-1}\sum_{k=1}^K\frac{1}{r_k} \leq 1
$.
The double bound \eqref{eq: sum_of_weights_bounds} shows that although some weights may be negative, the sum of all weights is always between the sum of all weights
in the log-linear pooling function in \eqref{eq: log_linear_pooling} (there, the sum was
$1$) and the sum of all weights in the multiplicative pooling function in \eqref{eq: multiplicative_pooling} (there, all weights were $1$,
and hence the sum was
$K$). 

Another conclusion  we can draw is that varying the number of shared observations $r_0$---while keeping the number of private observations $r_k$ fixed---corresponds to an ``interpolation" between the multiplicative pooling function  and the  log-linear pooling function. 
Consider first the case that the agents have the same number of private observations, i.e.,  $r_1=\cdots=r_K$. 
When $r_0=0$, a derivation similar to that  in Appendix~\ref{app:weightsd1} gives $w_k = 1$. 
This implies that when the agents do not share any observations, the pooling
function in \eqref{eq:fusedepbayesd1} corresponds exactly to the standard multiplicative pooling function in \eqref{eq: multiplicative_pooling}. 
On the other hand, as the number of shared observations $r_0$ increases, the pooling
function behaves closer to a symmetric log-linear pooling function (i.e., using $w_k = 1/K$). 
Indeed, it follows from \eqref{eq: weight_wk_ex_1} that
\begin{equation*}
    \lim_{r_0\rightarrow \infty} w_k = \frac{1}{K}.
\end{equation*}
If we remove the restriction that $r_1=\dots=r_K$, the connection to multiplicative pooling still holds; however, the connection to log-linear pooling only holds under the condition of nonnegative weights, i.e., $w_k\geq 0$ for all $k$, which may be violated if some agents hold only few private observations as compared to the total number of observations.

\end{example}


\subsection{Supra-Bayesian Fusion Rule for a Vector $\btheta$} 
\label{sec:gaussvector}

We can generalize Theorem~\ref{th:scalarbayesfusion} to a vector
$\btheta\in \mathbb{R}^{d_{\theta}}$ with $d_{\theta}>1$.
However, formally, the weights $w_k$ in \eqref{eq:powersbayes} become matrices $\mathbf{W}_k$ and thus cannot  be used as powers in a fusion rule.
Hence, the following fusion result is more complicated and the relation to the one-dimensional case is not obvious.
A proof is provided in Appendix~\ref{appendix:fuseliks}.
\begin{theorem} \label{th:fuseliks}
    Let $\ell_k(\btheta) = \lik(\meam_k \,\vert\, \btheta)$ denote the local observation likelihood functions given by \eqref{eq:loclicgaussnons1} for $k=1, \dots, K$ and let $\tlik(\btheta) = \lik(\sufst \, \vert \, \btheta)$ be the global likelihood function given by \eqref{eq:likgausssufstatinep}.
    Then 
    \begin{equation}
        \tlik (\btheta)
        \propto
        \xi_0(\btheta) \prod_{k=1}^K  \ell_k(\mathbf{W}_k\btheta)
        \label{eq:fuselikgaussvec}
    \end{equation}
    where
    \begin{equation}
        \mathbf{W}_k = (\mathbf{H}_k^\intercal
        \mathbf{\Sigma}_{k k}^{-1}
        \mathbf{H}_k)^{-1}
        (\mathbf{e}_k \otimes \mathbf{I}_{d_{\theta}})^\intercal\widetilde{\mathbf{\Sigma}}^{-1 }(\mathbf{1}_{K}\otimes \mathbf{I}_{d_{\theta}}),
        \label{eq:weightmatrix}
    \end{equation}
    with $\mathbf{e}_k$ denoting the $k$th unit vector in $\mathbb{R}^{K}$ and $\widetilde{\mathbf{\Sigma}} = \Ourinv\mathbf{\Sigma}{\Ourinv}^{\intercal}$, and
    \begin{align}\label{eq:defxi0}
         \xi_0(\btheta)
        & =
        \exp\bigg(
        {-} \frac{
        \btheta^\intercal \mathbf{G} \btheta
        }{2}
        \bigg).
    \end{align}
    Here,
    \begin{align}\label{eq:defG}
    \mathbf{G}
    & = 
    \widehat{\mathbf{\Sigma}}^{-1}
    - 
    \sum_{k=1}^K  \mathbf{W}_k^\intercal
        \mathbf{H}_k^\intercal
        \mathbf{\Sigma}_{k k}^{-1}
        \mathbf{H}_k 
        \mathbf{W}_k
    \end{align}
    with
    \begin{equation}
    \label{eq:sighatinv}
        \widehat{\mathbf{\Sigma}}^{-1} = (\mathbf{1}_{K}\otimes \mathbf{I}_{d_{\theta}})^\intercal \widetilde{\mathbf{\Sigma}}^{-1} (\mathbf{1}_{K}\otimes \mathbf{I}_{d_{\theta}})\,.
    \end{equation}
    Furthermore, for a given prior $\pri(\btheta)$
    and local posteriors $\pi_k(\btheta)= \poste(\btheta \,\vert\, \meam_k) \propto \pri(\btheta) \ell_k(\btheta)$, the supra-Bayesian fusion result
    $g[\pi_1, \dots, \pi_K](\btheta) = \poste(\btheta \,\vert\, \sufst) \propto \pri(\btheta) \tlik(\btheta)$ is given by
    \begin{align}
    g[\pi_1, \dots, \pi_K](\btheta)
    & \propto 
    \pri(\btheta) \xi_0(\btheta) \prod_{k=1}^K  \frac{\pi_k(\mathbf{W}_k\btheta)}{\pri(\mathbf{W}_k\btheta)}.
    \label{eq:fusedepbayes}
\end{align}
\end{theorem}

Finally, if the prior $\pri(\btheta)$ is Gaussian, then the supra-Bayesian fusion result
$ \poste(\btheta \,\vert\, \sufst)$ is again Gaussian and the fusion rule \eqref{eq:fusedepbayes} can be reduced to a second-order rule involving only the mean and covariance matrix:
\begin{corollary}\label{cor:fusegaussprior}
Under the assumptions of Theorem~\ref{th:fuseliks},
let the prior $\pri(\btheta)$ be Gaussian with mean $\bmu_0$ and covariance matrix $\mathbf{\Sigma}_0$,
i.e., $\pri(\btheta) = \mathcal{N}(\btheta; \bmu_0, \mathbf{\Sigma}_0)$.
Then the supra-Bayesian fusion result
$\poste(\btheta \,\vert\, \sufst)$ is again Gaussian, i.e., $\poste(\btheta \,\vert\, \sufst) =  \mathcal{N}(\btheta; \bmu_1, \mathbf{\Sigma}_1)$, with mean
\begin{equation}
    \bmu_1 = \big(\widehat{\mathbf{\Sigma}}^{-1} + \mathbf{\Sigma}_0^{-1} \big)^{-1} \big( (\mathbf{1}_{K}\otimes \mathbf{I}_{d_{\theta}})^\intercal \widetilde{\mathbf{\Sigma}}^{-1}\sufstr
    + 
    \mathbf{\Sigma}_0^{-1} \bmu_0 \big)
    \label{eq:mu1}
\end{equation}
and covariance matrix
\begin{equation}
    \mathbf{\Sigma}_1 = \big(\widehat{\mathbf{\Sigma}}^{-1} + \mathbf{\Sigma}_0^{-1} \big)^{-1}.
    \label{eq:sigma1}
\end{equation}
\end{corollary}
Here, we recall that $\sufstr = {[\sufstr_1^\intercal, \dots, \sufstr_K^\intercal]}^\intercal$ with $\sufstr_k$ given by \eqref{eq:sufstrgauss}.
A proof of Corollary~\ref{cor:fusegaussprior} is provided in Appendix~\ref{app:fusegaussprior}.

The supra-Bayesian fusion result in \eqref{eq:fusedepbayes} has an  intriguing structure
in that the agent pdfs are first preprocessed by a multiplication in the argument and then combined via a generalized multiplicative pooling function.
The relevance of this fusion rule beyond  the linear Gaussian  setting, especially for approximately  linear  Gaussian  observation models, is an open issue.

As in the scalar case, the supra-Bayesian fusion result $\poste(\btheta \,\vert\, \sufst)$
is in general different from the oracle posterior $\poste(\btheta \, \vert \, \meam)$.
 Again, the oracle posterior is proportional to the product of the prior $\pri(\btheta)$ and the global observation likelihood function $\lik (\meam \, \vert \, \btheta)$ in \eqref{eq:condpdfgauss};
it is easily seen that $\poste(\btheta \, \vert \, \meam)$ is also Gaussian but with mean
\begin{equation}
    \bmu_2 =  \big(\widehat{\mathbf{\Sigma}}_2^{-1} + \mathbf{\Sigma}_0^{-1}\big)^{-1}\big(\mathbf{H}^\intercal {\mathbf{\Sigma}}^{-1}\meam
    + 
    \mathbf{\Sigma}_0^{-1} \bmu_0  
    \big)
    \label{eq:mu2}
\end{equation}
and covariance matrix
\begin{equation*}
    \mathbf{\Sigma}_2 = \big(\widehat{\mathbf{\Sigma}}_2^{-1} + \mathbf{\Sigma}_0^{-1}\big)^{-1},
\end{equation*}
where
\begin{equation}
\label{eq:sigmahat2vec}
    \widehat{\mathbf{\Sigma}}_2^{-1} = \mathbf{H}^\intercal {\mathbf{\Sigma}}^{-1} \mathbf{H}\,.
\end{equation}
The difference can be better understood  by noting that in \eqref{eq:mu1}  
\begin{align*}
    (\mathbf{1}_{K}\otimes \mathbf{I}_{d_{\theta}})^\intercal \widetilde{\mathbf{\Sigma}}^{-1}\sufstr
    & = 
    \mathbf{H}^\intercal \Ourinv^\intercal(\Ourinv\mathbf{\Sigma}{\Ourinv}^{\intercal})^{-1} \Ourinv \meam
\end{align*}
and in \eqref{eq:sighatinv}
\begin{align*}
    (\mathbf{1}_{K}\otimes \mathbf{I}_{d_{\theta}})^\intercal \widetilde{\mathbf{\Sigma}}^{-1} (\mathbf{1}_{K}\otimes \mathbf{I}_{d_{\theta}})
    & = 
    \mathbf{H}^\intercal \Ourinv^\intercal(\Ourinv\mathbf{\Sigma}{\Ourinv}^{\intercal})^{-1} \Ourinv \mathbf{H}\,,
\end{align*}
where we used \eqref{eq:sufstrvecgauss}--\eqref{eq:ourinvleftinv}.
Comparing with $\mathbf{H}^\intercal {\mathbf{\Sigma}}^{-1}\meam$ and $\mathbf{H}^\intercal {\mathbf{\Sigma}}^{-1} \mathbf{H}$ in \eqref{eq:mu2} and \eqref{eq:sigmahat2vec}, respectively, we conclude that, as in the scalar case considered earlier, the difference between the oracle posterior and the supra-Bayesian posterior is that   ${\mathbf{\Sigma}}^{-1}$ is replaced by $\Ourinv^\intercal(\Ourinv\mathbf{\Sigma}{\Ourinv}^{\intercal})^{-1} \Ourinv$.

\section{Outlook}
\label{sec:outlook}


The fusion of pdfs presents numerous interesting aspects beyond those considered in our treatment. Moreover, 
certain
extensions can be envisioned.
In what follows, we suggest
some related 
directions of future research.

\begin{itemize}
\setlength\itemsep{1.5mm}

\item Our discussion of pdf fusion emphasized theoretical considerations. 
In practical implementations, a finite-dimensional representation or parametrization of 
the agent pdfs $q_k(\btheta)$ is required. 
Popular examples are Gaussian, Gaussian mixture, and particle representations \cite{Gunay16,rabin2011wasserstein, srivastava2015wasp}. 
Since these representations are usually approximations
of the true pdfs, a relevant issue is the tradeoff between low representation complexity (small number of parameters) and high accuracy of approximation.
Furthermore, 
algorithms implementing a given pooling function for a given type of parametric representation are required. Examples of 
finite-dimensional parametric fusion rules were considered
in Sections~\ref{sec: opinion_pooling_with_gaussians} and \ref{sec:lingaumeam}.

\item In the case of a centralized agent network where each agent pdf $q_k(\btheta)$ is transmitted to the fusion center via a channel, communication cost is another 
practical issue. 
Although a low-dimensional
parametric representation of the agent pdfs may be used to achieve a low communication cost,
the reduction of communication cost is ultimately a source coding (rate-distortion) problem.

\item In many cases, the aggregate pdf $q(\btheta) = g[q_1,\ldots,q_K](\btheta)$ is not used as the final result but arises as part of a method performing 
a statistical inference task such as estimation, detection, classification, or clustering. In this setting, the pooling function (or certain parameters within 
a given family of pooling functions) should be chosen or optimized 
such that the performance of the statistical inference method is maximized. Note that this is different from the optimization approach considered in 
Section~\ref{sec: pop_optim_appr}.

\item Our discussion assumed the existence of a fusion center that has access to all pdfs $q_k(\btheta)$. In a decentralized agent network, 
there is no fusion center and each agent is able to communicate only with certain neighboring agents. Besides the basic necessity of using a distributed 
communication-and-fusion protocol, 
challenging
aspects in the decentralized setting include communication cost, efficient representation of pdfs, and double counting of information
along cycles in the network graph.

\item In many 
scenarios,
the agent pdfs $q_k(\btheta)$ are time-varying and a temporal sequence $q_k^{(n)}(\btheta)$, where $n = 1,2,\ldots$ is a discrete time index, is 
available at the $k$th agent. This serial setting suggests a sequential variant of pdf fusion in which at each time $n$ the fused pdf is not calculated from scratch but
the previous fusion result is updated using the new set of $q_k^{(n)}(\btheta)$.
Practical implementations of 
sequential updating can be based on both 
parametric and nonparametric representations of the pdfs.

\item The fusion of multiobject pdfs or probability hypothesis densities
of finite point processes (random finite sets), especially in the context
of multitarget tracking, is a topic of active research
\cite{Clark10,Uney13,gao2020multiobject,li2020arithmetic,Li-T_19}. 
While the current focus is on the finite point process counterparts of the linear and log-linear pooling functions,
it would also be interesting to investigate the applicability of the other pooling functions 
considered in Sections~\ref{sec: probabilistic_opinion_pooling} and \ref{sec: pop_optim_appr}. In particular, the fact that the family of H\"older pooling functions offers fusion characteristics that are intermediate between those of the linear and log-linear pooling functions may be relevant to multitarget tracking.
Furthermore, it may be rewarding to reformulate and develop our results on supra-Bayesian pdf fusion in the context of finite point processes.

\item Big data problems allow a natural application of pdf fusion.
When the data to be processed are so large in size that they exceed the capacity of a single computer, it is logical to partition them and process the different parts separately.
Furthermore, data related to some quantity of interest may be available in heterogeneous form, so  that all of the data cannot be processed within a single framework and hence different parts have to be processed separately. 
In either case, the individual processing results can be represented as summaries, which then need to be fused into one overall summary. 
The concepts and techniques presented in this article provide suggestions regarding the construction and fusion of the summaries.
This is of particular interest in the context of modern machine learning methods \cite{neiswanger2014asymptotically,wang2013parallel,bardenet2017markov,qiu2016survey,vehtari2020expectation}.

\item  
Ensemble learning \cite{hastie2009friedman,sagi2018ensemble}, i.e., the combination of the results
of multiple
learning algorithms, is currently one of the most successful learning paradigms. 
At the same time, there is a growing demand for probabilistic machine learning methods that provide along with
a point estimate also a measure of reliability.
Until now, only few works have considered ensembles of probabilistic machine learning methods.
We conjecture that the 
success of the ensemble learning paradigm
will soon lead to its increased use
also in probabilistic machine learning. At that point, it is likely that probabilistic opinion pooling will outperform
the simple linear voting rules that are currently used to combine point estimates.

\item With a collaborative machine learning methodology known as federated learning,  a learning algorithm is trained across multiple decentralized edge devices or servers that hold local data, which are not exchanged  \cite{savazzi2020federated}. 
In other words, model parameters are learned collectively by many interconnected devices without sharing or disclosing local training data. 
The devices  send summaries instead of raw data to a server for fusion. 
Here, again, fusion plays a central role.
The fusion process can be challenging in the case of a large number of heterogeneous devices with different constraints.
Using pdfs to represent the local summaries enables the use of different pdf representations at the individual devices, from simple parametric models to complex kernel density estimates, which can still be combined in a meaningful way. 
Moreover, different levels of quality of the local data can be taken into account by using appropriate weights in the pooling function used for pdf fusion.

\item A potential theoretical basis of pdf fusion that has not been explored in this work is
information geometry, which studies probability theory and statistics  using tools from differential geometry \cite{tang2018information}. 
The focus of information geometry is on statistical manifolds whose points correspond to probability distributions. 
This theoretical framework can be exploited for fusion by assuming that local estimates are posterior pdfs  that correspond to a parametric family with the structure of a Riemannian manifold \cite{tang2018information}. 
One can then formulate pdf fusion, e.g., by considering the fused pdf to be an informative barycenter of the manifold \cite{kim2017wasserstein}.

\item Within the finite-dimensional supra-Bayesian setting,
an explicit fusion rule was obtained only for linear Gaussian observation models
(see Section~\ref{sec:lingaumeam}). 
This fusion rule can formally be used also for nonlinear/non-Gaussian models with known first and second moments. However, it is here unclear how close the obtained fusion result will be to the true supra-Bayesian fusion result.
A characterization of the error for approximately linear Gaussian observation models is an interesting topic for future research.
Another interesting topic is the derivation of explicit supra-Bayesian fusion rules for simple  nonlinear/non-Gaussian observation models.


\item Our supra-Bayesian framework is currently limited to a finite-dimensional setting. 
Although this is the setting most frequently encountered in practical applications, it would be
interesting to find a definition of a likelihood function for random pdfs that do not 
admit a finite-dimensional parameterization. For this, nonparametric Bayesian models \cite{ghosal2017fundamentals} appear to be a feasible starting point. The challenge is to 
model a useful and nontrivial dependence on the parameter $\btheta$ that accounts for the constraint that random pdfs must be nonnegative and integrate to one with probability one.

\end{itemize}

\section{Concluding Remarks}
\label{sec: conclusion}


The problem of fusing multiple pdfs $q_k(\btheta)$, $k = 1, \dots, K$ of a continuous random vector $\btheta$ into an aggregate pdf $q(\btheta) = g[q_1,\ldots,q_K](\btheta)$
has many possible solutions and, indeed, several different  approaches to this fusion problem have been developed in the past decades.
We have attempted to survey and study these approaches and the related  solutions
in a structured and coherent manner.
Our discussion has emphasized a first basic distinction between the axiomatic approach, the optimization approach,
and the conceptually more complex supra-Bayesian framework. 

Regarding the axiomatic approach, we formulated a set of axioms and determined the axioms satisfied by each considered 
pooling function. This analysis demonstrated the prominent role of the linear, log-linear, and multiplicative pooling functions within the axiomatic framework.
However, it also revealed that several desirable axioms are effectively incompatible and postulating those simultaneously implies a dictatorship pooling function.

Regarding the optimization approach, besides other results, we proved that the minimization of the weighted sum of $\alpha$-divergences 
yields the family of H\"{o}lder mean pooling functions. 
This 
family 
contains the two most popular pooling functions---the linear and log-linear pooling functions---as special cases.
Moreover, it offers an infinite number of further interesting pooling functions with different multimodality and tail decay characteristics depending on the choice of a single parameter.


The supra-Bayesian framework is different from the classical probabilistic opinion pooling framework in that 
the pdfs $q_k(\btheta)$
are modeled as random observations, and 
additional information regarding the statistical structure of $\btheta$ is available to the fusion center. 
In this framework,
the optimal aggregate pdf $q(\btheta)$ is the global posterior pdf of $\btheta$ given the pdfs $q_k(\btheta)$. 
Since random functions 
are difficult to work with, we introduced the finite-dimensional 
supra-Bayesian model based on
random ``local statistics.''
Using this framework, 
we formulated a general procedure for obtaining the supra-Bayesian posterior pdf conditioned on all the local statistics, and we derived explicit fusion rules for special cases.

While the theory of pdf fusion appears mature, interesting directions of future work are related to implementation and application aspects. 
We provided some suggestions including implementations using parametric representations, integration into probabilistic methods for multisensor signal processing and machine learning, and extensions to decentralized scenarios and point processes.

\section*{Acknowledgments}
\label{sec: ack}

We would like to thank the anonymous reviewers for insightful comments and constructive criticism, which have resulted in a significant improvement of this article. 
In particular, we are thankful for a comment on a conceptual relationship between pdf fusion and robust hypothesis testing as mentioned in Section \ref{sec: pop_optim_appr}.
We are also grateful to Mr.\ Thomas Kropfreiter for sharing his expertise in target tracking.

\appendices

\section{Proof of Theorem~\ref{th:linearpooling}}
\label{app:proof_linearpooling}

\subsection{Axioms Satisfied by the Linear Pooling Function}
We first show that all the mentioned axioms are satisfied by the linear pooling function.
Let $g[q_1, \dots, q_K](\btheta) = \sum_{k=1}^K w_k q_k(\btheta)$ with $(w_1, \dots, w_K)\in \mathcal{S}_K$.
We first show the ZPP (A\ref{ax:ZPP}).
Assume that for some event $\mA$, we have $Q_k(\mA)=0$ for all $k=1, \dots, K$. 
Because $Q_k(\mA) = \int_{\mA} q_k(\btheta) \, \mathrm{d}\btheta$ and $q_k(\btheta)$ is nonnegative, this implies
$q_k(\btheta)=0$ for almost all $\btheta \in \mA$ and all $k=1, \dots, K$.
Thus,
\begin{equation}
\notag 
   q(\btheta) = g[q_1, \dots, q_K](\btheta) = \sum_{k=1}^K w_k q_k(\btheta) = 0\,,
\end{equation}
for almost all $\btheta \in \mA$. 
Hence, $Q(\mA)= \int_{\mA} q(\btheta) \, \mathrm{d}\btheta=0$, which concludes the proof of the ZPP.

We next show unanimity preservation (A\ref{ax:UP}). 
To this end, assume that $q_k(\btheta) = q_0(\btheta)$ for all $k= 1, \dots, K$.
Then 
\begin{align*}
   q(\btheta) 
   & = g[q_1, \dots, q_K](\btheta) \\
   & = \sum_{k=1}^K w_k q_k(\btheta) \\
   & = q_0(\btheta)\sum_{k=1}^K w_k \\
   & = q_0(\btheta)\,,
\end{align*}
which shows unanimity preservation.

To show the SSFP  (A\ref{ax:SSFP}), we define 
$h \colon [0,1]^K \to [0,1]$ as
\begin{equation}
\label{eq:proofssfplin}
    h(p_1, \dots, p_K) \triangleq \sum_{k=1}^K w_k p_k\,.
\end{equation}
For an arbitrary set $\mA \subseteq \Theta$ and any opinion profile $(q_1, \dots, q_K)$, we have that
\begin{align*}
    Q(\mA) 
    & = \int_{\mA} q(\btheta) \, \mathrm{d}\btheta \\
    & = \int_{\mA} \sum_{k=1}^K w_k q_k(\btheta) \, \mathrm{d}\btheta \\
    & =  \sum_{k=1}^K w_k \int_{\mA} q_k(\btheta) \, \mathrm{d}\btheta \\
    & = \sum_{k=1}^K w_k Q_k(\mA) \\
    & = h(Q_1(\mA), \dots, Q_K(\mA))\,,
\end{align*}
i.e., $h$ satisfies the condition stated in A\ref{ax:SSFP}.

The WSFP (A\ref{ax:WSFP}) follows by setting $h_{\mA}=h$  with $h$ given in \eqref{eq:proofssfplin}.
The likelihood principles (A\ref{ax:LP} and A\ref{ax:WLP}) are obviously satisfied with
\begin{equation*}
    h(t_1, \dots, t_K)  \triangleq \sum_{k=1}^K w_k t_k
\end{equation*}
and $h_{\btheta}=h$, respectively, as is the symmetry statement. 

\subsection{Equivalence Statement}
We now prove the other direction, namely, that each of the assumptions  \ref{en:gssfp}--\ref{en:gwsfpup}  stated in Theorem~\ref{th:linearpooling} implies that $g$ is a linear pooling function.
More specifically, we will show the chain of implications \ref{en:gwsfpup} $\Rightarrow$ \ref{en:gwsfpzpp} $\Rightarrow$ \ref{en:gssfp} $\Rightarrow$ \ref{en:glinpool}. 
Because we already showed \ref{en:glinpool} $\Rightarrow$ \ref{en:gwsfpup}, this implies that \ref{en:glinpool}--\ref{en:gwsfpup} are equivalent, and thus concludes the proof.

\subsubsection{\ref{en:gwsfpup} implies \ref{en:gwsfpzpp}}
We assume that  $g$ satisfies \ref{en:gwsfpup}, i.e., the WSFP (A\ref{ax:WSFP}) and unanimity preservation (A\ref{ax:UP}).
We will show that this implies  that  $g$ satisfies the  ZPP (A\ref{ax:ZPP}), i.e., \ref{en:gwsfpup} implies \ref{en:gwsfpzpp}.
Let $h_{\mA}\colon [0,1]^K\rightarrow [0,1]$ denote the function satisfying \eqref{eq:wsft} for all opinion profiles.
For any set $\mA$ that satisfies $\lvert \mA^c \rvert>0$, let us choose $q_k(\btheta) = \mathbbm{1}_{\mA^c}(\btheta)/\lvert \mA^c \rvert$ for all $k=1, \dots, K$.
Then $Q_k(\mA) = \int_{\mA} q_k(\btheta) \mathrm{d}\btheta = 0$ for all $k$.
By unanimity preservation, this implies $Q(\mA)=0$.
On the other hand, we have
\begin{equation*}
    Q(\mA) \stackrel{\eqref{eq:wsft}}= h_{\mA}(Q_1(\mA), \dots, Q_K(\mA)) = h_{\mA}(0, \dots, 0)
\end{equation*}
and hence 
\begin{equation}
\label{eq:haiszero}
    h_{\mA}(0, \dots, 0) = 0
\end{equation}
for any set $\mA$ such that $\lvert \mA^c \rvert>0$.

To show the ZPP, assume that for a given opinion profile $(q_1, \dots, q_K)$, we have $Q_k(\mA)=0$ for all $k=1, \dots, K$.
Note that this is only possible if $\lvert \mA^c \rvert>0$ as otherwise $Q_k(\Theta) = Q_k(\mA)+ Q_k(\mA^c) = 0$.
Thus, we can calculate $Q(\mA)$ as 
\begin{equation*}
    Q(\mA) \stackrel{\eqref{eq:wsft}}= h_{\mA}(Q_1(\mA), \dots, Q_K(\mA)) = h_{\mA}(0, \dots, 0) 
    \stackrel{\eqref{eq:haiszero}}= 0\,,
\end{equation*}
which shows that the  ZPP (A\ref{ax:ZPP}) is satisfied. 

\subsubsection{\ref{en:gwsfpzpp} implies \ref{en:gssfp}}
Next we show that  \ref{en:gwsfpzpp}, i.e., the WSFP (A\ref{ax:WSFP}) and the ZPP (A\ref{ax:ZPP}), implies \ref{en:gssfp}, i.e., the SSFP (A\ref{ax:SSFP}). 
Let again $h_{\mA}\colon [0,1]^K\rightarrow [0,1]$ denote the function satisfying \eqref{eq:wsft} for all opinion profiles.
Our proof consists of three steps:
\begin{enumerate}
    \item \label{en:prlinpit1}Show that for two nontrivial events $\mA$ and $\mB$  (i.e., $\lvert \mA\rvert, \lvert \mA^c \rvert, \lvert \mB\rvert, \lvert \mB^c \rvert>0$) that have a nontrivial intersection and a nontrivial union, we have $h_{\mA} = h_{\mB}$.
    \item \label{en:prlinpit2}Show that for any  nontrivial events $\mA$ and $\mB$, there exists a nontrivial event $\mC$ such that $\mA\cap \mC$, $\mA\cup \mC$, $\mB\cap \mC$, and $\mB\cup \mC$ are all nontrivial. This implies by step~\ref{en:prlinpit1} that $h_{\mA} = h_{\mC}$ and $h_{\mB} = h_{\mC}$, and thus $h_{\mA} = h_{\mB}$.
    Thus, setting  $h \triangleq h_{\mA}$, we have $h_{\mA'} = h$  for all nontrivial events $\mA'$, and hence the same function $h$ satisfies  \eqref{eq: global_function_h} for all nontrivial events.
    \item \label{en:prlinpit3}Show that the function $h$ satisfies \eqref{eq: global_function_h} also for trivial events.
\end{enumerate}

To show step \ref{en:prlinpit1}, we consider two nontrivial events $\mA$ and $\mB$  that have a nontrivial intersection, in particular, $\lvert \mA\cap \mB \rvert >0$, and a nontrivial union, in particular, $\lvert (\mA\cup \mB)^c \rvert >0$.
We fix arbitrary $(p_1, \dots, p_K)\in [0,1]^K$ and will show that $h_{\mA}(p_1, \dots, p_K)=h_{\mB}(p_1, \dots, p_K)$.
Because $\lvert \mA\cap \mB \rvert >0$ and $\lvert (\mA\cup \mB)^c \rvert >0$,  there exists an opinion profile  $(q_1, \dots, q_K)$ such that 
\begin{equation}
\label{eq:qkacapb}
    Q_k(\mA\cap \mB) = p_k
\end{equation} 
and 
\begin{equation}
\label{eq:qkaxupbc}
    Q_k\big((\mA\cup \mB)^c\big) = 1-p_k\,,
\end{equation} 
for all $k=1, \dots, K$.
Because $\Theta = (\mA\cup \mB) \cup  (\mA\cup \mB)^c$ is a disjoint union and $Q_k(\Theta)=1$, \eqref{eq:qkaxupbc}  implies  
$Q_k(\mA\cup \mB) = p_k$.
Hence, as also $Q_k(\mA\cap \mB) = p_k$ by \eqref{eq:qkacapb}, the difference set $(\mA\cup \mB) \setminus (\mA\cap \mB)$ satisfies  
\begin{equation}
\label{eq:qkpfacbmacb}
    Q_k\big((\mA\cup \mB) \setminus (\mA\cap \mB)\big) = 0\,.
\end{equation}
Because $(\mA\cup \mB) \setminus (\mA\cap \mB) = (\mA\setminus \mB) \cup (\mB\setminus \mA)$, this implies $Q_k(\mA\setminus \mB) = 0$ and $Q_k(\mB\setminus \mA) = 0$.
Thus,
\begin{align}
    Q_k(\mA)
    & = Q_k\big((\mA\setminus \mB) \cup (\mA\cap \mB)\big) 
    \notag \\
    & = 
    Q_k(\mA\setminus \mB) + Q_k(\mA\cap \mB)
    \notag \\
    & = 
    p_k
    \label{eq:qkaispk}
\end{align} 
and, similarly,
\begin{equation}
    Q_k(\mB) = p_k\,.
    \label{eq:qkbispk}
\end{equation} 
Furthermore, 
\begin{align}
    Q_k\big((\mA\setminus \mB) \cup (\mB\setminus \mA)\big) 
    & =
    Q_k\big((\mA\cup \mB) \setminus (\mA\cap \mB)\big) 
    \, \stackrel{\hidewidth \eqref{eq:qkpfacbmacb} \hidewidth}= 
    \; 0\,,
    \label{eq:qkambubmaeq0}
\end{align}
for all $k= 1, \dots, K$.
By the ZPP, \eqref{eq:qkambubmaeq0} implies $Q\big((\mA\setminus \mB) \cup (\mB\setminus \mA)\big)=0$ and, in turn, $Q\big((\mA\setminus \mB)\big)=Q\big( (\mB\setminus \mA)\big)=0$.
Thus,
\begin{align}
    Q(\mA)
    & = 
    Q(\mA\setminus \mB) + Q(\mA\cap \mB)
    \notag \\
    & = 
    Q(\mA\cap \mB) 
    \notag \\
    & = 
    Q(\mB\setminus \mA) + Q(\mA\cap \mB)
    \notag \\
    & = 
    Q(\mB)\,.
    \label{eq:qaeqqc}
\end{align} 
For the functions $h_{\mA}$ and $h_{\mB}$, these properties imply
\begin{align*}
    h_{\mA}(p_1, \dots, p_K) 
    & \stackrel{\hidewidth \eqref{eq:qkaispk}\hidewidth} = \, h_{\mA}(Q_1(\mA), \dots, Q_K(\mA)) 
    \\
    & \stackrel{\hidewidth \eqref{eq:wsft}\hidewidth}= \,Q(\mA) 
    \\
    & \stackrel{\hidewidth \eqref{eq:qaeqqc}\hidewidth}= \, Q(\mB) 
    \\
    & \stackrel{\hidewidth \eqref{eq:wsft}\hidewidth}= \,h_{\mB}(Q_1(\mB), \dots, Q_K(\mB)) 
    \\
    & \stackrel{\hidewidth \eqref{eq:qkbispk}\hidewidth} = \, h_{\mB}(p_1, \dots, p_K)\,,
\end{align*}
i.e.,   
\begin{equation}
\label{eq:haeqhb}
    h_{\mA}(p_1, \dots, p_K)  = h_{\mB}(p_1, \dots, p_K)
\end{equation}
for any nontrivial events $\mA,\mB\subseteq \Theta$ that have a nontrivial intersection and a nontrivial union.

To show step \ref{en:prlinpit2},
we first construct a set 
$\mC\subseteq \mA\cup \mB$ such that $\mA\cap \mC$, $\mB\cap \mC$,  $\mA\cup \mC$, and $\mB\cup \mC$ are nontrivial.
If $\lvert \mA\cap \mB \rvert > 0$, then $\mC = \mA\cap \mB$ can easily be seen to satisfy these assumptions.
If $\lvert \mA\cap \mB \rvert = 0$, we choose 
$\mC=\mC_{\mA} \cup \mC_{\mB}$ where $\mC_{\mA}\subseteq \mA$ with $\lvert \mC_{\mA}\rvert,  \lvert \mA\setminus \mC_{\mA}\rvert>0$ and $\mC_{\mB}\subseteq \mB$ with $\lvert \mC_{\mB}\rvert,  \lvert \mB\setminus \mC_{\mB}\rvert>0$.
The separations $\mA = \mC_{\mA} \cup (\mA\setminus \mC_{\mA})$ and $\mB = \mC_{\mB} \cup (\mB\setminus \mC_{\mB})$ are possible because the Lebesgue measure is nonatomic, i.e., any set of positive Lebesgue measure can be separated into two disjoint sets of positive  Lebesgue measure.

We now choose 
\begin{equation}
\label{eq:choiceh}
    h(p_1, \dots, p_K) = h_{\mA}(p_1, \dots, p_K)
\end{equation}
for any nontrivial set $\mA$.
Then, for any nontrivial set $\mB\subseteq \Theta$, we construct $\mC$ as above and obtain
\begin{align}
    h(Q_1(\mB), \dots, Q_K(\mB)) \,
    & \stackrel{\hidewidth \eqref{eq:choiceh} \hidewidth}= \,h_{\mA}(Q_1(\mB), \dots, Q_K(\mB))
    \notag \\
    & \stackrel{\hidewidth \eqref{eq:haeqhb} \hidewidth}= \,h_{\mC}(Q_1(\mB), \dots, Q_K(\mB))
    \notag \\
    & \stackrel{\hidewidth \eqref{eq:haeqhb} \hidewidth}= \,h_{\mB}(Q_1(\mB), \dots, Q_K(\mB))
    \notag \\
    & \stackrel{\hidewidth \eqref{eq:wsft} \hidewidth}= \,Q(\mB)\,,
    \label{eq:hsatsssfp}
\end{align}
i.e., \eqref{eq: global_function_h} is satisfied for any nontrivial set $\mB$.

It remains to show step \ref{en:prlinpit3}, i.e.,  that with this choice of $h$, \eqref{eq: global_function_h} is also satisfied by trivial sets.
For trivial sets $\mA$, i.e., such that $\lvert \mA\rvert=0$ or $\lvert \mA^c \rvert=0$, we have $Q_k(\mA)=0$ or  $Q_k(\mA)=1$ for all $k=1, \dots, K$, respectively.
Also the fused result must satisfy $Q(\mA)=0$ or  $Q(\mA)=1$, respectively.
Thus, we have to show $h(0, \dots, 0) = 0$ and $h(1, \dots, 1) = 1$ for our choice of $h$ in \eqref{eq:choiceh}.
To this end, let $\mB\subseteq \Theta$ be any nontrivial set and choose an opinion profile $(q_1, \dots, q_K)$ such that $Q_k(\mB)=0$ for all $k=1, \dots, K$.
Then the ZPP implies $Q(\mB)=0$.
On the other hand, since $\mB$ is a nontrivial set and thus \eqref{eq:hsatsssfp} is satisfied, we have 
\begin{equation*}
    h(0, \dots, 0)
    = h(Q_1(\mB), \dots, Q_K(\mB))
    \stackrel{\eqref{eq:hsatsssfp}}= 
    Q(\mB)\,.
\end{equation*}
Thus, $h(0, \dots, 0)=0$.
Furthermore, $Q_k(\mB)=0$ and $Q(\mB)=0$ imply $Q_k(\mB^c)=1$ and $Q(\mB^c)=1$, respectively.
Hence, 
\begin{equation*}
    h(1, \dots, 1)
    = h(Q_1(\mB^c), \dots, Q_K(\mB^c))
    \stackrel{\eqref{eq:hsatsssfp}}= Q(\mB^c)\,.
\end{equation*}
Thus, $h(1, \dots, 1) = 1$.
Hence, we identified a function $h$ such that \eqref{eq: global_function_h} holds for all sets $\mA \subseteq \Theta$.
This concludes
the proof that  \ref{en:gwsfpzpp} implies \ref{en:gssfp}.

\subsubsection{\ref{en:gssfp} implies \ref{en:glinpool}}
Finally, we show that \ref{en:gssfp}, i.e., the SSFP (A\ref{ax:SSFP}) implies \ref{en:glinpool}, i.e., that $g$ is a linear pooling function.
Let $h$ denote the function satisfying \eqref{eq: global_function_h}.
Furthermore, let $\mA,\mB,\mC\subseteq \Theta$ be disjoint events of positive Lebesgue measure.
For arbitrary $p_1, \tilde{p}_1, \dots, p_K, \tilde{p}_K \in [0,1]$ satisfying $p_k+\tilde{p}_k\leq 1$ for all $k=1, \dots, K$, we define an opinion profile $(q_1, \dots, q_K)$ such that 
$Q_k(\mA)= p_k$, $Q_k(\mB) = \tilde{p}_k$, and $Q_k(\mC) = 1-p_k-\tilde{p}_k$.
Because $\mA$ and $\mB$ are disjoint, $Q_k(\mA\cup \mB)= p_k + \tilde{p}_k$ and $Q(\mA\cup \mB) = Q(\mA) + Q(\mB)$.
Thus,
\begin{align}
    & h(p_1+\tilde{p}_1, \dots, p_K+\tilde{p}_K) \notag \\*
    & \quad = h(Q_1(\mA\cup \mB), \dots, Q_K(\mA\cup \mB)) \notag \\
    & \quad \stackrel{\hidewidth \eqref{eq: global_function_h} \hidewidth }= Q(\mA\cup \mB) \notag \\
    & \quad = Q(\mA) + Q(\mB) \notag \\
    & \quad \stackrel{\hidewidth \eqref{eq: global_function_h} \hidewidth }= h(Q_1(\mA), \dots, Q_K(\mA)) + h(Q_1(\mB), \dots, Q_K(\mB)) \notag \\
    & \quad = h(p_1, \dots, p_K) + h(\tilde{p}_1, \dots, \tilde{p}_K)\,,\notag 
\end{align}
i.e., $h$ is an additive function on its domain $[0,1]^K$.
It can moreover be extended to an additive function on $\mathbb{R}^K$.
Because $h$ is also bounded by $1$ on $[0,1]^K$, it must be linear according to \cite[Th.~1, p.~215]{aczel1966lectures},
i.e., 
\begin{equation}
\label{eq:linpoolprobs}
    h(p_1, \dots, p_K) = \sum_{k=1}^K w_k p_k \,.
\end{equation}
Here, the weights $w_k$ must be in $[0,1]$ because $h(p_1, \dots, p_K) \in [0,1]$ for all $(p_1, \dots, p_K) \in  [0,1]^K$.
Furthermore, because 
$1=Q(\Theta) = h(Q_1(\Theta), \dots, Q_K(\Theta)) = h(1, \dots, 1) = \sum_{k=1}^K w_k$, the weights must sum to one.  
We thus have for any event $\mA \subseteq \Theta$
\begin{align*}
    \int_{\mA} q(\btheta) \, \mathrm{d}\btheta 
    & = 
    Q(\mA)
    \notag \\
    & \stackrel{\hidewidth \eqref{eq: global_function_h} \hidewidth}= h(Q_1(\mA), \dots, Q_K(\mA))
    \notag \\
    & \stackrel{\hidewidth \eqref{eq:linpoolprobs} \hidewidth}= \;\sum_{k=1}^K w_k Q_k(\mA)
    \notag \\
    & = \int_{\mA} \sum_{k=1}^K w_k q_k(\btheta) \, \mathrm{d}\btheta\,,
\end{align*}
which implies $q(\btheta) = \sum_{k=1}^K w_k q_k(\btheta)$.

\section{Proof of Theorem~\ref{th:genlinpool}}
\label{app:proof_genlinpool}

\subsection{Axioms Satisfied by the Generalized Linear Pooling Function}
We first show that all the mentioned axioms are satisfied by the generalized linear pooling function.
Let $g[q_1, \dots, q_K](\btheta) = \sum_{k=0}^K w_k q_k(\btheta)$ with $(w_0, \dots, w_K)\in \mathcal{S}_{K+1}$.
To show the WSFP (A\ref{ax:WSFP}), we define for an event $\mA \subseteq \Theta$
\begin{equation}
\label{eq:proofwsfpgenlin}
    h_{\mA}(p_1, \dots, p_K) = w_0 \int_{\mA} q_0(\btheta)\,\mathrm{d}\btheta + \sum_{k=1}^K w_k p_k\,,
\end{equation}
for all $(p_1, \dots, p_K)\in [0,1]^K$.
For any opinion profile $(q_1, \dots, q_K)$, we then have that
\begin{align*}
    Q(\mA) 
    & = \int_{\mA} q(\btheta) \, \mathrm{d}\btheta \\
    & = \int_{\mA} \sum_{k=0}^K w_k q_k(\btheta) \, \mathrm{d}\btheta \\
    & =  \sum_{k=0}^K w_k \int_{\mA} q_k(\btheta) \, \mathrm{d}\btheta \\
    & = w_0 \int_{\mA} q_0(\btheta) \, \mathrm{d}\btheta + \sum_{k=1}^K w_k Q_k(\mA) \\
    & \stackrel{\hidewidth\eqref{eq:proofwsfpgenlin}\hidewidth}= \,h_{\mA}(Q_1(\mA), \dots, Q_K(\mA))\,,
\end{align*}
i.e., $h_{\mA}$ satisfies the condition stated in A\ref{ax:WSFP}.
The weak likelihood principle (A\ref{ax:WLP}) is obviously satisfied
with 
\begin{equation*}
    h_{\btheta}(t_1, \dots, t_K) = w_0 q_0(\btheta) + \sum_{k=1}^K w_k t_k\,,
\end{equation*}
as is the symmetry statement in Theorem~\ref{th:genlinpool}. 

\subsection{Converse Statement}
We now prove the converse statement in Theorem~\ref{th:genlinpool}, i.e.,  that any pooling function $g$ that satisfies the WSFP (A\ref{ax:WSFP}) is a generalized linear pooling function.
For each event $\mA \subseteq \Theta$, let $h_{\mA}\colon [0,1]^K\rightarrow [0,1]$ denote the function satisfying \eqref{eq:wsft} for all opinion profiles.
Our proof consists of three steps:
First, we construct the pdf $q_0$ and the corresponding weight $w_0\leq 1$.
In the second step, we show that by adapting each function $h_{\mA}$ to $\tilde{h}_{\mA}= \big(h_{\mA}-w_0 \int_{\mA} q_0(\btheta)\, \mathrm{d}\btheta\big)/(1-w_0)$, we obtain a linear pooling function.
Finally, we show that this implies that $g$ is a generalized linear pooling function. 

\subsubsection*{Step 1: Construct $q_0$ and $w_0$}
We define 
\begin{equation}
\label{eq:defq0}
    Q_0(\mA)= h_{\mA}(0, \dots, 0)
\end{equation} 
for all nontrivial (i.e., $\lvert \mA\rvert, \lvert \mA^c \rvert >0$) events $\mA \subseteq \Theta$.
The pdf $q_0$ will be a weighted version of a density associated with $Q_0$.
Thus, we first show that $Q_0$ can be expressed as an integral
$Q_0(\mA) = \int_{\mA} \tilde{q}_{0}(\btheta) \, \mathrm{d}\btheta$.

Let $\mA_0$ be a fixed nontrivial event. 
Because $\lvert \mA_0^c \rvert >0$,
there exists an opinion profile $(q^{(\mA_0)}_1, \dots, q^{(\mA_0)}_K)$ such that $Q^{(\mA_0)}_k(\mA_0)= \int_{\mA_0} q^{(\mA_0)}_k(\btheta)\,\mathrm{d}\btheta = 0$ for all $k=1, \dots, K$.
We denote the fused pdf of this particular profile as $q^{(\mA_0)}(\btheta)$ and the resulting probability measure as 
\begin{equation}
\label{eq:defqa0}
    Q^{(\mA_0)}(\mA) = \int_{\mA_0} q^{(\mA_0)}(\btheta)\,\mathrm{d}\btheta\,.
\end{equation}
Then, 
\begin{equation*}
    Q^{(\mA_0)}(\mA)  \stackrel{\eqref{eq:wsft}}= h_{\mA}(Q^{(\mA_0)}_1(\mA), \dots, Q^{(\mA_0)}_K(\mA))\,.
\end{equation*}
In particular, for any event $\mA \subseteq \mA_0$, we have $Q^{(\mA_0)}_k(\mA)=0$ for all $k=1, \dots, K$ (because $Q^{(\mA_0)}_k(\mA)\leq Q^{(\mA_0)}_k(\mA_0)=0$), and thus we obtain further
\begin{align*}
    Q^{(\mA_0)}(\mA)
    & = h_{\mA}(0, \dots, 0) 
    \, \stackrel{\hidewidth\eqref{eq:defq0}\hidewidth}=\, Q_0(\mA)\,.
\end{align*}
Recalling \eqref{eq:defqa0}, we conclude that 
the fused pdf $q^{(\mA_0)}(\btheta)$ satisfies
\begin{equation}
\label{eq:q0a0}
    Q_0(\mA) = \int_{\mA} q^{(\mA_0)}(\btheta)\,\mathrm{d}\btheta \quad \text{for any event $\mA \subseteq \mA_0$.}
\end{equation}
Following the same steps with  $\mA_0$ replaced by  $\mA_0^c$, we obtain a pdf $q^{(\mA_0^c)}(\btheta)$ such that we have
\begin{equation}
\label{eq:q0a0c}
    Q_0(\mA) = \int_{\mA} q^{(\mA_0^c)}(\btheta)\,\mathrm{d}\btheta \quad \text{for any event $\mA \subseteq \mA_0^c$.}
\end{equation}
Now for an arbitrary nontrivial event $\mB\subseteq \Theta$, there exists an opinion profile $(q^{(\mB)}_1, \dots, q^{(\mB)}_K)$ such that $Q^{(\mB)}_k(\mB)= \int_{\mB} q^{(\mB)}_k(\btheta)\,\mathrm{d}\btheta = 0$ for all $k=1, \dots, K$.
Again we denote the fused probability measure as $Q^{(\mB)}$.
We thus obtain for  $Q_0(\mB)$ as defined by \eqref{eq:defq0}
\begin{align}
    Q_0(\mB) \,
    & = \, h_{\mB}(0, \dots, 0) 
    \notag \\
    & = h_{\mB}(Q^{(\mB)}_1(\mB), \dots, Q^{(\mB)}_K(\mB))
    \notag \\
    & \stackrel{\hidewidth\eqref{eq:wsft}\hidewidth}= 
    Q^{(\mB)}(\mB)\,.
    \label{eq:q0beqqb}
\end{align}
Because $\mB$ can be decomposed into disjoint subsets according to $\mB= (\mB\cap \mA_0) \cup (\mB\cap \mA_0^c)$, we further obtain from \eqref{eq:q0beqqb}
\begin{align*}
    Q_0(\mB)
    & =
    Q^{(\mB)}(\mB\cap \mA_0) + Q^{(\mB)}(\mB\cap \mA_0^c) \notag \\
    & \stackrel{\hidewidth\eqref{eq:wsft}\hidewidth}= 
    h_{\mB\cap \mA_0}(Q^{(\mB)}_1(\mB\cap \mA_0), \dots, Q^{(\mB)}_K(\mB\cap \mA_0)) 
    \notag \\* 
    & \quad + h_{\mB\cap \mA_0^c}(Q^{(\mB)}_1(\mB\cap \mA_0^c), \dots, Q^{(\mB)}_K(\mB\cap \mA_0^c))
    \notag \\
    & \stackrel{\hidewidth(a)\hidewidth}= h_{\mB\cap \mA_0}(0, \dots, 0) + h_{\mB\cap \mA_0^c}(0, \dots, 0)
    \notag \\
    & \stackrel{\hidewidth\eqref{eq:defq0}\hidewidth}= \, 
    Q_0(\mB\cap \mA_0) + Q_0(\mB\cap \mA_0^c)\,,
\end{align*}
where we used $Q^{(\mB)}_k(\mB)=0$ in $(a)$.
Using \eqref{eq:q0a0} with $\mA=\mB\cap \mA_0$ and \eqref{eq:q0a0c} with $\mA=\mB\cap \mA_0^c$, this implies 
\begin{align*}
    Q_0(\mB)
    & = \int_{\mB\cap \mA_0} q^{(\mA_0)}(\btheta) \,\mathrm{d}\btheta + \int_{\mB\cap \mA_0^c} q^{(\mA_0^c)}(\btheta) \,\mathrm{d}\btheta \notag \\
    & = \int_{\mB} \tilde{q}_{0}(\btheta) \,\mathrm{d}\btheta\,,
\end{align*}
where we defined 
\begin{equation*}
    \tilde{q}_{0}(\btheta)
    \triangleq
    \begin{cases}
        q^{(\mA_0)}(\btheta) & \text{ if $\btheta\in \mA_0$} \\
        q^{(\mA_0^c)}(\btheta) & \text{ if $\btheta\in \mA_0^c$.}
    \end{cases}
\end{equation*}
We thus found an integral representation for $Q_0$ and can define 
\begin{equation}
\label{eq:densQ0isq0tilde}
    Q_0(\mB) = \int_{\mB} \tilde{q}_{0}(\btheta) \,\mathrm{d}\btheta
\end{equation} 
also for trivial events $\mB$.
The nonnegativity of $q^{(\mA_0)}(\btheta)$ and $q^{(\mA_0^c)}(\btheta)$ implies that $\tilde{q}_{0}(\btheta)$ is nonnegative and, in turn, that $Q_0$ is a measure.
However, $\tilde{q}_{0}(\btheta)$ is not a pdf in general.

We define 
\begin{equation}
\label{eq:defw0inappb}
    w_0 \triangleq Q_0(\Theta)
\end{equation} 
(note that this implies $w_0\geq 0$) and 
\begin{equation}
\label{eq:q0tilde}
    q_{0}(\btheta) \triangleq \frac{\tilde{q}_0(\btheta)}{w_0},
\end{equation}
provided $w_0 \neq 0$. 
If $w_0= 0$, we choose $q_{0}(\btheta)$ as an arbitrary pdf.
We claim that   $w_0\leq 1$.
To prove this claim, let $\mB_n$ be a sequence of nontrivial events such that $\mB_n\subseteq \mB_{n+1}$ and $\lim_{n\to \infty} \mB_n= \Theta$.
For each $\mB_n$, there exists an opinion profile $(q^{(\mB_n)}_1, \dots, q^{(\mB_n)}_K)$ such that $Q^{(\mB_n)}_k(\mB_n)= \int_{\mB_n} q^{(\mB_n)}_k(\btheta)\,\mathrm{d}\btheta=0$ for all $k=1, \dots, K$.
Again we denote the sequence of fused probability measures as $Q^{(\mB_n)}$.
Following the steps in \eqref{eq:q0beqqb}, we have that 
\begin{equation*}
    Q_0(\mB_n)= Q^{(\mB_n)}(\mB_n) \leq 1\,,
\end{equation*} 
because $Q^{(\mB_n)}$ is a probability measure. 
The continuity from below of measures \cite[Lem.~3.4]{bartle1995elements} implies $w_0 = Q_0(\Theta) = \lim_{n\to\infty} Q_0(\mB_n) \leq 1$.

A similar argument can be employed to show (for later use) that for any nontrivial event $\mA \subseteq \Theta$ and arbitrary probabilities $p_k$
\begin{align}
    h_{\mA}(p_1, \dots, p_K)
    & \geq Q_0(\mA)\,.
    \label{eq:haboundha0}
\end{align}
Indeed, for any nontrivial event $\mA \subseteq \Theta$, let $\mB_n\subseteq \mA$ be a sequence satisfying $\lvert \mA\setminus \mB_n\rvert > 0$, $\mB_n\subseteq \mB_{n+1}$ for all $n\in \mathbb{N}$, and $\lim_{n\to \infty} \mB_n= \mA$.
Then for each $n\in \mathbb{N}$ there exists an opinion profile $(q^{(\mB_n)}_1, \dots, q^{(\mB_n)}_K)$ satisfying $Q^{(\mB_n)}_k(\mB_n)=0$, $Q^{(\mB_n)}_k(\mA\setminus \mB_n)=p_k$, and, in turn, $Q^{(\mB_n)}_k(\mA)=p_k$.
Again we denote the sequence of fused probability measures as $Q^{(\mB_n)}$.
Following the steps in \eqref{eq:q0beqqb}, $Q_0(\mB_n)= Q^{(\mB_n)}(\mB_n)$.
Thus, we have
\begin{align}
    h_{\mA}(p_1, \dots, p_K)
    & =
    h_{\mA}(Q^{(\mB_n)}_1(\mA), \dots, Q^{(\mB_n)}_K(\mA)) \notag \\
    &
    \stackrel{\hidewidth\eqref{eq:wsft}\hidewidth}= Q^{(\mB_n)}(\mA)  \notag \\
    & \geq Q^{(\mB_n)}(\mB_n) \notag \\
    & = Q_0(\mB_n)\,.
    \label{eq:haboundhb0}
\end{align}
Here,  $h_{\mA}(p_1, \dots, p_K)$ does not depend on $n$. 
Hence, we can take the limit on the right-hand side of \eqref{eq:haboundhb0} and obtain
\begin{align*}
    h_{\mA}(p_1, \dots, p_K)
    & \geq \lim_{n\to\infty}Q_0(\mB_n)
     = Q_0(\mA)\,,
\end{align*}
using again the continuity from below of $Q_0$.

\subsubsection*{Step 2: Define $\tilde{h}_{\mA}$ and prove that it defines a linear pooling function}
We  define 
\begin{equation}
\label{eq:tildeh}
    \tilde{h}_{\mA}(p_1, \dots, p_K)
    \triangleq \frac{h_{\mA}(p_1, \dots, p_K)- Q_0(\mA)}{1 - w_0}\,.
\end{equation} 
Here, we have to assume that $w_0 < 1$.
Thus, we first show that $g$ is a generalized linear pooling function in the case $w_0 = 1$.
In this case, for any nontrivial event $\mA \subseteq \Theta$ and arbitrary probabilities $p_k$, 
we choose an opinion profile that satisfies $Q_k(\mA)=p_k$ and hence  $Q_k(\mA^c)=1-p_k$ for all $k=1, \dots, K$.
We then have
\begin{align}
    1 & =  Q(\mA)+Q(\mA^c)
    \notag \\
    & \stackrel{\hidewidth\eqref{eq:wsft}\hidewidth}= h_{\mA}(p_1, \dots, p_K)
    + h_{\mA^c}(1-p_1, \dots, 1-p_K)
    \notag \\
    & \stackrel{\hidewidth\eqref{eq:haboundha0}\hidewidth}\geq \, Q_0(\mA) + Q_0(\mA^c) 
    \notag \\
    & \stackrel{\hidewidth\eqref{eq:defw0inappb}\hidewidth}=  \, w_0
    \notag \\
    & = 1\,.
    \notag
\end{align}
Thus, the inequality in the third line is actually an equality, which is only possible if $h_{\mA}(p_1, \dots, p_K)=Q_0(\mA)$.
Because $\mA$ and the $p_k$ were chosen arbitrarily, we have $h_{\mA}(p_1, \dots, p_K) = Q_0(\mA)$ independently of the probabilities $p_k$.
By \eqref{eq:wsft},
this further implies for any opinion profile $(q_1, \dots, q_K)$ that the aggregate pdf $q$  satisfies 
\begin{align*}
    \int_{\mA} q(\btheta) \,\mathrm{d}\btheta 
    & = Q(\mA) 
    \\
    & \stackrel{\hidewidth\eqref{eq:wsft}\hidewidth}= h_{\mA}(Q_1(\mA), \dots, Q_K(\mA))
    \\
    & = Q_0(\mA)
    \\
    & \stackrel{\hidewidth\eqref{eq:densQ0isq0tilde}\hidewidth}= \int_{\mA} \tilde{q}_0(\btheta) \,\mathrm{d}\btheta
\end{align*}
for all events $\mA$. 
Hence, $q(\btheta) = \tilde{q}_0(\btheta)$,
which implies that $g$ is a dogmatic pooling function (which is a special case of a generalized linear pooling function with weights $w_0=1$, $w_k=0$ for $k=1, \dots, K$).
This concludes the proof for the special case $w_0=1$, and thus we can assume $w_0<1$ in what follows.

We define  a new fusion rule $\tilde{g}$ by
\begin{equation}
\label{eq:tildeq}
    \tilde{g}[q_1, \dots, q_K](\btheta) 
    \triangleq \frac{g[q_1, \dots, q_K](\btheta)-   w_0 q_0(\btheta)}{1 - w_0} 
\end{equation}
and claim that it satisfies the WSFP with the functions $\tilde{h}_{\mA}$ defined by \eqref{eq:tildeh}.
Indeed, we have for any opinion profile $(q_1, \dots, q_K)$ and any event $\mA\subseteq \Theta$ that
\begin{align}
    \tilde{Q}(\mA) 
    & = 
    \int_{\mA} \tilde{g}[q_1, \dots, q_K](\btheta) \,\mathrm{d}\btheta
    \notag \\
    & =
    \int_{\mA} \frac{g[q_1, \dots, q_K](\btheta)-   w_0 q_0(\btheta)}{1 - w_0} \,\mathrm{d}\btheta
    \notag \\
    & \stackrel{\hidewidth(a)\hidewidth}=
    \, \frac{h_{\mA}(Q_1(\mA), \dots, Q_K(\mA))-  \int_{\mA} \tilde{q}_0(\btheta)\,\mathrm{d}\btheta}{1 - w_0} 
    \notag \\
    & \stackrel{\hidewidth\eqref{eq:densQ0isq0tilde}\hidewidth}=
    \, \frac{h_{\mA}(Q_1(\mA), \dots, Q_K(\mA))-  Q_0(\mA)}{1 - w_0} 
    \notag \\
    & \stackrel{\hidewidth\eqref{eq:tildeh}\hidewidth}=\, 
    \, \tilde{h}_{\mA}(Q_1(\mA), \dots, Q_K(\mA))\,,
    \label{eq:Qtildeeqqtilde}
\end{align}
where we used in $(a)$ that, by \eqref{eq:wsft}, $h_{\mA}(Q_1(\mA), \allowbreak \dots, \allowbreak Q_K(\mA)) = Q(\mA) = \int_{\mA} g[q_1, \dots, q_K](\btheta)$ and, by \eqref{eq:q0tilde},  $w_0 q_0(\btheta) = \tilde{q}_0(\btheta)$.
Furthermore, we claim that $\tilde{g}$ satisfies the ZPP.
To prove this, let $(q_1, \dots, q_K)$ be an opinion profile and $\mA$ a nontrivial event such that $Q_k(\mA)=0$ for all $k=1, \dots, K$.
Because $h_{\mA}(0, \dots, 0)= Q_0(\mA)$,
\begin{align*}
    \tilde{Q}(\mA) \, 
    & \stackrel{\hidewidth\eqref{eq:Qtildeeqqtilde}\hidewidth}=
    \, \tilde{h}_{\mA}(Q_1(\mA), \dots, Q_K(\mA))
    \\
    & = 
    \, \tilde{h}_{\mA}(0, \dots, 0)
    \\
    & \stackrel{\hidewidth\eqref{eq:tildeh}\hidewidth}= 
    \, \frac{h_{\mA}(0, \dots, 0)- Q_0(\mA)}{1 - w_0}
    \\
    & = 
    \, \frac{Q_0(\mA) -  Q_0(\mA)}{1 - w_0} 
    \\
    & = 0\,,
\end{align*}
proving the ZPP.

Finally, to see that $\tilde{g}$ is a valid pooling function, we first show that for any  $(p_1, \dots, p_K)\in [0,1]^K$, the function $\tilde{h}_{\mA}(p_1, \dots, p_K)$ is nonnegative. 
This follows from  \eqref{eq:tildeh}, \eqref{eq:haboundha0}, and our assumption $w_0 < 1$.
Hence, the measure $\tilde{Q}$ is nonnegative and thus also the associated density $\tilde{g}[q_1, \dots, q_K](\btheta)$ must be nonnegative.
The fact that $\tilde{g}[q_1, \dots, q_K](\btheta)$ integrates to one follows directly from the definition \eqref{eq:tildeq} and the fact that $g[q_1, \dots, q_K](\btheta)$ and $q_0$ are pdfs.

Because $\tilde{g}$ is a pooling function that satisfies  the WSFP and  the ZPP, Theorem~\ref{th:linearpooling}
 implies that it is a linear pooling function, i.e., 
 \begin{equation}
 \label{eq:tildeq2}
     \tilde{g}[q_1, \dots, q_K](\btheta) = \sum_{k=1}^K w_k q_k(\btheta)\,,
 \end{equation}
with $(w_1, \dots, w_K)\in \mathcal{S}_K$.
 
\subsubsection*{Step 3: Conclude that $\tilde{g}$ is a generalized linear pooling function}
Combining \eqref{eq:tildeq} and \eqref{eq:tildeq2}, we obtain
\begin{equation*}
    \sum_{k=1}^K w_k q_k(\btheta) = \frac{g[q_1, \dots, q_K](\btheta) -   w_0 q_0(\btheta)}{1 - w_0}  
\end{equation*}
or, equivalently,
\begin{equation*}
   g[q_1, \dots, q_K](\btheta) = w_0 q_{0}(\btheta) + \sum_{k=1}^K (1 - w_0)w_k q_k(\btheta)\,.
\end{equation*}
From $\sum_{k=1}^K  w_k=1$, it follows that 
$w_0 + \sum_{k=1}^K (1 - w_0)w_k$ is one.
Thus, $g$ is a generalized linear pooling function.

\section{Proof of the Equivalence Statement in Theorem~\ref{th:mult_pooling}}
\label{app:proof_mult_pooling}
We only show that \ref{en:glib}, i.e., 
$g$ satisfies individualized Bayesianity (A\ref{ax:IB}) and $g[q_0, \dots, q_0](\btheta)= q_0(\btheta)$ for some pdf $q_0$, implies \ref{en:gmultpool}, i.e., $g$ is a multiplicative pooling function.
The other direction is obvious.

Thus, let us assume that $g[q_0, \dots, q_0](\btheta)= q_0(\btheta)$ for some pdf $q_0$.
We have to show that, for 
any opinion profile  $(q_1, \dots, q_K)$ such that $q_k/q_0$ is bounded for all $k=1, \dots, K$ (recall that we only consider those opinion profiles in the multiplicative pooling function), $q$ is of the form \eqref{eq: multiplicative_pooling}, i.e., 
\begin{align*}
    g[q_1, \dots, q_K](\btheta)
    & \propto  (q_0(\btheta))^{1-K} \prod_{k=1}^K q_k(\btheta)\,.
\end{align*}
To this end, we first note that $q_k = q_0^{(\ell_k)}$ (see \eqref{eq: discrete_bayesian_update}) with $\ell_k = q_k/q_0$ for  all $k=1, \dots, K$.
Thus, 
\begin{align*}
    g[q_1, \dots, q_K](\btheta)
    & = g\Big[q_0^{(q_1/q_0)}, \dots, q_0^{(q_K/q_0)}\Big](\btheta)\,.
\end{align*}
By iteratively using individualized Bayesianity \eqref{eq:indivbayes} with $\ell=q_k/q_0$ for each $k=1, \dots, K$, we obtain further
\begin{align*}
    g[q_1, \dots, q_K](\btheta)
    & \propto g\Big[q_0, q_0^{(q_2/q_0)}, \dots, q_0^{(q_K/q_0)}\Big](\btheta)  \frac{q_1(\btheta)}{q_0(\btheta)} \notag \\
    & \propto g[q_0, \dots, q_0](\btheta) \prod_{k=1}^K \frac{q_k(\btheta)}{q_0(\btheta)} \notag \\
    & =  (q_0(\btheta))^{1-K} \prod_{k=1}^K q_k(\btheta)\,,
\end{align*}
which is \eqref{eq: multiplicative_pooling} and thus concludes the proof.

\section{Partial Proof of Theorem~\ref{th:dict}}
\label{app:proof_dict}

\subsubsection{\ref{en:gwsfpip} implies \ref{en:gssfpip}}
We first show that \ref{en:gwsfpip}, i.e., the WSFP  (A\ref{ax:WSFP}) and independence preservation (A\ref{ax:IP}), implies \ref{en:gssfpip}, i.e., the SSFP  (A\ref{ax:SSFP}) and independence preservation (A\ref{ax:IP}).
To this end, we show that independence preservation implies the ZPP (A\ref{ax:ZPP}).
The ZPP and the assumed WSFP in turn imply the SSFP by Theorem~\ref{th:linearpooling}.

To show  that independence preservation implies the ZPP, assume that for some event $\mA$, we have $Q_k(\mA)=0$ for all $k=1, \dots, K$. 
This implies that 
\begin{equation*}
    Q_k(\mA\cap \mA) = Q_k(\mA)= 0 = Q_k(\mA)Q_k(\mA)\,.
\end{equation*}
Independence preservation now implies that also $Q$ must satisfy
$Q(\mA\cap \mA) = Q(\mA)Q(\mA)$, and thus that either $Q(\mA) =0$ or  $Q(\mA) =1$.
In the first case, the proof of the ZPP is finished.
In the second case, i.e., $Q(\mA) = 1$,  there must exist a subset $\mB\subseteq \mA$ such that $Q(\mB) = 1/2$.
However, because $\mB\subseteq \mA$ and $Q_k(\mA)= 0$, we have that also $Q_k(\mB)=0$, and thus  we again have that $Q_k(\mB\cap \mB) = 0 = Q_k(\mB)Q_k(\mB)$.
This implies that $Q(\mB)$ is either $0$ or $1$, which is a contradiction to $Q(\mB) = 1/2$. 
Thus,  $Q(\mA) =0$ is the only valid conclusion, which proves that the ZPP is satisfied.

\subsubsection{\ref{en:gwsfpeb} implies \ref{en:gssfpeb}}
We next show that \ref{en:gwsfpeb}, i.e., the WSFP  (A\ref{ax:WSFP}) and external Bayesianity (A\ref{ax:EB}), implies \ref{en:gssfpeb}, i.e., the SSFP  (A\ref{ax:SSFP}) and external Bayesianity (A\ref{ax:EB}).
Thus, we have to show that the WSFP  and external Bayesianity imply   the SSFP.

By Theorem~\ref{th:genlinpool}, the WSFP implies the weak likelihood principle (A\ref{ax:WLP}). 
Furthermore, by Theorem~\ref{th:genloglinpool}, the weak likelihood principle and external Bayesianity imply that $g$ is a generalized log-linear pooling function, i.e.,
\begin{equation}
\label{eq:appdgenloglin}
    g[q_1,\ldots,q_K](\btheta) = c\, \xi_0(\btheta) \prod_{k=1}^K \left(q_k(\btheta)\right)^{w_k}
\end{equation}
for all positive opinion profiles.
Finally, by Theorem~\ref{th:genlinpool}, the WSFP implies that $g$ is also a generalized linear pooling function, i.e., 
\begin{equation}
\label{eq:appdgenlin}
    g[q_1,\ldots,q_K](\btheta) = \sum_{k=0}^K w'_{k} q_{k}(\btheta)
\end{equation}
for all opinion profiles.
Thus, combining \eqref{eq:appdgenloglin} and \eqref{eq:appdgenlin}, we have
\begin{equation}
\label{eq:genloglineqgenlin}
    c\, \xi_0(\btheta) \prod_{k=1}^K \left(q_k(\btheta)\right)^{w_k} = \sum_{k=0}^K w'_{k} q_{k}(\btheta)
\end{equation}
for all positive opinion profiles.
Note that $q_0(\btheta)$ is not necessarily positive, i.e., it may be zero for certain values of $\btheta$.

We choose an arbitrary positive pdf $\tilde{q}_0(\btheta)$ and $\varepsilon \in (0,1)$ and consider the  opinion profile 
$(\varepsilon \tilde{q}_0 + (1-\varepsilon) q_0, \dots, \varepsilon \tilde{q}_0 + (1-\varepsilon) q_0)$.
Since $\tilde{q}_0(\btheta)$ is positive, this is a positive opinion profile for any  $\varepsilon \in (0,1)$.
Using it in \eqref{eq:genloglineqgenlin} gives
\begin{align*}
    & c\, \xi_0(\btheta) \big(\varepsilon \tilde{q}_0(\btheta) + (1-\varepsilon) q_0(\btheta)\big) 
    \\*
    & \quad = w'_{0}q_0(\btheta) + (1- w'_{0}) \big(\varepsilon \tilde{q}_0(\btheta) + (1-\varepsilon) q_0(\btheta)\big),
\end{align*}
where $\sum_{k=1}^K w_k =1$ and $\sum_{k=0}^K w'_k =1$ were used, or, equivalently,
\begin{equation}
\label{eq:th8epstrick}
    c\, \xi_0(\btheta)
    = w'_{0}\frac{q_0(\btheta)}{\varepsilon \tilde{q}_0(\btheta) + (1-\varepsilon) q_0(\btheta)} + 1- w'_{0}\,.
\end{equation}
Taking the limit $\varepsilon \to 0$ in \eqref{eq:th8epstrick}, we obtain
\begin{equation}
\label{eq:cxiisone}
    c\, \xi_0(\btheta)
    = 
    \begin{cases}
    1 & \textrm{if } q_0(\btheta) > 0 \\
    1-w'_0 & \textrm{if } q_0(\btheta) = 0 \,.
    \end{cases}
\end{equation}
Inserting into \eqref{eq:appdgenloglin} and evaluating \eqref{eq:appdgenloglin} for the  opinion profile 
$(\tilde{q}_0, \dots, \tilde{q}_0)$ yields
\begin{equation*}
    g[\tilde{q}_0,\ldots,\tilde{q}_0](\btheta) = 
    \begin{cases}
    \tilde{q}_0(\btheta) & \textrm{if } q_0(\btheta) > 0 \\
    (1-w'_0)\tilde{q}_0(\btheta) & \textrm{if } q_0(\btheta) = 0 \,.
    \end{cases}
\end{equation*}
Because $g[\tilde{q}_0,\ldots,\tilde{q}_0](\btheta)$ is a pdf, this implies
\begin{align}
    1
    & = \int_{\Theta} g[\tilde{q}_0,\ldots,\tilde{q}_0](\btheta) \,\mathrm{d}\btheta
    \notag \\
    & = 
    \int_{\{\btheta\in \Theta: q_0(\btheta) > 0\}} \tilde{q}_0(\btheta) \,\mathrm{d}\btheta 
    \notag \\*
    & \quad 
    + 
    (1-w'_0)\int_{\{\btheta\in \Theta: q_0(\btheta) = 0\}} \tilde{q}_0(\btheta) \,\mathrm{d}\btheta \,.
\label{eq:goftq0isapdf}
\end{align}
On the other hand, because $\tilde{q}_0(\btheta)$ is a pdf, we have
\begin{align}
    1
    & = \int_{\Theta} \tilde{q}_0(\btheta) \,\mathrm{d}\btheta
    \notag \\
    & 
    = 
    \int_{\{\btheta\in \Theta: q_0(\btheta) > 0\}} \tilde{q}_0(\btheta) \,\mathrm{d}\btheta 
    + 
    \int_{\{\btheta\in \Theta: q_0(\btheta) = 0\}} \tilde{q}_0(\btheta) \,\mathrm{d}\btheta \,.
\label{eq:tq0isapdf}
\end{align}
Combining \eqref{eq:goftq0isapdf} and \eqref{eq:tq0isapdf}, we obtain
\begin{equation*}
    (1-w'_0)\int_{\{\btheta\in \Theta: q_0(\btheta) = 0\}} \tilde{q}_0(\btheta) \,\mathrm{d}\btheta = 
    \int_{\{\btheta\in \Theta: q_0(\btheta) = 0\}} \tilde{q}_0(\btheta) \,\mathrm{d}\btheta  
\end{equation*}
or equivalently
\begin{equation*}
    w'_0 \int_{\{\btheta\in \Theta: q_0(\btheta) = 0\}} \tilde{q}_0(\btheta) \,\mathrm{d}\btheta  
    = 0\,.
\end{equation*}
Since $\tilde{q}_0(\btheta)$ is a positive pdf on $\Theta$, this can only hold if either $w'_0=0$ or $\lvert \{\btheta\in \Theta: q_0(\btheta) = 0\}\rvert = 0$.
In the first case, \eqref{eq:appdgenlin} implies that $g$ is actually a linear pooling function and thus, by Theorem~\ref{th:linearpooling}, $g$ satisfies the SSFP.
In the second case, $q_0(\btheta) > 0$ almost everywhere and thus \eqref{eq:cxiisone} states that  $c\, \xi_0(\btheta)=1$.
Using the opinion profile $(q_1, \dots, q_K)= (\tilde{q}_0, \dots, \tilde{q}_0)$ 
in \eqref{eq:genloglineqgenlin} now gives
\begin{equation}
\label{eq:q1toseewp0is0}
    \tilde{q}_0(\btheta)
    = w'_{0}q_0(\btheta) + (1- w'_{0}) \tilde{q}_0(\btheta)\,.
\end{equation}
In particular, let us partition $\Theta$ into disjoint sets $\mA_1$, $\mA_2$ satisfying $\int_{\mA_1} q_0(\btheta) \mathrm{d}\btheta = \int_{\mA_2} q_0(\btheta) \mathrm{d}\btheta = 1/2$, and let us choose
\begin{equation*}
    \tilde{q}_0(\btheta)
    = 
    \begin{cases}
    \frac{3}{2} q_0(\btheta) & \textrm{if } \btheta\in \mA_1 \\
    \frac{1}{2} q_0(\btheta) & \textrm{if } \btheta\in \mA_2 \,.
    \end{cases}
\end{equation*}
Then \eqref{eq:q1toseewp0is0} yields for all $\btheta \in \mA_1$
\begin{equation*}
    \frac{3}{2} q_0(\btheta)
    = 
    w'_{0}q_0(\btheta) + (1- w'_{0}) \frac{3}{2} q_0(\btheta)
    = 
    \bigg(\frac{3}{2} - \frac{1}{2} w'_{0}\bigg) q_0(\btheta)\,.
\end{equation*}
This implies $w'_{0}=0$, and hence we again conclude from \eqref{eq:appdgenlin} that $g$ is actually a linear pooling function, and thus, by Theorem~\ref{th:linearpooling}, that $g$ satisfies the SSFP.

\subsubsection{\ref{en:gssfpgb} implies \ref{en:gdictpool}}
Finally, we prove that \ref{en:gssfpgb}, i.e.,  the SSFP  (A\ref{ax:SSFP}) and generalized Bayesianity (A\ref{ax:GB}), implies  \ref{en:gdictpool}, i.e.,  that $g$ is a dictatorship pooling function.
By Theorem~\ref{th:linearpooling}, the SSFP implies that $g$ is a linear pooling function, i.e., 
\begin{equation}
\label{eq:linpoolinproof}
    g[q_1,\ldots,q_K](\btheta) = \sum_{k=1}^K w_{k} q_{k}(\btheta)
\end{equation}
with $(w_1, \dots, w_K)\in \mathcal{S}_K$.
We will show that for an arbitrary $k$ the weight $w_k$ is either $0$ or $1$, which is equivalent to $g$ being a dictatorship pooling function. 

We first choose a positive function $f$ and two disjoint sets $\mA$ and $\mB$ such that 
$\Theta = \mA\cup \mB$ and
$\int_{\mA} f(\btheta)\,\mathrm{d}\btheta = \int_{\mB} f(\btheta)\,\mathrm{d}\btheta = 1$.
We fix an arbitrary $k$ and define an opinion profile $(q_1, \dots, q_K)$ by setting 
\begin{equation*}
    q_k(\btheta) =
    \begin{cases}
    \frac{1}{3}f(\btheta) & \textrm{if }\btheta \in \mA \\
    \frac{2}{3}f(\btheta) & \textrm{if }\btheta \in \mB
    \end{cases}
\end{equation*}
and $q_{k'}=q_0$ for all $k'\neq k$, where
\begin{equation}
\label{eq:defq0ashalff}
    q_0(\btheta) =\frac{1}{2}f(\btheta)\,.
\end{equation}
Inserting this opinion profile into the fusion rule \eqref{eq:linpoolinproof} and using $\sum_{k'\neq k}w_{k'}= 1-w_k$ gives 
\begin{align}
    g[q_1,\ldots,q_K](\btheta) 
    & = 
    \begin{cases}
    \big(\frac{1}{3}w_k + \frac{1}{2}(1-w_k)\big) f(\btheta) & \textrm{if }\btheta \in \mA \\
    \big(\frac{2}{3}w_k + \frac{1}{2}(1-w_k)\big) f(\btheta) & \textrm{if }\btheta \in \mB
    \end{cases} \notag \\
    & = 
    \begin{cases}
    \big(\frac{1}{2} - \frac{1}{6} w_k\big) f(\btheta) & \textrm{if }\btheta \in \mA \\
    \big(\frac{1}{2} + \frac{1}{6} w_k\big) f(\btheta) & \textrm{if }\btheta \in \mB\,.
    \end{cases}
    \label{eq:gpriepr}
\end{align}

Next, we use generalized Bayesianity with $\ell_{k'}=\ell$ for all $k'=1, \dots, K$, where
\begin{equation*}
    \ell(\btheta) =
    \begin{cases}
    1 & \textrm{if }\btheta \in \mA \\
    2 & \textrm{if }\btheta \in \mB\,.
    \end{cases}
\end{equation*}
We easily obtain (see~\eqref{eq: discrete_bayesian_update})
\begin{equation*}
    q_k^{(\ell)}(\btheta) =
    \begin{cases}
    \frac{1}{5}f(\btheta) & \textrm{if }\btheta \in \mA \\
    \frac{4}{5}f(\btheta) & \textrm{if }\btheta \in \mB\,,
    \end{cases}
\end{equation*}
and
\begin{equation}
\label{eq:q0ellasthirdsf}
    q_0^{(\ell)}(\btheta) =
    \begin{cases}
    \frac{1}{3}f(\btheta) & \textrm{if }\btheta \in \mA \\
    \frac{2}{3}f(\btheta) & \textrm{if }\btheta \in \mB\,.
    \end{cases}
\end{equation}
Now, \eqref{eq:linpoolinproof} gives
\begin{align}
    g[q_1^{(\ell)},\ldots,q_K^{(\ell)}](\btheta) 
    & = 
    \begin{cases}
    \big(\frac{1}{5}w_k + \frac{1}{3}(1-w_k)\big) f(\btheta) & \textrm{if }\btheta \in \mA \\
    \big(\frac{4}{5}w_k + \frac{2}{3}(1-w_k)\big) f(\btheta) & \textrm{if }\btheta \in \mB
    \end{cases} \notag \\
    & = 
    \begin{cases}
    \big(\frac{1}{3} - \frac{2}{15} w_k\big) f(\btheta) & \textrm{if }\btheta \in \mA \\
    \big(\frac{2}{3} + \frac{2}{15} w_k\big) f(\btheta) & \textrm{if }\btheta \in \mB\,.
    \end{cases}
    \label{eq:gpostepr}
\end{align}
On the other hand, because $g$ satisfies generalized Bayesianity, there exists a function $h[\ell, \dots, \ell]$ 
such that 
\begin{equation}
\label{eq:gsatgenb}
    g\big[q_1^{(\ell)},\ldots,q_K^{(\ell)}\big](\btheta) 
    = \frac{g[q_1,\ldots,q_K](\btheta) h[\ell, \dots, \ell](\btheta)}{c_{\ell}},
\end{equation}
where $c_{\ell} = \int_{\Theta} g[q_1,\ldots,q_K](\btheta) h[\ell, \dots, \ell](\btheta)\, \mathrm{d}\btheta$.
Inserting \eqref{eq:gpriepr} and \eqref{eq:gpostepr} into \eqref{eq:gsatgenb} gives
\begin{equation}
\label{eq:conw1eqA}
    \frac{1}{3} - \frac{2}{15} w_k = \frac{\big(\frac{1}{2} - \frac{1}{6} w_k\big) h[\ell, \dots, \ell](\btheta)}{c_{\ell}}
\end{equation}
for all $\btheta \in \mA$
and 
\begin{equation}
\label{eq:conw1eqB}
    \frac{2}{3} + \frac{2}{15} w_k = \frac{\big(\frac{1}{2} + \frac{1}{6} w_k\big) h[\ell, \dots, \ell](\btheta)}{c_{\ell}}
\end{equation}
for all $\btheta\in \mB$.

Using again the generalized Bayesianity of $g$, we also have
\begin{equation}
\label{eq:gofallq0}
    g\big[q_0^{(\ell)},\ldots,q_0^{(\ell)}\big](\btheta) 
    = \frac{g[q_0,\ldots,q_0](\btheta) h[\ell, \dots, \ell](\btheta)}{c_{0,\ell}},
\end{equation}
where $c_{0,\ell} = \int_{\Theta} g[q_0,\ldots,q_0](\btheta) h[\ell, \dots, \ell](\btheta) \,\mathrm{d}\btheta$.
Because linear pooling functions are unanimity preserving (see Theorem~\ref{th:linearpooling}), we have $g\big[q_0^{(\ell)},\ldots,q_0^{(\ell)}\big](\btheta)= q_0^{(\ell)}(\btheta)$ and $g[q_0,\ldots,q_0](\btheta)= q_0(\btheta)$, and thus \eqref{eq:gofallq0} is equivalent to 
\begin{equation*}
    q_0^{(\ell)}(\btheta) = \frac{q_0(\btheta) h[\ell, \dots, \ell](\btheta)}{c_{0,\ell}},
\end{equation*}
or, inserting \eqref{eq:q0ellasthirdsf} and \eqref{eq:defq0ashalff}, 
\begin{equation*}
    \frac{1}{3} = \frac{\frac{1}{2} h[\ell, \dots, \ell](\btheta)}{c_{0,\ell}}
\end{equation*}
for all $\btheta \in \mA$
and 
\begin{equation*}
    \frac{2}{3}  = \frac{\frac{1}{2}  h[\ell, \dots, \ell](\btheta)}{c_{0,\ell}}
\end{equation*}
for all $\btheta\in \mB$.
We thus obtain
\begin{equation*}
    h[\ell, \dots, \ell](\btheta) =
    \begin{cases}
    \frac{2}{3}c_{0,\ell} & \textrm{if }\btheta \in \mA \\
    \frac{4}{3}c_{0,\ell} & \textrm{if }\btheta \in \mB\,.
    \end{cases}
\end{equation*}
Inserting this into \eqref{eq:conw1eqA} and \eqref{eq:conw1eqB} yields
\begin{equation*}
    \frac{1}{3} - \frac{2}{15} w_k = \frac{\big(\frac{1}{2} - \frac{1}{6} w_k\big) \frac{2}{3}c_{0,\ell}}{c_{\ell}}
\end{equation*}
and 
\begin{equation*}
    \frac{2}{3} + \frac{2}{15} w_k = \frac{\big(\frac{1}{2} + \frac{1}{6} w_k\big) \frac{4}{3}c_{0,\ell}}{c_{\ell}}
\end{equation*}
or, equivalently,
\begin{equation*}
    \frac{\frac{1}{3} - \frac{2}{15} w_k}{\frac{1}{3} - \frac{1}{9} w_k}
    = \frac{c_{0,\ell}}{c_{\ell}}
    = \frac{\frac{2}{3} + \frac{2}{15} w_k}{\frac{2}{3} + \frac{2}{9} w_k}\,.
\end{equation*}
This amounts to the quadratic equation
$w_k^2 - w_k=0$,
which has the solutions $w_k=0$ and $w_k=1$.
Since $k$ was arbitrary, this concludes the proof.

\section{Proof of Lemma~\ref{lem:gbwsfpdog}}
\label{app:prooflemgbwsfpdog}

By Theorem~\ref{th:genlinpool}, the WSFP implies that $g$ is a generalized linear pooling function, i.e., 
\begin{equation}
\label{eq:genlinpoolinproof}
    g[q_1,\ldots,q_K](\btheta) = \sum_{k=0}^K w_{k} q_{k}(\btheta)
\end{equation}
with $(w_0, \dots, w_K)\in \mathcal{S}_{K+1}$.
We will show that $w_0$ is either $0$ or $1$, which is equivalent to $g$ being either a linear pooling function or a dogmatic pooling function. 

We first choose two disjoint sets $\mA$ and $\mB$ such that 
$\Theta = \mA\cup \mB$ and
$\int_{\mA} q_0(\btheta)\,\mathrm{d}\btheta = \int_{\mB} q_0(\btheta)\,\mathrm{d}\btheta = 1/2$.
Furthermore, we choose an opinion profile $(q_1, \dots, q_K)$ as
\begin{equation*}
    q_k(\btheta) =
    \begin{cases}
    \frac{2}{3}q_0(\btheta) & \textrm{if }\btheta \in \mA \\
    \frac{4}{3}q_0(\btheta) & \textrm{if }\btheta \in \mB
    \end{cases}
\end{equation*}
 for all $k = 1, \dots, K$.
Inserting this opinion profile into \eqref{eq:genlinpoolinproof} and using $\sum_{k=1}^K w_k = 1-w_0$
gives 
\begin{align}
    g[q_1,\ldots,q_K](\btheta) 
    & = 
    \begin{cases}
    \big(w_0 + \frac{2}{3}(1-w_0)\big) q_0(\btheta) & \textrm{if }\btheta \in \mA \\
    \big(w_0 + \frac{4}{3}(1-w_0)\big) q_0(\btheta) & \textrm{if }\btheta \in \mB
    \end{cases} \notag \\
    & = 
    \begin{cases}
    \big(\frac{2}{3} + \frac{1}{3} w_0\big) q_0(\btheta) & \textrm{if }\btheta \in \mA \\
    \big(\frac{4}{3} - \frac{1}{3} w_0\big) q_0(\btheta) & \textrm{if }\btheta \in \mB\,.
    \end{cases}
    \label{eq:g2priepr}
\end{align}

Next, we use generalized Bayesianity with $\ell_{k}=\ell$ for all $k = 1, \dots, K$, where
\begin{equation}
\label{eq:deflas1or2}
    \ell(\btheta) =
    \begin{cases}
    1 & \textrm{if }\btheta \in \mA \\
    2 & \textrm{if }\btheta \in \mB\,.
    \end{cases}
\end{equation}
We easily obtain (see~\eqref{eq: discrete_bayesian_update})
\begin{equation*}
    q_k^{(\ell)}(\btheta) =
    \begin{cases}
    \frac{2}{5}q_0(\btheta) & \textrm{if }\btheta \in \mA \\
    \frac{8}{5}q_0(\btheta) & \textrm{if }\btheta \in \mB
    \end{cases}
\end{equation*}
and then \eqref{eq:genlinpoolinproof} gives
\begin{align}
    g[q_1^{(\ell)},\ldots,q_K^{(\ell)}](\btheta) 
    & = 
    \begin{cases}
    \big(w_0 + \frac{2}{5}(1-w_0)\big) q_0(\btheta) & \textrm{if }\btheta \in \mA \\
    \big(w_0 + \frac{8}{5}(1-w_0)\big) q_0(\btheta) & \textrm{if }\btheta \in \mB
    \end{cases} \notag \\
    & = 
    \begin{cases}
    \big(\frac{2}{5} + \frac{3}{5} w_0\big) q_0(\btheta) & \textrm{if }\btheta \in \mA \\
    \big(\frac{8}{5} - \frac{3}{5} w_0\big) q_0(\btheta) & \textrm{if }\btheta \in \mB\,.
    \end{cases}
    \label{eq:g2postepr}
\end{align}
Because $g$ satisfies generalized Bayesianity, we have that there exists a function $h[\ell, \dots, \ell]$ such that 
\begin{equation}
\label{eq:g2satgenb}
    g\big[q_1^{(\ell)},\ldots,q_K^{(\ell)}\big](\btheta) 
    = \frac{g[q_1,\ldots,q_K](\btheta) h[\ell, \dots, \ell](\btheta)}{c_{\ell}},
\end{equation}
where $c_{\ell} = \int_{\Theta} g[q_1,\ldots,q_K](\btheta) h[\ell, \dots, \ell](\btheta) \mathrm{d}\btheta$.
Inserting \eqref{eq:g2priepr} and \eqref{eq:g2postepr} into \eqref{eq:g2satgenb} gives
\begin{equation}
\label{eq:con2w1eqA}
    \frac{2}{5} + \frac{3}{5} w_0
    = \frac{\big(\frac{2}{3} + \frac{1}{3} w_0\big) h[\ell, \dots, \ell](\btheta)}{c_{\ell}}
\end{equation}
for all $\btheta \in \mA$
and 
\begin{equation}
\label{eq:con2w1eqB}
    \frac{8}{5} - \frac{3}{5} w_0 
    = \frac{\big(\frac{4}{3} - \frac{1}{3} w_0\big) h[\ell, \dots, \ell](\btheta)}{c_{\ell}}
\end{equation}
for all $\btheta\in \mB$.

Using again the generalized Bayesianity of $g$, we also have
\begin{equation}
\label{eq:condgballq0}
    g\big[q_0^{(\ell)},\ldots,q_0^{(\ell)}\big](\btheta) 
    = \frac{g[q_0,\ldots,q_0](\btheta) h[\ell, \dots, \ell](\btheta)}{c_{0,\ell}},
\end{equation}
where $c_{0,\ell} = \int_{\Theta} g[q_0,\ldots,q_0](\btheta) h[\ell, \dots, \ell](\btheta) \mathrm{d}\btheta$.
Using \eqref{eq:deflas1or2} and \eqref{eq: discrete_bayesian_update}, we obtain
\begin{equation*}
    q_0^{(\ell)}(\btheta) =
    \begin{cases}
    \frac{2}{3}q_0(\btheta) & \textrm{if }\btheta \in \mA \\
    \frac{4}{3}q_0(\btheta) & \textrm{if }\btheta \in \mB\,.
    \end{cases}
\end{equation*}
Inserting into~\eqref{eq:genlinpoolinproof} yields
\begin{align}
    g\big[q_0^{(\ell)},\ldots,q_0^{(\ell)}\big](\btheta) 
    & = 
    \begin{cases}
    \big(w_0 + \frac{2}{3} (1-w_0)\big)q_0(\btheta) & \textrm{if }\btheta \in \mA \\
    \big(w_0 + \frac{4}{3} (1-w_0)\big)q_0(\btheta) & \textrm{if }\btheta \in \mB
    \end{cases}
    \notag \\
    & = 
    \begin{cases}
    \big(\frac{2}{3} +  \frac{1}{3} w_0\big)q_0(\btheta) & \textrm{if }\btheta \in \mA \\
    \big(\frac{4}{3} -  \frac{1}{3} w_0\big)q_0(\btheta) & \textrm{if }\btheta \in \mB\,.
    \end{cases}
    \label{eq:gq0ell}
\end{align}
Furthermore, again by~\eqref{eq:genlinpoolinproof},  $g[q_0,\ldots,q_0](\btheta) = q_0(\btheta)$.
Inserting this and \eqref{eq:gq0ell} into \eqref{eq:condgballq0}, we obtain
\begin{equation*}
    \frac{2}{3} +  \frac{1}{3} w_0 = \frac{ h[\ell, \dots, \ell](\btheta)}{c_{0,\ell}}
\end{equation*}
for all $\btheta \in \mA$
and 
\begin{equation*}
    \frac{4}{3} -  \frac{1}{3} w_0  = \frac{ h[\ell, \dots, \ell](\btheta)}{c_{0,\ell}}
\end{equation*}
for all $\btheta\in \mB$.
Thus,
\begin{equation*}
    h[\ell, \dots, \ell](\btheta) =
    \begin{cases}
    \big(\frac{2}{3} +  \frac{1}{3} w_0\big) c_{0,\ell} & \textrm{if }\btheta \in \mA \\
    \big(\frac{4}{3} -  \frac{1}{3} w_0\big) c_{0,\ell} & \textrm{if }\btheta \in \mB\,.
    \end{cases}
\end{equation*}
Inserting this into \eqref{eq:con2w1eqA} and \eqref{eq:con2w1eqB} gives
\begin{equation*}
    \frac{2}{5} + \frac{3}{5} w_0
    = \frac{\big(\frac{2}{3} + \frac{1}{3} w_0\big)^2 c_{0,\ell}}{c_{\ell}}
\end{equation*}
and 
\begin{equation*}
    \frac{8}{5} - \frac{3}{5} w_0 
    = \frac{\big(\frac{4}{3} - \frac{1}{3} w_0\big)^2 c_{0,\ell}}{c_{\ell}}
\end{equation*}
or, equivalently,
\begin{equation*}
    \frac{\frac{2}{5} + \frac{3}{5} w_0}{\big(\frac{2}{3} + \frac{1}{3} w_0\big)^2 }
    = \frac{c_{0,\ell}}{c_{\ell}}
    = \frac{\frac{8}{5} - \frac{3}{5} w_0}{\big(\frac{4}{3} - \frac{1}{3} w_0\big)^2 }\,.
\end{equation*}
This amounts to the cubic equation 
$
    w_0^3 
    -3 w_0^2
    +2 w_0
    = 0
$,
which has the solutions $w_0=0$, $w_0=1$, and $w_0=2$.
Since $w_0$ cannot be larger than one, only the solutions $w_0=0$ and $w_0=1$ remain.
In the first case, $g$ is a linear pooling function, which satisfies the SSFP by Theorem~\ref{th:linearpooling}.
Hence, since $g$ satisfies both the SSFP and generalized Bayesiantity, it reduces to a dictatorship pooling function by Theorem~\ref{th:dict}.
In the second case, $g$ is a dogmatic pooling function.

\section{Proof of Theorem~\ref{th:implstruct}}
\label{app:thimplstruct}


The implications in \ref{en:implSSFP} follow from Theorem~\ref{th:linearpooling} because the SSFP implies that $g$ is a linear pooling function and in turn satisfies the ZPP (A\ref{ax:ZPP}), unanimity preservation (A\ref{ax:UP}), the WSFP (A\ref{ax:WSFP}), the likelihood principle (A\ref{ax:LP}), and the weak likelihood principle (A\ref{ax:WLP}).
Similarly, the implications in \ref{en:implWSFP} follow from Theorem~\ref{th:genlinpool}.
Implication~\ref{en:implLP} follows directly from the concerned axioms.
Implication~\ref{en:implIP} is shown in the first part of the proof of Theorem~\ref{th:dict} in Appendix~\ref{app:proof_dict}.
It remains to show implication~\ref{en:implIB}, i.e.,  that individualized Bayesianity implies generalized Bayesianity.
This can easily be seen by defining
\begin{align}
    & h[\ell_1,\ldots,\ell_K](\btheta) 
    \triangleq \prod_{k=1}^K \ell_k(\btheta)\,.
    \label{eq:genbayeschoice}
\end{align}
Indeed, because $g$ satisfies individualized Bayesianity, iterative application of \eqref{eq:indivbayes} implies
\begin{align*}
    g[q_1^{(\ell_1)},\ldots,q_K^{(\ell_K)}](\btheta)
    & \propto g\Big[q_1, q_2^{(\ell_2)},\ldots,q_K^{(\ell_K)}\Big](\btheta)  \ell_1(\btheta) \notag \\
    & \propto g[q_1, \dots, q_K](\btheta) \prod_{k=1}^K \ell_k(\btheta) \notag \\
    & \propto g[q_1, \dots, q_K]^{\big(\prod_{k=1}^K \ell_k\big)}(\btheta).
\end{align*}
Thus, \eqref{eq:gnbayes} is satisfied by $h$ defined in \eqref{eq:genbayeschoice}.

\section{Proof of Theorem~\ref{th:optimal_hellinger}
(Constrained Minimization of the Weighted Average of $\alpha$-Divergences)}
\label{appendix: proof_weighted_hellinger_simplex}

Let $f_\alpha(x)=\frac{x^{\alpha}-1}{\alpha (\alpha-1)}$. The inverse function is given by
\begin{equation}
    \label{eq: f_alpha_inverse}
    f_{\alpha}^{-1}(x)=\left(x\alpha (\alpha-1)+1\right)^{1/\alpha}.
\end{equation}
Furthermore, we have that for two functions $p_1(\btheta)$ and $p_2(\btheta)$
\begin{align}
    \label{eq: f_alpha_decomp}
    f_\alpha\left(\frac{p_1(\btheta)}{p_2(\btheta)}\right)&=\frac{\left(\frac{p_1(\btheta)}{p_2(\btheta)}\right)^\alpha-1}{\alpha (\alpha-1)} \nonumber \\
    &=\frac{(p_1(\btheta))^\alpha- (p_2(\btheta))^\alpha}{(p_2(\btheta))^\alpha \alpha (\alpha-1)} \nonumber \\[2mm]
    &=\frac{(p_1(\btheta))^\alpha-1}{(p_2(\btheta))^\alpha \alpha (\alpha-1)}-\frac{(p_2(\btheta))^\alpha-1}{(p_2(\btheta))^\alpha \alpha (\alpha-1)} \nonumber \\[2mm]
    &=\frac{f_\alpha(p_1(\btheta))-f_\alpha(p_2(\btheta))}{(p_2(\btheta))^\alpha} . 
\end{align}
Therefore, the objective function in \eqref{eq: f_divergence_fusion_optimization} for $f(x) = f_\alpha(x)$ can be written as
\begin{equation*}
    \label{eq: Hellinger_objective}
    \begin{aligned}
    & \sum_{k=1}^K w_k \mathcal{D}_{\alpha}(q_k\|\varphi)
    \\
    & \rule{5mm}{0mm} = \;\sum_{k=1}^K w_k \int_{{\Theta}} \varphi(\btheta)  f_\alpha\left(\frac{q_k(\btheta)}{\varphi(\btheta)}\right) \mathrm{d}\btheta \\
    & \rule{5mm}{0mm} \stackrel{ \hidewidth \eqref{eq: f_alpha_decomp} \hidewidth } =  \; \sum_{k=1}^K w_k \int_{{\Theta}} \varphi(\btheta) \frac{f_\alpha(q_k(\btheta))-f_\alpha(\varphi(\btheta))}{\left(\varphi(\btheta)\right)^\alpha} \mathrm{d}\btheta.
    \end{aligned}
\end{equation*}
Interchanging the summation and the integral gives
\begin{align*}
    & \sum_{k=1}^K w_k \mathcal{D}_{\alpha}(q_k\|\varphi)
    \\
    &=\int_{{\Theta}} \varphi(\btheta)\sum_{k=1}^K w_k\frac{ f_\alpha(q_k(\btheta))- f_\alpha(\varphi(\btheta))}{\left(\varphi(\btheta)\right)^\alpha}  \mathrm{d}\btheta \\
    &\stackrel{ \hidewidth (a) \hidewidth }=\int_{{\Theta}} \varphi(\btheta)\frac{ \Big(\sum_{k=1}^K w_k f_\alpha(q_k(\btheta)) \Big)-  f_\alpha(\varphi(\btheta))}{\left(\varphi(\btheta)\right)^\alpha}  \mathrm{d}\btheta \\
    &=\int_{{\Theta}} \varphi(\btheta)\frac{f_{\alpha}\left(f^{-1}_\alpha\left(\sum_{k=1}^K w_k f_\alpha(q_k(\btheta))\right)\right)-f_\alpha(\varphi(\btheta))}{\left(\varphi(\btheta)\right)^\alpha}  \mathrm{d}\btheta \\
    &\stackrel{ \hidewidth \eqref{eq: f_alpha_decomp} \hidewidth }=\int_{{\Theta}} \varphi(\btheta)f_\alpha\left(\frac{f^{-1}_\alpha\left(\sum_{k=1}^K w_k f_\alpha(q_k(\btheta))\right)}{\varphi(\btheta)}\right)  \mathrm{d}\btheta ,
\end{align*}
where we used in $(a)$ that $\sum_{k=1}^K w_k=1$.
Since $\varphi$ is a pdf and $f_\alpha(x)=\frac{x^{\alpha}-1}{\alpha (\alpha-1)}$ is a convex function for $\alpha\in\mathbb{R}\setminus\{0, 1\}$, we can apply Jensen's inequality%
\footnote{Jensen's inequality \cite[Th.~3.3]{rudin1986} asserts that for a pdf $\varphi(\cdot)$, a measurable function $\zeta(\cdot)$, and a convex function $\psi(\cdot)$ we have that $\int \psi(\zeta(\btheta)) \varphi(\btheta) \mathrm{d}\btheta \geq \psi\big(\int \zeta(\btheta) \varphi(\btheta) \mathrm{d}\btheta\big)$, with equality if and only if  the function $\zeta$ is constant almost everywhere.} 
to obtain the following lower bound on the objective function:
\begin{align}
     \int_{{\Theta}} \varphi(\btheta)&f_\alpha\left(\frac{f^{-1}_\alpha\left(\sum_{k=1}^K w_k f_\alpha(q_k(\btheta))\right)}{\varphi(\btheta)}\right)  \mathrm{d}\btheta
    \notag \\
    &\geq 
    f_\alpha\left(\int_{{\Theta}} \varphi(\btheta)\frac{f^{-1}_\alpha\left(\sum_{k=1}^K w_k f_\alpha(q_k(\btheta))\right)}{\varphi(\btheta)} \mathrm{d}\btheta \right) 
    \notag \\
    &= 
    f_\alpha\left(\int_{{\Theta}} f^{-1}_\alpha\left(\sum_{k=1}^K w_k f_\alpha(q_k(\btheta))\right) \mathrm{d}\btheta \right),
    \label{eq: min_Hellinger_jensens}
\end{align}
with equality if and only if the function
\begin{equation*}
    \zeta(\btheta) \triangleq \frac{f^{-1}_\alpha\left(\sum_{k=1}^K w_k f_\alpha(q_k(\btheta))\right)}{\varphi(\btheta)}
\end{equation*}
is constant almost everywhere. Note that this is equivalent to $\varphi(\btheta)\propto f^{-1}_\alpha(\sum_{k=1}^K w_k f_\alpha(q_k(\btheta)))$.
Since the right-hand side of \eqref{eq: min_Hellinger_jensens} is independent of $\varphi$, it is a lower bound for any choice of $\varphi$, and hence the function $\varphi(\btheta)$ minimizing the objective function (which is the desired solution $q(\btheta)$ in \eqref{eq: f_divergence_fusion_optimization}) is the one for which this lower bound is achieved with equality, i.e.,
\begin{align*}
    q(\btheta)
    & \propto f^{-1}_\alpha\left(\sum_{k=1}^K w_k f_\alpha(q_k(\btheta))\right) \\
    & \stackrel{ \hidewidth \eqref{eq: f_alpha_inverse} \hidewidth } = 
    \, \left(\left(\sum_{k=1}^K w_k f_\alpha(q_k(\btheta))\right)\alpha(\alpha-1)+1\right)^{1/\alpha} \\
    & = 
    \, \left(\sum_{k=1}^K w_k \left(q_k(\btheta)\right)^{\alpha}-\sum_{k=1}^K w_k+1\right)^{1/\alpha} \\
    & =
    \, \left(\sum_{k=1}^K w_k \left(q_k(\btheta)\right)^{\alpha}\right)^{1/\alpha}.
\end{align*}
We conclude that the solution to \eqref{eq: f_divergence_fusion_optimization} when $f(x) = f_\alpha(x)$ is $q(\btheta) = c \big(\sum_{k=1}^K w_k (q_k(\btheta))^{\alpha}\big)^{1/\alpha}$, where  $c = 1/\int_{\Theta}\big(\sum_{k=1}^K w_k (q_k(\btheta))^{\alpha}\big)^{1/\alpha}\mathrm{d}\btheta$.

\section{Characterization of the Reverse $\alpha$-divergence}
\label{app:revalphadiv}
We will show that 
$\mathcal{D}_{\alpha}(\varphi \| q_k) = \mathcal{D}_{\alpha^*}(q_k\|\varphi)$, where $\alpha^*=1-\alpha$.
To this end, we will use \eqref{eq:invfdivergence} with $f(x) = f_\alpha(x) = \frac{x^\alpha-1}{\alpha(\alpha-1)}$.
By $f^*(x)=xf(1/x)$, we have 
\begin{align*}
    f_{\alpha}^*(x)
    & = x \frac{x^{-\alpha}-1}{\alpha(\alpha-1)}
    \\
    & = \frac{x^{-\alpha+1}-x}{\alpha(\alpha-1)}
    \\
    & = \frac{x^{-(\alpha-1)}-1}{\alpha(\alpha-1)} - \frac{1}{\alpha(\alpha-1)}(x-1)
    \\
    & = f_{\alpha^*}(x) - \frac{1}{\alpha(\alpha-1)}(x-1)\,.
\end{align*}
Thus, up to the additive term $- \frac{1}{\alpha(\alpha-1)}(x-1)$, the function $f_{\alpha}^*(x)$ is equal to $f_{\alpha^*}(x)$.
Now, by \cite[Prop.~1]{Sason_2018}, an $f$-divergence does not change if $f(x)$ is replaced by $f(x)+c(x-1)$ for an arbitrary $c\in \mathbb{R}$.
Hence, $f_{\alpha}^*$ and $f_{\alpha^*}$ result in the same $f$-divergence, and \eqref{eq:invfdivergence} together with \eqref{eq: df_alpha_def} implies
\begin{align*}
    \mathcal{D}_{\alpha}(\varphi \| q_k)
    & = \mathcal{D}_{f_\alpha} (\varphi\|q_k)
    \\
    & = \mathcal{D}_{f_\alpha^*} (q_k\|\varphi)
    \\
    & = \mathcal{D}_{f_{\alpha^*}} (q_k\|\varphi)
    \\
    & = \mathcal{D}_{\alpha^*}(q_k\|\varphi)\,.
\end{align*}

\section{Proof of Theorem~\ref{th:optimal_l2}
(Constrained Minimization of the Weighted Average of Squared $L_2$ Distances)}
\label{appendix: proof_weighted_l2_simplex}
We want to find 
\begin{equation}
\label{eq:optprobl2}
     q = \argmin_{\varphi\in\mathcal{P}}\sum_{k=1}^Kw_k\|q_k-\varphi\|_2^2.
\end{equation}
To this end, we note that 
\begin{align}
    & \min_{\varphi\in\mathcal{P}}\sum_{k=1}^K w_k\|q_k-\varphi\|_2^2
    \notag \\
    & \quad = 
    \min_{\varphi\in\mathcal{P}}\int_{\Theta}\sum_{k=1}^K w_k \big( q_k(\btheta) - \varphi(\btheta) \big)^2 \mathrm{d}\btheta
    \notag \\
    & \quad  \geq
    \int_{\Theta}\min_{\varphi(\btheta)\geq 0} \bigg\{\sum_{k=1}^K w_k \big( q_k(\btheta) - \varphi(\btheta) \big)^2 \bigg\} \mathrm{d}\btheta\,.
    \label{eq:pointminphistar}
\end{align}
For each fixed $\btheta$, the function value $\varphi(\btheta)$ that achieves the minimum $\min_{\varphi(\btheta)\geq 0} \sum_{k=1}^Kw_k \big( q_k(\btheta) - \varphi(\btheta)  \big)^2$ is easily seen to be 
\begin{equation*}
    \varphi^*(\btheta) = \sum_{k=1}^K w_k q_k(\btheta)\,.
\end{equation*}
Because $\varphi^* \in \mathcal{P}$ (due to $(w_1, \dots, w_K)\in \mathcal{S}_K$), we have that 
\begin{align}
    \sum_{k=1}^K w_k\| q_k- \varphi^*\|_2^2
    &\geq 
    \min_{\varphi\in\mathcal{P}}\sum_{k=1}^K w_k\|q_k-\varphi\|_2^2
    \notag \\
    & \stackrel{\hidewidth \eqref{eq:pointminphistar} \hidewidth} \geq
    \int_{\Theta}  \sum_{k=1}^Kw_k \big( q_k(\btheta) - \varphi^*(\btheta) \big)^2 \mathrm{d}\btheta
    \notag \\
    & = \sum_{k=1}^K w_k\| q_k- \varphi^*\|_2^2.
    \label{eq:phistaropt}
\end{align}
Thus, all inequalities in \eqref{eq:phistaropt} are actually equalities.
In particular,
\begin{align*}
    \sum_{k=1}^K w_k\| q_k- \varphi^*\|_2^2
    =
    & \min_{\varphi\in\mathcal{P}}\sum_{k=1}^K w_k\|q_k-\varphi\|_2^2,
\end{align*}
i.e., $q=\varphi^*$ solves \eqref{eq:optprobl2}.

\section{Proof of Theorem~\ref{th: generalized_pooling_frechet}
(Unconstrained Minimization of  the Weighted Average of General Distances)}
\label{appendix: proof_weighted_general_frechet}

Let $\warp_\varphi(\btheta)\triangleq \warp(\varphi(\btheta))$ and $\warp_{q_k}(\btheta)\triangleq \warp(q_k(\btheta))$.
We want to find 
\begin{equation}
\label{eq:optprobwarpl2}
     \tilde q = \argmin_{\varphi}\sum_{k=1}^K w_k \|\warp_{q_k}-\warp_\varphi\|_2^2.
\end{equation}
To this end, 
we first derive
\begin{equation}
\label{eq:warpstardef}
     \warp^* = \argmin_{\warp}\sum_{k=1}^K w_k \|\warp_{q_k}-\warp\|_2^2.
\end{equation}
Following the same steps as in Appendix~\ref{appendix: proof_weighted_l2_simplex} with $q_k$ replaced by $\warp_{q_k}$ and $\varphi$ replaced by $\warp$, it is easy to see that 
\begin{equation}
\label{eq:warpstarderev}
     \warp^*(\btheta) = \sum_{k=1}^K w_k\warp(q_k(\btheta))\,.
\end{equation}
Because $\warp(q_k(\btheta))\in (a,b)$, the convex combination $\sum_{k=1}^K w_k\warp(q_k(\btheta))$ is again in $(a,b)$.
Thus, $\warp^*(\btheta)$ is in the range of $\warp$ and we can define 
\begin{equation}
\label{eq:phistardef}
    \varphi^*(\btheta) \triangleq \warp^{-1}\big(\warp^*(\btheta)\big).
\end{equation}
This implies 
\begin{equation}
\label{eq:warpphistariswarpstar}
    \warp_{\varphi^*}(\btheta)= \warp(\varphi^*(\btheta))=\warp^*(\btheta)\,.
\end{equation}
We claim that $\tilde q$ defined in \eqref{eq:optprobwarpl2} equals  $\varphi^*$.
Indeed, we have for any $\varphi$
\begin{align*}
    \sum_{k=1}^K w_k \|\warp_{q_k}-\warp_\varphi\|_2^2
    & \geq \min_{\warp} \sum_{k=1}^K w_k \|\warp_{q_k}-\warp\|_2^2
    \\
    & \stackrel{\hidewidth \eqref{eq:warpstardef} \hidewidth}= \; \sum_{k=1}^K w_k \|\warp_{q_k}-\warp^*\|_2^2
    \\
    &  \stackrel{\hidewidth \eqref{eq:warpphistariswarpstar} \hidewidth}= \; \sum_{k=1}^K w_k \|\warp_{q_k}-\warp_{\varphi^*}\|_2^2,
\end{align*}
from which we conclude that $\varphi^*$ achieves the minimum in \eqref{eq:optprobwarpl2} and thus equals $\tilde q$.
We then obtain the optimal nonnormalized pooling function as
\begin{equation*}
    \tilde q(\btheta)
    = \varphi^*(\btheta)
    \stackrel{ \eqref{eq:phistardef} }= \warp^{-1}\big(\warp^*(\btheta)\big)
    \stackrel{ \eqref{eq:warpstarderev} }= \warp^{-1}\bigg(\sum_{k=1}^K w_k \warp(q_k(\btheta))\bigg).
\end{equation*}

\section{Proofs of the Fusion Rule for a Scalar Parameter}
\label{appendix:scalarbayesfusionall}

\subsection{Proof of Theorem~\ref{th:scalarbayesfusion}}
\label{appendix:scalarbayesfusion}
For $d_{\theta}=1$, 
the local observation likelihood functions from \eqref{eq:loclicgaussnons2} are given by
\begin{align}
    \ell_k(\theta) 
    & \propto \exp\bigg(
    {-}\frac{\theta^2\mathbf{h}_k^\intercal
    \mathbf{\Sigma}_{k k}^{-1}
    \mathbf{h}_k}{2 } + \theta\mathbf{h}_k^\intercal
    \mathbf{\Sigma}_{k k}^{-1}
    \mathbf{h}_k\sufstsc_k\bigg),
    \label{eq:likgaussd1proof}
\end{align}
where $\sufstsc_k = \ourinv_k\meamr_k=   \mathbf{h}_k^\intercal \mathbf{\Sigma}_{k k}^{-1}\meamr_k/  (\mathbf{h}_k^\intercal
    \mathbf{\Sigma}_{k k}^{-1}
    \mathbf{h}_k)$
    according to \eqref{eq:sufstsc_k} and \eqref{eq:ourinv_k}.
Furthermore, the global likelihood function \eqref{eq:likgausssufstatinep} can be rewritten as
\begin{align}
    \tlik(\theta) 
    & \propto \exp\bigg({-}\frac{(\sufstr-\mathbf{1}_{K} \theta)^\intercal \widetilde{\mathbf{\Sigma}}^{-1}(\sufstr-\mathbf{1}_{K} \theta)}{2}\bigg)
    \notag \\
    & \propto  \exp\bigg({-}\frac{
    \theta^2}{2 \widehat{\sigma}^2} +   \theta \mathbf{1}_{K}^\intercal \widetilde{\mathbf{\Sigma}}^{-1}\sufstr 
    \bigg),\label{eq:prooffuse1d}
\end{align}
where 
$
    \widehat{\sigma}^{2} = 1/(\mathbf{1}_{K}^\intercal \widetilde{\mathbf{\Sigma}}^{-1} \mathbf{1}_{K})
$.
The relation \eqref{eq:fusionliksgaussd1} follows from
\begin{align*}
    & \prod_{k=1}^K  (\ell_k(\theta))^{w_k}
    \notag \\
    & \quad \stackrel{\hidewidth \eqref{eq:likgaussd1proof} \hidewidth }\propto
    \exp\bigg(\sum_{k=1}^K w_k \bigg({-}\frac{\theta^2\mathbf{h}_k^\intercal
    \mathbf{\Sigma}_{k k}^{-1}
    \mathbf{h}_k}{2 } + \theta\mathbf{h}_k^\intercal
    \mathbf{\Sigma}_{k k}^{-1}
    \mathbf{h}_k\sufstsc_k 
     \bigg)
    \bigg)
    \notag \\
    & \quad \stackrel{\hidewidth \eqref{eq:powersbayes} \hidewidth }=
    \exp\bigg({-}\frac{
    \sum_{k=1}^K \theta^2\mathbf{1}_{K}^\intercal \widetilde{\mathbf{\Sigma}}^{-1}\mathbf{e}_k}{2} + \sum_{k=1}^K \theta  \mathbf{1}_{K}^\intercal \widetilde{\mathbf{\Sigma}}^{-1}\mathbf{e}_k \sufstsc_k 
    \bigg)
    \notag \\
    & \quad \stackrel{\hidewidth (a) \hidewidth }= 
    \exp\bigg({-}\frac{
      \theta^2 }{2\widehat{\sigma}^{2} } +  \theta  \mathbf{1}_{K}^\intercal \widetilde{\mathbf{\Sigma}}^{-1}  \sufst
    \bigg)
    \notag \\
    & \quad \stackrel{\hidewidth \eqref{eq:prooffuse1d} \hidewidth }\propto
    \tlik(\theta) \,,
\end{align*}
where we used in $(a)$ that $\sum_{k=1}^K \mathbf{e}_k =  \mathbf{1}_{K}$ and $\sum_{k=1}^K\mathbf{e}_k \sufstsc_k = \sufst$.
Finally, the fusion rule for the posteriors in \eqref{eq:fusedepbayesd1} easily follows from \eqref{eq:fusionliksgaussd1}:
\begin{align*}
    \poste(\theta \,\vert\, \sufst)
    & \propto  
    \pri(\theta) \tlik(\theta)
    \notag \\
    & \stackrel{\hidewidth \eqref{eq:fusionliksgaussd1} \hidewidth }\propto 
    \pri(\theta) \prod_{k=1}^K  (\ell_k(\theta))^{w_k}
    \notag \\
    & \propto 
    \pri(\theta) \prod_{k=1}^K  \bigg(\frac{\pi_k(\theta)}{\pri(\theta)}\bigg)^{w_k}
    \notag \\
    & =
    (\pri(\theta))^{1-\sum_{k=1}^K w_k}  \prod_{k=1}^K   (\pi_k( \theta))^{w_k}\,.
\end{align*}

\subsection{Calculation of the Weights in Example~\ref{ex:specialweightsd1}}\label{app:weightsd1}

We will show expression~\eqref{eq: weight_wk_ex_1} for $w_k$.
The vectors $\ourinv_k$ in \eqref{eq:ourinv_k} are given as
\begin{equation*}
    \ourinv_k = \frac{1}{\mathbf{1}_{r_0+r_k}^\intercal \mathbf{1}_{r_0+r_k}}\mathbf{1}_{r_0+r_k}^\intercal
    = \frac{1}{r_0+r_k}\mathbf{1}_{r_0+r_k}^\intercal,
\end{equation*}
and, in turn, the matrix $\widetilde{\mathbf{\Sigma}} $ in \eqref{eq:sigmatilde} is given by the entries
\begin{align*}
    \widetilde{{\Sigma}}_{k k'}
    & = \ourinv_k \mathbf{\Sigma}_{k k'}{\ourinvt_{k'} }
    \notag \\
    & = \frac{1}{r_0+r_k}\mathbf{1}_{r_0+r_k}^\intercal 
    \mathbf{\Sigma}_{k k'}
    \frac{1}{r_0+r_{k'}}\mathbf{1}_{r_0+r_{k'}}
    \notag \\
    & = \frac{r_0}{(r_0+r_k)(r_0+r_{k'})}
\end{align*}
for $k\neq k'$ and 
\begin{equation*}
    \widetilde{{\Sigma}}_{k k}=\frac{1}{r_0+r_k}.
\end{equation*}
It is easily verified that we can rewrite $ \widetilde{\mathbf{\Sigma}}$ as the following sum of a diagonal matrix and a rank one matrix 
\begin{align*}
    \widetilde{\mathbf{\Sigma}} 
    & = 
    \!\begin{pmatrix}
        \frac{r_1}{(r_0+r_1)^2} \\
         & \ddots \\
         & & \frac{r_K}{(r_0+r_K)^2}
    \end{pmatrix}
    \notag \\*
    & \quad\,
    + 
    \begin{pmatrix}
        \frac{1}{r_0+r_1} \\ \vdots \\ \frac{1}{r_0+r_K}
    \end{pmatrix}
    r_0
    \begin{pmatrix}
        \frac{1}{r_0+r_1} & \cdots & \frac{1}{r_0+r_K}
    \end{pmatrix} 
     .
\end{align*}
By the matrix inversion lemma \cite[eq.~(0.7.4.2)]{hojo13}, we can hence calculate  $\widetilde{\mathbf{\Sigma}}^{-1}$
as
\begin{align*}
    \widetilde{\mathbf{\Sigma}}^{-1}
    & = 
    \!\begin{pmatrix}
        \frac{(r_0+r_1)^2}{r_1} \\
         & \ddots \\
         & & \frac{(r_0+r_K)^2}{r_K}
    \end{pmatrix}
    \notag \\*
    & \quad\,
    - \bigg(\sum_{k=0}^K \frac{1}{r_k}\bigg)^{-1}
    \begin{pmatrix}
        \frac{r_0+r_1}{r_1}  \\ \vdots \\  \frac{r_0+r_K}{r_K} 
    \end{pmatrix}
    \begin{pmatrix}
         \frac{r_0+r_1}{r_1} & \cdots & \frac{r_0+r_K}{r_K}
    \end{pmatrix} 
     .
\end{align*}
To calculate the weights $w_k$ in \eqref{eq:powersbayes}, we  have to sum over the $k$th column of $\widetilde{\mathbf{\Sigma}}^{-1}$ and divide by $\mathbf{1}_{r_0+r_k}^\intercal \mathbf{1}_{r_0+r_k}=r_0+r_k$, i.e.,
\begin{align*}
    w_k 
    & = \frac{1}{r_0+r_k} \Bigg( 
        \frac{(r_0+r_k)^2}{r_k}
        - \frac{\sum_{k'=1}^K \frac{(r_0+r_k)(r_0+r_{k'})}{r_k r_{k'}}}%
        {\sum_{k'=0}^K \frac{1}{r_{k'}}}
        \Bigg)
    \notag \\
    & = \frac{r_0+r_k}{r_k}
        - \frac{\sum_{k'=1}^K \frac{r_0+r_{k'}}{r_k r_{k'}}}%
        {\sum_{k'=0}^K \frac{1}{r_{k'}}}
    \notag \\
    & = \frac{r_0+r_k}{r_k}
        - \frac{\frac{K}{r_k}+\frac{r_0}{r_k}\sum_{k'=1}^K \frac{1}{r_{k'}}}%
        {\sum_{k'=0}^K \frac{1}{r_{k'}}}
    \notag \\
    & = \frac{r_0+r_k}{r_k}
        - \frac{\frac{K-1}{r_k}+\frac{r_0}{r_k}\sum_{k'=0}^K \frac{1}{r_{k'}}}%
        {\sum_{k'=0}^K \frac{1}{r_{k'}}}
    \notag \\
    & = \frac{r_0}{r_k}
        + 1
        - \frac{\frac{K-1}{r_k}}%
        {\sum_{k'=0}^K \frac{1}{r_{k'}}}
        - \frac{r_0}{r_k}
    \notag \\
    & = 1 - \frac{K-1}{r_k} \bigg(\sum_{k'=0}^K \frac{1}{r_{k'}}\bigg)^{-1}.
\end{align*}

\section{Proofs of the Fusion Rule for a Vector Parameter}
\label{appendix:vectorbayesfusion}

\subsection{Proof of Theorem~\ref{th:fuseliks}}
\label{appendix:fuseliks}

We can rewrite \eqref{eq:likgausssufstatinep} as
\begin{align}
    & \tlik(\btheta) 
    \notag \\*
    & \quad \propto \exp\bigg({-}\frac{((\mathbf{1}_{K}\otimes \mathbf{I}_{d_{\theta}}) \btheta-\sufstr)^\intercal \widetilde{\mathbf{\Sigma}}^{-1}((\mathbf{1}_{K}\otimes \mathbf{I}_{d_{\theta}}) \btheta-\sufstr)}{2}\bigg)
    \notag \\
    & \quad \propto \exp\bigg(
    {-}\frac{\btheta^\intercal (\mathbf{1}_{K}\otimes \mathbf{I}_{d_{\theta}})^\intercal\widetilde{\mathbf{\Sigma}}^{-1}(\mathbf{1}_{K}\otimes \mathbf{I}_{d_{\theta}}) \btheta}{2}
    \notag \\*
    & \quad \qquad \qquad + \btheta^\intercal (\mathbf{1}_{K}\otimes \mathbf{I}_{d_{\theta}})^\intercal\widetilde{\mathbf{\Sigma}}^{-1} \sufstr
    \bigg)
    \notag \\
    & \quad \stackrel{ \hidewidth\eqref{eq:sighatinv} \hidewidth}=
    \exp\bigg({-}\frac{
    \btheta^\intercal \widehat{\mathbf{\Sigma}}^{-1} \btheta
    }{2}
    +
    \btheta^\intercal (\mathbf{1}_{K}\otimes \mathbf{I}_{d_{\theta}})^\intercal \widetilde{\mathbf{\Sigma}}^{-1}\sufstr
    \bigg) .
    \label{eq:sufstgiventhetagauss}
\end{align}
Furthermore, from \eqref{eq:loclicgaussnons2}, we see that
\begin{equation*}
    \ell_k(\btheta)
    \propto 
    \exp\bigg({-}\frac{\btheta^\intercal 
      \mathbf{H}_k^\intercal
    \mathbf{\Sigma}_{k k}^{-1}
    \mathbf{H}_k 
    \btheta}{2}
    +
    \btheta^\intercal \mathbf{H}_k^\intercal
    \mathbf{\Sigma}_{k k}^{-1}
    \mathbf{H}_k \sufstr_k
    \bigg),
\end{equation*}
where $\sufstr_k = \Ourinv_k\meamr_k$.
Thus, we have
\begin{align}
    & \prod_{k=1}^K  \ell_k(\mathbf{W}_k\btheta)
    \notag \\*
    & \quad \propto
    \exp\bigg({-}\frac{\btheta^\intercal \big(\sum_{k=1}^K  \mathbf{W}_k^\intercal
      \mathbf{H}_k^\intercal
    \mathbf{\Sigma}_{k k}^{-1}
    \mathbf{H}_k 
    \mathbf{W}_k\big)
    \btheta}{2}
    \notag \\*
    & \quad \quad
    +
    \btheta^\intercal \sum_{k=1}^K \mathbf{W}_k^\intercal \mathbf{H}_k^\intercal
    \mathbf{\Sigma}_{k k}^{-1}
    \mathbf{H}_k \sufstr_k
    \bigg)
    \notag \\
    & \quad \stackrel{\hidewidth \eqref{eq:defG} \hidewidth}=\,\exp\bigg({-}\frac{\btheta^\intercal \big( \widehat{\mathbf{\Sigma}}^{-1} - \mathbf{G}\big)
    \btheta}{2}
    +
    \btheta^\intercal \sum_{k=1}^K \mathbf{W}_k^\intercal \mathbf{H}_k^\intercal
    \mathbf{\Sigma}_{k k}^{-1}
    \mathbf{H}_k \sufstr_k
    \bigg)
    \notag \\
    & \quad = \frac{1}{\xi_0(\btheta)} \exp\bigg({-}\frac{\btheta^\intercal \widehat{\mathbf{\Sigma}}^{-1} \btheta}{2}
    +
    \btheta^\intercal \sum_{k=1}^K \mathbf{W}_k^\intercal \mathbf{H}_k^\intercal
    \mathbf{\Sigma}_{k k}^{-1}
    \mathbf{H}_k \sufstr_k
    \bigg),
    \label{eq:prodlikgaussbayesproof}
\end{align}
with $\xi_0(\btheta)$ as defined in \eqref{eq:defxi0}.
By comparing \eqref{eq:sufstgiventhetagauss} and \eqref{eq:prodlikgaussbayesproof}, we see that \eqref{eq:fuselikgaussvec} holds, provided that
\begin{equation}
    (\mathbf{1}_{K}\otimes \mathbf{I}_{d_{\theta}})^\intercal \widetilde{\mathbf{\Sigma}}^{-1}\sufstr
    =
    \sum_{k=1}^K \mathbf{W}_k^\intercal \mathbf{H}_k^\intercal
    \mathbf{\Sigma}_{k k}^{-1}
    \mathbf{H}_k \sufstr_k.
    \label{eq:meantoshow}
\end{equation}
Inserting \eqref{eq:weightmatrix} into the right-hand side of
\eqref{eq:meantoshow}, we obtain
\begin{align*}
     \sum_{k=1}^K \mathbf{W}_k^\intercal \mathbf{H}_k^\intercal
    \mathbf{\Sigma}_{k k}^{-1}
    \mathbf{H}_k \sufstr_k
    & = (\mathbf{1}_{K}\otimes \mathbf{I}_{d_{\theta}})^\intercal \widetilde{\mathbf{\Sigma}}^{-1}\sum_{k=1}^K  
    (\mathbf{e}_k \otimes \mathbf{I}_{d_{\theta}})\sufstr_k
    \notag \\
    &  =  (\mathbf{1}_{K}\otimes \mathbf{I}_{d_{\theta}})^\intercal \widetilde{\mathbf{\Sigma}}^{-1} \sufstr\,,
\end{align*}
concluding the proof of \eqref{eq:fuselikgaussvec}.

Finally, the fusion rule \eqref{eq:fusedepbayes} easily follows from \eqref{eq:fuselikgaussvec}:
\begin{align*}
    \poste(\btheta \,\vert\, \sufst)
    & \propto 
    \pri(\btheta) \tlik(\btheta)
    \notag \\
    & \propto 
    \pri(\btheta) \xi_0(\btheta) \prod_{k=1}^K  \ell_k(\mathbf{W}_k\btheta)
    \notag \\
    & \propto
    \pri(\btheta) \xi_0(\btheta) \prod_{k=1}^K  \frac{\pi_k(\mathbf{W}_k\btheta)}{\pri(\mathbf{W}_k\btheta)}\,.
\end{align*}

\subsection{Proof of Corollary~\ref{cor:fusegaussprior}}\label{app:fusegaussprior}

We start directly from $\poste(\btheta \,\vert\, \sufst) \propto \pri(\btheta) \tlik(\btheta)$.
By \eqref{eq:sufstgiventhetagauss} and our choice of prior $\pri(\btheta) = \mathcal{N}(\btheta; \bmu_0, \mathbf{\Sigma}_0)\propto \exp\big({-}\frac{
     \btheta^\intercal  \mathbf{\Sigma}_0^{-1} \btheta
    }{2}
    +
    \btheta^\intercal \mathbf{\Sigma}_0^{-1} \bmu_0
    \big)$, we have that 
\begin{align}
    \pri(\btheta) \tlik(\btheta)
    & \propto
    \exp\bigg({-}\frac{
    \btheta^\intercal \widehat{\mathbf{\Sigma}}^{-1} \btheta 
    + \btheta^\intercal  \mathbf{\Sigma}_0^{-1} \btheta
    }{2}
    \notag \\*
    & \quad 
    +
    \btheta^\intercal (\mathbf{1}_{K}\otimes \mathbf{I}_{d_{\theta}})^\intercal \widetilde{\mathbf{\Sigma}}^{-1}\sufstr
    + 
    \btheta^\intercal \mathbf{\Sigma}_0^{-1} \bmu_0
    \bigg) 
    \notag \\
    & =
    \exp\bigg({-}\frac{
    \btheta^\intercal \big(\widehat{\mathbf{\Sigma}}^{-1} + \mathbf{\Sigma}_0^{-1} \big) \btheta 
    }{2}
    \notag \\*
    & \quad 
    +
    \btheta^\intercal \big( (\mathbf{1}_{K}\otimes \mathbf{I}_{d_{\theta}})^\intercal \widetilde{\mathbf{\Sigma}}^{-1}\sufstr
    + 
    \mathbf{\Sigma}_0^{-1} \bmu_0 \big)
    \bigg) 
    \notag \\
    & =
    \exp\bigg({-}\frac{
    \btheta^\intercal  \mathbf{\Sigma}_1^{-1}  \btheta 
    }{2}
    +
    \btheta^\intercal \mathbf{\Sigma}_1^{-1} \bmu_1  
    \bigg) 
    \notag \\
    & \propto
    \exp\bigg({-}\frac{
    (\btheta-\bmu_1)^\intercal  \mathbf{\Sigma}_1^{-1}  (\btheta-\bmu_1)
    }{2}
    \bigg),
    \label{eq:proofgausspriend}
\end{align}
with $\bmu_1$ and $\mathbf{\Sigma}_1$ given by \eqref{eq:mu1} and \eqref{eq:sigma1}, respectively.
Expression \eqref{eq:proofgausspriend} is proportional to the pdf of a Gaussian with  mean $\bmu_1$ and covariance matrix $\mathbf{\Sigma}_1$.


\bibliographystyle{IEEEtran} 
\bibliography{refs}

\begin{thebibliography}{100}
\providecommand{\url}[1]{#1}
\csname url@samestyle\endcsname
\providecommand{\newblock}{\relax}
\providecommand{\bibinfo}[2]{#2}
\providecommand{\BIBentrySTDinterwordspacing}{\spaceskip=0pt\relax}
\providecommand{\BIBentryALTinterwordstretchfactor}{4}
\providecommand{\BIBentryALTinterwordspacing}{\spaceskip=\fontdimen2\font plus
\BIBentryALTinterwordstretchfactor\fontdimen3\font minus
  \fontdimen4\font\relax}
\providecommand{\BIBforeignlanguage}[2]{{%
\expandafter\ifx\csname l@#1\endcsname\relax
\typeout{** WARNING: IEEEtran.bst: No hyphenation pattern has been}%
\typeout{** loaded for the language `#1'. Using the pattern for}%
\typeout{** the default language instead.}%
\else
\language=\csname l@#1\endcsname
\fi
#2}}
\providecommand{\BIBdecl}{\relax}
\BIBdecl

\bibitem{bandyopadhyay2018distributed}
S.~Bandyopadhyay and S.-J. Chung, ``Distributed {B}ayesian filtering using
  logarithmic opinion pool for dynamic sensor networks,'' \emph{Automatica},
  vol.~97, pp. 7--17, 2018.

\bibitem{Clark10}
D.~Clark, S.~Julier, R.~Mahler, and B.~Risti\'c, ``Robust multi-object sensor
  fusion with unknown correlations,'' in \emph{Proc. Sens. Signal Process. Def.
  (SSPD 2010)}, London, UK, 2010.

\bibitem{hu2011diffusion}
J.~Hu, L.~Xie, and C.~Zhang, ``Diffusion {K}alman filtering based on covariance
  intersection,'' \emph{IEEE Trans. Signal Process.}, vol.~60, no.~2, pp.
  891--902, 2011.

\bibitem{maddern2016real}
W.~Maddern and P.~Newman, ``Real-time probabilistic fusion of sparse 3d lidar
  and dense stereo,'' in \emph{Proc. IEEE/RSJ Int. Conf. Intelligent Robots and
  Systems (IROS 2016)}, Daejeon, South Korea, 2016, pp. 2181--2188.

\bibitem{da2021recent}
K.~Da, T.~Li, Y.~Zhu, H.~Fan, and Q.~Fu, ``Recent advances in multisensor
  multitarget tracking using random finite set,'' \emph{Front. Inform. Technol.
  Electron. Eng.}, vol.~22, pp. 5--24, 2021.

\bibitem{fantacci2018robust}
C.~Fantacci, B.-N. Vo, B.-T. Vo, G.~Battistelli, and L.~Chisci, ``Robust fusion
  for multisensor multiobject tracking,'' \emph{IEEE Signal Process. Letters},
  vol.~25, no.~5, pp. 640--644, 2018.

\bibitem{meyer2015distributed}
F.~Meyer, O.~Hlinka, H.~Wymeersch, E.~Riegler, and F.~Hlawatsch, ``Distributed
  localization and tracking of mobile networks including noncooperative
  objects,'' \emph{IEEE Trans. Signal Inf. Process. Netw.}, vol.~2, no.~1, pp.
  57--71, 2015.

\bibitem{Uney13}
M.~{\"U}ney, D.~E. Clark, and S.~J. Julier, ``Distributed fusion of {PHD}
  filters via exponential mixture densities,'' \emph{IEEE J. Sel. Topics Signal
  Process.}, vol.~7, no.~3, pp. 521--531, 2013.

\bibitem{lakshminarayanan2017simple}
B.~Lakshminarayanan, A.~Pritzel, and C.~Blundell, ``Simple and scalable
  predictive uncertainty estimation using deep ensembles,'' \emph{Proc. Adv.
  Neural Inf. Process. Syst. (NIPS 2017)}, vol.~30, 2017.

\bibitem{lee2009neural}
H.~Lee, S.~Hong, and E.~Kim, ``Neural network ensemble with probabilistic
  fusion and its application to gait recognition,'' \emph{Neurocomputing},
  vol.~72, no. 7-9, pp. 1557--1564, 2009.

\bibitem{lu2020ensemble}
Q.~Lu, G.~Karanikolas, Y.~Shen, and G.~B. Giannakis, ``Ensemble {G}aussian
  processes with spectral features for online interactive learning with
  scalability,'' in \emph{International Conference on Artificial Intelligence
  and Statistics}.\hskip 1em plus 0.5em minus 0.4em\relax PMLR, 2020, pp.
  1910--1920.

\bibitem{thorgeirsson2021probabilistic}
A.~T. Thorgeirsson and F.~Gauterin, ``Probabilistic predictions with federated
  learning,'' \emph{Entropy}, vol.~23, no.~1, p.~41, 2021.

\bibitem{alam2017data}
F.~Alam, R.~Mehmood, I.~Katib, N.~N. Albogami, and A.~Albeshri, ``Data fusion
  and {IoT} for smart ubiquitous environments: A survey,'' \emph{IEEE Access},
  vol.~5, pp. 9533--9554, 2017.

\bibitem{kats2019soft}
E.~Kats, J.~Goldberger, and H.~Greenspan, ``A soft {STAPLE} algorithm combined
  with anatomical knowledge,'' in \emph{Int. Conf. Med. Image Comput.
  Comput.-Assist. Interv. (MICCAI 2019)}.\hskip 1em plus 0.5em minus
  0.4em\relax Shenzhen, China: Springer, 2019, pp. 510--517.

\bibitem{kolosz2013modelling}
B.~Kolosz, S.~Grant-Muller, and K.~Djemame, ``Modelling uncertainty in the
  sustainability of intelligent transport systems for highways using
  probabilistic data fusion,'' \emph{Environ. Model. Softw.}, vol.~49, pp.
  78--97, 2013.

\bibitem{chlingaryan2018machine}
A.~Chlingaryan, S.~Sukkarieh, and B.~Whelan, ``Machine learning approaches for
  crop yield prediction and nitrogen status estimation in precision
  agriculture: A review,'' \emph{Comput. Electron. Agr.}, vol. 151, pp. 61--69,
  2018.

\bibitem{li2013combined}
G.-B. Li, L.-L. Yang, Y.~Xu, W.-J. Wang, L.-L. Li, and S.-Y. Yang, ``A combined
  molecular docking-based and pharmacophore-based target prediction strategy
  with a probabilistic fusion method for target ranking,'' \emph{J. Mol. Graph.
  Model.}, vol.~44, pp. 278--285, 2013.

\bibitem{murphy1984probability}
A.~H. Murphy and R.~L. Winkler, ``Probability forecasting in meteorology,''
  \emph{J. Am. Stat. Assoc.}, vol.~79, no. 387, pp. 489--500, 1984.

\bibitem{marty2015combining}
R.~Marty, V.~Fortin, H.~Kuswanto, A.-C. Favre, and E.~Parent, ``Combining the
  {B}ayesian processor of output with {B}ayesian model averaging for reliable
  ensemble forecasting,'' \emph{J. R. Stat. Soc. Ser. C Appl. Stat.}, vol.~64,
  no.~1, pp. 75--92, 2015.

\bibitem{mitchell2005evaluating}
J.~Mitchell and S.~G. Hall, ``Evaluating, comparing and combining density
  forecasts using the {KLIC} with an application to the {B}ank of {E}ngland and
  {NIESR} `fan' charts of inflation,'' \emph{Oxf. Bull. Econ. Stat.}, vol.~67,
  pp. 995--1033, 2005.

\bibitem{moral2015model}
E.~Moral-Benito, ``Model averaging in economics: An overview,'' \emph{J. Econ.
  Surv.}, vol.~29, no.~1, pp. 46--75, 2015.

\bibitem{barak2017fusion}
S.~Barak, A.~Arjmand, and S.~Ortobelli, ``Fusion of multiple diverse predictors
  in stock market,'' \emph{Inform. Fusion}, vol.~36, pp. 90--102, 2017.

\bibitem{genest1986combining}
C.~Genest and J.~V. Zidek, ``Combining probability distributions: A critique
  and an annotated bibliography,'' \emph{Stat. Sci.}, vol.~1, no.~1, pp.
  114--135, 1986.

\bibitem{hall2007combining}
S.~G. Hall and J.~Mitchell, ``Combining density forecasts,'' \emph{Int. J.
  Forecast.}, vol.~23, no.~1, pp. 1--13, 2007.

\bibitem{gneiting2008probabilistic}
T.~Gneiting, ``Editorial: {P}robabilistic forecasting,'' \emph{J. R. Stat. Soc.
  Ser. A Stat. Soc.}, pp. 319--321, 2008.

\bibitem{stone1961opinion}
M.~Stone, ``The opinion pool,'' \emph{Ann. Math. Statist.}, vol.~32, no.~4, pp.
  1339--1342, 1961.

\bibitem{Bailey12}
T.~Bailey, S.~Julier, and G.~Agamennoni, ``On conservative fusion of
  information with unknown non-{G}aussian dependence,'' in \emph{Proc. Int.
  Conf. Inf. Fusion (FUSION 2012)}, Singapore, Singapore, 2012.

\bibitem{genest1984characterization}
C.~Genest, ``A characterization theorem for externally {B}ayesian groups,''
  \emph{Ann. Statist.}, vol.~12, no.~3, pp. 1100--1105, 1984.

\bibitem{hurley2002information}
M.~B. Hurley, ``An information theoretic justification for covariance
  intersection and its generalization,'' in \emph{Proc. Int. Conf. Inf. Fusion
  (FUSION 2002)}, Annapolis, MD, 2002.

\bibitem{Lehrer19}
N.~Lehrer, O.~Tslil, and A.~Carmi, ``Log-linear {C}hernoff fusion for
  distributed particle filtering,'' in \emph{Proc. Int. Conf. Inf. Fusion
  (FUSION 2019)}, Ottawa, ON, Canada, 2019.

\bibitem{Gunay16}
M.~Gunay, U.~Orguner, and M.~Demirekler, ``Chernoff fusion of {G}aussian
  mixtures based on sigma-point approximation,'' \emph{IEEE Trans. Aerosp.
  Electron. Syst.}, vol.~52, no.~6, pp. 2732--2746, 2016.

\bibitem{julier1997non}
S.~J. Julier and J.~K. Uhlmann, ``A non-divergent estimation algorithm in the
  presence of unknown correlations,'' in \emph{Proc. American Control Conf.},
  Albuquerque, NM, 1997.

\bibitem{Chang10}
K.-C. Chang, C.-Y. Chong, and S.~Mori, ``Analytical and computational
  evaluation of scalable distributed fusion algorithms,'' \emph{IEEE Trans.
  Aerosp. Electron. Syst.}, vol.~46, no.~4, pp. 2022--2034, 2010.

\bibitem{urteaga2016sequential}
I.~Urteaga, M.~F. Bugallo, and P.~M. Djuri{\'c}, ``Sequential {M}onte {C}arlo
  methods under model uncertainty,'' in \emph{2016 IEEE Stat. Signal Process.
  Workshop (SSP)}, Palma de Mallorca, Spain, 2016.

\bibitem{hlinka2014consensus}
O.~Hlinka, F.~Hlawatsch, and P.~M. Djuri{\'c}, ``Consensus-based distributed
  particle filtering with distributed proposal adaptation,'' \emph{IEEE Trans.
  Signal Process.}, vol.~62, no.~12, pp. 3029--3041, 2014.

\bibitem{savic2014belief}
V.~Savic, H.~Wymeersch, and S.~Zazo, ``Belief consensus algorithms for fast
  distributed target tracking in wireless sensor networks,'' \emph{Signal
  Processing}, vol.~95, pp. 149--160, 2014.

\bibitem{wymeersch2009cooperative}
H.~Wymeersch, J.~Lien, and M.~Z. Win, ``Cooperative localization in wireless
  networks,'' \emph{Proc. IEEE}, vol.~97, no.~2, pp. 427--450, 2009.

\bibitem{lopes2008diffusion}
C.~G. Lopes and A.~H. Sayed, ``Diffusion least-mean squares over adaptive
  networks: Formulation and performance analysis,'' \emph{IEEE Trans. Signal
  Process.}, vol.~56, no.~7, pp. 3122--3136, 2008.

\bibitem{Punska99}
O.~Punska, ``Bayesian approaches to multi-sensor data fusion,'' Master's
  thesis, University of Cambridge, Cambridge, UK, 1999.

\bibitem{battistelli2014kullback}
G.~Battistelli and L.~Chisci, ``{K}ullback-{L}eibler average, consensus on
  probability densities, and distributed state estimation with guaranteed
  stability,'' \emph{Automatica}, vol.~50, no.~3, pp. 707--718, Mar. 2014.

\bibitem{deng2012sequential}
Z.~Deng, P.~Zhang, W.~Qi, J.~Liu, and Y.~Gao, ``Sequential covariance
  intersection fusion {K}alman filter,'' \emph{Inf. Sci.}, vol. 189, pp.
  293--309, 2012.

\bibitem{tang2018information}
M.~Tang, Y.~Rong, J.~Zhou, and X.~R. Li, ``Information geometric approach to
  multisensor estimation fusion,'' \emph{IEEE Trans. Signal Process.}, vol.~67,
  no.~2, pp. 279--292, 2018.

\bibitem{li-battistelli2021distributed}
G.~Li, G.~Battistelli, L.~Chisci, W.~Yi, and L.~Kong, ``Distributed multi-view
  multi-target tracking based on {CPHD} filtering,'' \emph{Signal Process.},
  vol. 188, p. 108210, 2021.

\bibitem{li2020arithmetic}
T.~Li, X.~Wang, Y.~Liang, and Q.~Pan, ``On arithmetic average fusion and its
  application for distributed multi-{B}ernoulli multitarget tracking,''
  \emph{IEEE Trans. Signal Process.}, vol.~68, pp. 2883--2896, 2020.

\bibitem{li-hlawatsch2021distributed}
T.~Li and F.~Hlawatsch, ``A distributed particle-{PHD} filter using
  arithmetic-average fusion of {G}aussian mixture parameters,'' \emph{Inform.
  Fusion}, vol.~73, pp. 111--124, 2021.

\bibitem{yi2020distributed}
W.~Yi, G.~Li, and G.~Battistelli, ``Distributed multi-sensor fusion of {PHD}
  filters with different sensor fields of view,'' \emph{IEEE Trans. Signal
  Process.}, vol.~68, pp. 5204--5218, 2020.

\bibitem{gao2020fusion}
L.~Gao, G.~Battistelli, and L.~Chisci, ``Fusion of labeled {RFS} densities with
  minimum information loss,'' \emph{IEEE Trans. Signal Process.}, vol.~68, pp.
  5855--5868, 2020.

\bibitem{gostar2020cooperative}
A.~K~Gostar, T.~Rathnayake, R.~Tennakoon, A.~Bab-Hadiashar, G.~Battistelli,
  L.~Chisci, and R.~Hoseinnezhad, ``Cooperative sensor fusion in centralized
  sensor networks using {C}auchy-{S}chwarz divergence,'' \emph{Signal
  Process.}, vol. 167, p. 107278, 2020.

\bibitem{li-battistelli2019computationally}
S.~Li, G.~Battistelli, L.~Chisci, W.~Yi, B.~Wang, and L.~Kong,
  ``Computationally efficient multi-agent multi-object tracking with labeled
  random finite sets,'' \emph{IEEE Trans. Signal Process.}, vol.~67, no.~1, pp.
  260--275, 2019.

\bibitem{mahler2000optimal}
R.~P.~S. Mahler, ``Optimal/robust distributed data fusion: {A} unified
  approach,'' in \emph{Aerosense 2000: Signal and Image Processing (Proceedings
  of SPIE)}, vol. 4052, Orlando, FL, USA, 2000.

\bibitem{li2020federated}
T.~Li, A.~K. Sahu, A.~Talwalkar, and V.~Smith, ``Federated learning:
  Challenges, methods, and future directions,'' \emph{IEEE Signal Process.
  Mag.}, vol.~37, no.~3, pp. 50--60, 2020.

\bibitem{yurochkin2019bayesian}
M.~Yurochkin, M.~Agarwal, S.~Ghosh, K.~Greenewald, N.~Hoang, and Y.~Khazaeni,
  ``Bayesian nonparametric federated learning of neural networks,'' in
  \emph{Proc. Int. Conf. Mach. Learn. (ICML 2019)}, Long Beach, CA, 2019.

\bibitem{clemen1989combining}
R.~T. Clemen, ``Combining forecasts: A review and annotated bibliography,''
  \emph{Int. J. Forecast.}, vol.~5, no.~4, pp. 559--583, 1989.

\bibitem{armstrong2001combining}
J.~S. Armstrong, ``Combining forecasts,'' in \emph{Principles of Forecasting},
  J.~S. Armstrong, Ed.\hskip 1em plus 0.5em minus 0.4em\relax Boston, MA:
  Springer, 2001, pp. 417--439.

\bibitem{wallis2005combining}
K.~F. Wallis, ``Combining density and interval forecasts: A modest proposal,''
  \emph{Oxf. Bull. Econ. Stat.}, vol.~67, no.~s1, pp. 983--994, 2005.

\bibitem{winkler2019probability}
R.~L. Winkler, Y.~Grushka-Cockayne, K.~C. Lichtendahl~Jr, and V.~R.~R. Jose,
  ``Probability forecasts and their combination: A research perspective,''
  \emph{Decis. Anal.}, vol.~16, no.~4, pp. 239--260, 2019.

\bibitem{hoeting1999bayesian}
J.~A. Hoeting, D.~Madigan, A.~E. Raftery, and C.~T. Volinsky, ``Bayesian model
  averaging: A tutorial,'' \emph{Stat. Sci.}, vol.~14, no.~4, pp. 382--401,
  1999.

\bibitem{fragoso2018bayesian}
T.~M. Fragoso, W.~Bertoli, and F.~Louzada, ``{B}ayesian model averaging: A
  systematic review and conceptual classification,'' \emph{Int. Stat. Rev.},
  vol.~86, no.~1, pp. 1--28, 2018.

\bibitem{posada2008jmodeltest}
D.~Posada, ``{jModelTest}: Phylogenetic model averaging,'' \emph{Mol. Biol.
  Evol.}, vol.~25, no.~7, pp. 1253--1256, 2008.

\bibitem{darriba2012jmodeltest}
D.~Darriba, G.~L. Taboada, R.~Doallo, and D.~Posada, ``{jModelTest} 2: More
  models, new heuristics and parallel computing,'' \emph{Nat. Methods}, vol.~9,
  no.~8, pp. 772--772, 2012.

\bibitem{steel2020model}
M.~F.~J. Steel, ``Model averaging and its use in economics,'' \emph{J. Econ.
  Lit.}, vol.~58, no.~3, pp. 644--719, 2020.

\bibitem{dormann2018model}
C.~F. Dormann, J.~M. Calabrese, G.~Guillera-Arroita, E.~Matechou, V.~Bahn,
  K.~Barto{\'n}, C.~M. Beale, S.~Ciuti, J.~Elith, K.~Gerstner, J.~Guelat,
  P.~Keil, J.~J. Lahoz-Monfort, L.~J. Pollock, B.~Reineking, D.~R. Roberts,
  B.~Schr\"oder, W.~Thuiller, D.~I. Warton, B.~A. Wintle, S.~N. Wood, R.~O.
  W\"uest, and F.~Hartig, ``Model averaging in ecology: A review of {B}ayesian,
  information-theoretic, and tactical approaches for predictive inference,''
  \emph{Ecol. Monogr.}, vol.~88, no.~4, pp. 485--504, 2018.

\bibitem{turkheimer2003undecidability}
F.~E. Turkheimer, R.~Hinz, and V.~J. Cunningham, ``On the undecidability among
  kinetic models: From model selection to model averaging,'' \emph{J. Cereb.
  Blood Flow Metab.}, vol.~23, no.~4, pp. 490--498, 2003.

\bibitem{montgomery2010bayesian}
J.~M. Montgomery and B.~Nyhan, ``{B}ayesian model averaging: Theoretical
  developments and practical applications,'' \emph{Political Anal.}, vol.~18,
  no.~2, pp. 245--270, 2010.

\bibitem{neiswanger2014asymptotically}
W.~Neiswanger, C.~Wang, and E.~P. Xing, ``Asymptotically exact, embarrassingly
  parallel {MCMC},'' in \emph{Proc. Conf. Uncertain. Artif. Intell. (UAI
  2014)}, Arlington, VA, 2014.

\bibitem{wang2013parallel}
X.~Wang and D.~B. Dunson, ``Parallelizing {MCMC} via {W}eierstrass sampler,''
  \emph{arXiv preprint arXiv:1312.4605}, 2013.

\bibitem{bardenet2017markov}
R.~Bardenet, A.~Doucet, and C.~Holmes, ``On {M}arkov chain {M}onte {C}arlo
  methods for tall data,'' \emph{J. Mach. Learn. Res.}, vol.~18, no.~1, pp.
  1515--1557, 2017.

\bibitem{rabin2011wasserstein}
J.~Rabin, G.~Peyr{\'e}, J.~Delon, and M.~Bernot, ``Wasserstein barycenter and
  its application to texture mixing,'' in \emph{Proc. Int. Conf. Scale Space
  Var. Methods Comput. Vis. (SSVM 2011)}, Ein-Gedi, Israel, 2011.

\bibitem{srivastava2015wasp}
S.~Srivastava, V.~Cevher, Q.~Dinh, and D.~Dunson, ``{WASP}: {S}calable {B}ayes
  via barycenters of subset posteriors,'' in \emph{Proc. Int. Conf. Artif.
  Intell. Stat. (AISTATS 2015)}, San Diego, CA, 2015.

\bibitem{wolpert1992stacked}
D.~H. Wolpert, ``Stacked generalization,'' \emph{Neural Netw.}, vol.~5, no.~2,
  pp. 241--259, 1992.

\bibitem{breiman1996stacked}
L.~Breiman, ``Stacked regressions,'' \emph{Mach. Learn.}, vol.~24, no.~1, pp.
  49--64, 1996.

\bibitem{shazeer2017outrageously}
N.~Shazeer, A.~Mirhoseini, K.~Maziarz, A.~Davis, Q.~V. Le, G.~E. Hinton, and
  J.~Dean, ``Outrageously large neural networks: The sparsely-gated
  mixture-of-experts layer,'' in \emph{Proc. Int. Conf. Learn. Represent. (ICLR
  2017)}, Toulon, France, 2017.

\bibitem{hoang2019collective}
M.~Hoang, N.~Hoang, B.~K.~H. Low, and C.~Kingsford, ``Collective model fusion
  for multiple black-box experts,'' in \emph{Proc. Int. Conf. Mach. Learn.
  (ICML 2019)}, Long Beach, CA, 2019.

\bibitem{liu2020gaussian}
H.~Liu, Y.~Ong, X.~Shen, and J.~Cai, ``When {G}aussian process meets big data:
  A review of scalable {GP}s,'' \emph{IEEE Trans. Neural Netw. Learn. Syst.},
  vol.~31, no.~11, pp. 4405--4423, 2020.

\bibitem{wozniak2014survey}
M.~Wo{\'z}niak, M.~Gra{\~n}a, and E.~Corchado, ``A survey of multiple
  classifier systems as hybrid systems,'' \emph{Inform. Fusion}, vol.~16, pp.
  3--17, 2014.

\bibitem{burnham2002practical}
K.~P. Burnham and D.~R. Anderson, \emph{Model Selection and Multimodel
  Inference: A Practical Information-Theoretic Approach}, 2nd~ed.\hskip 1em
  plus 0.5em minus 0.4em\relax New York, NY: Springer, 2002.

\bibitem{wallis2011combining}
K.~F. Wallis, ``Combining forecasts--forty years later,'' \emph{Appl. Financial
  Econ.}, vol.~21, no. 1--2, pp. 33--41, 2011.

\bibitem{cooke1991experts}
R.~M. Cooke, \emph{Experts in Uncertainty: Opinion and Subjective Probability
  in Science}.\hskip 1em plus 0.5em minus 0.4em\relax New York, NY: Oxford
  University Press, 1991.

\bibitem{clemen1999combining}
R.~T. Clemen and R.~L. Winkler, ``Combining probability distributions from
  experts in risk analysis,'' \emph{Risk Anal.}, vol.~19, no.~2, pp. 187--203,
  1999.

\bibitem{dietrichprobabilistic}
F.~Dietrich and C.~List, ``Probabilistic opinion pooling,'' in \emph{The Oxford
  Handbook of Probability and Philosophy}.\hskip 1em plus 0.5em minus
  0.4em\relax Oxford, UK: Oxford University Press, 2016.

\bibitem{stewart2018probabilistic}
R.~T. Stewart and I.~O. Quintana, ``Probabilistic opinion pooling with
  imprecise probabilities,'' \emph{J. Philos. Log.}, vol.~47, no.~1, pp.
  17--45, 2018.

\bibitem{stewart2018learning}
------, ``Learning and pooling, pooling and learning,'' \emph{Erkenntnis},
  vol.~83, no.~3, pp. 369--389, 2018.

\bibitem{abbas2009kullback}
A.~E. Abbas, ``A {K}ullback-{L}eibler view of linear and log-linear pools,''
  \emph{Decis. Anal.}, vol.~6, no.~1, pp. 25--37, 2009.

\bibitem{garg2004generalized}
A.~Garg, T.~S. Jayram, S.~Vaithyanathan, and H.~Zhu, ``Generalized opinion
  pooling,'' in \emph{Proc. Int. Symp. Artif. Intell. Math. (ISAIM 2004)}, Fort
  Lauderdale, FL, 2004.

\bibitem{da2019kullback}
K.~Da, T.~Li, Y.~Zhu, H.~Fan, and Q.~Fu, ``Kullback-{L}eibler averaging for
  multitarget density fusion,'' in \emph{Proc. Int. Symp. Distrib. Comput.
  Artif. Intell. (DCAI 2019)}, \'Avila, Spain, 2019.

\bibitem{agueh2011barycenters}
M.~Agueh and G.~Carlier, ``Barycenters in the {W}asserstein space,'' \emph{SIAM
  J. Appl. Math.}, vol.~43, no.~2, pp. 904--924, 2011.

\bibitem{winkler1968consensus}
R.~L. Winkler, ``The consensus of subjective probability distributions,''
  \emph{Manag. Sci.}, vol.~15, no.~2, pp. 61--75, 1968.

\bibitem{morris1977combining}
P.~A. Morris, ``Combining expert judgments: {A} {B}ayesian approach,''
  \emph{Manag. Sci.}, vol.~23, no.~7, pp. 679--693, 1977.

\bibitem{winkler1981combining}
R.~L. Winkler, ``Combining probability distributions from dependent information
  sources,'' \emph{Manag. Sci.}, vol.~27, no.~4, pp. 479--488, 1981.

\bibitem{lindley1983reconciliation}
D.~Lindley, ``Reconciliation of probability distributions,'' \emph{Oper. Res.},
  vol.~31, no.~5, pp. 866--880, 1983.

\bibitem{clemen1985limits}
R.~T. Clemen and R.~L. Winkler, ``Limits for the precision and value of
  information from dependent sources,'' \emph{Oper. Res.}, vol.~33, no.~2, pp.
  427--442, 1985.

\bibitem{anderson1979optimal}
B.~Anderson and J.~Moore, \emph{{Optimal Filtering}}.\hskip 1em plus 0.5em
  minus 0.4em\relax Englewood Cliffs, NJ: Prentice-Hall, 1979.

\bibitem{ristic2003beyond}
B.~Ristic, S.~Arulampalam, and N.~Gordon, \emph{Beyond the Kalman Filter:
  Particle Filters for Tracking Applications}.\hskip 1em plus 0.5em minus
  0.4em\relax Boston, MA, USA: Artech House, 2003.

\bibitem{chong2001convex}
C.~Chong and S.~Mori, ``Convex combination and covariance intersection
  algorithms in distributed fusion,'' in \emph{Proc. Int. Conf. Inform.
  Fusion}, Montr\'eal, Canada, 2001.

\bibitem{Bar-Shalom95}
Y.~Bar-Shalom and X.~Li, \emph{Multitarget-Multisensor Tracking: Principles and
  Techniques}.\hskip 1em plus 0.5em minus 0.4em\relax Storrs, CT: YBS
  Publishing, 1995.

\bibitem{blackman1999design}
S.~Blackman and R.~Popoli, \emph{Design and Analysis of Modern Tracking
  Systems}.\hskip 1em plus 0.5em minus 0.4em\relax Artech House, 1999.

\bibitem{Mahler07}
R.~Mahler, \emph{Statistical Multisource-Multitarget Information Fusion}.\hskip
  1em plus 0.5em minus 0.4em\relax Norwood, MA: Artech House, 2007.

\bibitem{Bar-Shalom11}
Y.~Bar-Shalom, P.~K. Willett, and X.~Tian, \emph{Tracking and Data
  Fusion}.\hskip 1em plus 0.5em minus 0.4em\relax Storrs, CT: YBS Publishing,
  2011.

\bibitem{Challa2011fundamentals}
S.~Challa, M.~R. Morelande, D.~Mu\v{s}icki, and R.~J. Evans, \emph{Fundamentals
  of Object Tracking}.\hskip 1em plus 0.5em minus 0.4em\relax New York, NY,
  USA: Cambridge University Press, 2011.

\bibitem{mahler2014advances}
R.~P.~S. Mahler, \emph{{Advances in Statistical Multisource-Multitarget
  Information Fusion}}.\hskip 1em plus 0.5em minus 0.4em\relax Boston, MA, USA:
  Artech House, 2014.

\bibitem{Koch14}
W.~Koch, \emph{Tracking and Sensor Data Fusion}.\hskip 1em plus 0.5em minus
  0.4em\relax Berlin, Germany: Springer, 2016.

\bibitem{meyer2018message}
F.~Meyer, T.~Kropfreiter, J.~L. Williams, R.~Lau, F.~Hlawatsch, P.~Braca, and
  M.~Z. Win, ``Message passing algorithms for scalable multitarget tracking,''
  \emph{Proc. IEEE}, vol. 106, no.~2, pp. 221--259, 2018.

\bibitem{Daley03}
D.~J. Daley and D.~Vere-Jones, \emph{An Introduction to the Theory of Point
  Processes: Volume I: Elementary Theory and Methods}.\hskip 1em plus 0.5em
  minus 0.4em\relax New York, NY: Springer, 2003.

\bibitem{kaplan2008assignment}
L.~M. Kaplan, Y.~Bar-Shalom, and W.~D. Blair, ``Assignment costs for multiple
  sensor track-to-track association,'' \emph{IEEE Trans. Aerosp. Electron.
  Syst.}, vol.~44, no.~2, pp. 655--677, 2008.

\bibitem{maresca2014maritime}
S.~Maresca, P.~Braca, J.~Horstmann, and R.~Grasso, ``Maritime surveillance
  using multiple high-frequency surface-wave radars,'' \emph{IEEE Trans. Geosc.
  Remote Sens.}, vol.~52, no.~8, pp. 5056--5071, 2014.

\bibitem{uney2010monte}
M.~\"Uney, S.~Julier, D.~Clark, and B.~Risti\'c, ``{M}onte {C}arlo realisation
  of a distributed multi-object fusion algorithm,'' in \emph{Sensor Signal
  Processing for Defence (SSPD 2010)}, London, UK, Sep. 2010, pp. 1--5.

\bibitem{mahler2013toward}
R.~Mahler, ``Toward a theoretical foundation for distributed fusion,'' in
  \emph{Distributed Data Fusion for Network-Centric Operations}, D.~Hall, C.-Y.
  Chong, J.~Llinas, and M.~Liggins~II, Eds.\hskip 1em plus 0.5em minus
  0.4em\relax CRC Press, 2013, ch.~8.

\bibitem{li-fan2019second}
T.~Li, H.~Fan, J.~Garc\'ia, and J.~M. Corchado, ``Second-order statistics
  analysis and comparison between arithmetic and geometric average fusion:
  {A}pplication to multi-sensor target tracking,'' \emph{Information Fusion},
  vol.~51, pp. 233--243, 2019.

\bibitem{yu2016distributed}
J.~Y. Yu, M.~Coates, and M.~Rabbat, ``Distributed multi-sensor {CPHD} filter
  using pairwise gossiping,'' in \emph{2016 IEEE Int. Conf. Acoust., Speech,
  Signal Process. (ICASSP)}, 2016, pp. 3176--3180.

\bibitem{Li-T_19}
T.~Li, J.~M. Corchado, and S.~Sun, ``Partial consensus and conservative fusion
  of {G}aussian mixtures for distributed {PHD} fusion,'' \emph{IEEE Trans.
  Aerosp. Electron. Syst.}, vol.~55, no.~5, pp. 2150--2163, 2019.

\bibitem{gao2020multiobject}
L.~Gao, G.~Battistelli, and L.~Chisci, ``Multiobject fusion with minimum
  information loss,'' \emph{IEEE Signal Process. Lett.}, vol.~27, pp. 201--205,
  2020.

\bibitem{fantacci2015consensus}
C.~Fantacci, B.-N. Vo, B.-T. Vo, G.~Battistelli, and L.~Chisci, ``Consensus
  labeled random finite set filtering for distributed multi-object tracking,''
  \emph{ArXiv}, vol. abs/1501.01579, 2015.

\bibitem{li-yi2018robust}
S.~Li, W.~Yi, R.~Hoseinnezhad, G.~Battistelli, B.~Wang, and L.~Kong, ``Robust
  distributed fusion with labeled random finite sets,'' \emph{IEEE Trans.
  Signal Process.}, vol.~66, no.~2, pp. 278--293, 2018.

\bibitem{kropfreiter2020probabilistic}
T.~Kropfreiter and F.~Hlawatsch, ``A probabilistic label association algorithm
  for distributed labeled multi-{B}ernoulli filtering,'' in \emph{2020 IEEE
  23rd International Conference on Information Fusion (FUSION)}, Rustenburg,
  South Africa, 2020, pp. 1--8.

\bibitem{ghahramani2015probabilistic}
Z.~Ghahramani, ``Probabilistic machine learning and artificial intelligence,''
  \emph{Nature}, vol. 521, no. 7553, pp. 452--459, 2015.

\bibitem{gal2016dropout}
Y.~Gal and Z.~Ghahramani, ``Dropout as a {B}ayesian approximation:
  {R}epresenting model uncertainty in deep learning,'' in \emph{international
  conference on machine learning}.\hskip 1em plus 0.5em minus 0.4em\relax PMLR,
  2016, pp. 1050--1059.

\bibitem{krems2019bayesian}
R.~Krems, ``Bayesian machine learning for quantum molecular dynamics,''
  \emph{Physical Chemistry Chemical Physics}, vol.~21, no.~25, pp.
  13\,392--13\,410, 2019.

\bibitem{leibig2017leveraging}
C.~Leibig, V.~Allken, M.~S. Ayhan, P.~Berens, and S.~Wahl, ``Leveraging
  uncertainty information from deep neural networks for disease detection,''
  \emph{Scientific reports}, vol.~7, no.~1, pp. 1--14, 2017.

\bibitem{begoli2019need}
E.~Begoli, T.~Bhattacharya, and D.~Kusnezov, ``The need for uncertainty
  quantification in machine-assisted medical decision making,'' \emph{Nature
  Machine Intelligence}, vol.~1, no.~1, pp. 20--23, 2019.

\bibitem{kendall2015bayesian}
A.~Kendall, V.~Badrinarayanan, and R.~Cipolla, ``Bayesian {S}eg{N}et: {M}odel
  uncertainty in deep convolutional encoder-decoder architectures for scene
  understanding,'' \emph{arXiv preprint arXiv:1511.02680}, 2015.

\bibitem{ching2019constructing}
J.~Ching and K.-K. Phoon, ``Constructing site-specific multivariate probability
  distribution model using {B}ayesian machine learning,'' \emph{Journal of
  Engineering Mechanics}, vol. 145, no.~1, p. 04018126, 2019.

\bibitem{ovadia2019can}
Y.~Ovadia, E.~Fertig, J.~Ren, Z.~Nado, D.~Sculley, S.~Nowozin, J.~Dillon,
  B.~Lakshminarayanan, and J.~Snoek, ``Can you trust your model's uncertainty?
  evaluating predictive uncertainty under dataset shift,'' in \emph{Advances in
  Neural Information Processing Systems}, 2019, pp. 13\,991--14\,002.

\bibitem{abdar2021review}
M.~Abdar, F.~Pourpanah, S.~Hussain, D.~Rezazadegan, L.~Liu, M.~Ghavamzadeh,
  P.~Fieguth, X.~Cao, A.~Khosravi, U.~R. Acharya \emph{et~al.}, ``A review of
  uncertainty quantification in deep learning: Techniques, applications and
  challenges,'' \emph{Information Fusion}, 2021.

\bibitem{wang2016towards}
H.~Wang and D.~Yeung, ``Towards bayesian deep learning: A framework and some
  existing methods,'' \emph{IEEE Transactions on Knowledge and Data
  Engineering}, vol.~28, no.~12, pp. 3395--3408, 2016.

\bibitem{wilson2020bayesian}
A.~G. Wilson and P.~Izmailov, ``Bayesian deep learning and a probabilistic
  perspective of generalization,'' \emph{arXiv preprint arXiv:2002.08791},
  2020.

\bibitem{rasmussen2006gaussian}
C.~E. Rasmussen and C.~K. Williams, \emph{Gaussian Processes for Machine
  Learning}.\hskip 1em plus 0.5em minus 0.4em\relax MIT press Cambridge, MA,
  2006, vol.~2, no.~3.

\bibitem{salimbeni2017doubly}
H.~Salimbeni and M.~Deisenroth, ``Doubly stochastic variational inference for
  deep gaussian processes,'' in \emph{Advances in Neural Information Processing
  Systems}, 2017, pp. 4588--4599.

\bibitem{hastie2009friedman}
T.~Hastie, R.~Tibshirani, and J.~Friedman, \emph{The Elements of Statistical
  Learning; Data Mining, Inference and Prediction}.\hskip 1em plus 0.5em minus
  0.4em\relax New York, NY: Springer, 2009.

\bibitem{murphy2012machine}
K.~P. Murphy, \emph{Machine Learning: A Probabilistic Perspective}.\hskip 1em
  plus 0.5em minus 0.4em\relax Cambridge, MA: MIT Press, 2012.

\bibitem{sagi2018ensemble}
O.~Sagi and L.~Rokach, ``Ensemble learning: {A} survey,'' \emph{Wiley
  Interdiscip. Rev. Data Min. Knowl. Discov.}, vol.~8, no.~4, 2018.

\bibitem{opitz1999popular}
D.~Opitz and R.~Maclin, ``Popular ensemble methods: {A}n empirical study,''
  \emph{Journal of artificial intelligence research}, vol.~11, pp. 169--198,
  1999.

\bibitem{polikar2006ensemble}
R.~Polikar, ``Ensemble based systems in decision making,'' \emph{IEEE Circuits
  and systems magazine}, vol.~6, no.~3, pp. 21--45, 2006.

\bibitem{rothe2016comparison}
S.~Rothe and D.~S{\"o}ffker, ``Comparison of different information fusion
  methods using ensemble selection considering benchmark data,'' in \emph{2016
  19th International Conference on Information Fusion (FUSION)}.\hskip 1em plus
  0.5em minus 0.4em\relax IEEE, 2016, pp. 73--78.

\bibitem{savazzi2020federated}
S.~Savazzi, M.~Nicoli, and V.~Rampa, ``Federated learning with cooperating
  devices: {A} consensus approach for massive {IoT} networks,'' \emph{IEEE
  Internet Things J.}, vol.~7, no.~5, pp. 4641--4654, 2020.

\bibitem{mohri2019agnostic}
M.~Mohri, G.~Sivek, and A.~T. Suresh, ``Agnostic federated learning,'' in
  \emph{Proc. Int. Conf. Mach. Learn. (ICML 2019)}, Long Beach, CA, 2019.

\bibitem{vehtari2020expectation}
A.~Vehtari, A.~Gelman, T.~Sivula, P.~Jyl{\"a}nki, D.~Tran, S.~Sahai,
  P.~Blomstedt, J.~P. Cunningham, D.~Schiminovich, and C.~P. Robert,
  ``Expectation propagation as a way of life: A framework for {B}ayesian
  inference on partitioned data.'' \emph{J. Mach. Learn. Res.}, vol.~21, pp.
  17--1, 2020.

\bibitem{brockwell2016introduction}
P.~J. Brockwell and R.~A. Davis, \emph{Introduction to time series and
  forecasting}, 3rd~ed.\hskip 1em plus 0.5em minus 0.4em\relax Springer, 2016.

\bibitem{bermowitz1979automated}
R.~J. Bermowitz and E.~A. Zurndorfer, ``Automated guidance for predicting
  quantitative precipitation,'' \emph{Mon. Weather Rev.}, vol. 107, no.~2, pp.
  122--128, 1979.

\bibitem{tay2000density}
A.~S. Tay and K.~F. Wallis, ``Density forecasting: a survey,'' \emph{J.
  Forecast.}, vol.~19, no.~4, pp. 235--254, 2000.

\bibitem{freiberger1965formulation}
W.~Freiberger and U.~Grenander, ``On the formulation of statistical
  meteorology,'' \emph{Revue de l'Institut International de Statistique}, pp.
  59--86, 1965.

\bibitem{epstein1969stochastic}
E.~S. Epstein, ``Stochastic dynamic prediction,'' \emph{Tellus}, vol.~21,
  no.~6, pp. 739--759, 1969.

\bibitem{gneiting2013combining}
T.~Gneiting and R.~Ranjan, ``Combining predictive distributions,''
  \emph{Electron. J. Stat.}, vol.~7, pp. 1747--1782, 2013.

\bibitem{raftery2005using}
A.~E. Raftery, T.~Gneiting, F.~Balabdaoui, and M.~Polakowski, ``Using
  {B}ayesian model averaging to calibrate forecast ensembles,'' \emph{Mon.
  Weather Rev.}, vol. 133, no.~5, pp. 1155--1174, 2005.

\bibitem{glahn2009mos}
B.~Glahn, M.~Peroutka, J.~Wiedenfeld, J.~Wagner, G.~Zylstra, B.~Schuknecht, and
  B.~Jackson, ``{MOS} uncertainty estimates in an ensemble framework,''
  \emph{Mon. Weather Rev.}, vol. 137, no.~1, pp. 246--268, 2009.

\bibitem{bassetti2018bayesian}
F.~Bassetti, R.~Casarin, and F.~Ravazzolo, ``Bayesian nonparametric calibration
  and combination of predictive distributions,'' \emph{J. Am. Stat. Assoc.},
  vol. 113, no. 522, pp. 675--685, 2018.

\bibitem{clements2011combining}
M.~P. Clements and D.~I. Harvey, ``Combining probability forecasts,''
  \emph{Int. J. Forecast.}, vol.~27, no.~2, pp. 208--223, 2011.

\bibitem{zarnowitz1969new}
V.~Zarnowitz, ``The new {ASA-NBER} survey of forecasts by economic
  statisticians,'' \emph{Am. Stat.}, vol.~23, no.~1, pp. 12--16, 1969.

\bibitem{tastu2013probabilistic}
J.~Tastu, P.~Pinson, P.-J. Trombe, and H.~Madsen, ``Probabilistic forecasts of
  wind power generation accounting for geographically dispersed information,''
  \emph{IEEE Trans. Smart Grid}, vol.~5, no.~1, pp. 480--489, 2013.

\bibitem{baran2016mixture}
S.~Baran and S.~Lerch, ``Mixture {EMOS} model for calibrating ensemble
  forecasts of wind speed,'' \emph{Environmetrics}, vol.~27, no.~2, pp.
  116--130, 2016.

\bibitem{ho2016probabilistic}
T.~Hong and S.~Fan, ``Probabilistic electric load forecasting: A tutorial
  review,'' \emph{Int. J. Forecast.}, vol.~32, no.~3, pp. 914--938, 2016.

\bibitem{nowotarski2018recent}
J.~Nowotarski and R.~Weron, ``Recent advances in electricity price forecasting:
  A review of probabilistic forecasting,'' \emph{Renew. Sustain. Energy Rev.},
  vol.~81, pp. 1548--1568, 2018.

\bibitem{doubleday2020probabilistic}
K.~Doubleday, S.~Jascourt, W.~Kleiber, and B.-M. Hodge, ``Probabilistic solar
  power forecasting using {B}ayesian model averaging,'' \emph{IEEE Transactions
  on Sustainable Energy}, vol.~12, no.~1, pp. 325--337, 2020.

\bibitem{genest1984pooling}
C.~Genest, ``Pooling operators with the marginalization property,'' \emph{Can.
  J. Stat.}, vol.~12, no.~2, pp. 153--163, 1984.

\bibitem{genest1986characterization}
C.~Genest, K.~J. McConway, and M.~J. Schervish, ``Characterization of
  externally {B}ayesian pooling operators,'' \emph{Ann. Stat.}, vol.~14, no.~2,
  pp. 487--501, 1986.

\bibitem{bullen2013handbook}
P.~S. Bullen, \emph{Handbook of Means and Their Inequalities}.\hskip 1em plus
  0.5em minus 0.4em\relax Dordrecht, The Netherlands: Springer, 2013.

\bibitem{mcconway1981marginalization}
K.~J. McConway, ``Marginalization and linear opinion pools,'' \emph{J. Am.
  Stat. Assoc.}, vol.~76, no. 374, pp. 410--414, 1981.

\bibitem{laddaga1977lehrer}
R.~Laddaga, ``Lehrer and the consensus proposal,'' \emph{Synthese}, vol.~36,
  no.~4, pp. 473--477, 1977.

\bibitem{madansky1964externally}
A.~Madansky, ``Externally {B}ayesian groups,'' RAND Corporation, Santa Monica,
  CA, Tech. Rep., 1964.

\bibitem{rufo2012log}
M.~J. Rufo, J.~Martin, and C.~J. P{\'e}rez, ``Log-linear pool to combine prior
  distributions: {A} suggestion for a calibration-based approach,''
  \emph{Bayesian Anal.}, vol.~7, no.~2, pp. 411--438, 2012.

\bibitem{genest1984conflict}
C.~Genest, ``A conflict between two axioms for combining subjective
  distributions,'' \emph{J. R. Stat. Soc. Ser. B Methodol.}, vol.~46, no.~3,
  pp. 403--405, 1984.

\bibitem{csiszar1964informationstheoretische}
I.~Csisz{\'a}r, ``Eine informationstheoretische {U}ngleichung und ihre
  {A}nwendung auf den {B}eweis der {E}rgodizit\"at von {M}arkoffschen
  {K}etten,'' \emph{Magyer Tud. Akad. Mat. Kutat\'o Int. K\"ozl.}, vol.~8, pp.
  85--108, 1963.

\bibitem{csiszar1967information}
------, ``Information-type measures of difference of probability distributions
  and indirect observation,'' \emph{Stud. Sci. Math. Hung.}, vol.~2, pp.
  299--318, 1967.

\bibitem{vajda1972f}
I.~Vajda, ``On the $f$-divergence and singularity of probability measures,''
  \emph{Period. Math. Hung.}, vol.~2, no. 1--4, pp. 223--234, 1972.

\bibitem{Sason_2018}
I.~Sason, ``On $f$-divergences: Integral representations, local behavior, and
  inequalities,'' \emph{Entropy}, vol.~20, no.~5, 2018.

\bibitem{veeravalli1994minimax}
V.~Veeravalli, T.~Basar, and H.~Poor, ``Minimax robust decentralized
  detection,'' \emph{IEEE Trans. Inform. Theory}, vol.~40, no.~1, pp. 35--40,
  1994.

\bibitem{fauss2021minimax}
M.~Fau{\ss}, A.~M. Zoubir, and H.~V. Poor, ``Minimax robust detection:
  {C}lassic results and recent advances,'' \emph{IEEE Trans. Signal Process.},
  vol.~69, pp. 2252--2283, 2021.

\bibitem{liese2006divergences}
F.~Liese and I.~Vajda, ``On divergences and informations in statistics and
  information theory,'' \emph{IEEE Trans. Inf. Theory}, vol.~52, no.~10, pp.
  4394--4412, 2006.

\bibitem{van2014renyi}
T.~Van~Erven and P.~Harremos, ``R{\'e}nyi divergence and {K}ullback-{L}eibler
  divergence,'' \emph{IEEE Trans. Inf. Theory}, vol.~60, no.~7, pp. 3797--3820,
  2014.

\bibitem{dedecius2016sequential}
K.~Dedecius and P.~M. Djuri{\'c}, ``Sequential estimation and diffusion of
  information over networks: A {B}ayesian approach with exponential family of
  distributions,'' \emph{IEEE Trans. Signal Process.}, vol.~65, no.~7, pp.
  1795--1809, 2016.

\bibitem{zhu1995information}
H.~Zhu and R.~Rohwer, ``Information geometric measurements of generalisation,''
  Aston University, Birmingham, UK, Tech. Rep., 1995.

\bibitem{minka2005divergence}
T.~Minka, ``Divergence measures and message passing,'' Microsoft Research,
  Cambridge, UK, Tech. Rep., 2005.

\bibitem{cichocki2011generalized}
A.~Cichocki, S.~Cruces, and S.~Amari, ``Generalized alpha-beta divergences and
  their application to robust nonnegative matrix factorization,''
  \emph{Entropy}, vol.~13, no.~1, pp. 134--170, 2011.

\bibitem{pearson1900x}
K.~Pearson, ``X. {O}n the criterion that a given system of deviations from the
  probable in the case of a correlated system of variables is such that it can
  be reasonably supposed to have arisen from random sampling,'' \emph{Lond.
  Edinb. Dubl. Phil. Mag.}, vol.~50, no. 302, pp. 157--175, 1900.

\bibitem{malladi1997new}
D.~P. Malladi and J.~L. Speyer, ``A new approach to multiple model adaptive
  estimation,'' in \emph{Proc. IEEE Conf. Decis. Control (CDC 1997)}, San
  Diego, CA, 1997.

\bibitem{wang2015multivariate}
J.~Wang and M.~R. Taaffe, ``Multivariate mixtures of normal distributions:
  Properties, random vector generation, fitting, and as models of market daily
  changes,'' \emph{INFORMS J. Comput.}, vol.~27, no.~2, pp. 193--203, 2015.

\bibitem{hanlon2000multiple}
P.~D. Hanlon and P.~S. Maybeck, ``Multiple-model adaptive estimation using a
  residual correlation {K}alman filter bank,'' \emph{IEEE Trans. Aerosp.
  Electron. Syst.}, vol.~36, no.~2, pp. 393--406, 2000.

\bibitem{robert2013monte}
C.~Robert and G.~Casella, \emph{{M}onte Carlo Statistical Methods}.\hskip 1em
  plus 0.5em minus 0.4em\relax New York, NY: Springer, 2013.

\bibitem{hinrichs2019curse}
A.~Hinrichs, J.~Prochno, and M.~Ullrich, ``The curse of dimensionality for
  numerical integration on general domains,'' \emph{J. Complex.}, vol.~50, pp.
  25--42, 2019.

\bibitem{genest1990allocating}
C.~Genest and K.~J. McConway, ``Allocating the weights in the linear opinion
  pool,'' \emph{J. Forecast.}, vol.~9, no.~1, pp. 53--73, 1990.

\bibitem{degroot1991optimal}
M.~H. DeGroot and J.~Mortera, ``Optimal linear opinion pools,'' \emph{Manag.
  Sci.}, vol.~37, no.~5, pp. 546--558, 1991.

\bibitem{clemen2008comment}
R.~T. Clemen, ``Comment on {C}ooke's classical method,'' \emph{Reliab. Eng.
  Syst. Saf.}, vol.~93, no.~5, pp. 760--765, 2008.

\bibitem{heskes1998selecting}
T.~Heskes, ``Selecting weighting factors in logarithmic opinion pools,'' in
  \emph{Proc. Adv. Neural Inf. Process. Syst. (NIPS 1997)}, Denver, CO, 1997.

\bibitem{de2015choosing}
L.~M. De~Carvalho, D.~A.~M. Villela, F.~C. Coelho, and L.~S. Bastos,
  ``Combining probability distributions: Extending the logarithmic pooling
  approach,'' \emph{arXiv preprint arXiv:1502.04206}, 2020.

\bibitem{bunn1981two}
D.~W. Bunn, ``Two methodologies for the linear combination of forecasts,''
  \emph{J. Oper. Res. Soc.}, vol.~32, no.~3, pp. 213--222, 1981.

\bibitem{bunn1982synthesis}
D.~W. Bunn and E.~Kappos, ``Synthesis or selection of forecasting models,''
  \emph{Eur. J. Oper. Res.}, vol.~9, no.~2, pp. 173--180, 1982.

\bibitem{barlow1986combination}
R.~E. Barlow, R.~W. Mensing, and N.~G. Smiriga, ``Combination of experts'
  opinions based on decision theory,'' in \emph{Reliability and Quality
  Control}, A.~P. Basu, Ed.\hskip 1em plus 0.5em minus 0.4em\relax Amsterdam,
  The Netherlands: Elsevier, 1986, pp. 9--19.

\bibitem{degroot1974reaching}
M.~H. DeGroot, ``Reaching a consensus,'' \emph{J. Am. Stat. Assoc.}, vol.~69,
  no. 345, pp. 118--121, 1974.

\bibitem{carvalho2013consensual}
A.~Carvalho and K.~Larson, ``A consensual linear opinion pool,'' in \emph{Proc.
  Int. Jt. Conf. Artif. Intell. (IJCAI 2013)}, Beijing, China, 2013.

\bibitem{rufo2012bayesian}
M.~J. Rufo, C.~J. P{\'e}rez, and J.~Mart{\'\i}n, ``A {B}ayesian approach to
  aggregate experts' initial information,'' \emph{Electron. J. Stat.}, vol.~6,
  pp. 2362--2382, 2012.

\bibitem{casella2002statistical}
G.~Casella and R.~L. Berger, \emph{Statistical Inference}, 2nd~ed.\hskip 1em
  plus 0.5em minus 0.4em\relax Pacific Grove, CA: Duxbury, 2002.

\bibitem{brown1986fundamentals}
L.~D. Brown, \emph{Fundamentals of Statistical Exponential Families: With
  Applications in Statistical Decision Theory}.\hskip 1em plus 0.5em minus
  0.4em\relax Hayward, CA: Institute of Mathematical Statistics, 1986, vol.~9.

\bibitem{evangar92measure}
L.~C. Evans and R.~F. Gariepy, \emph{Measure Theory and Fine Properties of
  Functions}.\hskip 1em plus 0.5em minus 0.4em\relax Boca Raton, FL: CRC Press,
  1992.

\bibitem{qiu2016survey}
J.~Qiu, Q.~Wu, G.~Ding, Y.~Xu, and S.~Feng, ``A survey of machine learning for
  big data processing,'' \emph{EURASIP J. Adv. Signal Process.}, vol. 2016,
  no.~67, 2016.

\bibitem{kim2017wasserstein}
Y.-H. Kim and B.~Pass, ``Wasserstein barycenters over {R}iemannian manifolds,''
  \emph{Adv. Math.}, vol. 307, pp. 640--683, 2017.

\bibitem{ghosal2017fundamentals}
S.~Ghosal and A.~Van~der Vaart, \emph{Fundamentals of nonparametric {B}ayesian
  inference}.\hskip 1em plus 0.5em minus 0.4em\relax Cambridge, UK: Cambridge
  University Press, 2017, vol.~44.

\bibitem{aczel1966lectures}
J.~Acz{\'e}l, \emph{Lectures on Functional Equations and Their
  Applications}.\hskip 1em plus 0.5em minus 0.4em\relax New York, NY: Academic
  Press, 1966.

\bibitem{bartle1995elements}
R.~G. Bartle, \emph{The Elements of Integration and Lebesgue Measure}.\hskip
  1em plus 0.5em minus 0.4em\relax New York, NY: Wiley, 1995.

\bibitem{rudin1986}
W.~Rudin, \emph{Real and Complex Analysis}, 3rd~ed.\hskip 1em plus 0.5em minus
  0.4em\relax New York, NY: McGraw-Hill, 1986.

\bibitem{hojo13}
R.~A. Horn and C.~R. Johnson, \emph{{M}atrix {A}nalysis}, 2nd~ed.\hskip 1em
  plus 0.5em minus 0.4em\relax Cambridge, UK: Cambridge Univ. Press, 2013.

\end{thebibliography}

\end{document}